\documentclass[12pt,english]{article}
\usepackage[T1]{fontenc}
\usepackage[utf8]{inputenc}
\usepackage{geometry}
\usepackage{float}
\usepackage{amsmath}
\usepackage{amsthm}
\usepackage{amssymb}
\usepackage{graphicx}
\usepackage{xcolor}
\usepackage{setspace}
\usepackage{dsfont}
\usepackage{caption}
\usepackage{hyperref}
\usepackage[authoryear,round]{natbib}
\usepackage[multiple]{footmisc}
\allowdisplaybreaks
\usepackage{threeparttable}

\usepackage{xr, refcount}

\makeatletter
\providecommand{\tabularnewline}{\\}

\theoremstyle{plain}
\newtheorem{assumption}{Assumption}

\newtheorem{condition}{Condition}

\newtheorem{theorem}{Theorem}

\newtheorem{lem}{Lemma}

\AtEndDocument{\refstepcounter{lem}\label{finallem}}

\theoremstyle{remark}

\newtheorem{remark}{Remark}

\DeclareMathOperator*{\argmax}{arg\,max}

\begin{document}


\title{\Large \textsc{A jackknife bias correction for nonlinear network data models with fixed effects}\thanks{David W. Hughes: dw.hughes@bc.edu. I am
extremely grateful for the support and advice of Whitney Newey, Anna
Mikusheva, and Alberto Abadie. I would also like to thank the editor, anonymous referees,  Isaiah Andrews,
Victor Chernozhukov, Ivan Fernandez-Val, Claire Lazar Reich, Ben Deaner, Sylvia Klosin, and 
many seminar participants, for helpful feedback and suggestions.}}
\author{David W. Hughes\\
\small Boston College}

\date{}
\maketitle

\begin{abstract}
I introduce a new method for bias correction of dyadic models with agent-specific
fixed effects, including the dyadic link formation model with homophily and degree heterogeneity. 
The proposed approach uses a jackknife procedure to deal with the incidental parameters problem. 
The method can be applied to both directed and undirected networks, allows for non-binary
outcome variables, and can be used to bias correct estimates of average effects and counterfactual outcomes. 
I also show how the jackknife can be used to bias correct fixed-effect averages over functions that depend on multiple nodes, e.g. triads
or tetrads in the network. As an example, I implement specification
tests for dependence across dyads, such as reciprocity or transitivity.
Finally, I demonstrate the usefulness of the estimator in an application
to a gravity model for import/export relationships across countries. \\


\end{abstract}



\section{Introduction}

Networks are common in both economic and social contexts, and it is
important to understand the factors that play a role in both the formation
and strength of links between agents. The econometric analysis
of networks faces a number of challenges that have received much attention
in recent literature (see \citet{dePaula2020} and \citet{Graham2020a}
for reviews of this literature). One common modeling approach is to
assume a dyadic network structure, in which decisions are made
bilaterally between agents, but allow for linking decisions to depend
on unobserved agent-specific heterogeneity. These models are common
in practice since they are straightforward to implement while still
being able to capture important aspects of observed networks. Controlling
for agent-specific heterogeneity is important since in many real world
networks agents vary significantly in the number and strength of connections
made. Ignoring this heterogeneity can lead to large biases in estimated
effects.

In this paper, I consider the estimation of dyadic models in which
the presence of unobserved heterogeneity is accounted for by two sets
of agent-specific fixed effects -- a sender and a receiver effect. Allowing for heterogeneity using fixed effects is appealing as it does not require strong
assumptions about the unobserved component as in random effects models. The large number of fixed effect parameters (proportional
to the square root of the sample size) creates an incidental parameters
problem \citep{Neyman1948}. This paper proposes a jackknife approach
to bias correction that is based on combining the full-sample estimate with an average over various subsample estimators. I demonstrate the consistency and asymptotic normality of the jackknife estimator under asymptotic sequences in which a single network grows in size. To do this, I extend the asymptotic expansions for two-way fixed effects models derived by \citet{Fernandez-Val2016} to higher-order. In the binary outcome setting, the assumptions imply that the network is `dense', in the sense that the degree of each agent remains proportional to the total number of agents. 

The jackknife approach proposed here has a number of advantages over existing methods of bias correction. Importantly, the jackknife is easily accessible to applied researchers, since it relies only on the ability to estimate the parameter of interest and does not require any special estimation algorithms or the derivation or estimation of any additional quantities. I derive results for a general set of M-estimators, so that the method is easily adaptable to a range of settings, including models for binary, count or continuous outcome variables, as well as maximum likelihood or nonlinear least-squares estimators. 
Additionally, in comparison to conditional likelihood methods, which avoid estimation of the fixed effect parameters, the jackknife bias correction estimates all parameters, allowing for the construction of average effects, and other objects that depend on the fixed effects.
I also extend the jackknife method to the bias correction of averages over functions of multiple observations
(e.g.\ triads or tetrads in the network). These types of averages are useful for estimating network measures, such as clustering coefficients, as well as for
constructing specification test statistics, such as tests for the
presence of certain strategic interactions like reciprocity or transitivity.

As a demonstration of the jackknife in an empirical setting, I estimate
a model for the formation of country-level trade relationships, producing bias-corrected estimates of both model parameters and marginal effects. Gravity models have
been a workhorse model in the trade literature for many years, and
the importance of including country-specific fixed effects is well
known \citep{Anderson2003}. 
I also conduct tests for the presence of excess transitivity in both a country-level trade network, and a network 
of professional friendships. 
The impact of bias correction on the transitivity test statistic is large.
While the uncorrected statistic would imply strong evidence of excess transitivity in the trade network, the bias-corrected statistic is close to zero (for the professional network I find evidence of strong transitivity in line with common theories of social network formation).

There are several alternative approaches to address the incidental
parameters problem in the setting in which the outcome is a binary variable that denotes the presence of a link between two agents. \citet{Graham2017}, \citet{Charbonneau2017},
and \citet{Jochmans2018} all consider versions of this model in which
the latent disturbances follow a logistic distribution, and use conditioning
arguments to remove dependence on the fixed effects. The conditioning
approach has the advantage of being applicable under certain \emph{sparse
network} asymptotic sequences, where the degree of agents remains bounded asymptotically, but is limited to models in which sufficient
statistics for the fixed effects exist, and is not able to recover
counterfactuals or average effects. \citet{Yan2019} also studies
the logistic model and provides asymptotic results for the incidental
parameters. \citet{Graham2017} also considers an analytical correction
for the logistic model, while \citet{Dzemski2019} derives the analytical
correction for a probit model. Similarly to the jackknife, the analytical bias correction approach
is limited to \emph{dense network} sequences, but is 
able to recover average effects. The advantage
of the jackknife correction relative to an analytical approach is
that it provides an off-the-shelf procedure that researchers may apply
to new settings, without the need to first derive bias expressions.
\citet{Candelaria2020}, \citet{Toth2017}, and \citet{Gao2020} study
identification of the common parameters without a known parametric
form for the disturbance term, while \citet{Zeleneev2020} allows
for nonparametric structure in the unobserved heterogeneity term.

Several papers also consider tests of the dyadic network formation model. \citet{Dzemski2019} 
provides a test for the presence of excess transitivity in a probit model, and derives an analytical bias correction
for the test statistic. \citet{Graham2020d} derive the locally best similar test for a class
of alternatives in a logit model, using conditioning arguments. Jackknife bias-correction of such network statistics has not 
been previously considered in the literature, and I demonstrate that a range of such statistics, 
including that of \citet{Dzemski2019}, can be bias corrected using the jackknife, extending the set of tests available 
to researchers, as well as the range of models they can be applied to. 

The network jackknife extends previous results on jackknife bias correction
in panel data. \citet{Hahn2004} introduced a jackknife correction
for panel estimators with individual fixed effects; \citet{Dhaene2015}
present a split-sample version of this idea; and, \citet{Fernandez-Val2016}
develop a general framework that allows for both time and individual
fixed effects. Analogously to the panel data setting, the network jackknife is constructed
by forming `leave-out' estimates that exclude certain subsets of the
data. \citet{Fernandez-Val2016} derive the asymptotic distribution of a split-sample jackknife bias correction
that involves computing parameter estimates in various half-samples (i.e.\ dropping half of all senders or receivers). 
Although such split-sample corrections have the same asymptotic distribution as other bias-corrected estimators, 
in finite samples they can remove less bias and have much larger variance, particularly when the network contains many agents with few links, as is shown in simulations (see Section 5).\footnote{\citet{Hahn2024} show that the higher-order bias and variance of the split-sample bias-corrected estimator is larger 
than that of the analytical or leave-one-out jackknife in a panel data setting.} \citet{Cruz-Gonzalez2017} suggest
a jackknife approach for network data based on a leave-one-out approach
that drops a single agent at a time, although the formal properties have not been derived. I propose a different approach
to jackknifing network data that is based on a novel partitioning
of the data set and formally establish consistency and asymptotic normality. This is done by extending the asymptotic expansions of \citet{Fernandez-Val2016} to higher order. 
In addition, I introduce a weighted version of the jackknife estimator, that down-weights noisy leave-out 
estimates and improves the finite-sample performance of the estimator. Simulations suggest that the  jackknife proposed here (particularly the weighted version) is more robust
to settings with meaningful levels of unobserved heterogeneity and in networks with lower density.
The method is likely to be relevant in other settings, such as panel data models with so-called `rare outcomes', in which some leave-out samples may lead to substantially noisier estimates than others. 

The rest of the paper is organized as follows. Section 2 introduces
the network model, Section 3 discusses implementation of the jackknife procedure
for the estimation of model parameters, while Section 4 discusses
estimation of average effects, and the construction of specification
tests. In Section 5 I report simulation results that investigate the finite-sample properties of the bias correction, while in Section 6 I implement the method using a country-level trade network.\footnote{The appendix in this paper contains proofs of results based on the first-order asymptotic expansions of the estimators. Further details, including additional simulation results and the full higher-order expansions are contained in a supplementary appendix that is available at \url{https://arxiv.org/abs/2203.15603}.}

\section{The dyadic linking model}

The researcher observes a network of $N$ agents, which may be individuals, firms, countries, and so on. 
For each of the $N(N-1)/2$ unordered pairs of agents $\{i,j\}$, called a \emph{dyad}, we observe the outcome of two
linking decisions $Y_{ij}$ and $Y_{ji}$. The network may be directed, in which case $Y_{ij}\ne Y_{ji}$
in general, or undirected, in which case the two outcomes are equal. Following the literature, I term $i$ the `sender' and $j$ the
`receiver' in link $Y_{ij}$.
The outcome variable $Y_{ij}$ may be binary, capturing the
presence (or absence) of a link between two agents, or may represent a measure of the strength of
the link between agents, in which case $Y_{ij}$ is non-binary (e.g.\ the value of
exports from country $i$ to country $j$ in a particular year, or
the number of times agents $i$ and $j$ interacted in some period).

The researcher additionally observes a set of link-specific covariates $X_{ij}$, which
capture determinants of linking decisions such as \emph{homophily}, the tendency for agents
to link with other agents that are similar to themselves. 
Agents are endowed with two fixed effects, a `sender' effect $\alpha_{i}$ and a `receiver' effect $\gamma_{i}$,
which allow for \emph{degree heterogeneity}.
The fixed effects will be treated as unknown parameters in the model, 
and the researcher does not need to specify their distribution or how they relate to the covariates.

The data $(Y_{ij},X_{ij})$ are assumed to be drawn from some model, which is left unspecified except 
that the model parameters $(\beta_{0},\alpha_{0},\gamma_{0})$ are known to be solutions to the population maximization problem
\begin{align}
&\max_{(\beta,\alpha,\gamma)\in\mathbb{R}^{\dim\beta+2N}}\bar{E}\big[\mathcal{L}(\beta,\alpha,\gamma)\big], \nonumber\\
\mathcal{L}(\beta,\alpha,\gamma)=\frac{1}{N-1}&\sum_{i}\sum_{j\ne i}\ell(Y_{ij},X_{ij},\beta,\alpha_{i}+\gamma_{j})-\frac{b}{2N}\big(\sum_{i}\alpha_{i}-\sum_{i}\gamma_{i}\big)^{2},\label{eq:objective}
\end{align}
where $\bar{E}$ represents expectation conditional on covariates and fixed effects, and $\ell$ is a known function that
is maximized in expectation at the true parameters. The objective function in (\ref{eq:objective}) covers many 
common estimators such as MLE/ quasi-MLE, in which $\ell$ is a log-likelihood, and nonlinear least squares estimators, in which $\ell$ captures squared deviations of outcomes from their conditional expectations. In particular, this includes the commonly used directed link formation model, in which
\begin{align}
Y_{ij} & =\mathds{1}\{X_{ij}'\beta+\alpha_{i}+\gamma_{j}-\varepsilon_{ij}\geq0\}, \label{eq:formation}
\end{align}
where $\varepsilon_{ij}$ follows a known distribution $F$. This is an extension of the linking model of \citet{Holland1981}
and has been used extensively in empirical literature.

The unobserved effects are assumed to enter in an additively separable
manner, i.e.\ as $\alpha_{i}+\gamma_{j}$. This assumption is common in the literature, and is
consistent with a model in which utility is transferable between directly connected agents \citep{Bloch2007}.
The additive separability does rule out some other kinds of interactions, notably homophily based on unobservables.\footnote{\citet{Zeleneev2020} studies a model that allows for general forms of interactions in the individual-specific effects. \citet{Chen2021} derive expansions for models with interactive fixed effects. The expansions derived in that paper share a similar structure to the expansions in \citet{Fernandez-Val2016}, which form the basis for the proofs in this paper, and so it is hypothesized that the jackknife approach may be valid in that setting also.}
Identification of the two
sets of fixed effects parameters requires a normalization, for which
we choose $\sum_{i}\alpha_{i}=\sum_{i}\gamma_{i}$.
The term $\frac{b}{2N}\big(\sum_{i}\alpha_{i}-\sum_{i}\gamma_{i}\big)^{2}$
is a penalty term intended to impose this normalization on the fixed
effect parameters, where $b>0$ is an arbitrary constant.\footnote{In practice the constraint could simply be eliminated by substituting it into the objective. I follow \citet{Fernandez-Val2016} in choosing
this normalization as it is convenient in proofs.} 

Throughout, we let $\ell_{ij}(\beta,\pi_{ij})=\ell(Y_{ij},X_{ij},\beta,\pi_{ij})$, where $\pi_{ij} = \alpha_{i}+\gamma_{j}$ is the additive index for the fixed effects. Denote by $X$ the vector of link-specific covariates $X_{ij}$, and $\phi'=(\alpha',\gamma')$ the $2N\times 1$ vector of fixed effect parameters. We denote derivatives of the function $\ell$ with respect
to parameters by $\partial_{\beta}\ell_{ij}(\beta,\pi_{ij})=\partial\ell_{ij}(\beta,\pi_{ij})/\partial\beta$,
$\partial_{\pi^{q}}\ell_{ij}(\beta,\pi_{ij})=\partial^{q}\ell_{ij}(\beta,\pi_{ij})/\partial\pi_{ij}^{q}$
etc. When evaluating these objects at the true parameter values, we
simply write $\partial_{\pi^{q}}\ell_{ij}$ and so on. The corresponding derivatives of the objective function are 
denoted by $\partial_\beta\mathcal{L}=\frac{1}{N-1}\sum_i\sum_{j\ne i} \partial_{\beta}\ell_{ij}$, 
$\partial_{\phi\phi'}\mathcal{L}=\frac{1}{N-1}\sum_i\sum_{j\ne i} \partial_{\phi\phi'}\ell_{ij}$ etc. Finally, we denote 
expectations of objects (conditional on the fixed effects) by $\partial_\beta\bar{\mathcal{L}}=\bar{E}[\partial_\beta\mathcal{L}]$, and deviations from 
these expectations as $\partial_\beta\tilde{\mathcal{L}}=\partial_\beta\mathcal{L}-\bar{E}[\partial_\beta\mathcal{L}]$.

The asymptotic framework is one in which a single network of $N$
agents grows in size, i.e. $N\to\infty$. We will impose two key assumptions on the structure of the network.\footnote{Additional regularity conditions and notation are discussed in Section \ref{sec:beta_asymptotics}}

\begin{assumption}
\label{assu:network}
(i) For all dyads $\{i,j\}$, it is assumed that outcomes are drawn independently from some distribution, conditional on covariates and fixed effects
\begin{align} \label{eq:cond_ind}
	(Y_{ij}, Y_{ji}) \vert X,\phi \sim f(Y_{ij}, Y_{ji} \vert X,\phi)
\end{align}

(ii) Let $\mathcal{B}_{\varepsilon}$ be an $\varepsilon$-neighborhood that contains $(\beta_0,\alpha_{0i},\gamma_{0j})$ for all $i,j$ and $N$.\footnote{As stated in \citet{Fernandez-Val2016}, $\mathcal{B}_{\varepsilon}$ could be chosen as the Cartesian product
	of the $\epsilon$-ball around $\beta_{0}$ with the interval $[\pi_{\min},\pi_{\max}]$,
	where $\pi_{\min}\leq\pi_{ij}-\varepsilon$ and $\pi_{\max}\geq\pi_{ij}+\varepsilon$
	for all $(i,j)$.} 
For some $\varepsilon>0$, there exist constants $b_{min}$ and $b_{max}$ such that for
all $(\beta,\pi_{ij})\in\mathcal{B}_{\varepsilon,ij}$
\[
0<b_{min}\leq-\bar{E}[\partial_{\pi^{2}}\ell_{ij}(\beta,\pi_{ij})]\leq b_{max}
\]
 a.s. uniformly over $i,j$ and $N$.
\end{assumption}

Assumption \ref{assu:network} (i) implies that linking decisions are bilateral in nature. 
This allows for dependence between the two links within a pair of agents (e.g.\ allowing for reciprocity in links),
but implies the linking decision between $i$ and $j$ is independent of that
between $i$ and $k$ for instance. Importantly, this independence
is conditional on the covariates and agent-specific fixed effects, so that $Y_{ij}$ 
and $Y_{ik}$ are still unconditionally correlated through the covariates and shared 
sender effect $\alpha_{i}$. Note that the distribution from which outcomes are drawn need not
be fully specified, so long as the model parameters are solutions to the maximization (\ref{eq:objective}). 
In many settings, the inclusion of
fixed effects will be important in establishing the plausibility of
conditional independence. Assumption \ref{assu:network} (i) may not be appropriate
in situations where linking decisions are strategic, but the dyadic model presented here 
nonetheless represents an important baseline model, and can be used to construct tests for
the presence of strategic interactions against the null hypothesis
of conditional independence. We discuss examples of such tests in Section
\ref{subsec:avg_effects}.

Assumption \ref{assu:network} (ii) is used to ensure that all parameters are point identified. This requires sufficient variation
in the outcomes across both dimensions i.e., variation in $Y_{ij}$
over $j$ (for fixed sender $i$) and over $i$ (for fixed receiver
$j$). 
In the binary outcome model in (\ref{eq:formation}), we have
\[
	-\bar{E}[\partial_{\pi^{2}}\ell_{ij}(\beta,\pi_{ij})]=\frac{f(X_{ij}'\beta + \pi_{ij})^{2}}{F(X_{ij}'\beta + \pi_{ij})\big(1-F(X_{ij}'\beta + \pi_{ij})\big)},
\]
where $f(u) = dF(u)/du$, so that Assumption \ref{assu:network} (ii) requires that the probability of any link forming is bounded away from 0 and 1, and that the density $f$ is bounded away from zero over the support of $X_{ij}'\beta + \alpha_i + \gamma_j$. This implies that the model generates
a \emph{dense network}, one in which the number of links formed by
each node is proportional to $N$. The
assumption may not be reasonable in all empirical settings; however, 
it seems a natural condition when we hope to estimate 
parameters that depend on the unobserved heterogeneity, such as marginal 
effects or counterfactual outcomes.\footnote{Section \ref{sec:Simulations} reports  
simulations that investigate the robustness of the estimator to sparsity
in finite samples.} For other models the requirements may be weaker, for example in the nonlinear least-squares estimator with $\bar{E}[Y_{ij}] = G(X_{ij}'\beta + \alpha_i + \gamma_j)$, Assumption \ref{assu:network} (ii) simply requires that $G'(X_{ij}'\beta + \alpha_ i+ \gamma_j)$ is non-zero, and has bounded second moment for all $(i,j)$.


\section{Jackknife bias correction}

\subsection{Jackknife for dyadic data}

As is well known, nonlinear estimators with fixed effects suffer from an incidental
parameters problem \citep{Neyman1948}. In total, the model contains $\dim(\beta)+2N$ parameters to be estimated,
from the $N(N-1)$ observations $(Y_{ij},X_{ij}')$; we refer to these ordered pairs of agents using the shorthand $(i,j)$. The incidental parameters generate an asymptotic bias in the network model analogous to the panel setting with both $N$ and $T$
growing to infinity at the same rate. Similarly to the panel case, it can be shown that the maximizer 
of the sample objective (\ref{eq:objective}) has first order bias of the form 
\begin{align}\label{eq:bias_N}
\bar{E}[\widehat{\beta}]-\beta\approx\frac{B}{N-1},
\end{align}
so that the asymptotic distribution of the estimator is not centered around zero, i.e.\ 
$N(\widehat{\beta}-\beta_{0}) \overset{d}{\longrightarrow}\mathcal{N}\big(B,V\big)$.

Motivated by the form of this bias, a jackknife bias-corrected estimator can be constructed as a linear combination of the full-sample parameter estimate and an average of `leave-out' estimators that exclude certain observations in the
data set.
To give some intuition for the jackknife correction, note that the bias in (\ref{eq:bias_N}) 
is inversely proportional to the number of elements in the summation (\ref{eq:objective}) that depend 
on a particular fixed effect, e.g.\ there are $(N-1)$ observations that are functions of any particular $\alpha_i$.\footnote{This is analogous to the $1/T$ bias in the panel data model with individual effects.} 
Suppose we drop observations from our data set in such a way
that for every agent $i$ we exclude one observation in which $i$
is the sender, and one observation in which $i$ is the receiver
(recall that a single observation is a directed $(i,j)$ pair,
of which we observe $N(N-1)$ in total). The estimator that uses 
this `leave-out' sample, $\tilde{\beta}$, has a first-order bias given by
\[
\bar{E}[\tilde{\beta}]-\beta\approx\frac{B}{N-2},
\]
since $\tilde{\beta}$ is estimated using only $N-2$
observations per fixed effect. Taking advantage of the fact that the estimate $\tilde{\beta}$ has
a larger bias than the full-sample estimate $\widehat{\beta}$ by
the factor $\frac{N-1}{N-2}$, we can construct a new estimator $\widehat{\beta}_{jack}=(N-1)\widehat{\beta}-(N-2)\tilde{\beta}$
which has no asymptotic bias.

To describe the construction of these leave-out estimators, we first
define a partition of the $N(N-1)$ observations of directed pairs
$(i,j)$ into $N-1$ distinct sets 
\[
\mathcal{I}_{k}= \big\{ (s,s+k):s=1,\dots,N-k \big\} \bigcup \big\{  (N-k+s,s):s=1,\dots,k  \big\},
\]
for $k=1,\dots,N-1$. We define the $k$-th `leave-out sample' to be the set of observations that \emph{excludes} the observations $(i,j)\in\mathcal{I}_k$.

\begin{figure} \centering
\caption{Diagram of the sets $\mathcal{I}_{k}$ for $k=1,2,3$.}
\includegraphics[width=0.7\paperwidth]{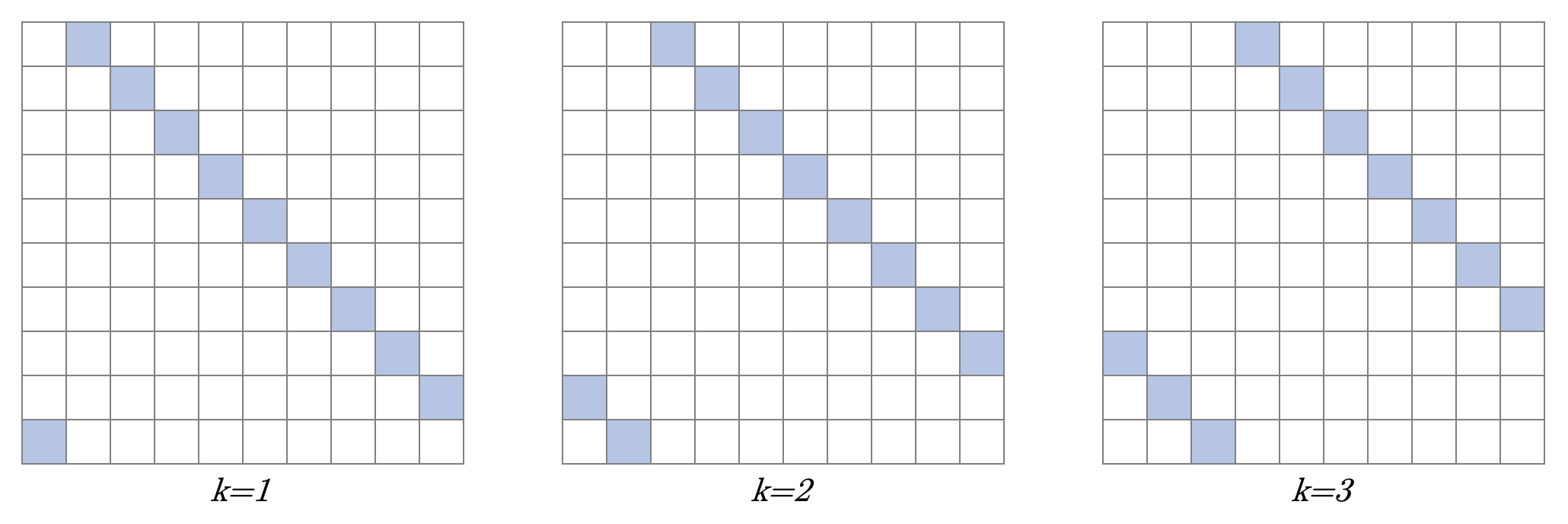}
\caption*{\small Observation $(i,j)$ in the network is represented by the corresponding position in 
each matrix. The shaded squares are the observations contained in the leave-out sets $\mathcal{I}_k$}
\label{fig:leaveout_sets-1}
\end{figure}

Figure \ref{fig:leaveout_sets-1} represents the structure of the
first three sets $\mathcal{I}_{k}$ for a network of $N=10$ agents.
Observations are ordered in an $N\times N$ matrix so that the $(i,j)$
cell represents the corresponding observation in the network (the
diagonal elements are empty since there are no $(i,i)$ observations).\footnote{In principle, the jackknife method presented here can be applied to settings with self-links, so long as these are treated identically to other observations in the network. In this case, the set of self-links forms an $N$-th set $\mathcal{I}_N$, and the jackknife estimator would use the linear combination $N\widehat{\beta}-(N-1)\tilde{\beta}$.}
The leave-out sets take diagonal slices from this data matrix. As is clear from Figure \ref{fig:leaveout_sets-1},
constructing the sets this way ensures that each $\mathcal{I}_{k}$ contains exactly
one observation related to each sender and receiver fixed effect,
i.e. there is one observation taken from every row and every column of the data matrix.

Let $1_{ij}^{k}=\mathds{1}\{(i,j)\not\in\mathcal{I}_{k}\}$ be an
indicator variable that is equal to one whenever the observation $(i,j)$
is included in the $k$-th leave-out sample (i.e. not in $\mathcal{I}_k$). The $k$-th leave-out
estimates are obtained by maximizing the leave-out sample objective function, i.e.\
\[
(\widehat{\beta}_{(k)},\widehat{\alpha}_{(k)},\widehat{\gamma}_{(k)}) = \argmax_{(\beta,\alpha,\gamma)\in\mathbb{R}^{\dim(\beta)+2N}}\frac{1}{N-2}\sum_{i}\sum_{j\ne i}\ell_{ij}(\beta,\alpha_i + \gamma_j)\times 1_{ij}^{k},
\]
subject to $\sum_{i}\alpha_{i}=\sum_{i}\gamma_{i}$. We can then construct the jackknife bias-corrected estimator
\begin{equation}
\widehat{\beta}_{J}=(N-1)\widehat{\beta}-(N-2)\frac{1}{N-1}\sum_{k=1}^{N-1}\widehat{\beta}_{(k)}.\label{eq:jack}
\end{equation}

The construction of the leave-out estimators is analogous to jackknife
bias correction in the panel data setting; however, the structure
of the jackknife proposed here is new. The procedure relies on dropping
sets of observations that contain a single observation related to
every sender fixed effect $\alpha_{i}$ as well as every receiver
fixed effect $\gamma_{i}$. In this way, the bias from both types
of fixed effects can be addressed simultaneously, while holding the
distribution of fixed effects constant across the leave-out samples.
This is in contrast to an approach which drops all observations from
a single agent, which removes that agent's fixed effects from the
sample and alters the distribution of unobserved heterogeneity.
We provide simulation evidence that the method proposed here is more robust
to networks that have greater levels of unobserved heterogeneity or lower average degree.

\begin{remark}
While asymptotic bias could be removed using only a single leave-out estimate $\widehat{\beta}_{(k)}$, the resulting estimator would have larger variance than (\ref{eq:jack}). In Section \ref{sec:beta_asymptotics} it is shown that averaging across all $(N-1)$ leave-out estimates results in an asymptotically unbiased estimator with the same asymptotic variance as the MLE. This result relies on the fact that each observation is excluded from exactly one of the leave-out estimates, ensuring balanced treatment of the observations.\footnote{In the panel data setting, \citet{Dhaene2015} note that forming a
jackknife using overlapping subpanels (across time) results in an
inflation of the asymptotic variance, since some time periods are
used more that others.}
\end{remark}

\begin{remark}
The construction of the sets $\mathcal{I}_k$ depends on the labeling of
nodes, which is assumed to be independent of the observed data. One could remove this arbitrariness by computing the jackknife estimate for a number of different random node labelings and averaging the results.\footnote{Averaging across estimators that
depend on random sample partitions is common; see, for example, \citet{Hirano2017} who provide a formal argument in the case
of model selection.} In simulations (see Supplementary Appendix D for details), I find that in dense networks, this averaging has almost no impact on the distribution of the estimator. In networks featuring many agents with few links, there can be some advantage to computing the jackknife over even a small number of random node labelings in order to rule out the possibility of outlier cases, although this does not seem to be necessary for the weighted jackknife (introduced in the next section).
\end{remark}

\begin{remark}
The jackknife can be applied to undirected networks for which the objective function takes the form $\mathcal{L}(\beta,\alpha)=\frac{2}{N-1}\sum_i\sum_{j>i}\ell(Y_{ij},X_{ij},\beta,\alpha_i+\alpha_j)$. 
	In this case, we construct a directed network in which $(Y_{ij},X_{ij})=(Y_{ji},X_{ji})$ and then estimate the model using the directed network data with both sender and receiver fixed effects as in (\ref{eq:objective}).\footnote{Estimates using the `directed' network data are numerically identical to estimates from the undirected network, although standard errors should be rescaled by $1/\sqrt{2}$ to account for the doubling in sample size.}\footnote{The structure of the leave-out sets makes uses of the symmetry of the dyadic network. Given this, the approach is not suitable for settings that do not feature this symmetry such as bipartite networks or those in which data is only observed for a subset of dyads. Alternative jackknife corrections may be possible in these settings, although this is left for future work.}
\end{remark}

The jackknife estimator requires $N$ estimations of the model, and
so may be computationally intensive for large networks, although speed
may be improved by computing the leave-out estimates in parallel,
and using good starting values such as the full sample estimates.
As an alternative, a `leave-$l$-out' version of the jackknife can also be used,
which reduces the number of additional estimations of the model by
dropping $l$ observations per fixed effect, as opposed to just one.
To describe the estimator, let $N_{l}=\frac{N-1}{l}$ (we assume here
that $N-1$ is divisible by $l$ for simplicity). We can construct a partition of the 
data into $N_{l}$ non-overlapping sets $\mathcal{I}_{k}^{l}$, defined by combining 
$l$ of the original sets as follows $\mathcal{I}_{k}^{l}=\cup_{j=0}^{l-1}\mathcal{I}_{k+jN_{l}}$,
for $k=1,\dots,N_{l}$. This results in $N_l$ 
leave-$l$-out samples, with corresponding estimates $\hat{\beta}_{l,(k)}$,
which are the estimates from using all observations except those in
the $k$-th set $\mathcal{I}_{k}^{l}$. A jackknife
bias-corrected estimate can then be constructed as
\begin{equation}
\hat{\beta}_{J,l}=\frac{N-1}{l}\hat{\beta}-\frac{N-1-l}{l}\frac{1}{N_{l}}\sum_{k=1}^{N_{l}}\hat{\beta}_{l,(k)}.\label{eq:jack_l}
\end{equation}

\begin{remark}

The leave-$l$-out jackknife bias correction has the same asymptotic
variance as the standard leave-one-out jackknife and the full-sample
estimator. However, there may be some finite-sample efficiency loss,
particularly when $l$ is large or when the network is not sufficiently
dense. \citet{Hahn2024} show in the panel case that the leave-one-out
jackknife has smaller higher-order variance and bias than the split-sample jackknife
(i.e. its variance to $O(T^{-1})$ is smaller), and it is likely that
the same result applies here, although this is beyond the scope of
the present paper.

\end{remark}


\subsection{A weighted jackknife}

In finite samples, it is possible for some of the leave-out samples to produce outlier estimates that have an unduly large 
impact on the average of the $(N-1)$ leave-out estimates. In very dense settings, in which each agent has a large number of
links, this is unlikely to be the case. However, if $N$ is small, or the network has low density, then it is possible for some 
leave-out samples to generate very noisy estimates of $\beta$. In order to reduce the impact of these potential outlier estimates,
I propose taking a weighted average of the estimates $\widehat{\beta}_{(k)}$ as a way to 
improve the performance of the jackknife. Define the weights
\[
\widehat{W}_{(k)}=-\frac{1}{N}\Big(\partial_{\beta\beta'}\widehat{\mathcal{L}}_{(k)}-(\partial_{\beta\phi'}\widehat{\mathcal{L}}_{(k)})(\partial_{\phi\phi'}\widehat{\mathcal{L}}_{(k)})^{-1}(\partial_{\phi\beta'}\widehat{\mathcal{L}}_{(k)})\Big),
\]
where  $\partial_{\beta\phi'}\widehat{\mathcal{L}}$ etc.\ represent derivatives of the objective function with respect to the parameters. The weighted-jackknife estimator is given by
\begin{equation}
\widehat{\beta}_{wJ}=(N-1)\widehat{\beta}-(N-2)\sum_{k=1}^{N-1}\big( \widehat{W}_{J}^{-1}\widehat{W}_{(k)} \big)\widehat{\beta}_{(k)},\label{eq:jack_weighted}
\end{equation}
where $\widehat{W}_{J}=\sum_{k=1}^{N-1}\widehat{W}_{(k)}$.

The weights $\widehat{W}_{(k)}$ are the Hessian for $\beta$, after
concentrating out the fixed effects in the leave-out sample. In the special case that $\mathcal{L}$
is a log-likelihood function, the weights are estimators of the Fisher information for
$\beta$, in which case we are using an inverse variance weighting scheme that 
down-weights leave-out samples that produce particularly noisy estimates
of the common parameters. The weighting scheme is equally applicable
to non-likelihood settings, although it no longer carries the inverse
variance interpretation. 
In the following section it is shown that the weighted version of the jackknife
has the same asymptotic distribution as the standard jackknife under dense network asymptotics. However, in simulations (see Section \ref{sec:Simulations} for more details), the weighting is shown to significantly improve the performance of the estimator in sparser networks.


\subsection{Asymptotic analysis for the common parameters \label{sec:beta_asymptotics}}

Before stating the asymptotic result for the jackknife estimator for $\beta$, I first  state some regularity conditions that are assumed to hold, in addition to Assumption \ref{assu:network}. The conditions are standard for 
M-estimators and follow the assumptions made in \citet{Fernandez-Val2016}.
\begin{assumption}[Regularity conditions]
\label{assu:beta_jack}Let 
$\mathcal{B}_{\varepsilon}$ be a subset of $\mathbb{R}^{\dim\beta+1}$
that contains an $\varepsilon$-neighborhood of $(\beta{}_{0},\pi_{0,ij})$
for all $(i,j)$ and $N$, for some $\varepsilon>0$.

(i) For all $i,j$ and $N$ we have that $\bar{E}[\partial_{\beta}\ell_{ij}]=\bar{E}[\partial_{\pi}\ell_{ij}]=0$.
For all $N$, the objective function $\mathcal{L}$ is strictly concave
over $\mathbb{R}^{\dim\beta+2N}$, and the matrix $-\partial_{\phi\phi'}\bar{\mathcal{L}}$
is positive definite.

(ii) There exists an $\varepsilon>0$ such that, for all $(i,j)$, the function $\ell_{ij}(\beta,\pi)$
is six times continuously differentiable with respect to $(\beta,\pi)$ over $\mathcal{B}_{\varepsilon}$
almost surely, with derivatives that are 
bounded in absolute value uniformly over $(\beta,\pi)\in\mathcal{B}_{\varepsilon}$
by a function $M(Z_{ij})>0$ a.s., where $\max_{i,j}\bar{E}[M(Z_{ij})^{16}]\leq C$ and $Z_{ij}'=(Y_{ij},X_{ij}')$.

(iii) The quantities
\begin{align*}
W  & =\textup{plim}_{N\to\infty} \bigg[-\frac{1}{N}\big(\partial_{\beta\beta'}\bar{\mathcal{L}}-(\partial_{\beta\phi'}\bar{\mathcal{L}})(\partial_{\phi\phi'}\bar{\mathcal{L}})^{-1}(\partial_{\phi\beta'}\bar{\mathcal{L}})\big) \bigg]\\
\Omega  & =\textup{plim}_{N\to\infty} \bar{E}\big[\big(\partial_{\beta}\mathcal{L}-(\partial_{\beta\phi'}\bar{\mathcal{L}})(\partial_{\phi\phi'}\bar{\mathcal{L}})^{-1}(\partial_{\phi}\mathcal{L})\big)\big(\partial_{\beta}\mathcal{L}-(\partial_{\beta\phi'}\bar{\mathcal{L}})(\partial_{\phi\phi'}\bar{\mathcal{L}})^{-1}(\partial_{\phi}\mathcal{L})\big)'\big].
\end{align*}
exist for $W$ and $\Omega$ positive definite matrices.
\end{assumption}

Assumption \ref{assu:beta_jack}(i) 
requires that the true parameters $(\beta_{0},\alpha_{0},\gamma_{0})$
are solutions to the first-order conditions of the objective function.
Concavity of the objective function ensures that these are unique solutions. 
This is satisfied in many common nonlinear
models, including the class of regression models with log-concave
densities (as well as censored and truncated versions of these models),
which includes probit, logit, ordered probit, Tobit, gamma and beta
models among others (see \citet{Pratt1981}, \citet{Newey1994a}).

Assumption \ref{assu:beta_jack}(ii) provides basic smoothness conditions for the objective
function, which are required to ensure
the validity of the asymptotic expansions. Analysis of the jackknife requires
higher order expansions than are required for characterization of
the analytical bias and first-order asymptotic properties of the estimator,
and so Assumption \ref{assu:beta_jack}(ii) is somewhat stronger than the equivalent
assumption employed in \citet{Fernandez-Val2016}. 
Finally, Assumption \ref{assu:beta_jack}(iii) ensures that the asymptotic variance of $\widehat{\beta}$ is
non-degenerate. 

To describe the estimator of this variance, first define $D_{\beta}\ell_{ij}=\partial_{\beta}\ell_{ij}-(\partial_{\beta\phi'}\bar{\mathcal{L}})(\partial_{\phi\phi'}\bar{\mathcal{L}})^{-1}(\partial_{\phi}\ell_{ij})$.
This term is the score for $\beta$ after partialling out the fixed
effect parameters. Estimators of the variance terms can be created
in the usual way by plugging in estimates of the model parameters,
i.e.
\begin{equation}
\begin{aligned}\widehat{W} & =-\frac{1}{N}\big(\partial_{\beta\beta'}\widehat{\mathcal{L}}-(\partial_{\beta\phi'}\widehat{\mathcal{L}})(\partial_{\phi\phi'}\widehat{\mathcal{L}})^{-1}(\partial_{\phi\beta'}\widehat{\mathcal{L}})\big),\\
\widehat{\Omega} & =\frac{1}{N(N-1)}\sum_{i}\sum_{j<i}(\widehat{D_{\beta}\ell_{ij}}+\widehat{D_{\beta}\ell_{ji}})^{2},
\end{aligned}
\label{eq:var_est}
\end{equation}
where the terms $\partial_{\beta\beta'}\widehat{\mathcal{L}}$, $\widehat{D_{\beta}\ell_{ij}}$
etc. are evaluated at the estimates $\widehat{\beta}$, $\widehat{\alpha}$,
$\widehat{\gamma}$. Note that the estimator $\widehat{\Omega}$
allows for correlation between the $Y_{ij}$ and $Y_{ji}$ outcomes.

We now state the main theorem of the paper, on the asymptotic distribution
of the jackknife bias-corrected estimator.
\begin{theorem}[Jackknife estimation of $\beta$]
\label{thm:beta}Let Assumptions \ref{assu:network} and  \ref{assu:beta_jack} hold. Then, for
$\widehat{\beta}_{J}$ the jackknife bias-corrected estimator (\ref{eq:jack}), and $\widehat{\beta}_{wJ}$ the weighted-jackknife estimator (\ref{eq:jack_weighted})
\begin{align*}
N(\widehat{\beta}_{J}-\beta_{0})\overset{d}{\longrightarrow}\mathcal{N}(0,V), \\
N(\widehat{\beta}_{wJ}-\beta_{0})\overset{d}{\longrightarrow}\mathcal{N}(0,V)
\end{align*}
where $V =W^{-1}\Omega W^{-1}$. In addition, the estimator of the asymptotic variance $\widehat{V}=\widehat{W}^{-1}\widehat{\Omega}\widehat{W}^{-1}$, for $\widehat{W}$
and $\widehat{\Omega}$ the plug-in estimators shown in (\ref{eq:var_est}), is consistent, i.e. $\widehat{V}\overset{p}{\longrightarrow} V$.
\end{theorem}

\begin{proof}
	See Appendix \ref{sec:App_B} for proof.
\end{proof}

The jackknife estimator is asymptotically normally distributed and
unbiased. It also has the same asymptotic variance as the non-bias-corrected
estimator based on maximization of (\ref{eq:objective}), which is the usual sandwich
form one and is easily computed. In the case of maximum likelihood
we will have that $W=\Omega$ so that the variance
simplifies to $V=W^{-1}$, although in general this will not
be true, for example the researcher may wish to allow for correlation
between $Y_{ij}$ and $Y_{ji}$ by clustering 
at the dyad level as in (\ref{eq:var_est}).

When the objective function (\ref{eq:objective}) is a log-likelihood, it is interesting to consider what the jackknife estimates when the 
model is in fact misspecified. The asymptotic result relies on the existence of a set of `pseudo-true'
 parameter values that set the first-order conditions of (\ref{eq:objective}) to zero. In this sense, we can expect the jackknife to provide an asymptotically 
 unbiased estimate of the same `pseudo-true' parameter values that are the probability limit of the MLE \citep{White1982}. This is a well-defined
 object of interest that is well understood and may still be of economic interest.
 In the Supplementary Appendix (section D) I show the results of simulations that compare 
 the jackknife to the CMLE estimator of \citet{Jochmans2018} under both a correctly specified (logit) and misspecified (probit) DGP. In the 
 misspecified case it is clear that the estimators are estimating very different quantities. Although both quantities may be of some interest, 
 the interpretation of network parameters under misspecification may be a relevant consideration in the choice of estimator in applied settings.


\section{\label{subsec:avg_effects}Estimating average effects}

In addition to estimation of the common parameter $\beta$, researchers
may also be interested in estimating certain averages over the distribution
of exogenous regressors and fixed effects. An important advantage
of the jackknife bias correction over methods based on conditioning
on sufficient statistics (e.g. \citet{Graham2017}, \citet{Jochmans2018})
is that we are able to construct asymptotically unbiased 
estimators for these averages. In the network setting, these averages may be over functions that 
depend not just on a single observation, but on dyads, triads (groups of three nodes),
or other network patterns. As an example, we focus on how these objects
can be used to construct tests of the assumption of independent link
formation (Assumption \ref{assu:network} (i)), but they may have wider
relevance in empirical work.

\subsection{Jackknife fixed-effect averages}

Here we describe a general fixed effect average parameter that covers each of the cases discussed above. 
Let $\lambda$ be a set of observations in the network; for example,
$\lambda=\{(i,j)\}$ contains a single observation, while $\lambda=\{(i,j),(j,k),(k,i)\}$ collects a sequence of
link outcomes between three nodes. Let the number of unique observations contained in $\lambda$ be $r$, 
and let the set of all possible $\lambda$ formed by permuting the nodes for a network
of size $N$ be $\Lambda_{N}$. We consider averages of the form
\begin{equation}
\begin{aligned}
\delta = E\big[ \Delta(\beta_0,\phi_0) \big], \quad 
\overline{\Delta} = \bar{E}\big[ \Delta(\beta_0,\phi_0) \big] \\
\Delta(\beta_0,\phi_0) =\frac{1}{\vert\Lambda_{N}\vert}\sum_{\lambda}m(Y_{\lambda},X_{\lambda},\beta_0,\pi_{0,\lambda}),
\end{aligned}
\label{eq:avg_eff2}
\end{equation}
where $m$ is a given function of interest, and $Y_{\lambda}=\{Y_{ij}\}_{(i,j)\in\lambda}$, $X_{\lambda}=\{X_{ij}\}_{(i,j)\in\lambda}$,
and $\pi_{\lambda}=\{\alpha_{i}+\gamma_{j}\}_{(i,j)\in\lambda}$ collect
the outcomes, covariates and fixed effects for the observations in
$\lambda$. As earlier, we use the shorthand $m_\lambda =  m(Y_{\lambda},X_{\lambda},\beta_0,\pi_{0,\lambda})$. 
We distinguish between the average in the observed sample 
$\Delta$, the expectation of this average conditional on  covariates and 
fixed effects $\overline{\Delta}$, and a population expectation $\delta$. Each object may be of interest in particular applications and the choice does not affect the estimator, but does impact the asymptotic variance, which we discuss further in Section \ref{sec:avg_asymptotics}.

Estimators of the averages in (\ref{eq:avg_eff2}) suffer from an incidental parameters bias, but may be bias corrected using the jackknife. 
The asymptotic bias in the plug-in estimator $\widehat{\Delta}=\Delta(\widehat{\beta},\widehat{\phi})$
arises not just from bias in the common parameter estimates, but also from averaging over a nonlinear function of noisy
fixed effect estimates, as well as correlation between the fixed-effect estimation error and the average effect terms $m_\lambda$. The jackknife takes care of each of these sources of bias in a single step. Define the leave-out estimates as
\begin{equation}
\widehat{\Delta}_{(k)}=\frac{N-1}{N-r-1}\frac{1}{\vert\Lambda_{N}\vert}\sum_{\lambda}m(Y_{\lambda},X_{\lambda},\widehat{\beta}_{(k)},\widehat{\pi}_{\lambda,(k)})\times 1_{\lambda}^{k}. \label{eq:leaveout_m2}
\end{equation}
where $1_{\lambda}^{k}=\prod_{(i,j)\in\lambda}\mathds{1}\{ (i,j) \not\in \mathcal{I}_k \}$
is an indicator that is equal to one whenever all of the observations
in $\lambda$ are included in the $k$-th leave-out sample, but zero when one or more 
of the observations are in $\mathcal{I}_{k}$. The 
leave-out estimate $\widehat{\Delta}_{(k)}$ replaces parameter estimates with their corresponding leave-out estimators 
$\widehat{\beta}_{(k)}$ and $\widehat{\pi}_{\lambda,(k)} = \widehat{\alpha}_{(k),i} + \widehat{\gamma}_{(k),j}$, and takes the 
average only over sets $\lambda$ for which all $(i,j)\in\lambda$ are contained in the corresponding leave-out sample. 
The factor $\frac{N-1}{N-r-1}$ accounts for the fact that $\lambda$ is dropped from the average
whenever any of the $r$ observations it contains are dropped.\footnote{When $m$ does not depend on $Y_\lambda$ and so $\bar{\Delta}=\Delta$ (e.g.\ marginal effects), the leave-out estimates may instead be constructed as $\widehat{\Delta}_{(k)}=\frac{1}{\vert\Lambda_{N}\vert}\sum_{\lambda}m(Y_{\lambda},X_{\lambda},\widehat{\beta}_{(k)},\widehat{\pi}_{\lambda,(k)})$.} The jackknife bias-corrected estimator is 
\begin{equation}
\widehat{\Delta}_{J}=(N-1)\widehat{\Delta}-(N-2)\frac{1}{N-1}\sum_{k}\widehat{\Delta}_{(k)}.\label{eq:avg_jack}
\end{equation}
In Section \ref{sec:avg_asymptotics}, we show that the bias-corrected estimator
is asymptotically normal and unbiased. Before presenting the asymptotic result, we briefly discuss some examples of 
the usefulness of such average effects.

\paragraph{Example: Marginal effects}
Researchers are often interested in computing marginal effects of covariates, rather than in the coefficients $\beta$. For example, in the 
link formation model in (\ref{eq:formation}), we may be interested in the average derivative with respect to the $l$-th covariate, for which we 
have $\lambda=\{(i,j)\}$ and $m_{ij} = \beta_l \partial F(X_{ij}'\beta+\alpha_{i}+\gamma_{j}))$, where $\partial F$ is the derivative of the 
CDF $F$.

\paragraph{Example: Clustering coefficients}

Clustering coefficients are common measures of the degree to which agents tend 
to cluster together in a network (see for example \citealp{Jackson2008}). There are many definitions for 
such network measures that can be expressed in the form of averages as in (\ref{eq:avg_eff2}), or ratios of such averages. For example, in a model of directed link formation as in (\ref{eq:formation}), we may be 
interested in the expected number of transitive triangles in the network (defined as a set of links between three agents $i\to j$, $i\to k$, $k\to j$). 

In a dyadic link formation model in which $P(Y_{ij} = 1 \vert X,\phi) = F(X_{ij}'\beta + \alpha_i + \gamma_j)=p_{ij}$ say, the expected fraction of transitive triangles that will form is
\begin{equation} \label{eq:triangle}
T=\frac{1}{N(N-1)(N-2)}\sum_{i=1}^{N}\sum_{j\ne i}\sum_{k\not\in\{i,j\}}p_{ij}p_{ik}p_{kj},
\end{equation}
which is of the form (\ref{eq:avg_eff2}), with $\lambda=\{(i,j),(i,k),(k,j)\}$ and $m_\lambda = p_{ij}p_{ik}p_{kj}$. This can be constructed under various counterfactuals.

\paragraph{Example: Specification testing}

Another useful application of fixed-effect averages is in specification testing. \citet{Dzemski2019} proposes a test for the presence of strategic interactions that is based on comparison of the observed fraction of transitive triangles 
to the expected fraction of triangles that would form in a dyadic model, i.e. (\ref{eq:triangle}). 
The test statistic is
\[
\widehat{T}=\frac{1}{N(N-1)(N-2)}\sum_{i=1}^{N}\sum_{j\ne i}\sum_{k\not\in\{i,j\}}\Big(Y_{ij}Y_{ik}Y_{kj}-\widehat{p}_{ij}\widehat{p}_{ik}\widehat{p}_{kj}\Big),
\]
where $\widehat{p}_{ij}=\Phi(X_{ij}'\widehat{\beta} +\widehat{\alpha}_i + \widehat{\gamma}_j)$, for $\Phi$ the standard normal CDF. 
Under the null hypothesis that the dyadic model is correct, the test statistic  is mean zero when evaluated at $(\beta_0,\phi_0)$. In contrast, failure of the conditional independence assumption (Assumption \ref{assu:network} (i) in this paper), as would occur in the presence of strategic decision making by agents, will result in the test staistic converging in probability to a non-zero limit. 
\citet{Dzemski2019} derived an analytical bias correction for this statistic (and another based on cyclic triangles) for the case of a probit model. This statistic can also be jackknife bias-corrected as shown in this paper. In addition, the jackknife allows for specification tests like this to be extended to a range of other models, including models with non-binary outcomes, as well as a range of other statistics. For example, we could consider test statistics based on the covariance between  $(Y_{ij} - \bar{E}[Y_{ij} ])$ and some network statistic $S_{ij}$
\begin{equation}
T=\frac{1}{N(N-1)}\sum_{i}\sum_{j\ne i}(Y_{ij}-\bar{E}[Y_{ij} ])S_{ij}.\label{eq:stat_T}
\end{equation}
The choice of the statistic $S_{ij}$ may be motivated by the exact type of deviation from conditional 
independence the researcher is interested in testing. For example, $S_{ij}=Y_{ji}$ gives a test of reciprocity, i.e. the tendency of agents to be more likely to form a directed link if the corresponding link in the opposing direction already exists.\footnote{Note that reciprocity is allowed for under Assumption \ref{assu:network}, so that a rejection of the null hypothesis of no reciprocity does not affect interpretation of the dyadic model estimates.} 
Using $S_{ij}=\frac{1}{N-2}\sum_{k\not\in\{i,j\}}Y_{ik}Y_{kj}$ gives a test for transitivity, similar to the one presented above. \citet{Graham2020d} consider a similar set of tests for a logit model, where sufficient statistics for the fixed effect parameters 
exist, and derive some optimal test statistics. It should be noted that these statistics test for the presence of strategic interactions, conditional on the specified model being correct, and in general could reject due to any form of misspecification (for example, estimating a logit specification when the truth is a probit).\footnote{\citet{auerbach2022} proposes a two-sample Kolmogorov-Smirnov test for network data that can be used to test whether a network was generated by a particular parametric model.}
It would seem useful to consider the properties of 
such test statistics in a broader set of models, although this is beyond the scope of this paper. We note only that many interesting statistics are likely to have the form of the general fixed effect averages 
considered in this paper, and so the jackknife bias correction may prove useful in forming such tests.


\subsection{Asymptotic analysis for fixed effect averages \label{sec:avg_asymptotics}}

Here we present asymptotic results for the general fixed-effect averages. We can decompose the difference between our jackknife estimate and the population average into three sources
\begin{equation}
\widehat{\Delta}_{J}-\delta=(\widehat{\Delta}_{J}-\Delta)+(\Delta-\overline{\Delta})+(\overline{\Delta}-\delta).\label{eq:avg_decomp}
\end{equation}
The first term, $\widehat{\Delta}_{J}-\Delta$, represents variation
caused by parameter estimation error. The next term, $\Delta-\overline{\Delta}$, is variation
of the sample outcomes $m_{\lambda}$ around their conditional expectations
$\bar{m}_{\lambda}= \bar{E}[m_{\lambda}]$. In the case that $m$ does not
depend on outcomes $Y_{\lambda}$, as for marginal effects, we will have $m_{\lambda}=\bar{m}_{\lambda}$
and this second term will vanish. Finally, $\overline{\Delta}-\delta$
captures differences in the distribution of covariates and fixed effects
in the observed network, relative to the population. In the case that
the full network is observed, or whenever $\overline{\Delta}=0$, as
is the case for specification tests discussed above, we will have that $\overline{\Delta}=\delta$. In addition to Assumptions \ref{assu:network} and \ref{assu:beta_jack}, we 
will impose regularity conditions on the choices of $\lambda$ and $m$.

\begin{assumption}[Regularity conditions for fixed-effect averages]
\label{assu:averageeffects}Let $\lambda$ be a set of $r$ observations
$(i,j)$ corresponding to an ordered set of $p$ distinct agents. Let
$\mathcal{B}_{\varepsilon}$ be a subset of $\mathbb{R}^{\dim\beta+r}$
that contains an $\varepsilon$-neighborhood of $(\beta{}_{0},\pi_{0,\lambda})$
for all $N$ and $\lambda$, for some $\varepsilon>0$.

(i) The set $\Lambda_{N}$ contains every $\lambda$ corresponding to the set of $\frac{N!}{(N-p)!}$ permutations of $p$ agents.

(ii) The function $m$ depends on $(\alpha,\gamma)$ only through
$\pi_{\lambda}=\{\alpha_{i}+\gamma_{j}\}_{(i,j)\in\lambda}$. There exists an $\varepsilon>0$ such that, for
all $\lambda$, $m(Z_{\lambda},\beta,\pi_{\lambda})$ is six times continuously differentiable 
with respect to $\beta$ and $\pi_\lambda$ over $\mathcal{B}_{\varepsilon}$
a.s., with derivatives that are
bounded in absolute value uniformly over $(\beta,\pi_{\lambda})\in\mathcal{B}_{\varepsilon}$
by a function $M(Z_{\lambda})>0$ a.s., where $\max_{\lambda}\bar{E}[M(Z_{\lambda})^{9}]\leq C$ and $Z_\lambda'=(Y_\lambda',X_\lambda')$.

(iii) We have that $0<\min_{\lambda}Var(m_{\lambda})\leq\max_{\lambda}Var(m_{\lambda})<\infty$
uniformly over $N$.
\end{assumption}
Assumption \ref{assu:averageeffects} (i) restricts $m$ to be a function
of a fixed number of edges in the network, and assumes
that $\Delta$ is an average of all possible arrangements of the
nodes in $\lambda$. This ensures that sufficient averaging occurs over both dimensions 
of the fixed effects. For example, an average of the form $\frac{1}{N}\sum_{j\ne i}m(X_{ij},\alpha_{i},\gamma_{j})$
is not allowed since we are only averaging over the receiver dimension,
while holding $i$ fixed. Assumption \ref{assu:averageeffects} (ii) is analogous to Assumption
\ref{assu:beta_jack} (ii), and imposes differentiability
and moment conditions on $m$. Finally, (iii) ensures
that the unconditional second moments of $m$ are well defined.

\subsubsection*{Inference on the conditional average $\overline{\Delta}$}\label{sec:inf_delta1}

The asymptotic distribution of $\widehat{\Delta}_J$ depends on the choice of target parameter,
either a conditional or population average. To describe the asymptotic variance when the target parameter is $\overline{\Delta}$, define, for $\theta'=(\beta',\alpha',\gamma')$,
\begin{align*}
	h_{ij}&=-N(\partial_{\theta}\overline{\Delta})(\partial_{\theta\theta'}\bar{\mathcal{L}})^{-1}\big((\partial_{\theta}\ell_{ij})+(\partial_{\theta}\ell_{ji})\big) \\
	s_{ij}&=\frac{(N-p)!}{(N-2)!}\sum_{\lambda\in\bar{\Lambda}_{ij}}(m_{\lambda}-\bar{m}_{\lambda}),\qquad \bar{\Lambda}_{ij}=\{\lambda:(i,j)\in\lambda\text{ or }(j,i)\in\lambda\}
\end{align*}
whenever $E\big[(m_{\lambda}-\bar{m}_{\lambda})(m_{\lambda'}-\bar{m}_{\lambda'})\big]\ne0$
for sets $\lambda$ and $\lambda'$ that share exactly one observation in common, and $s_{ij}=0$ otherwise. The asymptotic variance of 
the jackknife estimator is
\begin{align}\label{eq:V_Delta}
	V_{\Delta}=\lim_{N\to\infty}\frac{1}{N(N-1)}\sum_{i}\sum_{j<i}\bar{E}\big[(h_{ij}+s_{ij})^{2}\big]
\end{align}

The following theorem states the asymptotic result for the jackknife bias-corrected estimator
of the conditional fixed effect average $\overline{\Delta}$.
\begin{theorem}[Jackknife estimator for fixed-effect averages]
\label{thm:avg_eff}Let Assumptions \ref{assu:network}, \ref{assu:beta_jack} and \ref{assu:averageeffects}
hold, and let $\widehat{\Delta}_{J}$ be the jackknife bias-corrected estimator (\ref{eq:avg_jack}). 
Then
\[
N(\widehat{\Delta}_{J}-\overline{\Delta}) \overset{d}{\longrightarrow} \mathcal{N}(0,V_{\Delta}).
\]
Let $\widehat{V}_{\Delta}$ be the plug-in estimator
for $V_{\Delta}$ that replaces the unknown $\theta$ with estimates
$\widehat{\theta}$. Then $\widehat{V}_{\Delta} \overset{p}{\longrightarrow} V_{\Delta}$.
\end{theorem}
Some explanation for the form of the variance may be useful. The two
terms $h_{ij}$ and $s_{ij}$ relate to the first two components of
(\ref{eq:avg_decomp}). The term $h_{ij}$ captures 
variation from the estimation of the parameters $\theta=(\beta,\alpha',\gamma')$, that is, $N(\widehat{\Delta}_{J}-\Delta) =\frac{1}{N-1}\sum_{i}\sum_{j\ne i}h_{ij}+o_{p}(1)$. Whenever $m$ is not a function of $Y_{ij}$, as 
for marginal effects, 
this will be the only term in the variance, and is equal to the standard delta-method variance formula.

Whenever $m_\lambda \ne \bar{m}_\lambda$, the term $s_{ij}$ represents variation from the difference between the 
average effect term and its conditional mean, $\Delta-\overline{\Delta}$. 
This term is a U-statistic (conditional on $X$ and fixed effects) and its variance is dominated by covariances between
sets  $\lambda$ and $\lambda'$ that share exactly one observation
in common. Although $\lambda$ and $\lambda'$ are correlated whenever they share at least one common dyad,  
there are an order of magnitude fewer combinations that share two or more dyads 
and so these represent smaller order contributions that do not appear
in the asymptotic variance. In settings where $E\big[(m_{\lambda}-\bar{m}_{\lambda})(m_{\lambda'}-\bar{m}_{\lambda'})\big]=0$
for $\lambda$ and $\lambda'$ that share exactly one observation
in common, $\Delta-\overline{\Delta}$ is a degenerate U-statistic
and its variance is asymptotically of smaller order than the variance
from parameter estimation, i.e. $\widehat{\Delta}_{J}-\Delta$, and so may be ignored asymptotically.

\subsubsection*{Inference on the population average $\delta$}

When the object of interest is the unconditional average effect $\delta$, the convergence
of the estimator will be dominated by variation from the third component
in (\ref{eq:avg_decomp}), $\overline{\Delta}-\delta$. To describe
the average in this setting, it is useful to use its U-statistic
representation. We will additionally assume that $X_{ij}=h(X_{i},X_{j})$, 
a condition that appears in other work on dyadic models, for example
in \citet{Graham2017}. When $X_{ij}$ measures the similarity (or
difference) between $i$ and $j$ in some measure, we will commonly
have $X_{ij}=d(X_{i}-X_{j})$, for $d$ some distance function. Alternatively,
if $X_{ij}$ captures common membership in some group, we may have
$X_{ij}=X_{i}X_{j}$ where $X_{i}$ is an indicator for $i$'s membership.

We may write $\overline{\Delta}-\delta$ as
\begin{align} \label{eq:U_stat}
\overline{\Delta}-\delta =\binom{N}{p}^{-1} \sum_{\tau}u_{\tau}, \qquad u_{\tau}=\frac{1}{p!}\sum_{\lambda\in\tau}(\bar{m}_{\lambda} -E[m_\lambda])
\end{align}
where $\tau$ represents a set of $p$ unique agents and $u_\tau$ is a sum over the $p!$ $\lambda$ that share the same set of agents (but in different order).
The term $u_{\tau}$ is a mean-zero symmetric function of $\{X_{i},\alpha_{i},\gamma_{i}\}$
for $p$ agents $i$. Assuming that the $\{X_{i},\alpha_{i},\gamma_{i}\}$
are i.i.d. over agents, $\overline{\Delta}-\delta$ is a U-statistic
of order $p$ and we may apply standard theory on such statistics
to compute its asymptotic distribution.

\begin{theorem}[Population average effects]
\label{thm:avg_pop}Let Assumptions \ref{assu:network}, \ref{assu:beta_jack} and \ref{assu:averageeffects}
hold. Additionally, assume that $X_{ij}=h(X_{i},X_{j})$ where $X_{i}$
is an observed agent-specific characteristic, and also that $(\alpha_{i},\gamma_{i},X_{i})$
is i.i.d. over $i$. Let $\Sigma_{1}=Cov\big(u_{\tau},u_{\tau'}\big)$, 
for $\tau,\tau'$ such that $\tau\cap\tau'=\{i\}$. Then, for $V_{\delta}=p^{2}\Sigma_{1}$
\[
\sqrt{N}(\widehat{\Delta}_{J}-\delta) \overset{d}{\longrightarrow}\mathcal{N}(0,V_{\delta}).
\]
Additionally, the variance estimator $\widehat{V}_{\delta}$ in (\ref{eq:V_delta_hat})
is consistent, i.e. $\widehat{V}_{\delta} \overset{p}{\longrightarrow} V_{\delta}$.
\end{theorem}
The convergence rate in Theorem \ref{thm:avg_pop} is slower than
the rate in Theorem \ref{thm:avg_eff}. While $m_{\lambda}$ and $m_{\lambda'}$
are \emph{conditionally} independent when $\lambda$ and $\lambda'$
share no dyads in common, the two are \emph{unconditionally }independent
only when they share no agents in common. One can think of the effective observation level here as an 
agent $i$, rather than a link $(i,j)$, and hence $\sqrt{N}$ is the standard convergence rate. Similarly to Theorem \ref{thm:avg_eff}, 
the variance is dominated
by covariances between sets $\lambda$ that share exactly one node
in common, $\Sigma_{1}$. A consistent estimator for this quantity is given by
\begin{align} 
\widehat{V}_{\delta} =\frac{1}{N}\sum_{i}\tilde{\mu}_{i}^{2}, \qquad
\tilde{\mu}_{i} =\frac{(N-p)!}{(N-1)!}\sum_{\lambda:i\in\lambda}(\widehat{\bar{m}}_{\lambda}-\widehat{\Delta}),
\label{eq:V_delta_hat}
\end{align}
where $\widehat{\bar{m}}_{\lambda}$ is a plug-in estimator for $\bar{m}_{\lambda}$.

Since the rate of convergence in Theorem \ref{thm:avg_pop} is $N^{-1/2}$,
there is in fact no asymptotic bias generated by the incidental parameters. 
Nonetheless, the bias 
correction step is still recommended as it is likely to improve
the finite sample properties of inference by correctly centering 
confidence sets, with little or no cost in terms of additional
variance. In the panel data setting, \citet{Fernandez-Val2016} report
such improvements in simulations.


\section{\label{sec:Simulations}Simulations}

Here I demonstrate the effectiveness of the jackknife in simulations, by simulating the estimation of a probit model of network formation.\footnote{Supplementary Appendix D reports the results of additional simulations, including those from a logit specification. The results for the logit model match those presented here for the probit specification.}
This allows for comparison with the analytical correction of \citet{Dzemski2019}, in addition to some alternative jackknife methods. The design is taken from that paper and has
also been used in a number of other network papers. The binary outcome
$Y_{ij}$ is determined by
\[
Y_{ij}=1\{\beta X_{ij}+\alpha_{i}+\gamma_{j}>\varepsilon_{ij}\}
\]
where $\beta=1$ and $\varepsilon_{ij}\sim \mathcal{N}(0,1)$. Individual $i$
is characterized by the binary scalar $X_{i}=1-2\times1\{i\text{ is odd}\}$,
and the homophily variable is given by $X_{ij}=X_{i}X_{j}$, i.e.
it is one for pairs with the same sign and minus one for pairs with
opposing signs. The fixed effects are given by an equally spaced 
sequence from $C_{N}^{l}$ to $C_{N}^{u}$, with $\alpha_i = \gamma_i$.\footnote{The DGP sets sender and receiver fixed effects to be equal, but this is not imposed in the estimation.} The value of
$(C_{N}^{l},C_{N}^{u})$ is intended to control the sparsity of the
network, and I consider four choices, shown in Table \ref{tab:settings}. 
The first setting generates a dense network in which around half of all links are formed.
Subsequent settings feature increasingly sparse networks in which
some nodes remain well connected, while others make few links. In
the sparsest setting only around 3 per cent of all links are
formed and the networks feature large numbers of disconnected agents. 

\renewcommand{\arraystretch}{1.2}
\begin{table}
\centering \small
\caption{\label{tab:settings}Fixed effect distributions}
\begin{tabular}{lcccccccc}
\hline 
$(C_{N}^{l},C_{N}^{u})$ & Density & Connected & Min & 1st quart.& Median & 3rd quart. & Max \tabularnewline
\hline 
$(\pm \log\log N)$   & 0.50 & 50.0 & 8.6 & 16.9 & 24.1 & 31.7 & 40.3 \\
$(-\log\log N, 0)$    & 0.19 & 49.9 & 2.0 & 6.3 & 9.2 & 12.5 & 18.9 \\
$(-\log^{1/2}N, 0)$ & 0.12 & 45.5 & 0.1 & 2.4 & 4.9 & 8.4 & 15.2 \\
$(-\log N, 0)$          & 0.03 & 20.7 & 0.0 & 0.0 & 0.5 & 2.6 & 8.2 \\
\hline 
\end{tabular}
\vspace{2px}
\caption*{\footnotesize *Values are averages over 1000 simulations in a network of $N=50$  nodes. Density is the proportion of all potential links formed; `Connected' counts the number of agents that have at least one link in each direction; the remaining statistics are the min, max and quartiles of the out-degree across the 50 agents ($\alpha_i=\gamma_i$ so using in-degree gives the same results).} 
\end{table}

I report results for the uncorrected maximum likelihood estimator, the analytical bias correction, the jackknife and weighted jackknife presented in this paper, as well as two additional jackknife bias corrections that have been previously proposed.
Specifically, \citet{Cruz-Gonzalez2017} suggest two
jackknife bias corrections for dyadic network models: (1) a split-sample correction based on dividing the sample into halves along first the sender and then the receiver dimension; and (2) a leave-one-agent-out jackknife, based on subsamples of the data which drop all observations from a particular agent in the network. Details for the construction of these estimators are given in \citet{Cruz-Gonzalez2017}.

\begin{table}[H]
\centering \small
\caption{\label{tab:Sim1_mean}Simulation results - estimation of $\beta$}
\begin{tabular}{lccccccccc}
\hline 
 & Bias & SD & RMSE & Rej(5\%) &  & Bias & SD & RMSE & Rej(5\%)\tabularnewline
\hline 
 & \multicolumn{4}{c}{$-\log(\log n),\log(\log n)$} &  & \multicolumn{4}{c}{$-\log(\log n),0$}\tabularnewline
\cline{2-5}\cline{7-10}
MLE & 0.060 & 0.044 & 0.075 & 0.273 &  & 0.074 & 0.062 & 0.096 & 0.249\tabularnewline
BC & -0.001 & 0.040 & 0.040 & 0.038 &  & 0.007 & 0.057 & 0.057 & 0.047\tabularnewline
Jackknife & -0.009 & 0.039 & 0.040 & 0.041 &  & -0.011 & 0.057 & 0.058 & 0.054\tabularnewline
Jackknife (weighted) & -0.006 & 0.039 & 0.040 & 0.037 &  & -0.001 & 0.056 & 0.056 & 0.049\tabularnewline
Jackknife (agent) & -0.012 & 0.039 & 0.041 & 0.042 &  & -0.021 & 0.067 & 0.071 & 0.069\tabularnewline
Jackknife (split) & -0.016 & 0.042 & 0.045 & 0.070 &  & -0.037 & 0.100 & 0.107 & 0.119\tabularnewline
\hline 
 & \multicolumn{4}{c}{$-(\log n)^{1/2},0$} &  & \multicolumn{4}{c}{$-\log n,0$}\tabularnewline
\cline{2-5}\cline{7-10}
MLE & 0.110 & 0.126 & 0.167 & 0.216 &  & 0.613 & 0.654 & 0.896 & 0.067\tabularnewline
BC & 0.022 & 0.112 & 0.114 & 0.038 &  & -21.692 & 41.259 & 46.614 & 0.013\tabularnewline
Jackknife & -0.040 & 0.163 & 0.168 & 0.078 &  & -0.513 & 1.098 & 1.212 & 0.503\tabularnewline
Jackknife (weighted) & 0.008 & 0.143 & 0.143 & 0.054 &  & 0.295 & 0.533 & 0.610 & 0.046\tabularnewline
Jackknife (agent) & -0.096 & 0.262 & 0.279 & 0.137 &  & -1.323 & 1.661 & 2.123 & 0.640\tabularnewline
Jackknife (split) & -0.223 & 0.316 & 0.387 & 0.410 &  & -1.560 & 1.476 & 2.148 & 0.658\tabularnewline
\hline 
\end{tabular}
\caption*{\footnotesize *Estimators are the MLE, analytical bias correction (BC) (Dzemski, 2019), jackknife (J), and weighted jackknife (WJ).  Bias, standard deviation, root mean squared error, and rejection rate for a 5\% nominal size test across 1000 simulations.}
\end{table}

Table \ref{tab:Sim1_mean} presents the bias, standard deviation, root-mean-squared error and rejection rate (for a standard t-test at the 5\% level) for each estimator of $\beta$ over 1000 simulations. It is clear that the MLE estimates are substantially biased, which results in over-rejection of the t-test in all of the simulation designs. The top left panel of the table reports results from the densest design. In this setting, all of the bias-corrected estimators substantially reduce bias, have standard deviations below that of the MLE (and subsequently much lower RMSE) and have rejection rates close to the nominal level. In other designs, as the network becomes less dense, there are clear differences between the bias corrections. In particular, the split-sample bias correction removes less bias and has much larger standard deviation than other methods -- the split-sample estimator actually has larger RMSE than the MLE in three designs. Unless the original network is sufficiently dense, the half-sample networks can be very sparse, and hence quite noisy. This feature of split-sample jackknife estimators has been noted elsewhere for panel data model (see \citet{Fernandez-Val2018, Hahn2024}).\footnote{\citet{Hahn2024} demonstrate that split-sample jackknife estimators have larger high-order variance in a panel data model with individual fixed effects. We would expect such a result to also be true here, although this is beyond the scope of this paper.} 
Comparing the other jackknife estimators, we see that the jackknife method proposed in this paper is more successful at removing bias in the first three simulation designs than the leave-one-agent-out style jackknife, particularly in the third design, in which many observations only contain a few links. The weighted jackknife performs extremely well in each of the simulation designs, being close to unbiased in the first three designs, and having rejection rates at or below the nominal level in all cases. It is clear that the weighting of the leave-out samples becomes more valuable as the network design becomes sparser. The bias-corrected estimator performs similarly well -- we emphasize that the value of the jackknife correction is that it can be applied to a wide range of models and objects of interest for which analytical corrections are not derived or built into existing software packages. Finally, it should be noted that each of the bias corrections relies on dense network asymptotics. Although the jackknife performs well, even in relatively sparse settings, its performance is clearly not guaranteed in very sparse networks, as shown in the fourth simulation design.

Next, we investigate the performance of the jackknife correction for the transitivity test statistic described in (\ref{eq:stat_T}), using $S_{ij} = \frac{1}{N-2}\sum_{k\not\in\{i,j\}} Y_{ik}Y_{kj}$.  Table \ref{tab:sim_stat} reports the results for the MLE statistic, the jackknife and the weighted jackknife.\footnote{The weighted jackknife for the transitivity statistic is analgous to the version described for $\widehat{\beta}$, but with $\widehat{W}_{(k)}=\frac{1}{N(N-1)}\sum_i\sum_{j<i} \widehat{h}_{(k),ij}^2$, where $ h_{(k),ij}$ is as defined at the start of section \ref{sec:inf_delta1} and evaluated at the leave-out parameter estimates.}
Following \citet{Dzemski2019}, I report rejection rates both for the statistic as described in Theorem \ref{thm:avg_eff}, as well as a version that uses a bootstrap estimator for the standard error of the statistic.\footnote{The bootstrap samples are generated according to $Y^*_{ij} = 1\{ U_{ij} \leq \widehat{p}_{ij} \}$, where $U_{ij}$ is a uniformly distributed random variable.} In each of the designs, the bias for MLE version of the statistic is substantial (up to 8 standard deviations) so that we essentially reject with probability one. In contrast, the jackknife version of the statistic has low bias and in the denser simulation designs has rejection rates close to the nominal 5\% (the bootstrap standard errors make a small improvement in coverage for these first three designs). As for inference about the common parameter, the weighted version of the jackknife performs even better than the standard version, particularly for the less dense simulation designs, with approximately correct coverage in all designs (interestingly, the bootstrap does not appear to improve inference for this estimator).  In all, it is clear that the incidental parameters bias issue is quite severe when it comes to the estimation of these statistics, making bias correction important.

\begin{table}[H]
	\centering \small
	\caption{Simulation results - transitivity statistic}  \label{tab:sim_stat}
	\begin{tabular}{lccccccc}
		\hline 
		& Bias/SD & Rej (5\%) & Rej(b) (5\%) &  & Bias/SD & Rej (5\%) & Rej(b) (5\%)\tabularnewline
		\hline 
		& \multicolumn{3}{c}{$-\log(\log n),\log(\log n)$} &  & \multicolumn{3}{c}{$-\log(\log n),0$}\tabularnewline
		\cline{2-4}\cline{6-8}
		MLE & -7.039 & 1.000 & 1.000 &  & -8.134 & 1.000 & 1.000\tabularnewline
		Jackknife & -0.084 & 0.064 & 0.044 &  & -0.058 & 0.074 & 0.064\tabularnewline
		Jackknife (weighted) & -0.094 & 0.064 & 0.044 &  & 0.362 & 0.062 & 0.072\tabularnewline
		\hline 
		& \multicolumn{3}{c}{$-(\log n)^{/12},0$} &  & \multicolumn{3}{c}{$-\log n,0$}\tabularnewline
		\cline{2-4}\cline{6-8}
		MLE & -6.697 & 1.000 & 1.000 &  & -4.578 & 0.998 & 1.000\tabularnewline
		Jackknife & -0.117 & 0.080 & 0.078 &  & -0.675 & 0.180 & 0.208\tabularnewline
		Jackknife (weighted) & 0.248 & 0.064 & 0.072 &  & -0.226 & 0.074 & 0.100\tabularnewline
		\hline 
	\end{tabular}
	\caption*{\footnotesize *Estimators are the MLE, jackknife (J), and weighted jackknife (WJ).  Bias relative to standard deviation, and rejection rate for a 5\% nominal size test across 500 simulations. Rejection rates are shown for the t-statistics using standard errors obtained by either the formula for $\widehat{V}_\Delta$ described in Theorem \ref{thm:avg_eff} (Rej) or the bootstrap (Rej(b)).}
\end{table}


\section{Empirical example}

I apply the jackknife procedure to a country-level trade network, using data on relationships between 136 
countries, taken from \citet{Santos2006} (additional details on the data and its construction can be found in that
paper).\footnote{In Supplementary Appendix D, I also estimate a zero-inflated negative binomial model, as in \cite{Burger2009}, using the same data.} The outcome variable is a binary indicator for the presence 
of a trading relationship, which is predicted using several covariates that capture homophily
in trade relationships: \emph{log distance}, the log
of the distance between the capitals of the countries; \emph{border},
an indicator of whether the countries share a common border; \emph{language},
an indicator for whether the countries share a language; \emph{colony},
and indicator for whether either country had colonized the other at
some point in history; and \emph{trade agreement}, an indicator for
the presence of a joint preferential trade agreement between the two
countries. This same model is also estimated in \citet{Jochmans2018}. The network is dense, with around half of all potential trade links forming; however, there 
is significant heterogeneity in both in and out degree. For example, the 10th percentile of out degree is just 
26, while the 90th percentile is 133.

Table \ref{tab:gravity} presents the estimates of the model using a probit specification. 
The signs  and general magnitudes of the coefficients agree with those estimated elsewhere (e.g.\ \citet{Helpman2008}, \citet{Jochmans2018}), and so here we focus on the 
effect of the bias correction. The bias correction has the largest impact on the geographic distance coefficient, which is reduced by 0.74 of a standard deviation, while the impact is smaller on other coefficients. 
We find that the effect of bias correction is less significant here than in \citet{Jochmans2018} who uses a conditional logit estimator to correct for the incidental parameters bias.\footnote{Similar results to those presented in Table \ref{tab:gravity} were found using a logit specification as well as using the analytical bias correction. We also present results of simulations calibrated to this example in the supplementary appendix.}

\renewcommand{\arraystretch}{1.3}
\begin{table}[H]
\centering \small
\caption{\label{tab:gravity}Gravity model estimates}
\begin{tabular}{lrrrrcrrrr}
\hline 
 & \multicolumn{4}{c}{Parameter estimates} & & \multicolumn{4}{c}{Marginal effects}  \\ \cline{2-5}  \cline{7-10} 
 & \multicolumn{1}{c}{MLE} & \multicolumn{1}{c}{Jack} & \multicolumn{1}{c}{\; SE \;} & \multicolumn{1}{c}{Bias/SE} &
  & \multicolumn{1}{c}{MLE} & \multicolumn{1}{c}{Jack} & \multicolumn{1}{c}{\; SE \;} & \multicolumn{1}{c}{Bias/SE} \tabularnewline
\hline 
log distance 	  & -0.730 & -0.710 & 0.027 & -0.744 && -0.109 & -0.109 & 0.004 & 0.024 \\ 
language 		  & 0.320 & 0.311 & 0.051 & 0.174 && 0.048 & 0.048 & 0.008 & -0.003 \\ 
colony 		  & 0.300 & 0.292 & 0.053 & 0.150 && 0.045 & 0.045 & 0.008 & -0.008 \\ 
border       & -0.655 & -0.635 & 0.116 & -0.176 & &  -0.094 & -0.094 & 0.016 & 0.017 \\
trade agree.	  & 1.105 & 1.069 & 0.169 & 0.213 && 0.171 & 0.171 & 0.026 & 0.022 \\ 
\hline 
\end{tabular}
\vspace{2px}
\caption*{\footnotesize *Bias/SE is ratio of the difference in the Jackknife and MLE estimates to the standard error.} 
\end{table}

In addition to bias correction for the parameters of the model, the jackknife method allows bias correction of the marginal effects. For the (continuous) log distance variable, I compute the marginal effect as the average derivative $\varphi(\beta'X_{ij} + \alpha_i + \gamma_j)\beta_{dist}$, while for the binary variables I compute the average of the difference between the 
fitted probabilities setting that variable to one versus zero, $\Phi(\beta'X^{(1)}_{ij} + \alpha_i + \gamma_j) - \Phi(\beta'X^{(0)}_{ij} + \alpha_i + \gamma_j)$.
The results are contained in the right panel of Table \ref{tab:gravity}. In general, the magnitude of the incidental parameters bias appears smaller when considering marginal effects, this has been noted elsewhere, for example \citealp{Fernandez-Val2009}.

Finally, I also conduct tests for the presence of transitivity. Here I compute tests for the trade network discussed above, 
as well as for a professional network of attorneys from \citet{lazega2001}. In this network the outcome is an indicator for friendship between 71 attorneys at a law firm; transitivity is often suggested to play an important role in social networks such as this. 
The network is sparser than the trade one, with only 11 per cent of all potential friendships 
forming. I include covariates for whether the attorneys are of the same gender, professional status, or work in the same office, as well as measures of 
the difference in their ages and tenure. I compute transitivity test statistics based on (\ref{eq:stat_T}) using $S_{ij} = \frac{1}{N-2}\sum_{k\not\in \{i,j\}} Y_{ik} Y_{kj}$. 
Following \citet{Dzemski2019}, I compute standard errors for the statistics using a parametric bootstrap procedure, although using $\widehat{V}_\Delta$ gives identical conclusions. The results of the transitivity tests are presented in Table \ref{tab:statistic} -- the table presents t-statistics (ratio of $T$ to its bootstrap standard error) that are asymptotically standard normal and so can be compared against the usual critical values. 

\renewcommand{\arraystretch}{1.3}
\begin{table}[H]
\centering \small
\caption{\label{tab:statistic}Transitivity test statistics}
\begin{tabular}{lrr}
\hline 
 & \multicolumn{1}{c}{MLE} & \multicolumn{1}{c}{Jackknife} \tabularnewline
\hline 
trade network    & -8.85 & 0.70 \\
professional network  & 14.90 & 21.03 \\
\hline 
\end{tabular}
\end{table}

In both networks, jackknife bias correction has a large and important impact on the size of the t-statistics. As seen in the simulations, the more complicated average effects that make up the test statistics appear to suffer from the incidental parameters bias to a much greater extent than the simple average effects. In the 
case of the trade network, the uncorrected test statistic would suggest a strong rejection of the null hypothesis of a dyadic network formation model. 
In contrast, the bias-corrected t-statistic is close to zero, and suggests that the dyadic model does a good job of capturing the 
observed transitivity. For the professional network the statistic becomes larger after bias correction, so that we do reject 
the dyadic model. This aligns with common findings that social networks tend to exhibit high degrees of transitivity and clustering.


\section{Conclusion}

This paper presents a new method for bias correcting nonlinear dyadic
network models with fixed effects. I provide a novel formulation
of the jackknife method that applies to networks with both sender
and receiver fixed effects. The jackknife method provides an `off-the-shelf'
procedure for bias correction that is easy to apply, and applicable
to a wide set of models. It allows for discrete multivalued and continuous
outcome variables, and is able to obtain estimates of average effects
and counterfactual outcomes. In simulations, I show that the jackknife performs well, even in
relatively low density networks, and outperforms previous suggestions
for jackknife procedures.

In addition, I show how the jackknife can be used to bias correct
averages of functions that depend on multiple observations, including
dyads, triads, and tetrads in the network. These averages can be used
to produce a wide array of test statistics, including tests for the presence of strategic
interactions in the network, such as reciprocity or transitivity.

\bibliographystyle{myplainnat} 
\bibliography{Network}  

@article{Hirano2017,
	author = {Keisuke Hirano and Jonathan H. Wright},
	doi = {10.3982/ECTA13372},
	issn = {0012-9682},
	issue = {2},
	journal = {Econometrica},
	number = {2},
	pages = {617-643},
	title = {Forecasting With Model Uncertainty: Representations and Risk Reduction},
	volume = {85},
	year = {2017}}

@article{Santos2006,
	author = {J. M. C. Santos Silva and Silvana Tenreyro},
	doi = {10.1162/rest.88.4.641},
	issn = {0034-6535},
	issue = {4},
	journal = {The Review of Economics and Statistics},
	month = {11},
	number = {4},
	pages = {641-658},
	title = {The Log of Gravity},
	volume = {88},
	year = {2006}}

@book{VanderVaart1998,
	author = {A. W. van der Vaart},
	doi = {10.1017/CBO9780511802256},
	isbn = {9780511802256},
	month = {10},
	publisher = {Cambridge University Press},
	title = {Asymptotic Statistics},
	url = {https://www.cambridge.org/core/product/identifier/9780511802256/type/book},
	year = {1998}}

@article{Hahn2004,
	author = {Jinyong Hahn and Whitney Newey},
	doi = {10.1111/j.1468-0262.2004.00533.x},
	issn = {0012-9682},
	issue = {4},
	journal = {Econometrica},
	month = {7},
	number = {4},
	pages = {1295-1319},
	title = {Jackknife and Analytical Bias Reduction for Nonlinear Panel Models},
	url = {http://www.blackwell-synergy.com/links/doi/10.1111%2Fj.1468-0262.2004.00533.x},
	volume = {72},
	year = {2004}}

@article{Newey1994a,
	author = {Whitney K. Newey and Daniel McFadden},
	doi = {10.1016/S1573-4412(05)80005-4},
	journal = {Handbook of Econometrics},
	pages = {2111-2245},
	title = {Large sample estimation and hypothesis testing},
	url = {https://linkinghub.elsevier.com/retrieve/pii/S1573441205800054},
	volume = {4},
	year = {1994}}

@article{dePaula2020,
	author = {{\'A}ureo de Paula},
	doi = {10.1146/annurev-economics-093019-113859},
	issn = {1941-1383},
	issue = {1},
	journal = {Annual Review of Economics},
	month = {8},
	pages = {775-799},
	title = {Econometric Models of Network Formation},
	url = {https://www.annualreviews.org/doi/10.1146/annurev-economics-093019-113859},
	volume = {12},
	year = {2020}}

@article{Neyman1948,
	author = {J. Neyman and Elizabeth L. Scott},
	doi = {10.2307/1914288},
	issn = {00129682},
	journal = {Econometrica},
	number = {1},
	pages = {1-32},
	title = {Consistent Estimates Based on Partially Consistent Observations},
	volume = {16},
	year = {1948}}

@article{Dhaene2015,
	author = {Geert Dhaene and Koen Jochmans},
	doi = {10.1093/restud/rdv007},
	issn = {0034-6527},
	issue = {3},
	journal = {The Review of Economic Studies},
	month = {7},
	number = {3},
	pages = {991-1030},
	title = {Split-panel Jackknife Estimation of Fixed-effect Models},
	url = {https://academic.oup.com/restud/article-lookup/doi/10.1093/restud/rdv007},
	volume = {82},
	year = {2015}}

@article{Fernandez-Val2009,
	author = {Iv{\'a}n Fern{\'a}ndez-Val},
	doi = {10.1016/j.jeconom.2009.02.007},
	issn = {03044076},
	issue = {1},
	journal = {Journal of Econometrics},
	month = {5},
	number = {1},
	pages = {71-85},
	title = {Fixed effects estimation of structural parameters and marginal effects in panel probit models},
	url = {https://linkinghub.elsevier.com/retrieve/pii/S0304407609000463},
	volume = {150},
	year = {2009}}

@article{Burger2009,
	author = {Martijn Burger and Frank van Oort and Gert-Jan Linders},
	doi = {10.1080/17421770902834327},
	issn = {1742-1772},
	issue = {2},
	journal = {Spatial Economic Analysis},
	month = {6},
	pages = {167-190},
	title = {On the Specification of the Gravity Model of Trade: Zeros, Excess Zeros and Zero-inflated Estimation},
	url = {https://www.tandfonline.com/doi/full/10.1080/17421770902834327},
	volume = {4},
	year = {2009}}

@article{Anderson2003,
	author = {James E Anderson and Eric van Wincoop},
	doi = {10.1257/000282803321455214},
	issn = {0002-8282},
	issue = {1},
	journal = {American Economic Review},
	month = {2},
	number = {1},
	pages = {170-192},
	title = {Gravity with Gravitas: A Solution to the Border Puzzle},
	url = {https://pubs.aeaweb.org/doi/10.1257/000282803321455214},
	volume = {93},
	year = {2003}}

@article{Gao2020,
	author = {Wayne Yuan Gao},
	doi = {10.1016/j.jeconom.2019.09.005},
	issn = {03044076},
	issue = {2},
	journal = {Journal of Econometrics},
	month = {4},
	number = {2},
	pages = {399-413},
	title = {Nonparametric identification in index models of link formation},
	url = {https://linkinghub.elsevier.com/retrieve/pii/S0304407619302039},
	volume = {215},
	year = {2020}}

@article{Holland1981,
	author = {Paul W. Holland and Samuel Leinhardt},
	doi = {10.2307/2287037},
	issn = {01621459},
	issue = {373},
	journal = {Journal of the American Statistical Association},
	month = {3},
	number = {373},
	pages = {33-50},
	title = {An Exponential Family of Probability Distributions for Directed Graphs},
	url = {https://www.jstor.org/stable/2287037?origin=crossref},
	volume = {76},
	year = {1981}}

@article{Dzemski2019,
	author = {Andreas Dzemski},
	doi = {10.1162/rest_a_00805},
	issn = {0034-6535},
	issue = {5},
	journal = {The Review of Economics and Statistics},
	month = {12},
	number = {5},
	pages = {763-776},
	title = {An Empirical Model of Dyadic Link Formation in a Network with Unobserved Heterogeneity},
	url = {https://direct.mit.edu/rest/article/101/5/763/58549/An-Empirical-Model-of-Dyadic-Link-Formation-in-a},
	volume = {101},
	year = {2019}}

@article{Fernandez-Val2016,
	author = {Iv{\'a}n Fern{\'a}ndez-Val and Martin Weidner},
	doi = {10.1016/j.jeconom.2015.12.014},
	issn = {03044076},
	issue = {1},
	journal = {Journal of Econometrics},
	month = {5},
	number = {1},
	pages = {291-312},
	title = {Individual and time effects in nonlinear panel models with large {N,T}},
	url = {https://linkinghub.elsevier.com/retrieve/pii/S0304407615002997},
	volume = {192},
	year = {2016}}

@article{Candelaria2020,
	author = {Luis E. Candelaria},
	journal = {arXiv:2007.05403},
	month = {7},
	title = {A Semiparametric Network Formation Model with Unobserved Linear Heterogeneity},
	url = {http://arxiv.org/abs/2007.05403},
	year = {2020}}

@article{Graham2020a,
	author = {Bryan S. Graham},
	doi = {10.1016/bs.hoe.2020.05.001},
	journal = {Handbook of Econometrics},
	pages = {111-218},
	title = {Network data},
	url = {https://linkinghub.elsevier.com/retrieve/pii/S1573441220300015},
	volume = {7},
	year = {2020}}

@article{Graham2017,
	author = {Bryan S. Graham},
	doi = {10.3982/ECTA12679},
	issn = {0012-9682},
	issue = {4},
	journal = {Econometrica},
	number = {4},
	pages = {1033-1063},
	title = {An Econometric Model of Network Formation With Degree Heterogeneity},
	url = {https://www.econometricsociety.org/doi/10.3982/ECTA12679},
	volume = {85},
	year = {2017}}

@article{Jochmans2018,
	author = {Koen Jochmans},
	doi = {10.1080/07350015.2017.1286242},
	issn = {0735-0015},
	issue = {4},
	journal = {Journal of Business {\&} Economic Statistics},
	month = {10},
	number = {4},
	pages = {705-713},
	title = {Semiparametric Analysis of Network Formation},
	url = {https://www.tandfonline.com/doi/full/10.1080/07350015.2017.1286242},
	volume = {36},
	year = {2018}}

@article{Helpman2008,
	author = {Elhanan Helpman and Marc Melitz and Yona Rubinstein},
	doi = {10.1162/qjec.2008.123.2.441},
	issn = {0033-5533},
	issue = {2},
	journal = {Quarterly Journal of Economics},
	month = {5},
	number = {2},
	pages = {441-487},
	title = {Estimating Trade Flows: Trading Partners and Trading Volumes},
	url = {https://academic.oup.com/qje/article-lookup/doi/10.1162/qjec.2008.123.2.441},
	volume = {123},
	year = {2008}}

@article{Charbonneau2017,
	author = {Karyne B. Charbonneau},
	doi = {10.1111/ectj.12093},
	issn = {1368-4221},
	issue = {3},
	journal = {The Econometrics Journal},
	month = {10},
	number = {3},
	pages = {S1-S13},
	title = {Multiple fixed effects in binary response panel data models},
	url = {https://academic.oup.com/ectj/article/20/3/S1/5056388},
	volume = {20},
	year = {2017}}

@article{Yan2019,
	author = {Ting Yan and Binyan Jiang and Stephen E. Fienberg and Chenlei Leng},
	doi = {10.1080/01621459.2018.1448829},
	issn = {0162-1459},
	issue = {526},
	journal = {Journal of the American Statistical Association},
	month = {4},
	number = {526},
	pages = {857-868},
	title = {Statistical Inference in a Directed Network Model With Covariates},
	url = {https://www.tandfonline.com/doi/full/10.1080/01621459.2018.1448829},
	volume = {114},
	year = {2019}}

@article{Bloch2007,
	author = {Francis Bloch and Matthew O. Jackson},
	doi = {10.1016/j.jet.2005.10.003},
	issn = {00220531},
	issue = {1},
	journal = {Journal of Economic Theory},
	month = {3},
	number = {1},
	pages = {83-110},
	title = {The formation of networks with transfers among players},
	url = {https://linkinghub.elsevier.com/retrieve/pii/S0022053105002413},
	volume = {133},
	year = {2007}}

@article{Cruz-Gonzalez2017,
	author = {Mario Cruz-Gonzalez and Iv{\'a}n Fern{\'a}ndez-Val and Martin Weidner},
	doi = {10.1177/1536867X1701700301},
	issn = {1536-867X},
	issue = {3},
	journal = {The Stata Journal: Promoting communications on statistics and Stata},
	month = {9},
	number = {3},
	pages = {517-545},
	title = {Bias Corrections for Probit and Logit Models with Two-way Fixed Effects},
	url = {http://journals.sagepub.com/doi/10.1177/1536867X1701700301},
	volume = {17},
	year = {2017}}

@article{White1982,
	author = {Halbert White},
	doi = {10.2307/1912526},
	issn = {00129682},
	issue = {1},
	journal = {Econometrica},
	month = {1},
	number = {1},
	pages = {1-25},
	title = {Maximum Likelihood Estimation of Misspecified Models},
	url = {https://www.jstor.org/stable/1912526?origin=crossref},
	volume = {50},
	year = {1982}}

@article{Pratt1981,
	author = {John W. Pratt},
	doi = {10.2307/2287052},
	issn = {01621459},
	issue = {373},
	journal = {Journal of the American Statistical Association},
	month = {3},
	number = {373},
	pages = {103-106},
	title = {Concavity of the Log Likelihood},
	url = {https://www.jstor.org/stable/2287052?origin=crossref},
	volume = {76},
	year = {1981}}

@article{Chen2021,
	author = {Mingli Chen and Iv{\'a}n Fern{\'a}ndez-Val and Martin Weidner},
	doi = {10.1016/j.jeconom.2020.04.004},
	issn = {03044076},
	issue = {2},
	journal = {Journal of Econometrics},
	month = {2},
	number = {2},
	pages = {296-324},
	title = {Nonlinear factor models for network and panel data},
	url = {https://linkinghub.elsevier.com/retrieve/pii/S0304407620301238},
	volume = {220},
	year = {2021}}

@unpublished{Toth2017,
	author = {Peter Toth},
	note = {SSRN 2988698},
	title = {Semiparametric Estimation in Network Formation Models with Homophily and Degree Heterogeneity},
	year = {2017}}

@unpublished{Graham2020d,
	author = {Andrin Pelican and Bryan Graham},
	city = {Cambridge, MA},
	doi = {10.3386/w27793},
	institution = {National Bureau of Economic Research},
	month = {9},
	note = {National Bureau of Economic Research, Working Paper 27793},
	title = {An Optimal Test for Strategic Interaction in Social and Economic Network Formation between Heterogeneous Agents},
	url = {http://www.nber.org/papers/w27793.pdf},
	year = {2020}}

@unpublished{Zeleneev2020,
	author = {Andrei Zeleneev},
	note = {Working Paper},
	title = {Identification and Estimation of Network Models with Nonparameteric Unobserved Heterogeneity},
	year = {2020}}

@book{Jackson2008,
	author = {Matthew O. Jackson},
	doi = {10.2307/j.ctvcm4gh1},
	isbn = {9781400833993},
	month = {11},
	publisher = {Princeton University Press},
	title = {Social and Economic Networks},
	url = {http://www.jstor.org/stable/10.2307/j.ctvcm4gh1},
	year = {2010}}

@book{lazega2001,
	author = {Emmanuel Lazega},
	doi = {10.1093/acprof:oso/9780199242726.001.0001},
	isbn = {9780199242726},
	issue = {3},
	journal = {Administrative Science Quarterly},
	month = {9},
	pages = {525-529},
	publisher = {Oxford University Press},
	title = {The Collegial Phenomenon},
	url = {https://academic.oup.com/book/6928},
	volume = {48},
	year = {2001}}

@article{Hahn2024,
	author = {Jinyong Hahn and David W. Hughes and Guido Kuersteiner and Whitney K. Newey},
	doi = {10.3982/QE2350},
	issn = {1759-7323},
	issue = {3},
	journal = {Quantitative Economics},
	number = {3},
	pages = {783-816},
	title = {Efficient bias correction for cross‐section and panel data},
	volume = {15},
	year = {2024}}

@article{Auerbach2022,
	author = {Eric Auerbach},
	doi = {10.3982/ECTA18093},
	issn = {0012-9682},
	issue = {3},
	journal = {Econometrica},
	number = {3},
	pages = {1205-1223},
	title = {Testing for Differences in Stochastic Network Structure},
	volume = {90},
	year = {2022}}

@article{Fernandez-Val2018,
	author = {Iv{\'a}n Fern{\'a}ndez-Val and Martin Weidner},
	doi = {10.1146/annurev-economics-080217-053542},
	issn = {1941-1383},
	issue = {1},
	journal = {Annual Review of Economics},
	month = {8},
	number = {1},
	pages = {109-138},
	title = {Fixed Effects Estimation of Large-{T} Panel Data Models},
	volume = {10},
	year = {2018}}

\newpage

\appendix
\let\clearpage\relax


\section{Asymptotic expansions}
The results in this paper are based on asymptotic expansions of the
objective function, following those derived by \cite{Fernandez-Val2016} (FW16), but 
extended to higher order.\footnote{In the supplementary appendix it is shown that the 
	expansions in FW16 apply to the model studied here; this follows  \cite{Dzemski2019}, who verified 
	the conditions for a dyadic probit model.}
Derivations of higher-order terms and their bounds are quite long
and largely similar to the derivations in FW16, and so are provided
in the supplementary appendix, which is available at \url{https://arxiv.org/abs/2203.15603}. This appendix contains analysis based
on the first-order expansions and focuses on results related to the
jackknife, which are of most interest.

\subsection*{Notation}

Here we collect the notation used throughout the paper and in the
appendix for easy reference.
Directed pairs of agents (observations) are referred to as $(i,j)$,
while unordered pairs (dyads) are referred to using $\{i,j\}$. We use $\ell_{ij}$ as shorthand for the function $\ell(Z_{ij},\beta_{0},\pi_{0,ij})$,
with $Z_{ij}=(Y_{ij},X_{ij})$ and $\pi_{0,ij}=\alpha_{0i}+\gamma_{0j}$. 

Partial derivatives of the objective function are shown using subscripts,
so that $\partial_{\beta}\mathcal{L}(\beta,\phi)$ denotes $\partial\mathcal{L}(\beta,\phi)/\partial\beta$
and so on, where $\phi'=(\alpha',\gamma')$. When functions are evaluated
at $\beta_{0},\phi_{0}$ the dependence on these arguments is dropped.
We use a bar to refer to expectations conditional on covariates and fixed effects,
e.g. $\partial_{\beta}\bar{\mathcal{L}}=\bar{E}[\partial_{\beta}\mathcal{L}(\beta_{0},\phi_{0})]$,
and a tilde for deviations from these expectations, $\partial_{\beta}\tilde{\mathcal{L}}=\partial_{\beta}\mathcal{L}-\partial_{\beta}\bar{\mathcal{L}}$.

Let $\mathcal{S}(\beta,\phi)=\partial_{\phi}\mathcal{L}(\beta,\phi)$,
and $\mathcal{S}_{\beta}(\beta,\phi)=\partial_{\beta}\mathcal{L}(\beta,\phi)$
be the first derivatives of the objective function. We also write
$\mathcal{H}(\beta,\phi)=-\partial_{\phi\phi'}\mathcal{L}(\beta,\phi)$
for the negative of the Hessian with respect to the fixed effects,
a $2N\times2N$ matrix. We use $\lVert\cdot\rVert$ to refer to the
Euclidean norm for 
vectors (norms for higher dimensional objects are discussed in the
supplementary appendix).

For describing fixed effect averages, we let $\lambda$ stand for
a set of $r$ observations that depends on $p$ ordered agents,with
$\Lambda_{N}$ the set of all $\lambda$ that can be formed by
selecting $p$ ordered agents from the $N$ total agents. We let
$m_{\lambda}$ denote a function that depends on the observations
in set $\lambda$, and $Z_{\lambda}$ be the vector containing $Z_{ij}$
for each $(i,j)\in\lambda$.

\subsection{Asymptotic expansion}

The next lemma gives the asymptotic expansion for the estimated fixed
effects and common parameters. It is largely a restatement of Theorem
B.1 and B.3 in \cite{Fernandez-Val2016} (here the remainder term
is split into two parts, which are dealt with separately in the jackknife proof). 
The full expressions for the remainder terms are provided in Supplementary Appendix E. 
\begin{lem}
	\label{lem:exp_beta1}Let Assumptions 1 and 2 hold. Then,
	\[
	N\bar{W}_{N}(\widehat{\beta}-\beta_{0})=U^{(0)}+U^{(1)}+R_{\beta}+\tilde{R}_{\beta}
	\]
	where
	\begin{align*}
		U^{(0)} & =(\partial_{\beta}\mathcal{L})+(\partial_{\beta\phi'}\bar{\mathcal{L}})\bar{\mathcal{H}}^{-1}\mathcal{S}\\
		U^{(1)} & =(\partial_{\beta\phi'}\tilde{\mathcal{L}})\bar{\mathcal{H}}^{-1}\mathcal{S}-(\partial_{\beta\phi'}\bar{\mathcal{L}})\bar{\mathcal{H}}^{-1}\tilde{\mathcal{H}}\bar{\mathcal{H}}^{-1}\mathcal{S}\\
		& +\frac{1}{2}\sum_{g=1}^{\dim\phi}\big((\partial_{\beta\phi\phi_{g}}\bar{\mathcal{L}})+(\partial_{\beta\phi'}\bar{\mathcal{L}})\bar{\mathcal{H}}^{-1}(\partial_{\phi\phi'\phi_{g}}\bar{\mathcal{L}})\big)\big[\bar{\mathcal{H}}^{-1}\mathcal{S}\big]_{g}\bar{\mathcal{H}}^{-1}\mathcal{S}
	\end{align*}
	and the remainders satisfy $\lVert R_{\beta}\rVert=o_{p}(1)$
	and $\lVert\tilde{R}_{\beta}\rVert=o_{p}(N^{-1})$.
	Further, the parameter estimates are consistent, i.e.
	\[
	\lVert\widehat{\phi}-\phi_{0}\rVert  \overset{p}{\longrightarrow} 0 \qquad
	\lVert\widehat{\beta}-\beta_{0}\rVert \overset{p}{\longrightarrow} 0
	\]
\end{lem}
%


\section{\label{sec:App_B}Jackknife results for common parameter}

Here we allow for a more general construction of the leave-out sets,
but impose two important conditions.
\begin{condition}
	\label{cond:losets}Let $\mathcal{I}_{k}$ for $k=1,\dots N_{k}$
	be a partition of the observations in a network of size $N$, that is fixed and not dependent 
	on the data. Define
	$1_{ij}^{k}=1\{(i,j)\not\in\mathcal{I}_{k}\}$ as an indicator that
	the observation $(i,j)$ is included in the $k$-th leave-out sample.
	We impose the following conditions:
	
	(i) $\sum_{k=1}^{N_{k}}(1-1_{ij}^{k})=1$ for all $(i,j)$
	
	(ii) $\sum_{j\ne i}(1-1_{ij}^{k})=\sum_{i\ne j}(1-1_{ij}^{k})=l$
	for all $i$, $j$, and $k$, for some fixed $l$.
\end{condition}
Condition \ref{cond:losets} imposes two important constraints on the
sets $\mathcal{I}_{k}$. First, that every observation $(i,j)$ appears in exactly one of the $\mathcal{I}_k$, 
so that all observations are used equally in the jackknife, which is important for showing that
the asymptotic variance of the estimator is not affected by the jackknife.
Second, that each $\mathcal{I}_k$ contains exactly $l$ observations related to each fixed effect, 
so that they are affected
equally in the leave-out sets. We assume in the proofs below
that $l=1$, i.e. that we are using the leave-one-out jackknife as
this simplifies notation. All of the results still hold for the leave-$l$-out
style jackknife for fixed $l$. Note that we consider the partitioning to be fixed and independent 
of the data, and condition on it throughout.

Let $A$ be some function of the data and model parameters, 
and $A_{(k)}$ the same quantity in the $k$-th leave-out sample. We define the jackknife operator
$\mathbf{J}$ as 
\[
\mathbf{J}[A]=A_{J}=(N-1)A-(N-2)\frac{1}{N-1}\sum_{k=1}^{N-1}A_{(k)}
\]
Additionally, we define a set of indicators that count the number
of unique sets $\mathcal{I}_{k}$ that a group of observations
$(i_{1},j_{1}),\dots,(i_{t},j_{t})$ are contained in. Let 
\[
I_{(i_{1},j_{1}),\dots,(i_{t},j_{t})}^{r}=\begin{cases}
	1 & (i_{1},j_{1}),\dots,(i_{t},j_{t})\text{ span exactly \ensuremath{r} of the sets }\mathcal{I}_{k}\\
	0 & \text{otherwise}
\end{cases}
\]

We next state an expansion for the leave-out estimate $\widehat{\beta}_{(k)}$ (see Supplementary Appendix I for full details of all terms). 
The expansion is the same as in Lemma \ref{lem:exp_beta1}, after redefining $\ell_{ij}$ in the objective 
function to include indicator variables for the dropped observations. Specifically, let $\ell^{(k)}(Y_{ij},X_{ij},\beta,\alpha_{i}+\gamma_{j})=\ell(Y_{ij},X_{ij},\beta,\alpha_{i}+\gamma_{j})1_{ij}^{k}\frac{N-1}{N-2}$, and $\mathcal{L}_{(k)}$ the correspesponding objective function.
The new objective function satisfies the conditions of Assumptions
1 and 2 whenever the full-sample version does (e.g.\ we have independence of $\ell_{ij}^{(k)}$
(and its derivatives) from $\ell_{st}^{(k)}$ (and its derivatives),
conditional on $(X,\phi)$ for all $(i,j)\ne(s,t),(t,s)$). We use the subscript $(k)$ to denote functions of the $k$-th
leaveout sample, e.g.\ $\partial_{\beta\phi'}\bar{\mathcal{L}}_{(k)}=\bar{E}[\partial_{\beta\phi'}\mathcal{L}_{(k)}]$.
\begin{lem}
	\label{lem:exp_beta1_k}Let Contdition \ref{cond:losets} and Assumptions \ref{assu:network} and \ref{assu:beta_jack} hold. Then,
	\[
	N\overline{W}_{N,(k)}(\widehat{\beta}_{(k)}-\beta_{0})=U_{(k)}^{(0)}+U_{(k)}^{(1)}+R_{\beta,(k)}+\tilde{R}_{\beta,(k)}
	\]
	where
	\begin{align*}
		U_{(k)}^{(0)} & =(\partial_{\beta}\mathcal{L}_{(k)})+(\partial_{\beta\phi'}\bar{\mathcal{L}})\bar{\mathcal{H}}^{-1}\mathcal{S}_{(k)}\\
		U_{(k)}^{(1)} & =(\partial_{\beta\phi'}\tilde{\mathcal{L}}_{(k)})\bar{\mathcal{H}}^{-1}\mathcal{S}_{(k)}-(\partial_{\beta\phi'}\bar{\mathcal{L}})\bar{\mathcal{H}}^{-1}\tilde{\mathcal{H}}_{(k)}\bar{\mathcal{H}}^{-1}\mathcal{S}_{(k)}\\
		& +\frac{1}{2}\sum_{g=1}^{\dim\phi}\big((\partial_{\beta\phi\phi_{g}}\bar{\mathcal{L}})+(\partial_{\beta\phi'}\bar{\mathcal{L}})\bar{\mathcal{H}}^{-1}(\partial_{\phi\phi'\phi_{g}}\bar{\mathcal{L}})\big)\big[\bar{\mathcal{H}}^{-1}\mathcal{S}_{(k)}\big]_{g}\bar{\mathcal{H}}^{-1}\mathcal{S}_{(k)}
	\end{align*}
	and the remainders satisfy $\lVert R_{\beta,(k)}\rVert=o_{p}(1)$
	and $\lVert\tilde{R}_{\beta,(k)}\rVert=o_{p}(N^{-1})$.
\end{lem}

\subsection{Jackknifing first-order expansion terms}

The next lemmas demonstrate the effect of the jackknife on general sums of the forms in $U^{(0)}$ and $U^{(1)}$.
\begin{lem}
	\label{lem:jackS} Let $1_{ij}^{k}$ satisfy Condition \ref{cond:losets}.
	For $A_{ij}$ a mean-zero random variable, let
	\[
	A = \frac{1}{N-1}\sum_i\sum_{s\ne i}A_{is},\qquad A_{(k)} = \frac{1}{N-2}\sum_i\sum_{s\ne i}A_{is}1_{is}^{k}
	\]
	Then $A_{J}=\mathbf{J}[A]=A$.
\end{lem}
\begin{proof}
	By Condition \ref{cond:losets}, $\sum_{k}1_{is}^{k}=N-2$ so that 
	\[
	\frac{1}{N-1}\sum_k A_{(k)} = \frac{1}{(N-1)(N-2)}\sum_k \sum_{i}\sum_{s\ne i}A_{is}1_{is}^{k}=\frac{1}{N-1}\sum_i\sum_{s\ne i}A_{is}=A,
	\]
	from which the result follows.
\end{proof}
\begin{lem}
	\label{lem:jack_r0} Let $1_{ij}^{k}$ satisfy Condition \ref{cond:losets}.
	For $A_{ij}$ a mean-zero random variable with bounded fourth moment,
	let
	\begin{align*}
		A & =\frac{1}{N-1}\big(\{\sum_{s\ne i}A_{is}\}_{i=1,\dots,N},\{\sum_{s\ne j}A_{sj}\}_{j=1,\dots,N}\big)
		=(A_{\alpha},A_{\gamma})\\
		A_{k} & =\frac{1}{N-2}\big(\{\sum_{s\ne i}A_{is}1_{is}^{k}\}_{i=1,\dots,N},\{\sum_{s\ne j}A_{sj}1_{sj}^{k}\}_{j=1,\dots,N}\big)
		=(A_{\alpha,k},A_{\gamma,k})
	\end{align*}
	and let $B$ and $B_{k}$ be similarly defined vectors
	of sums involving mean-zero random variables $B_{ij}$. Assume that
	$(A_{ij},B_{ij})$ are independent of $(A_{st},B_{st})$ for $(i,j)\not\in\{(s,t),(t,s)\}$.
	Let $M$ be a non-random $2N \times 2N$ matrix that has $O_{p}(1)$ elements on
	its diagonal and $O_{p}(N^{-1})$ off-diagonal terms. Then,
	\[
	\mathbf{J}[A'MB]= o_{p}(1)
	\]

\end{lem}
\begin{proof}
	We can decompose the term $A'MB$ as
	\[
	A'MB  =A_{\alpha}'M_{\alpha\alpha}B_{\alpha}+A_{\alpha}'M_{\alpha\gamma}B_{\gamma}
	+A_{\gamma}'M_{\gamma\alpha}B_{\alpha}+A_{\gamma}'M_{\gamma\gamma}B_{\gamma}
	\]
	where $M_{\alpha\alpha}$ is the upper left $N\times N$ block of $M$, $M_{\alpha\gamma}$ the upper right block, and so on. Recall that $I_{(ij)(st)}^{1}$ is one whenever $(i,j)$ and $(s,t)$
	are contained in the same $\mathcal{I}_{k}$, and so $\sum_{k}1_{ij}^{k}1_{st}^{k}=(N-2)I_{(ij)(st)}^{1}+(N-3)(1-I_{(ij)(st)}^{1})$. For the first term we have
	\begin{align*}
		A_{\alpha}'M_{\alpha\alpha}B_{\alpha} & =\sum_{i,j}A_{\alpha,i}(M_{\alpha\alpha})_{ij}B_{\alpha,j}\\
		& =\frac{1}{(N-1)^{2}}\sum_{i,j}\sum_{s\ne i}\sum_{t\ne j}(M_{\alpha\alpha})_{ij}A_{is}B_{jt}\\
		\frac{1}{N-1}\sum_{k}A_{k,\alpha}'M_{\alpha\alpha}B_{k,\alpha} & =\frac{1}{(N-1)(N-2)^{2}}\sum_{k}\sum_{i,j}\sum_{s\ne i}\sum_{t\ne j}(M_{\alpha\alpha})_{ij}A_{is}B_{jt}1_{is}^{k}1_{jt}^{k}\\
		& =\frac{1}{(N-1)(N-2)}\sum_{i,j}\sum_{s\ne i}\sum_{t\ne j}(M_{\alpha\alpha})_{ij}A_{is}B_{jt}I_{(is)(jt)}^{1}\\
		& +\frac{N-3}{(N-1)(N-2)^{2}}\sum_{i,j}\sum_{s\ne i}\sum_{t\ne j}(M_{\alpha\alpha})_{ij}A_{is}B_{jt}(1-I_{(is)(jt)}^{1})
	\end{align*}
	We can write the jackknife version of this first term is
	\begin{align*}
		 \mathbf{J}[A_\alpha'M_{\alpha\alpha}B_\alpha] = & (N-1)A_{\alpha}'M_{\alpha\alpha}B_{\alpha}-\frac{N-2}{N-1}\sum_{k}A_{\alpha,k}'M_{\alpha\alpha}B_{\alpha,k}\\
		= & \frac{1}{(N-1)(N-2)}\sum_{i}\sum_{j}\sum_{s\ne i}\sum_{t\ne j}(M_{\alpha\alpha})_{ij}\big(A_{is}B_{jt}\big)(1-I_{(is)(jt)}^{1})
	\end{align*}
	Similar computations for the other three elements gives
	\begin{align*}
		 \mathbf{J}[A'MB]
		= & \frac{1}{(N-1)(N-2)}\sum_{i}\sum_{j}\sum_{s\ne i}\sum_{t\ne j}\Gamma_{ijst}A_{is}B_{jt}(1-I_{(is)(jt)}^{1})
	\end{align*}
	where $\Gamma_{ijst}=(M_{\alpha\alpha})_{ij}+(M_{\alpha\gamma})_{it}+(M_{\gamma\alpha})_{sj}+(M_{\gamma\gamma})_{st}$.
	From the properties of $M$,  $\Gamma_{ijst}$ is $O_{p}(1)$ only when $i=j$ or $s=t$ and is $O_{p}(N)$ otherwise. 
	We can decompose  $\mathbf{J}[A'MB]$ as
	\begin{align*}
		 \mathbf{J}[A'MB] & =\frac{1}{(N-1)(N-2)}\sum_{i}\sum_{j}\Big(\sum_{s<i}\sum_{t<j}\Gamma_{ijst}A_{is}B_{jt}+\sum_{s<i}\sum_{t>j}\Gamma_{ijst}A_{is}B_{jt}\\
		& +\sum_{s>i}\sum_{t<j}\Gamma_{ijst}A_{is}B_{jt}+\sum_{s>i}\sum_{t>j}\Gamma_{ijst}A_{is}B_{jt}\Big)(1-I_{(is)(jt)}^{1})\\
		& =\mathcal{J}_{1}+\mathcal{J}_{2}+\mathcal{J}_{3}+\mathcal{J}_{4}
	\end{align*}
	Then we have
	\begin{align*}
		\bar{E}[\mathcal{J}_{1}^{2}] & =\frac{1}{(N-1)^{2}(N-2)^{2}}\sum_{i}\sum_{j}\sum_{k}\sum_{l}\sum_{s<i}\sum_{t<j}\sum_{p<k}\sum_{q<l}\Gamma_{ijst}\Gamma_{klpq}\\
		& \times E\big[A_{is}B_{jt}A_{kp}B_{lq}\big](1-I_{(is)(jt)}^{1})(1-I_{(kp)(lq)}^{1})\\
		& =\frac{1}{(N-1)^{2}(N-2)^{2}}\sum_{i}\sum_{j}\sum_{s<i}\sum_{t<j}\Gamma_{ijst}^{2}\bar{E}[A_{is}^{2}B_{jt}^{2}](1-I_{(is)(jt)}^{1})\\
		& +\frac{1}{(N-1)^{2}(N-2)^{2}}\sum_{i}\sum_{j}\sum_{s<i}\sum_{t<j}\Gamma_{ijst}\Gamma_{stij}\bar{E}[A_{is}B_{is}A_{jt}B_{jt}\big](1-I_{(is)(jt)}^{1})\\
		& =\frac{1}{(N-1)^{2}(N-2)^{2}}\sum_{i}\sum_{s<i}\sum_{t<i}\Gamma_{iist}^{2}\bar{E}[A_{is}^{2}B_{it}^{2}](1-I_{(is)(it)}^{1})\\
		& +\frac{1}{(N-1)^{2}(N-2)^{2}}\sum_{i}\sum_{j\ne i}\sum_{s<(i\wedge j)}\Gamma_{ijss}^{2}\bar{E}[A_{is}^{2}]\bar{E}[B_{js}^{2}](1-I_{(is)(js)}^{1})\\
		& +\frac{1}{N^{2}(N-1)^{2}(N-2)^{2}}\sum_{i}\sum_{j\ne i}\sum_{s<i}\sum_{t<j,t\ne s}N^{2}\Gamma_{ijst}^{2}\bar{E}[A_{is}^{2}]\bar{E}[B_{jt}^{2}](1-I_{(is)(jt)}^{1})\\
		& +\frac{1}{(N-1)^{2}(N-2)^{2}}\sum_{i}\sum_{s<i}\sum_{t<i}\Gamma_{iist}\Gamma_{stii}\bar{E}[A_{is}B_{is}]\bar{E}[A_{it}B_{it}\big](1-I_{(is)(it)}^{1})\\
		& +\frac{1}{(N-1)^{2}(N-2)^{2}}\sum_{i}\sum_{j\ne i}\sum_{s<(i\wedge j)}\Gamma_{ijss}\Gamma_{ssij}\bar{E}[A_{is}B_{is}]\bar{E}[A_{js}B_{js}\big](1-I_{(is)(js)}^{1})\\
		& +\frac{1}{N^{2}(N-1)^{2}(N-2)^{2}}\sum_{i}\sum_{j\ne i}\sum_{s<i}\sum_{t<j,t\ne s}N^{2}\Gamma_{ijst}\Gamma_{stij}\bar{E}[A_{is}B_{is}]\bar{E}[A_{jt}B_{jt}\big](1-I_{(is)(jt)}^{1})\\
		& =O_{p}(N)
	\end{align*}
	Where the last line follows from the properties of $\Gamma_{ijst}$.
	The same result holds for $\bar{E}[\mathcal{J}_{2}^{2}]$, $\bar{E}[\mathcal{J}_{3}^{2}]$,
	and $\bar{E}[\mathcal{J}_{4}^{2}]$, hence $ \mathbf{J}[A'MB]=O_{p}(N^{-1/2})=o_{p}(1)$ by the Markov inequality.
\end{proof}
The following lemma derives the forms of $\mathbf{J}[U^{(0)}]$ and
$\mathbf{J}[U^{(1)}]$.
\begin{lem}
	\label{lem:jack1}Let Assumptions \ref{assu:network} and \ref{assu:beta_jack} hold. Then
	\[
	\mathbf{J}[U^{(0)}]  =U^{(0)} \qquad
	\mathbf{J}[U^{(1)}]  =o_{p}(1)
	\]
\end{lem}
\begin{proof}
	For $U^{(0)}=(\partial_{\beta}\mathcal{L})+(\partial_{\beta\phi'}\bar{\mathcal{L}})\bar{\mathcal{H}}^{-1}\mathcal{S}$
	we can appeal to Lemma \ref{lem:jackS} with $A=(\partial_{\beta}\mathcal{L})=\frac{1}{N-1}\sum_i\sum_{j\ne i}\partial_{\beta}\ell_{ij}$
	and with 
	\[
		A = \frac{1}{N-1} \sum_i \sum_{j\ne i} \big\{ \big[(\partial_{\beta\alpha'}\bar{\mathcal{L}}) \bar{\mathcal{H}}^{-1}_{\alpha\alpha} + (\partial_{\beta\gamma'}\bar{\mathcal{L}}) \bar{\mathcal{H}}^{-1}_{\gamma\alpha}\big]_{i} +
			 \big[(\partial_{\beta\alpha'}\bar{\mathcal{L}}) \bar{\mathcal{H}}^{-1}_{\alpha\gamma} + (\partial_{\beta\gamma'}\bar{\mathcal{L}}) \bar{\mathcal{H}}^{-1}_{\gamma\gamma}\big]_{j} \big\} ( \partial_\pi \ell_{ij})
	\]
	where $[V]_i$ is the $i$-th element of vector $V$.
	
	Now, for $U^{(1)}$ we can appeal to Lemma \ref{lem:jack_r0}. Firstly, we can apply Lemma D.1 from FW16 to show that
	\begin{align*}
		\lVert \bar{\mathcal{H}}^{-1} - \mathcal{D}^{-1} \rVert = O_p(N^{-1})
	\end{align*}
	where $\mathcal{D}$ is a $2N \times 2N$ block diagonal matrix with upper left block given by $\mathcal{D}_{\alpha\alpha}=\text{diag}(-\frac{1}{N-1}\sum_{j\ne i} \bar{E}[\partial_{\pi^2}\ell_{ij}])$ and lower right block by  $\mathcal{D}_{\gamma\gamma}\text{diag}(-\frac{1}{N-1}\sum_{j\ne i} \bar{E}[\partial_{\pi^2}\ell_{ji}])$. It then follows that $ \bar{\mathcal{H}}^{-1}$ satsifies the conditions for the matrix $M$ in Lemma \ref{lem:jack_r0}.
	
	For the first term, $(\partial_{\beta\phi'}\tilde{\mathcal{L}})\bar{\mathcal{H}}^{-1}\mathcal{S}$
	we set $A_{ij}=\partial_{\beta\pi}\tilde{\ell}_{ij}$ and $B_{ij}=\partial_{\pi}\ell_{ij}$.
	For the second term, $(\partial_{\beta\phi'}\bar{\mathcal{L}})\bar{\mathcal{H}}^{-1}\tilde{\mathcal{H}}\bar{\mathcal{H}}^{-1}\mathcal{S}$
	we note that
	\[
	(\partial_{\beta\phi'}\bar{\mathcal{L}})\bar{\mathcal{H}}^{-1}\tilde{\mathcal{H}}=\big((\partial_{\beta\phi'}\bar{\mathcal{L}})\bar{\mathcal{H}}^{-1}\tilde{\mathcal{H}}_{\cdot,\alpha},(\partial_{\beta\phi'}\bar{\mathcal{L}})\bar{\mathcal{H}}^{-1}\tilde{\mathcal{H}}_{\cdot,\gamma}\big)
	\]
	where the $i$-th element of $(\partial_{\beta\phi'}\bar{\mathcal{L}})\bar{\mathcal{H}}^{-1}\tilde{\mathcal{H}}_{\cdot,\alpha}$
	is
	\begin{align*}
		(\partial_{\beta\phi'}\bar{\mathcal{L}})\bar{\mathcal{H}}^{-1}\tilde{\mathcal{H}}_{\cdot,\alpha_{i}} & =(\partial_{\beta\alpha'}\bar{\mathcal{L}})\bar{\mathcal{H}}_{\alpha\alpha}^{-1}\tilde{\mathcal{H}}_{\alpha\alpha_{i}}+(\partial_{\beta\alpha'}\bar{\mathcal{L}})\bar{\mathcal{H}}_{\alpha\gamma}^{-1}\tilde{\mathcal{H}}_{\gamma\alpha_{i}}\\
		& +(\partial_{\beta\gamma'}\bar{\mathcal{L}})\bar{\mathcal{H}}_{\gamma\alpha}^{-1}\tilde{\mathcal{H}}_{\alpha\alpha_{i}}+(\partial_{\beta\gamma'}\bar{\mathcal{L}})\bar{\mathcal{H}}_{\gamma\gamma}^{-1}\tilde{\mathcal{H}}_{\gamma\alpha_{i}}\\
		& =-(\partial_{\beta\alpha'}\bar{\mathcal{L}})\bar{\mathcal{H}}_{\alpha\alpha_{i}}^{-1}\frac{1}{N-1}\sum_{j\ne i}\partial_{\pi^{2}}\tilde{\ell}_{ij}-\frac{1}{N-1}\sum_{j\ne i}(\partial_{\beta\alpha'}\bar{\mathcal{L}})\bar{\mathcal{H}}_{\alpha\gamma_{j}}^{-1}\partial_{\pi^{2}}\tilde{\ell}_{ij}\\
		& -(\partial_{\beta\gamma'}\bar{\mathcal{L}})\bar{\mathcal{H}}_{\gamma\alpha_{i}}^{-1}\frac{1}{N-1}\sum_{j\ne i}\partial_{\pi^{2}}\tilde{\ell}_{ij}-\frac{1}{N-1}\sum_{j\ne i}(\partial_{\beta\gamma'}\bar{\mathcal{L}})\bar{\mathcal{H}}_{\gamma\gamma_{j}}^{-1}\partial_{\pi^{2}}\tilde{\ell}_{ij}
	\end{align*}
	and similarly for elements of $(\partial_{\beta\phi'}\bar{\mathcal{L}})\bar{\mathcal{H}}^{-1}\tilde{\mathcal{H}}_{\cdot,\gamma}$.
	So we can let 
	\[
	A_{ij}=\big((\partial_{\beta\alpha'}\bar{\mathcal{L}})\bar{\mathcal{H}}_{\alpha\alpha_{i}}^{-1}+(\partial_{\beta\alpha'}\bar{\mathcal{L}})\bar{\mathcal{H}}_{\alpha\gamma_{j}}^{-1}+(\partial_{\beta\gamma'}\bar{\mathcal{L}})\bar{\mathcal{H}}_{\gamma\alpha_{i}}^{-1}+(\partial_{\beta\gamma'}\bar{\mathcal{L}})\bar{\mathcal{H}}_{\gamma\gamma_{j}}^{-1}\big)\partial_{\pi^{2}}\tilde{\ell}_{ij}
	\]
	and $B_{ij}=\partial_{\pi}\ell_{ij}$ in Lemma \ref{lem:jack_r0} (we have $A_{ij}$ and $B_{ij}$ mean zero and the moment
	condition also holds by assumption). For the final term, we begin by
	demonstrating that both $\bar{\mathcal{H}}^{-1}\sum_{g}(\partial_{\beta\phi'\phi_{g}}\bar{\mathcal{L}})\bar{\mathcal{H}}^{-1}$and
	$\bar{\mathcal{H}}^{-1}\sum_{g}(\partial_{\phi\phi'\phi_{g}}\bar{\mathcal{L}})\big[\bar{\mathcal{H}}^{-1}(\partial_{\beta\phi}\bar{\mathcal{L}})\big]_{g}\bar{\mathcal{H}}^{-1}$are
	matrices that satisfy the requirements for $M$ in Lemma \ref{lem:jack_r0}. First note that if $M_1$ and $M_2$ are two matrices satisfying the conditions, then 
	their product $M_1 M_2$ will also. Hence we must show that the central parts of the two terms satsify the requirement.
	Here we show the proof for $\sum_{g}(\partial_{\phi\phi'\phi_{g}}\bar{\mathcal{L}})\big[\bar{\mathcal{H}}^{-1}(\partial_{\beta\phi}\bar{\mathcal{L}})\big]_{g}$, the other term can be shown similarly. Firstly,
	\begin{align*}
		\sum_{g}  (\partial_{\phi\phi'\phi_{g}}\bar{\mathcal{L}})\big[\bar{\mathcal{H}}^{-1}(\partial_{\beta\phi}\bar{\mathcal{L}})\big]_{g}
		&=\sum_{s,t}(\partial_{\phi\phi'\alpha_{s}}\bar{\mathcal{L}})\big(\bar{\mathcal{H}}_{\alpha\alpha}^{-1}\big)_{st}(\partial_{\beta\alpha_{t}}\bar{\mathcal{L}})+\sum_{s,t}(\partial_{\phi\phi'\gamma_{s}}\bar{\mathcal{L}})\big(\bar{\mathcal{H}}_{\gamma\alpha}^{-1}\big)_{st}(\partial_{\beta\alpha_{t}}\bar{\mathcal{L}})\\
		& +\sum_{s,t}(\partial_{\phi\phi'\alpha_{s}}\bar{\mathcal{L}})\big(\bar{\mathcal{H}}_{\alpha\gamma}^{-1}\big)_{st}(\partial_{\beta\gamma_{t}}\bar{\mathcal{L}})+\sum_{s,t}(\partial_{\phi\phi'\gamma_{s}}\bar{\mathcal{L}})\big(\bar{\mathcal{H}}_{\gamma\gamma}^{-1}\big)_{st}(\partial_{\beta\gamma_{t}}\bar{\mathcal{L}})
	\end{align*}
	Taking the first of these terms, the $(i,j)$ element is given by $\big[\sum_{s,t} (\partial_{\phi\phi'\alpha_{s}}\bar{\mathcal{L}})\big(\bar{\mathcal{H}}_{\alpha\alpha}^{-1}\big)_{st}(\partial_{\beta\alpha_{t}}\bar{\mathcal{L}})\big]_{ij}$.
	Now if $\phi_{i}=\phi_{j}=\alpha_{i}$ then
	\[
	\sum_{t}(\partial_{\alpha_{i}\alpha_{i}'\alpha_{i}}\bar{\mathcal{L}})(\bar{\mathcal{H}}_{\alpha\alpha})_{it}(\partial_{\beta\alpha_{t}}\bar{\mathcal{L}})  =O_{p}(1)
	\]
	and similarly if $\phi_{i}=\phi_{j}=\gamma_{i}$. If $i\ne j$ then we have $\sum_{t}(\partial_{\alpha_{i}\gamma_{j}'\alpha_{i}}\bar{\mathcal{L}})(\bar{\mathcal{H}}_{\alpha\alpha})_{it}(\partial_{\beta\alpha_{t}}\bar{\mathcal{L}})$ is either zero or $O_{p}(N^{-1})$.
	Identical results apply to the other elements in $\sum_{g}(\partial_{\phi\phi'\phi_{g}}\bar{\mathcal{L}})\big[\bar{\mathcal{H}}^{-1}(\partial_{\beta\phi}\bar{\mathcal{L}})\big]_{g}$
	and hence we can conclude that the matrix has $O_{p}(1)$ diagonal
	elements and $O_{p}(N^{-1})$ off-diagonal elements as required.
	Then, we can apply Lemma \ref{lem:jack_r0} with $A=B=\mathcal{S}$
	to give the result.
\end{proof}
\begin{lem}
	\label{lem:betaJ_approx}Let Assumptions \ref{assu:network} and \ref{assu:beta_jack} hold, and let $\widehat{\beta}_{J}$
	be the jackknife estimator. Then, a first-order approximation
	to the estimator is given by
	\[
	\overline{W}_{N}N(\widehat{\beta}_{J}-\beta_{0})=(\partial_{\beta}\mathcal{L})+(\partial_{\beta\phi'}\bar{\mathcal{L}})\bar{\mathcal{H}}^{-1}\mathcal{S}+o_p(1)=U^{(0)}+o_{p}(1).
	\]
\end{lem}
\begin{proof}
	Recall from Lemma \ref{lem:exp_beta1} that
	\[
	N\overline{W}_{N}(\widehat{\beta}-\beta_{0})=U^{(0)}+U^{(1)}+R_{\beta}+\tilde{R}_{\beta}
	\]
	Since $\overline{W}_{N}$ is fixed across leave-out samples, we can focus
	on the jackknife operator applied to the RHS. By Lemma \ref{lem:jack1}
	we have that $\mathbf{J}[U^{(0)}+U^{(1)}]=U^{(0)}+o_{p}(1)$, while
	in Supplementary Appendix I it is shown that $\mathbf{J}[R_{\beta}]=o_{p}(1)$.
	Finally,
	\[
	\mathbf{J}[\tilde{R}_{\beta}] =(N-1)\tilde{R}_{\beta}-(N-2)\frac{1}{N-1}\sum_{k}\tilde{R}_{\beta,(k)}
	=o_{p}(1)
	\]
	since each remainder term in the above is $o_{p}(N^{-1})$.
\end{proof}

\subsection{Proof of Theorem 1} \label{subsec:thm1proof}

From Lemma \ref{lem:betaJ_approx} we have that
\begin{align*}
	\overline{W}_{N}N(\widehat{\beta}_{J}-\beta_{0}) & =(\partial_{\beta}\mathcal{L})+(\partial_{\beta\phi'}\bar{\mathcal{L}})\bar{\mathcal{H}}^{-1}\mathcal{S}+o_{p}(1)\\
	& =\frac{1}{N-1}\sum_{i}\sum_{j<i}\big(D_{\beta}\ell_{ij}+D_{\beta}\ell_{ji}\big)+o_{p}(1)
\end{align*}
where $D_{\beta}\ell_{ij}=\partial_{\beta}\ell_{ij}-\partial_{\pi}\ell_{ij}\Xi_{ij}$
for $\Xi_{ij} =-\frac{1}{N-1}\sum_{s}\sum_{t\ne s}\Gamma_{ijst}\bar{E}[\partial_{\beta\pi}\ell_{st}]$. 
The result then follows from a standard CLT argument since elements
in the summation are conditionally uncorrelated with bounded second moments. To show consistency
of the plug-in estimator $\widehat{\Omega}_{N}$, we note that by
Assumptions \ref{assu:network} and \ref{assu:beta_jack} (iii), $D_{\beta}\ell_{ij}$ has a first-order
Taylor approximation that is a continuously differentiable function
with respect to the parameters. Then by the continuous mapping theorem
and consistency of the parameters $\lVert\widehat{\beta}-\beta_{0}\rVert \overset{p}{\longrightarrow} 0$
and $\lVert\widehat{\phi}-\phi\rVert_{\infty} \overset{p}{\longrightarrow} 0$, $\widehat{\Omega}_{N} \overset{p}{\longrightarrow} \Omega$
as required (see for example Lemma S.1 in FW16). Similarly, $W$ is also a sample average of functions that are continuously differentiable with respect to parameters, so that $\widehat{W}\to W$, and hence $\widehat{V}\to V$.

For the weighted jackknife we define
\[
W_{(k)}  =\frac{1}{N}\partial_{\beta\beta}\mathcal{L}_{(k)}+\frac{1}{N}(\partial_{\beta\phi'}\mathcal{L}_{(k)})\mathcal{H}_{(k)}^{-1}(\partial_{\beta\phi}\mathcal{L}_{(k)}), \qquad 
W_{J}  =\frac{1}{N-1}\sum_{k}W_{(k)}
\]
and $\widehat{W}_{(k)}$ and $\widehat{W}_J$ the plug-in versions of these. The weighted jackknife estimator is
\begin{align*}
	\widehat{W}_{J}(\widehat{\beta}_{wJ}-\beta) & =(N-1)W_{N}(\widehat{\beta}-\beta)-(N-2)\frac{1}{N-1}\sum_{k}W_{(k)}(\widehat{\beta}_{(k)}-\beta)\\
	& \quad+(N-1)(\widehat{W}_{J}-W_{J})(\widehat{\beta}-\beta)-(N-2)\frac{1}{N-1}\sum_{k}(\widehat{W}_{(k)}-W_{(k)})(\widehat{\beta}_{(k)}-\beta)\\
	& \quad+(N-1)(W_{J}-W_{N})(\widehat{\beta}-\beta)
\end{align*}
In Supplementary Appendix I.3 it is shown
that the second and third lines are both $o_{p}(N^{-1})$. For the first
line we have
\[
(N-1)W_{N}N(\widehat{\beta}-\beta)-(N-2)\frac{1}{N-1}\sum_{k}W_{(k)}N(\widehat{\beta}_{(k)}-\beta)=U^{(0)}+o_{p}(1)
\]
which follows directly from the proof of the first-order expansion
of the unweighted jackknife (note that the expansion for $\widehat{\beta}_{(k)}$
is naturally in terms of $W_{(k)}\widehat{\beta}_{(k)}$ and the proof
of the jackknife involves proof of the statement given here, in combination
with $\lVert W_{(k)}-\overline{W}_{N}\rVert=o_{p}(1)$, which allows us
to write the expansion in Lemma \ref{lem:exp_beta1_k}). Then, we
have
\[
N\widehat{W}_{J}(\widehat{\beta}_{wJ}-\beta)  =U^{(0)}+o_{p}(1)
\]
Supplementary Appendix I.3 also shows that $\lVert\widehat{W}_J-\overline{W}_{N}\rVert=o_{p}(1)$,
from which the statement $N\overline{W}_{N}(\widehat{\beta}_{wJ}-\beta)=U^{(0)}+o_{p}(1)$
follows, and hence the same asymptotic result as the standard jackknife
applies to the weighted jackknife.


\section{Jackknife results for average effects}

We begin by stating a first-order asymptotic expansion for the average
effect estimator that will be used in the proof of Theorem 3. The
proof of this result is provided in the Supplementary Appendix J.
\begin{lem}
	\label{lem:expansion_avg}Let Assumptions 1, 2 and 3 hold. Then
	\begin{align*}
		N(\widehat{\Delta}-\Delta) & =\Big[(\partial_{\beta}\bar{\Delta})+(\partial_{\phi'}\bar{\Delta})\mathcal{\bar{H}}^{-1}(\partial_{\beta\phi}\bar{\mathcal{L}})\Big]\bar{W}_{N}^{-1}\big(U^{(0)}+U^{(1)}\big)\\
		& +N(\partial_{\phi'}\bar{\Delta})\mathcal{\bar{H}}^{-1}\mathcal{S}
		+N(\partial_{\phi'}\tilde{\Delta})\mathcal{\bar{H}}^{-1}\mathcal{S}-N(\partial_{\phi'}\bar{\Delta})\mathcal{\bar{H}}^{-1}\tilde{\mathcal{H}}\mathcal{\bar{H}}^{-1}\mathcal{S}\\
		& +\frac{1}{2}N\mathcal{S}'\bar{\mathcal{H}}^{-1}\Big((\partial_{\phi\phi'}\bar{\Delta})+\sum_{g}(\partial_{\phi\phi'\phi_{g}}\bar{\mathcal{L}})[(\partial_{\phi'}\bar{\Delta})\mathcal{\bar{H}}^{-1}]_{g}\Big)\mathcal{\bar{H}}^{-1}\mathcal{S}
		+R_{\Delta}+\tilde{R}_{\Delta}
	\end{align*}
	where $\lVert R_{\Delta}\rVert=o_{p}(1)$, and $\lVert\tilde{R}_{\Delta}\rVert=o_{p}(N^{-1})$,
	and 
	\begin{align*}
		N(\widehat{\Delta}_{(k)}-\Delta_{(k)}) & =\Big[(\partial_{\beta}\bar{\Delta})+(\partial_{\phi'}\bar{\Delta})\mathcal{\bar{H}}^{-1}(\partial_{\beta\phi}\bar{\mathcal{L}})\Big]\bar{W}_{N}^{-1}\big(U_{(k)}^{(0)}+U_{(k)}^{(1)}\big)\\
		& +N(\partial_{\phi'}\bar{\Delta})\mathcal{\bar{H}}^{-1}\mathcal{S}_{(k)}
		+N(\partial_{\phi'}\tilde{\Delta}_{(k)})\mathcal{\bar{H}}^{-1}\mathcal{S}_{(k)}-N(\partial_{\phi'}\bar{\Delta})\mathcal{\bar{H}}^{-1}\tilde{\mathcal{H}}_{(k)}\mathcal{\bar{H}}^{-1}\mathcal{S}_{(k)}\\
		& +\frac{1}{2}N\mathcal{S}_{(k)}'\bar{\mathcal{H}}^{-1}\Big((\partial_{\phi\phi'}\bar{\Delta})+\sum_{g}(\partial_{\phi\phi'\phi_{g}}\bar{\mathcal{L}})[(\partial_{\phi'}\bar{\Delta})\mathcal{\bar{H}}^{-1}]_{g}\Big)\mathcal{\bar{H}}^{-1}\mathcal{S}_{(k)}
		+R_{(k),\Delta}+\tilde{R}_{(k),\Delta}
	\end{align*}
	where $\lVert R_{(k),\Delta}\rVert=o_{p}(1)$, and $\lVert\tilde{R}_{(k),\Delta}\rVert=o_{p}(N^{-1})$.
\end{lem}

The next set of lemmas demonstrates the effect of the jackknife operation on terms that involve derivatives of the 
average effects parameter.

\begin{lem}
	\label{lem:jack_r0_Delta} Let $1_{ij}^{k}$ satisfy Condition \ref{cond:losets}
	and let $1_{\lambda}^{k}=\prod_{(i,j)\in\lambda}1_{ij}^{k}$ and define
	$\Lambda_{is}=\{\lambda:(i,s)\in\lambda\}$, i.e. the set of all $\lambda$
	containing observation $(i,s)$. Let $A=(A_{\alpha},A_{\gamma})$ and $A_{k}=(A_{\alpha,k},A_{\gamma,k})$ be defined as in Lemma \ref{lem:jack_r0}, and let $B$ and $B_{k}$ be defined as
	\begin{align*}
		B & =\frac{N}{\vert\Lambda_{N}\vert}\big(\{\sum_{s\ne i}\sum_{\lambda\in\Lambda_{is}}B_{\lambda}\}_{i=1,\dots,N},\{\sum_{s\ne j}\sum_{\lambda\in\Lambda_{sj}}B_{\lambda}\}_{i=1,\dots,N}\big)
		=(B_{\alpha},B_{\gamma})\\
		B_{k} & =\frac{N-1}{N-r-1}\frac{N}{\vert\Lambda_{N}\vert}\big(\{\sum_{s\ne i}\sum_{\lambda\in\Lambda_{is}}B_{\lambda}1_{\lambda}^{k}\}_{i=1,\dots,N},\{\sum_{s\ne j}\sum_{\lambda\in\Lambda_{sj}}B_{\lambda}1_{\lambda}^{k}\}_{i=1,\dots,N}\big)
		=(B_{\alpha,k},B_{\gamma,k})
	\end{align*}
	for mean zero $B_{\lambda}$ with bounded fourth moment. Assume that
	$A_{ij}$ is independent of $A_{st}$ for $(s,t)\not\in\{(i,j),(j,i)\}$,
	and independent of $B_{\lambda}$ whenever $\lambda$ does not contain
	either $(i,j)$ or $(j,i)$. Let $M$ be a non-random matrix that has $O_{p}(1)$ elements on
	its diagonal and $O_{p}(N^{-1})$ off-diagonal terms. Then
	\[
	\mathbf{J}[A'MB] =  (N-1)A'MB-\frac{N-2}{N-1}\sum_{k}A_{(k)}'MB_{(k)}=o_{p}(1).
	\]
\end{lem}
\begin{proof}
	Define the sets $\Lambda_{(is)(jt)} = \Lambda_{is} \cap \Lambda_{jt}$ and $\bar{\Lambda}_{is} = \Lambda_{is} \cup \Lambda_{si}$.
	We have
	\begin{align*}
		A'MB & =A_{\alpha}'M_{\alpha\alpha}B_{\alpha}+A_{\alpha}'M_{\alpha\gamma}B_{\gamma}
		+A_{\gamma}'M_{\gamma\alpha}B_{\alpha}+A_{\gamma}'M_{\gamma\gamma}B_{\gamma}
	\end{align*}
	The full sample and leave-out versions of the first term are
	\begin{align*}
		A_{\alpha}'M_{\alpha\alpha}B_{\alpha} & =\frac{N}{(N-1)\vert\Lambda_{N}\vert}\sum_{i,j}\sum_{s\ne i}\sum_{t\ne j}\sum_{\lambda\in\Lambda_{jt}}(M_{\alpha\alpha})_{ij}A_{is}B_{\lambda}\\
		\frac{1}{N-1} & \sum_{k}A_{k,\alpha}'M_{\alpha\alpha}B_{k,\alpha}\\
		& =\frac{N}{(N-r-1)(N-2)\vert\Lambda_{N}\vert}\sum_{k}\sum_{i,j}\sum_{s\ne i}\sum_{t\ne j}\sum_{\lambda\in\Lambda_{(is)(jt)} }(M_{\alpha\alpha})_{ij}A_{is}B_{\lambda}1_{\lambda}^{k}\\
		& +\frac{N}{(N-r-1)(N-2)\vert\Lambda_{N}\vert}\sum_{k}\sum_{i,j}\sum_{s\ne i}\sum_{t\ne j}\sum_{\lambda\in(\Lambda_{jt}\backslash\Lambda_{is})}(M_{\alpha\alpha})_{ij}A_{is}B_{\lambda}1_{is}^{k}1_{\lambda}^{k}
	\end{align*}
	and the jackknifed term is
	\begin{align*}
		\mathcal{J}_{\alpha\alpha}= & (N-1)A_{\alpha}'M_{\alpha\alpha}B_{\alpha}-\frac{N-2}{N-1}\sum_{k}A_{\alpha,k}'M_{\alpha\alpha}B_{\alpha,k}\\
		= & \frac{N}{\vert\Lambda_{N}\vert}\sum_{i,j}\sum_{s\ne i}\sum_{t\ne j}\sum_{\lambda \in \Lambda_{(is)(jt)}}
		(M_{\alpha\alpha})_{ij}A_{is}B_{\lambda}\Big(1-\frac{\sum_{k}1_{\lambda}^{k}}{N-r-1}\Big)\\
		& +\frac{N}{\vert\Lambda_{N}\vert}\sum_{k}\sum_{i,j}\sum_{s\ne i}\sum_{t\ne j}\sum_{\lambda\in(\Lambda_{jt}\backslash\Lambda_{is})}(M_{\alpha\alpha})_{ij}A_{is}B_{\lambda}\Big(1-\frac{\sum_{k}1_{is}^{k}1_{\lambda}^{k}}{N-r-1}\Big)
	\end{align*}
	
	Letting $\Gamma_{ijst}=(M_{\alpha\alpha})_{ij}+(M_{\alpha\gamma})_{it}+(M_{\gamma\alpha})_{sj}+(M_{\gamma\gamma})_{st}$,
	similar computation for the other three elements gives
	\begin{align*}
		\mathbf{J}[A'MB]  = & (N-1)A'MB-\frac{N-2}{N-1}\sum_{k}A_{k}'MB_{k}\\
		= & \frac{N}{\vert\Lambda_{N}\vert}\sum_{i,j}\sum_{s\ne i}\sum_{t\ne j}\sum_{\lambda \in \Lambda_{(is)(jt)}} \Gamma_{ijst}A_{is}B_{\lambda}\Big(1-\frac{\sum_{k}1_{\lambda}^{k}}{N-r-1}\Big)\\
		& \frac{N}{\vert\Lambda_{N}\vert}\sum_{i,j}\sum_{s\ne i}\sum_{t\ne j}\sum_{\lambda\in(\Lambda_{(is)(jt)}\backslash\Lambda_{is})}\Gamma_{ijst}A_{is}B_{\lambda}\Big(1-\frac{\sum_{k}1_{is}^{k}1_{\lambda}^{k}}{N-r-1}\Big)\\
		& +\frac{N}{\vert\Lambda_{N}\vert}\sum_{i,j}\sum_{s\ne i}\sum_{t\ne j}\sum_{\lambda\in(\Lambda_{jt}\backslash\bar{\Lambda}_{is})}\Gamma_{ijst}A_{is}B_{\lambda}\Big(1-\frac{\sum_{k}1_{is}^{k}1_{\lambda}^{k}}{N-r-1}\Big)\\
		& =\mathcal{J}_{1}+\mathcal{J}_{2}+\mathcal{J}_{3}
	\end{align*}
	Next, recall that $\Gamma_{ijst}=O_{p}(1)$ whenever $i=j$ or $s=t$,
	and is $O_{p}(N^{-1})$ otherwise. Then, for the third term
	\begin{align*}
		\bar{E}[\mathcal{J}_{3}^{2}] & =\frac{N^{2}}{\vert\Lambda_{N}\vert^{2}}\sum_{i,j,k,l}\sum_{s\ne i}\sum_{t\ne j}\sum_{p\ne k}\sum_{q\ne l}\sum_{\lambda\in(\Lambda_{jt}\backslash\bar{\Lambda}_{is})}\sum_{\lambda\in(\Lambda_{ql}\backslash\bar{\Lambda}_{pk})}\\
		& \quad\Gamma_{ijst}\Gamma_{klpq}\bar{E}[A_{is}B_{\lambda}A_{pk}B_{\lambda'}]\Big(1-\frac{\sum_{k}1_{is}^{k}1_{\lambda}^{k}}{N-r-1}\Big)\Big(1-\frac{\sum_{k}1_{is}^{k}1_{\lambda'}^{k}}{N-r-1}\Big)\\
		& =\frac{N^{2}}{\vert\Lambda_{N}\vert^{2}}\sum_{i,j,l}\sum_{s\ne i}\sum_{t\ne j}\sum_{q\ne l}\sum_{\lambda\in(\Lambda_{jt}\backslash\bar{\Lambda}_{is})}\sum_{\lambda'\in((\Lambda_{ql}\cap\bar{\Lambda}_{jt})\backslash\bar{\Lambda}_{is})}\\
		& \quad\Gamma_{ijst}\bar{E}[A_{is}(\Gamma_{ilsq}A_{is}+\Gamma_{sliq}A_{si})]\bar{E}[B_{\lambda}B_{\lambda'}]\Big(1-\frac{\sum_{k}1_{is}^{k}1_{\lambda}^{k}}{N-r-1}\Big)\Big(1-\frac{\sum_{k}1_{is}^{k}1_{\lambda'}^{k}}{N-r-1}\Big)\\
		& +\frac{N^{2}}{\vert\Lambda_{N}\vert^{2}}\sum_{i,j,k,l}\sum_{s\ne i}\sum_{t\ne j}\sum_{p\ne k}\sum_{q\ne l}\sum_{\lambda\in((\Lambda_{jt}\cap\bar{\Lambda}_{pk})\backslash\bar{\Lambda}_{is})}\sum_{\lambda\in((\Lambda_{ql}\cap\bar{\Lambda}_{is})\backslash\bar{\Lambda}_{pk})}\\
		& \quad\Gamma_{ijst}\Gamma_{klpq}\bar{E}[A_{is}B_{\lambda'}]\bar{E}[A_{pk}B_{\lambda}]\Big(1-\frac{\sum_{k}1_{is}^{k}1_{\lambda}^{k}}{N-r-1}\Big)\Big(1-\frac{\sum_{k}1_{is}^{k}1_{\lambda'}^{k}}{N-r-1}\Big)
	\end{align*}
	
	Note that $\frac{N}{\vert\Lambda_{N}\vert}\vert\lambda\in(\Lambda_{jt}\backslash\bar{\Lambda}_{is})\vert=O(N^{-1})$,
	while $\frac{N}{\vert\Lambda_{N}\vert}\vert\lambda'\in((\Lambda_{ql}\cap\bar{\Lambda}_{jt})\backslash\bar{\Lambda}_{is})\vert$
	is $O(N^{-1})$ if $(q,l)$ equals $(t,j)$ or $(j,t)$, $O(N^{-2})$
	if either $q\in\{t,j\}$ or $l\in\{t,j\}$ and $O(N^{-3})$ otherwise.
	Also, 
	\[
	\Big(1-\frac{\sum_{k}1_{is}^{k}1_{\lambda}^{k}}{N-r-1}\Big)\Big(1-\frac{\sum_{k}1_{is}^{k}1_{\lambda'}^{k}}{N-r-1}\Big)=O(N^{-2})
	\]
	Combining these facts with, $\Gamma_{ijst}=O_{p}(1)$ whenever $i=j$
	or $s=t$, and $O_{p}(N^{-1})$ otherwise gives $\bar{E}[\mathcal{J}_{3}^{2}]=o_{p}(1)$.
	An almost identical analysis applies to $\mathcal{J}_{1}$ and $\mathcal{J}_{2}$,
	giving the result $\mathbf{J}[A'MB]=o_{p}(1)$.
\end{proof}
The next lemma states the first-order approximation for the jackknife
bias-corrected average effect estimator.
\begin{lem}
	\label{lem:avg_approx1}Let Assumptions 1, 2 and 3 hold and let $\widehat{\Delta}_{J}$
	be the jackknife bias-corrected estimator in (14). Then,
	\begin{align*}
		N(\widehat{\Delta}_{J}-\Delta) & =\Big[(\partial_{\beta}\bar{\Delta})-N(\partial_{\phi'}\bar{\Delta})\mathcal{\bar{H}}^{-1}(\partial_{\beta\phi}\bar{\mathcal{L}})\Big]\bar{W}_{N}^{-1}U^{(0)}
		+N(\partial_{\phi'}\bar{\Delta})\mathcal{\bar{H}}^{-1}\mathcal{S}+o_{p}(1)
	\end{align*}
\end{lem}
\begin{proof}
	First note that 
	\begin{align*}
	\widehat{\Delta}_{J}-\Delta & =(N-1)(\widehat{\Delta}-\Delta)-(N-2)\frac{1}{N-1}\sum_{k=1}^{N-1}(\widehat{\Delta}_{(k)}-\Delta_{(k)})\\
	& \quad+\frac{N-2}{N-1}\sum_{k=1}^{N-1}(\Delta_{(k)}-\Delta)
	\end{align*}
	
	We show that the first term is equal to the expression given in the lemma, while the second term is $o_p(1)$. To analyze the first term, we apply the results in Lemmas \ref{lem:jackS}, \ref{lem:jack_r0},  \ref{lem:jack1}, and \ref{lem:jack_r0_Delta}
	to the expression in the expansion of Lemma \ref{lem:expansion_avg}.
	For the first term, we apply the result in Lemma \ref{lem:jack1}
	to give
	\begin{align*}
		\mathbf{J}\Big[ & \big((\partial_{\beta}\bar{\Delta})-N(\partial_{\phi'}\bar{\Delta})\mathcal{\bar{H}}^{-1}(\partial_{\beta\phi}\bar{\mathcal{L}})\big)\bar{W}_{N}^{-1}\big(U^{(0)}+U^{(1)}\big)\Big]\\
		& =\big((\partial_{\beta}\bar{\Delta})-N(\partial_{\phi'}\bar{\Delta})\mathcal{\bar{H}}^{-1}(\partial_{\beta\phi}\bar{\mathcal{L}})\big)\bar{W}_{N}^{-1}U^{(0)}+o_{p}(1)
	\end{align*}
	Similarly, Lemma \ref{lem:jackS} implies that jackknifing the second
	term gives $N(\partial_{\phi'}\bar{\Delta})\mathcal{\bar{H}}^{-1}\mathcal{S}$.
	
	For the next two terms, we can apply Lemma  \ref{lem:jack_r0_Delta} with $M=\bar{\mathcal{H}}^{-1}$,
	$A=\mathcal{S}$, and  $B=N(\partial_{\phi'}\tilde{\Delta})$,
	and apply Lemma \ref{lem:jack_r0} with $M=\bar{\mathcal{H}}^{-1}$,
	$A=\mathcal{S}$, and $B=N(\partial_{\phi'}\bar{\Delta})\mathcal{\bar{H}}^{-1}\tilde{\mathcal{H}}$ (Lemma 28 in the supplementary appendix shows that $N(\partial_{\phi'}\bar{\Delta})$ is a vector with $O_p(1)$ terms so that $B$ satisfies the requirements of the lemma). 
	
	For the final term, we first show that 
	\[
	N\bar{\mathcal{H}}^{-1}\Big((\partial_{\phi\phi'}\bar{\Delta})+\sum_{g}(\partial_{\phi\phi'\phi_{g}}\bar{\mathcal{L}})[(\partial_{\phi'}\bar{\Delta})\mathcal{\bar{H}}^{-1}]_{g}\Big)\mathcal{\bar{H}}^{-1}
	\]
	satisfies the conditions for $M$ in Lemma \ref{lem:jack_r0}, from
	which we will be able to conclude that the jackknifed term will be
	$o_{p}(1)$.  
	
	By Lemma 28 we have that $N(\partial_{\phi\phi'}\bar{\Delta})$ is a $2\times2$ block-diagonal matrix with each block
	having $O_{p}(1)$ diagonal elements and $O_{p}(N^{-1})$ off-diagonal elements. We also have that $\bar{\mathcal{H}}^{-1}$ satisfies the conditions for $M$ (see Lemma D.1 in \cite{Fernandez-Val2016}), from which it follows that the product $N\bar{\mathcal{H}}^{-1}(\partial_{\phi\phi'}\bar{\Delta})\mathcal{\bar{H}}^{-1}$ satsifies the conditions.
	For $N\sum_{g}(\partial_{\phi\phi'\phi_{g}}\bar{\mathcal{L}})[(\partial_{\phi'}\bar{\Delta})\mathcal{\bar{H}}^{-1}]_{g}$, the proof follows identically to that in Lemma \ref{lem:jack1}, replacing $(\partial_{\beta\phi}\bar{\mathcal{L}})$ by $N(\partial_{\phi'}\bar{\Delta})$ (which by Lemma 28 is made up of $O_p(1)$ terms). 
	
	Finally, in the supplementary appendix
	(J.4) it is shown that $\mathbf{J}[R_{\Delta}]=o_{p}(1)$, and by
	$\tilde{R}_{\Delta,(k)}=o_{p}(N^{-1})$ for each $k$ (and in the
	full sample) we have that $(N-2)\tilde{R}_{\Delta,(k)}=o_{p}(1)$
	and $(N-1)\tilde{R}_{\Delta}=o_{p}(1)$ so that the jackknifed version
	of this term is also $o_{p}(1)$. 
	
	We next demonstrate that $(N-2)\frac{1}{N-1}\sum_{k}(\Delta_{(k)}-\Delta)=o_{p}(1)$. This follows from
	\begin{align*}
		\frac{N-2}{N-1}\sum_{k}(\Delta_{(k)}-\Delta) & =\frac{N-2}{N-1}\sum_{k}\frac{1}{\lvert\Lambda_{N}\rvert}\sum_{\lambda}m_{\lambda}\Big(\frac{N-1}{N-r-1}1_{\lambda}^{k}-1\Big)\\
		& =\frac{N-2}{N-1}\frac{1}{\lvert\Lambda_{N}\rvert}\sum_{\lambda}m_{\lambda}\Big(\frac{N-1}{N-r-1}\sum_{k}1_{\lambda}^{k}-(N-1)\Big)
	\end{align*}
	For most $\lambda$, on the order of $\lvert\Lambda_{N}\rvert$, we have that any $\lambda$ spans $r$ different leaveout sets. Relative to these sets, the number of $\lambda$ in which two observations are constrained to fall within the same leaveout set are of order at most $N^{-1}\lvert\Lambda_{N}\rvert$, and so on. Letting $I_{\lambda,t}$ be an indicator equal to one when $\lambda$ spans $t$ leaveout sets, we can write
\begin{align*}
	\frac{N-2}{N-1}\sum_{k}(\Delta_{(k)}-\Delta) & =\frac{N-2}{N-1}\frac{1}{\lvert\Lambda_{N}\rvert}\sum_{\lambda}m_{\lambda}\Big(\frac{N-1}{N-r-1}\sum_{k}1_{\lambda}^{k}-(N-1)\Big)\\
	& =\frac{N-2}{N-r-1}\frac{1}{\lvert\Lambda_{N}\rvert}\sum_{\lambda}m_{\lambda}1_{\lambda,r-1}\\
	& +\dots+(r-1)\frac{N-2}{N-r-1}\frac{1}{\lvert\Lambda_{N}\rvert}\sum_{\lambda}m_{\lambda}1_{\lambda,1}\\
	& =o_{p}(1)
\end{align*}
	and hence we get the result.

\end{proof}

\subsection{Proof of Theorem 2}

We can decompose $\widehat{\Delta}_{J}-\bar{\Delta}$ into $\widehat{\Delta}_{J}-\bar{\Delta}=(\widehat{\Delta}_{J}-\Delta)+(\Delta-\bar{\Delta})$.
From Lemma \ref{lem:avg_approx1} we have 
\begin{align*}
	N(\widehat{\Delta}_{J}-\Delta) & =\Big[(\partial_{\beta}\bar{\Delta})-(\partial_{\phi'}\bar{\Delta})\mathcal{\bar{H}}^{-1}(\partial_{\beta\phi}\bar{\mathcal{L}})\Big]\bar{W}_{N}^{-1}U^{(0)}
	+N(\partial_{\phi'}\bar{\Delta})\mathcal{\bar{H}}^{-1}\mathcal{S}+o_{p}(1)\\
	&=\frac{1}{N-1}\sum_{i}\sum_{j < i} h_{ij} + o_p(1)
\end{align*}
for $h_{ij}=-N(\partial_{\theta}\bar{\Delta})(\partial_{\theta\theta'}\bar{\mathcal{L}})^{-1}(\partial_{\theta}\ell_{ij}+\partial_{\theta}\ell_{ji})$. 
Similarly, we may write
\begin{align*}
	N(\Delta-\bar{\Delta}) & =\frac{N}{\vert\Lambda_{N}\vert}\sum_{\lambda}(m_{\lambda}-\bar{m}_{\lambda}) \\
	&= \frac{1}{d}\frac{N}{\vert\Lambda_{N}\vert}\sum_i\sum_{j < i} \sum_{\lambda\in \bar{\Lambda}_{ij}}\tilde{m}_{\lambda}
	=\frac{1}{N-1}\sum_{i}\sum_{j < i} \frac{1}{d}s_{ij}
\end{align*}
for $s_{ij}=\frac{(N-p)!}{(N-2)!}\sum_{\lambda\in\bar{\Lambda}_{ij}}\tilde{m}_{\lambda}$,
where $d$ represents the number of times each $\lambda$ appears in $\bar{\Lambda}_{ij}$ over all $i$ and $j$ (this is the number of unique observations that can be obtained by alternating the senders and receivers in $\lambda$). We then have
\[
N(\widehat{\Delta}_{J}-\bar{\Delta})=\frac{1}{N-1}\sum_{i}\sum_{j < i}(h_{ij}+\frac{1}{d} s_{ij})+o_{p}(1)
\]
which is asymptotically normal by a standard CLT. To compute the variance of this term, first note that $\sum_{s}\sum_{t < s}\bar{E}[h_{ij}h_{st}]=\bar{E}[h_{ij}^2]$.
Also,
\begin{align*}
	\frac{1}{d}\sum_{s}\sum_{t < s}\bar{E}[h_{ij}s_{st}]  & =\frac{(N-p)!}{(N-2)!}\sum_{\lambda}\bar{E}[h_{ij}\tilde{m}_{\lambda}]\\
	& =\frac{(N-p)!}{(N-2)!}\sum_{\lambda\in\bar{\Lambda}_{ij}}\bar{E}[h_{ij}\tilde{m}_{\lambda}]
	=\bar{E}[h_{ij}s_{ij}].
\end{align*}
Finally, let $\rho_{\lambda}$ represent the set of agents contained in $\lambda$
(recall that $\lambda$ is a set of observations $(i,j)$). The term
$E[\tilde{m}_{\lambda}\tilde{m}_{\lambda'}]$ can only be non-zero
so long as $\vert\rho_{\lambda}\cap\rho_{\lambda'}\vert\geq2$, since
$\lambda$ and $\lambda'$ are independent unless they have a dyad
in common. Letting $n_{\lambda,\lambda'}=\vert\rho_{\lambda}\cap\rho_{\lambda'}\vert$
represent the number of agents in common between the observations
in $\lambda$ and $\lambda'$, we have
\begin{align*}
	Var\Big(\frac{N}{\vert\Lambda_{N}\vert}\sum_{\lambda}\tilde{m}_{\lambda}\Big) & =\frac{N^{2}}{\vert\Lambda_{N}\vert^{2}}\sum_{\lambda}\sum_{\lambda'}E[\tilde{m}_{\lambda}\tilde{m}_{\lambda'}]\\
	& =\frac{N^{2}}{\vert\Lambda_{N}\vert^{2}}\sum_{\lambda}\sum_{\lambda':n_{\lambda,\lambda'}=2}E[\tilde{m}_{\lambda}\tilde{m}_{\lambda'}]
	+\frac{N^{2}}{\vert\Lambda_{N}\vert^{2}}\sum_{\lambda}\sum_{\lambda':n_{\lambda,\lambda'}>2}E[\tilde{m}_{\lambda}\tilde{m}_{\lambda'}]
\end{align*}
We have that $\vert\lambda':n_{\lambda,\lambda'}=2\vert=\binom{p}{2}\frac{(N-p-1)!}{(N-2p+1)!}=O(N^{p-2})$,
so that the first term above is $O(1)$.
Similar calculations for $n_{\lambda,\lambda'}>3$ show that the second
term is $O(N^{-1})$; e.g. considering the case $n_{\lambda,\lambda'}=3$
(assuming $p\geq3$), we find $\vert\lambda':n_{\lambda,\lambda'}=3\vert=\binom{p}{3}\frac{(N-p-1)!}{(N-2p+1)!}=O(N^{p-3})$
so that $\frac{N^{2}}{\vert\Lambda_{N}\vert^{2}}\sum_{\lambda}\sum_{\lambda':n_{\lambda,\lambda'}=3}E[\tilde{m}_{\lambda}\tilde{m}_{\lambda'}]=O(N^{-1})$.
Therefore
\begin{align*}
	Var\Big(\frac{1}{N-1}\sum_{i}\sum_{j < i}\frac{1}{d} s_{ij}\Big) & =\frac{N^{2}}{\vert\Lambda_{N}\vert^{2}}\sum_{\lambda}\sum_{\lambda':n_{\lambda,\lambda'}=2}\bar{E}[\tilde{m}_{\lambda}\tilde{m}_{\lambda'}]+o(1)\\
	& =\frac{N^{2}}{\vert\Lambda_{N}\vert^{2}}\sum_{i}\sum_{j<i}\sum_{\lambda\in\bar{\Lambda}_{ij}}\sum_{\lambda'\in\bar{\Lambda}_{ij}}\bar{E}[\tilde{m}_{\lambda}\tilde{m}_{\lambda'}]+o(1)\\
	& =\frac{1}{(N-1)^{2}}\sum_{i}\sum_{j<i}\bar{E}\Big[\big(\frac{(N-p)!}{(N-2)!}\sum_{\lambda\in\bar{\Lambda}_{ij}}\tilde{m}_{\lambda}\Big)^{2}\Big]+o(1)\\
	& =\frac{1}{(N-1)^{2}}\sum_{i}\sum_{j<i}\bar{E}[s_{ij}^{2}]+o(1)
\end{align*}

Combining the above results gives
\begin{align*}
	Var\Big(\frac{1}{N-1}\sum_{i}\sum_{j < i}(h_{ij}+\frac{1}{d}s_{ij})\Big) &  =\frac{1}{(N-1)^{2}}\sum_{i}\sum_{j<i}\bar{E}[(h_{ij}+s_{ij})^{2}] + o(1)
\end{align*}

The asymptotic variance
of $N(\widehat{\Delta}_{J}-\bar{\Delta})$ is given by the limit
of this expression. Assumptions 2 (iii) and 3 (ii) guarantee that
both $h_{ij}$ and $s_{ij}$ have first-order approximations that
are continuously differentiable in the parameters, so that the continuous
mapping theorem and the consistency results $\lVert\widehat{\beta}-\beta_{0}\rVert \overset{p}{\longrightarrow} 0$
and $\lVert\widehat{\phi}-\phi\rVert_{\infty} \overset{p}{\longrightarrow} 0$, imply consistency
of the plug-in estimator for $V_{\Delta}$ (see for example Lemma
S.1 in FW16).

\subsection{Proof of Theorem 3}

We begin with a U-statistic representation of $\bar{\Delta}$,
which will allow us to apply standard asymptotic results on U-statistics. As in equation (\ref{eq:U_stat}) 
we have
\[
\bar{\Delta}-\delta =\binom{N}{p}^{-1}\sum_{\tau}u_{\tau}, \qquad u_{\tau}=\frac{1}{p!}\sum_{\lambda:\rho_\lambda=\tau}(\bar{m}_{\lambda}-\delta)
\]
where $\tau$ is a set of $p$ agents, and $u_\tau$ sums over the $p!$ $\lambda$ that share a common set of $p$ agents. The variable $u_{\tau}$ is symmetric function of $\{X_{i},\alpha_{i},\gamma_{i}\}$ for $p$ agents $i$. Since we assume that $\{X_{i},\alpha_{i},\gamma_{i}\}$ are i.i.d. over agents,
$\bar{\Delta}-\delta$ is a U-statistic of order $p$. Standard theory, e.g. Theorem 12.3 in \cite{VanderVaart1998} gives
\[
\sqrt{N}(\bar{\Delta}-\delta) \overset{d}{\longrightarrow} \mathcal{N}(0,p^{2}\zeta_{1})
\]
where  $\zeta_{1}=Cov(u_{\tau},u_{\tau'})$ for $\tau$ and $\tau'$ sharing exactly one agent in common. Define
$t_i = p\binom{N-1}{p-1}^{-1}\sum_{\tau:i\in\tau}u_{\tau} = \frac{(N-p)!}{(N-1)!}\sum_{\lambda:i\in\lambda} (\bar{m}_\lambda-\delta)$. We next show that $Var(t_i) \overset{p}{\longrightarrow} p^2 \zeta_1$, the variance of the average effect. The variance of $t_{i}$ is given by
\begin{align*}
	Var(t_{i}) & =p^2\binom{N-1}{p-1}^{-2}\sum_{\tau:i\in\tau}\sum_{\tau':i\in\tau'}E[u_{\tau}u_{\tau'}]\\
	& =p^2 \binom{N-1}{p-1}^{-2}\sum_{\tau:i\in\tau}\sum_{\tau':\tau\cap\tau'=\{i\}}E[u_{\tau}u_{\tau'}]
	+p^2 \binom{N-1}{p-1}^{-2}\sum_{\tau:i\in\tau}\sum_{\substack{\tau':i\in\tau'\\
			\vert\tau\cap\tau'\vert>1}}E[u_{\tau}u_{\tau'}]\\
	& =p^2\frac{(N-p)!(p-1)!}{(N-1)!}\frac{(N-p)!}{(N-2p+1)!(p-1)!}\zeta_{1}+o(1)\\
	& =p^2\zeta_{1}+o(1)
\end{align*}
where the second to last line follows from the fact that there are $\binom{N-1}{p-1}$ ways to
choose $\tau$ containing $i$, and $\binom{N-p}{p-1}$ ways to choose
the remaining $p-1$ agents in $\tau'$ so that $\tau$ and $\tau'$
share only agent $i$ in common, while the number of $(\tau,\tau')$ that share more than one agent in common is of an order $N$ smaller. Hence, a LLN implies that $\frac{1}{N}\sum_i t_i^2 \overset{p}{\longrightarrow} p^2\zeta_{1}$. Let 
\[
\widehat{t}_i = \frac{(N-p)!}{(N-1)!}\sum_{\lambda:i\in\lambda} (\widehat{\bar{m}}_\lambda-\widehat{\Delta}_N)
\]
Since $\bar{m}_\lambda$ is continuously differentiable with respect to the parameters, consistency of parameters $\lVert\widehat{\phi}-\phi_{0}\rVert_{\infty} \overset{p}{\longrightarrow} 0$, and $\widehat{\Delta}_N \overset{p}{\longrightarrow} \delta$ ensures that $\widehat{t}_i^2 \overset{p}{\longrightarrow} t_i^2$ and hence $\widehat{V}_{\delta} \overset{p}{\longrightarrow} p^2 \zeta_1$.

\newpage

%
%
%

\newpage

\appendix
\setcounter{section}{3}

{\bf \Large \centering Supplementary Appendix to ``A jackknife bias correction for nonlinear network data models with fixed effects''}

This supplementary appendix contains a number of sections with additional results and proofs. 
Section D contains additional empirical results (estimation of a zero-inflated negative binomial model using the trade data from the main paper) as well as additional simulations (a logit model with and without misspecification, the effect of averaging the jackknife over multiple random node labelings, and a model calibrated to the trade data used in the empirical example). 

Sections E.1-E.3 present the expansion of the objective function. Since the expansion is extremely long, I present it in terms of derivatives of various higher-level terms for succinctness. In E.4-E.9 I show the full expressions for the derivatives of these higher-level terms, Section E.10 then shows that expansion for $\widehat{\beta}$. 
In Section F I present a series of lemmas that will be useful in bounding the derivatives displayed in E. Section G then applies these lemmas to bound the derivatives of the higher-level terms. In Section H, these results are in turn applied to bound the derivatives of the objective function.
Section I discusses the expansion for the jackknife estimator and provides the additional results that complement the proof of Theorem 1 in the main appendix.
Section J provides an expansion for the average effect estimator. J.1 provides an expansion for the fixed effect estimates, J.2 provides additional lemmas to bound derivatives of the average effect terms, J.3 gives the expansion and J.4 discusses the expansion applied to the jackknife estimate.

\section{Additional simulations and empirical results} \label{app:additional}

\subsection{Zero-inflated binomial model}

Here we provide another example of how the jackknife correction can
be used. As in the main empirical example, the data are taken from
\cite{Santos2006}, and additional details on their construction can
be found in their paper. The outcome variable $Y_{ij}$ is the \emph{value}
of exports from country $i$ to country $j$. We also use several
covariates to capture homophily in trade relationships, which include:
\emph{log distance}, the log of the distance between the capitals
of the countries\emph{; border}, an indicator of whether the countries
share a common border; \emph{language}, an indicator for whether the
countries share a language; \emph{colonial}, and indicator for whether
either country had colonized the other at some point in history; and
\emph{trade agreement}, an indicator for the presence of a joint preferential
trade agreement between the two countries.

\cite{Burger2009} propose a zero-inflated negative binomial model.
The value of trade between $i$ and $j$ is given by the product of
two variables $Y_{ij}=z_{ij}Y_{ij}^{*}$, where $z_{ij}\in\{0,1\}$
is a binary decision to enter into a trading relationship, while $Y_{ij}^{*}$
is the value of exports that will be realized, conditional on $z_{ij}=1$.
The binary decision is modeled using as a probit function, while the
latent outcome $Y_{ij}^{*}$ is modeled as a negative binomial variable.

In this example, the objective function is given by
\begin{align*}
f(Y_{ij}\vert X_{ij},\theta) & =\mathds{1}\{Y_{ij}=0\}p_{ij}+(1-p_{ij})g(Y_{ij}\vert X_{ij},\theta)
\end{align*}
where $\theta=(\beta,\alpha,\gamma,\nu)$, and
\begin{align*}
p_{ij} & =\Phi(X_{ij}'\beta^{z}+\alpha_{i}^{z}+\gamma_{j}^{z})\\
g(Y_{ij}\vert X_{ij},\theta) & =\frac{\Gamma(Y_{ij}+\nu)}{\Gamma(\nu)Y_{ij}!}\big(\frac{\nu}{\nu+\mu_{ij}}\big)^{\nu}\big(\frac{\mu_{ij}}{\nu+\mu_{ij}}\big)^{Y_{ij}}\\
\mu_{ij} & =\exp(X_{ij}'\beta^{y}+\alpha_{i}^{y}+\gamma_{j}^{y})
\end{align*}
The parameter $\nu$ captures the degree of overdispersion in the
model for $Y_{ij}^{*}$, with $\nu\to\infty$ resulting in a model
with equal mean and variance (as in the Poisson), while smaller $\nu$
lead to greater degrees of dispersion.

Estimates of the parameters in the model are presented in Table \ref{tab:zinb_trade}.
Most variables change by only small amounts after bias correction.
However, the effect of sharing a common border on the probability
of engaging in zero trade changes significantly after bias correction;
while the maximum likelihood estimate suggests that common borders
are important for link formation, the bias corrected estimate is no
longer significant. This suggests that the sharing a common border
has little effect on the likelihood of engaging in trade, but does
affect the volume of trade. The results also suggests a substantial
impact of trade agreements, both on the probability of engaging in
trade and on the volume of trade, a result that is robust to bias
correction. The overdispersion parameter $\nu$ is less than a half,
suggesting a significant amount of overdispersion, i.e. export volumes
have far greater variation across country pairs than suggested by
a Poisson model.

The rightmost column in the table reports the difference between the
MLE and jackknife bias-corrected estimators in terms of their standard
errors. For a number of variables in the model of export volumes,
the change in estimate is around three-quarters of the standard deviation
or more, which has an important impact on inference. To give some
idea of the scale of these biases, a bias of three-quarters of a standard
error is enough for a five per cent test two reject around 12 per
cent of the time (more than twice nominal size), while bias of 1.5
standard errors leads to a rejection rate of more than 30 per cent.

\noindent 
\begin{table}[H]
\noindent \centering\small
\begin{threeparttable}

\noindent \caption{\label{tab:zinb_trade}Zero-inflated negative binomial model}

\noindent %
\begin{tabular}{lcccc} 
\hline 
 & \multicolumn{4}{c}{Coefficients}\tabularnewline
 & MLE & Jackknife & SE & (Bias/SE)\tabularnewline
\hline 
Zero model &  &  &  & \tabularnewline
\emph{log distance} & 0.721 & 0.721 & 0.029 & 0.00\tabularnewline
\emph{border} & 0.628 & 0.157 & 0.120 & 3.93\tabularnewline
\emph{language} & -0.330 & -0.306 & 0.053 & 0.45\tabularnewline
\emph{colonial} & -0.305 & -0.282 & 0.056 & 0.41\tabularnewline
\emph{trade agreement} & -1.168 & -1.126 & 0.180 & 0.24\tabularnewline
\hline 
Volume model &  &  &  & \tabularnewline
\emph{log distance} & -1.243 & -1.218 & 0.033 & 0.77\tabularnewline
\emph{border} & 0.437 & 0.483 & 0.129 & 0.36\tabularnewline
\emph{language} & 0.405 & 0.418 & 0.068 & 0.18\tabularnewline
\emph{colonial} & 0.399 & 0.335 & 0.073 & 0.88\tabularnewline
\emph{trade agreement} & 1.055 & 0.960 & 0.131 & 0.73\tabularnewline
$\nu$ & 0.492 & 0.459 & 0.008 & 4.38\tabularnewline
\hline 
\end{tabular}

\end{threeparttable}
\end{table}

Table \ref{tab:zinb_avg} contains estimates of the average effect
of a regressor on expected export volume, conditional on non-zero
trade, over the distribution of regressors and fixed effects for trading
country pairs. That is, we calculate (for $n_{1}=\sum_{i}\sum_{j\ne i}\mathds{1}\{Y_{ij}>0\}$)
\[
\Delta_{N}=\frac{1}{n_{1}}\sum_{i}\sum_{j\ne i}\mathds{1}\{Y_{ij}>0\}\beta_{dist}\exp(X_{ij}'\beta^{y}+\alpha_{i}^{y}+\gamma_{j}^{y})
\]
for the continuous regressor \emph{log distance} and
\[
\Delta_{N}=\frac{1}{n_{1}}\sum_{i}\sum_{j\ne i}\mathds{1}\{Y_{ij}>0\}\big(\exp({\beta^{y}}'X_{ij}^{(1)}+\alpha_{i}^{y}+\gamma_{j}^{y})-\exp({\beta^{y}}'X_{ij}^{(0)}+\alpha_{i}^{y}+\gamma_{j}^{y})\big)
\]
for binary regressors, where $X_{ij}^{(1)}$ sets the binary regressor
of interest to one for all $(i,j)$ and $X_{ij}^{(0)}$ sets it to
zero (leaving other regressors unchanged). Again, the jackknife bias
correction has an important impact on two of the effects; for example,
the effect of a trade agreement on expected export volumes decreases
by about a quarter (more than a one standard error change in magnitude).
Note that, as is the case here, a small bias in the coefficient on
some variable does not necessarily imply low bias in the corresponding
marginal effect.

\noindent 
\begin{table}[H]
\noindent \centering \small
\begin{threeparttable}

\noindent \caption{\label{tab:zinb_avg}Zero-inflated negative binomial model}

\noindent %
\begin{tabular}{lcccc}
\hline 
 & \multicolumn{4}{c}{Average effects}\tabularnewline
 & MLE & Jackknife & SE & (Bias/SE)\tabularnewline
\hline 
\emph{log distance} & -116.2 & -113.8 & 9.08 & 0.26\tabularnewline
\emph{border} & 47.5 & 50.2 & 16.95 & 0.16\tabularnewline
\emph{language} & 43.4 & 41.0 & 8.64 & 0.28\tabularnewline
\emph{colonial} & 44.4 & 31.0 & 10.00 & 1.34\tabularnewline
\emph{trade agreement} & 140.2 & 107.1 & 28.10 & 1.18\tabularnewline
\hline 
\end{tabular}

\end{threeparttable}
\end{table}

\subsection{Simulations of misspecified logit model}

In order to demonstrate how the jackknife bias correction from the conditional logit estimator under misspecification, I compare the estimators
in simulations of a network formation using both logistic (correctly specified) and normally distributed (misspecified) errors. The simulation design
is the same as that used in the main paper, although the fixed effects distributions are adjusted slightly for the logit model in order to 
match the degree distributions in the probit specification. 

As can be seen in the table, under correct specification, all estimators are approximately unbiased, other than the uncorrected MLE estimator. Suprisingly, the weighted jackknife has lower bias than the CMLE in the sparest design, but this would not be expected to hold in general, or in even sparser settings. In the misspecified case the jackknife and CMLE estimators deliver different results. Bias is reported relative to the true maximizer of the expected value of the log-likelihood (the object targeted by the MLE). It is clear that the jackknife estimators are approximately 
unbiased relative to these pseudo-true values. In contrast, CMLE estimates a different pseudo-true value, since it is based on maximum likelihood estimation of the two-way differences $(Y_{ij}-Y_{il}) - (Y_{kj}-Y_{kl})$. The choice of bias correction can be important under misspecification, and so the researcher should carefully consider to the different interpretations of the pseudo-true parameter values that are estimated by the jackknife relative to the CMLE.

\begin{table}[H]
\centering \small
\caption{Misspecified model (true DGP probit)}
\begin{tabular}{lccccccccc}
	\hline 
	& Bias & SD & RMSE & Rej (5\%) &  & Bias & SD & RMSE & Rej (5\%)\tabularnewline
	\hline 
	& \multicolumn{4}{c}{$-\log(\log n),\log(\log n)$} &  & \multicolumn{4}{c}{$-\log(\log n),0$}\tabularnewline
	\cline{2-10}
	MLE & 0.105 & 0.081 & 0.133 & 0.243 &  & 0.109 & 0.127 & 0.167 & 0.115\tabularnewline
	Jackknife & -0.030 & 0.072 & 0.078 & 0.044 &  & -0.035 & 0.135 & 0.139 & 0.068\tabularnewline
	Jackknife (weighted) & -0.023 & 0.073 & 0.076 & 0.040 &  & -0.010 & 0.121 & 0.122 & 0.059\tabularnewline
	CMLE & 0.087 & 0.110 & 0.140 & 0.032 &  & 0.087 & 0.189 & 0.209 & 0.007\tabularnewline
	\hline 
	& \multicolumn{4}{c}{$-(\log n)^{1/2},0$} &  & \multicolumn{4}{c}{$-\log n,0$}\tabularnewline
	\cline{2-10}
	MLE & 0.179 & 0.370 & 0.411 & 0.101 &  & 1.792 & 2.309 & 2.923 & 0.014\tabularnewline
	Jackknife & -0.123 & 0.587 & 0.600 & 0.085 &  & -1.704 & 4.382 & 4.702 & 0.523\tabularnewline
	Jackknife (weighted) & 0.033 & 0.373 & 0.375 & 0.063 &  & 1.232 & 2.286 & 2.597 & 0.054\tabularnewline
	CMLE & 0.311 & 0.792 & 0.851 & 0.068 &  & 2.455 & 1.481 & 2.868 & 0.770\tabularnewline
	\hline 
\end{tabular}

\end{table}

\begin{table}[H]
\centering \small \caption{Correctly specified model (true DGP logit)}
\begin{tabular}{lccccccccc}
	\hline 
	& Bias & SD & RMSE & Rej (5\%) &  & Bias & SD & RMSE & Rej (5\%)\tabularnewline
	\hline 
	& \multicolumn{4}{c}{$-\log(\log n),\log(\log n)$} &  & \multicolumn{4}{c}{$-\log(\log n),0$}\tabularnewline
	\cline{2-10}
	MLE & 0.057 & 0.058 & 0.081 & 0.143 &  & 0.056 & 0.075 & 0.093 & 0.129\tabularnewline
	Jackknife & -0.002 & 0.054 & 0.054 & 0.034 &  & -0.004 & 0.070 & 0.070 & 0.056\tabularnewline
	Jackknife (weighted) & 0.000 & 0.054 & 0.054 & 0.034 &  & 0.001 & 0.070 & 0.070 & 0.055\tabularnewline
	CMLE & 0.007 & 0.058 & 0.058 & 0.021 &  & 0.009 & 0.072 & 0.073 & 0.046\tabularnewline
	\hline 
	& \multicolumn{4}{c}{$-(\log n)^{1/2},0$} &  & \multicolumn{4}{c}{$-\log n,0$}\tabularnewline
	\cline{2-10}
	MLE & 0.069 & 0.099 & 0.121 & 0.119 &  & 0.128 & 0.227 & 0.261 & 0.116\tabularnewline
	Jackknife & -0.008 & 0.091 & 0.091 & 0.040 &  & -0.048 & 0.264 & 0.269 & 0.046\tabularnewline
	Jackknife (weighted) & 0.001 & 0.092 & 0.092 & 0.043 &  & 0.002 & 0.211 & 0.211 & 0.046\tabularnewline
	CMLE & 0.014 & 0.097 & 0.098 & 0.026 &  & 0.040 & 0.219 & 0.223 & 0.014\tabularnewline
	\hline 
\end{tabular}

\end{table}

\subsection{Averaging over multiple random node labelings}
Table \ref{tab:label} investigates the impact of computing the average of jackknife estimators over a number of random node labelings. The table shows the bias, standard deviation, and 5-95th percentile range of estimators that average 1, 10 or 50 different jackknife estimates. In each case, the bias of the estimator is essentially unaffected. In the dense network cases it is clear that the additional averaging also has very little impact on the dispersion of the estimator, with standard deviations being very similar across the estimators. In the sparsest design, the standard deviation of the jackknife is around 5\% smaller when averaging over 10 relabelings (increasing this to 50 seems to have little effect). The weighted jackknife appears unaffected by the relabeling in any of the designs. Given these findings, it is suggested that in less dense settings, recomputing the jackknife a small number of times may be useful to ensure that outlier cases are avoided, while in denser settings this is likely unnecessary.

\begin{table}[H]
\centering\caption{\label{tab:label}Simulations using multiple node labelings} \small
\begin{threeparttable}
\begin{tabular}{llcccccccc}
	\hline 
	&  & \multicolumn{4}{c}{Jackknife} & \multicolumn{4}{c}{Weighted jackknife}\tabularnewline
	\cline{3-10}
	& $k$ & A & B & C & D & A & B & C & D\tabularnewline
	\hline 
	Bias & 1 & -0.009 & -0.011 & -0.040 & -0.513 & -0.006 & -0.001 & 0.008 & 0.295\tabularnewline
	& 10 & -0.009 & -0.011 & -0.038 & -0.458 & -0.006 & -0.001 & 0.008 & 0.305\tabularnewline
	& 50 & -0.009 & -0.011 & -0.038 & -0.457 & -0.006 & -0.001 & 0.007 & 0.307\tabularnewline
	\hline 
	SD & 1 & 0.039 & 0.057 & 0.163 & 1.098 & 0.039 & 0.056 & 0.143 & 0.533\tabularnewline
	& 10 & 0.039 & 0.057 & 0.147 & 1.041 & 0.039 & 0.056 & 0.125 & 0.517\tabularnewline
	& 50 & 0.039 & 0.057 & 0.146 & 1.042 & 0.039 & 0.056 & 0.123 & 0.516\tabularnewline
	\hline 
	$5-95^{**}$ & 1 & 0.131 & 0.176 & 0.246 & 3.027 & 0.132 & 0.179 & 0.286 & 1.244\tabularnewline
	& 10 & 0.131 & 0.174 & 0.242 & 2.961 & 0.132 & 0.180 & 0.283 & 1.266\tabularnewline
	& 50 & 0.131 & 0.174 & 0.237 & 2.981 & 0.132 & 0.180 & 0.284 & 1.264\tabularnewline
	\hline 
\end{tabular}
\begin{tablenotes}
	\item *A = $(-\log(\log n), \log(\log n))$, B=$(-\log(\log n),0)$, C=$(-(\log n)^{1/2},0)$,  and D=$(-\log n,0)$.
	\item **difference between the 95th and 5th percentiles of simulated distribution
\end{tablenotes}
\end{threeparttable}
\end{table}

\subsection{Simulations calibrated to empirical data}

Table \ref{tab:emp_sims} reports simulation results from a data generating process calibrated to the empirical example presented in the main paper (the trade data of Silva and Tenreyo, 2006). To generate the data sets, we draw random samples of $N=(50,70,100,134)$ countries without replacement (from the 134 total countries) and construct the network links according to $Y^*_{ij} = 1\{U_{ij}\leq p_{ij}\}$, where $U_{ij}$ is a uniform random variable and $p_{ij}=\Phi(\beta'X_{ij}+\alpha_i +\gamma_j)$ are the fitted probabilities estimated using the full sample of 134 countries. For the only continuous regressor, log distance, the bias in the MLE is close to one standard deviation in magnitude for each of the sample sizes so that the simulated rejection rate is two or three times larger than the nominal 5\%. Both the standard and weighted jackknife estimators achieve low bias and rejection rates below the nominal rate. For the remaining variables, it appears that the bias is only small  relative to the standard error so that he MLE estimator only slightly over-rejects, which matches the results found in the empirical analysis. The jackknife estimator performs well in all cases expecting the estimation of the coefficient on the trade agreement dummy variable when the sample size is smaller ($N=50$). Only around 2\% of all country pairs have a trade agreement in the data set, so that in the $N=50$ samples the majority of countries have no agreements with any other country in the samples. This is of course worsened in the leaveout subsamples used in the jackknife and highlights that extreme lack of variation in regressors can also be problematic for jackknife bias correction.  Nonetheless, the weighted jackknife still has a reasonable reject rate of 7.2\%.

\begin{table}
	\caption{Simulation results - trade calibration} \label{tab:emp_sims}
	\centering \small
	\begin{threeparttable}
	\begin{tabular}{ccccccccccc}
		\hline 
		&  & \multicolumn{3}{c}{Bias} & \multicolumn{3}{c}{SD} & \multicolumn{3}{c}{Rej (5\%)}\tabularnewline
		\cline{3-11}
		Coefficient & $N$ & MLE & Jack. & Jack.(w) & MLE & Jack. & Jack.(w) & MLE & Jack. & Jack.(w)\tabularnewline
		\hline 
		log dist. & 50 & -0.072 & 0.005 & -0.006 & 0.088 & 0.078 & 0.079 & 0.168 & 0.048 & 0.040\tabularnewline
		& 70 & -0.041 & 0.007 & 0.002 & 0.058 & 0.054 & 0.054 & 0.108 & 0.040 & 0.036\tabularnewline
		& 100 & -0.031 & 0.000 & -0.002 & 0.038 & 0.037 & 0.037 & 0.132 & 0.040 & 0.036\tabularnewline
		& 134 & -0.021 & 0.001 & 0.000 & 0.027 & 0.027 & 0.027 & 0.108 & 0.036 & 0.036\tabularnewline
		\hline 
		language & 50 & 0.034 & 0.003 & 0.008 & 0.169 & 0.153 & 0.156 & 0.068 & 0.032 & 0.032\tabularnewline
		& 70 & 0.023 & 0.003 & 0.005 & 0.105 & 0.097 & 0.098 & 0.040 & 0.036 & 0.036\tabularnewline
		& 100 & 0.003 & -0.009 & -0.008 & 0.069 & 0.067 & 0.067 & 0.044 & 0.040 & 0.044\tabularnewline
		& 134 & 0.008 & -0.001 & -0.001 & 0.051 & 0.050 & 0.050 & 0.044 & 0.040 & 0.040\tabularnewline
		\hline 
		colony & 50 & 0.021 & -0.009 & -0.005 & 0.182 & 0.164 & 0.167 & 0.076 & 0.044 & 0.044\tabularnewline
		& 70 & 0.004 & -0.015 & -0.014 & 0.117 & 0.109 & 0.110 & 0.076 & 0.052 & 0.060\tabularnewline
		& 100 & 0.021 & 0.008 & 0.009 & 0.074 & 0.071 & 0.071 & 0.044 & 0.040 & 0.040\tabularnewline
		& 134 & 0.004 & -0.005 & -0.004 & 0.056 & 0.055 & 0.055 & 0.060 & 0.064 & 0.056\tabularnewline
		\hline 
		border & 50 & -0.046 & 0.002 & 0.016 & 0.384 & 0.402 & 0.365 & 0.072 & 0.048 & 0.052\tabularnewline
		& 70 & -0.040 & 0.006 & 0.000 & 0.242 & 0.227 & 0.229 & 0.048 & 0.036 & 0.036\tabularnewline
		& 100 & -0.034 & -0.004 & -0.006 & 0.173 & 0.166 & 0.167 & 0.060 & 0.056 & 0.060\tabularnewline
		& 134 & -0.038 & -0.016 & -0.017 & 0.124 & 0.121 & 0.121 & 0.068 & 0.064 & 0.068\tabularnewline
		\hline 
		trade agr. & 50 & 0.454 & -3.831 & -3.018 & 1.319 & 9.012 & 9.094 & 0.064 & 0.180 & 0.072\tabularnewline
		& 70 & 0.154 & -0.614 & -0.450 & 0.574 & 4.067 & 4.062 & 0.052 & 0.044 & 0.024\tabularnewline
		& 100 & 0.058 & -0.023 & 0.019 & 0.240 & 0.293 & 0.236 & 0.024 & 0.016 & 0.016\tabularnewline
		& 134 & 0.031 & -0.011 & 0.000 & 0.179 & 0.172 & 0.174 & 0.060 & 0.056 & 0.052\tabularnewline
		\hline 
	\end{tabular}
	\end{threeparttable}
\end{table}

\section{Asymptotic expansion} \label{SA:expansion}

\subsection{Application of results in Fern\'andez-Val and Weidner (2016)}
\cite{Fernandez-Val2016} (FW16) derive asymptotic expansion for two-fixed effect models. The results in this paper are based on extensions of these expansion to higher-order. I make use of the many of the results that have already been proven in that paper (sometimes using the stronger moment conditions assumed in this paper to allow for the higher-order terms). \cite{Dzemski2019} provides details on how the FW16 results can be applied to a network model. In particular, note that under Assumptions 1 and 2 in this paper, Assumption B.1 in FW16 holds with $q=15$ and $\epsilon=1/30$ (see for example proofs in Lemma S.7 of FW16, or in the results or Lemma A.2 in \cite{Dzemski2019}). A key result that follows from this is consistency of both the common parameters and fixed effects (see Corollary B.3 in FW16). More details on these key results can be found in Section \ref{sec:basic_lemmas}.

\subsection*{Notation}
The notation in the Supplementary Appendix follows that in the main paper. In addition, we 
define the set of norms which will be used to bound expansion terms. We follow FW16 in using the Euclidean
norm $\lVert\cdot\rVert$ for $\dim\beta$ vectors, and the norm induced
by the Euclidean norm for matrices and tensors, i.e. 
\[
\lVert\partial_{\beta\beta\beta}\mathcal{L}(\beta,\phi)\rVert=\max_{u,v\in\mathbb{R}^{\dim\beta}:\lVert u\rVert=1,\lVert v\rVert=1}\lVert\sum_{k,l=1}^{\dim\beta}u_{k}v_{l}\partial_{\beta\beta_{k}\beta_{l}}\mathcal{L}(\beta,\phi)\rVert
\]
In the proofs we sometimes take $\beta$ to be a scalar to simplify
notation, although the results apply to any vector of fixed size.
Since the number of fixed effect parameters in the model grows with
$N$, the choice of norm for $\dim\phi$ vectors and matrices is important.
Following FW16, we choose the $\ell_{q}$-norm for $\dim\phi$ vectors
and the corresponding induced norms for matrices and tensors 
\[
\lVert\partial_{\phi\phi\phi}\mathcal{L}(\beta,\phi)\rVert=\max_{u,v\in\mathbb{R}^{\dim\phi}:\lVert u\rVert=1,\lVert v\rVert=1}\lVert\sum_{k,l=1}^{\dim\phi}u_{k}v_{l}\partial_{\phi\phi_{k}\phi_{l}}\mathcal{L}(\beta,\phi)\rVert_{q}
\]
See FW16 for more details on these norms. When omit the subscript $q$ whenever $q=2$ for vectors and matrices, i.e.\ for the 
Euclidean norm and spectral norm resepctively.
We define the sets $\mathcal{B}(r,\beta_{0})=\{\beta:\lVert\beta-\beta_{0}\rVert\leq r\}$,
for $r>0$, and $\mathcal{B}_{q}(r,\phi_{0})=\{\phi:\lVert\phi-\phi_{0}\rVert_{q}\leq r\}$.

\subsection{Expansion of the objective function}

The asymptotic expansion is based on that derived in FW16. For convenience in the
jackknife results that follow, we make a small change to the Legendre
transformed objective function described in S.3.1 of FW16, by also
creating a dual for the common parameter $\beta$. We describe the
transform only briefly as it differs little from FW16.

Let $\mathcal{B}(r_{\beta},\beta_{0})\times\mathcal{B}_{q}(r_{\phi},\phi_{0})$
be a shrinking neighborhood of the true parameters, and define the
Legendre transformation of the objective function
\begin{align*}
\mathcal{L}^{*}(b,s) & =\max_{\beta\in\mathcal{B}(r_{\beta},\beta_{0}),\phi\in\mathcal{B}_{q}(r_{\phi},\beta_{0})}\big[\mathcal{L}(\beta,\phi)-b\beta-\phi's\big]\\
\big(B(b,s),\Phi(b,s)\big) & =\arg\max_{\beta\in\mathcal{B}(r_{\beta},\beta_{0}),\phi\in\mathcal{B}_{q}(r_{\phi},\beta_{0})}\big[\mathcal{L}(\beta,\phi)-b\beta-\phi's\big]
\end{align*}

Since $B,\Phi$ are the optimal values defined by the first-order
conditions, we have $\mathcal{L}^{*}(b,s)=\mathcal{L}(B(b,s),\Phi(b,s))-B(b,s)b-\Phi(b,s)'s$,
and the first-order conditions
\begin{align*}
\partial_{\beta}\mathcal{L}(B,\Phi) & =b\\
\partial_{\phi}\mathcal{L}(B,\Phi) & =s
\end{align*}

We can differentiate these two FOCs with respect to $b$ and $s$
to give
\begin{equation}
\begin{aligned}\partial_{\beta\beta}\mathcal{L}(B,\Phi)(\partial_{b}B)+\partial_{\beta\phi'}\mathcal{L}(B,\Phi)(\partial_{b}\Phi) & =1\\
\partial_{\beta\phi}\mathcal{L}(B,\Phi)(\partial_{b}B)+\partial_{\phi\phi'}\mathcal{L}(B,\Phi)(\partial_{b}\Phi) & =0\\
\partial_{\beta\beta}\mathcal{L}(B,\Phi)(\partial_{s'}B)+\partial_{\beta\phi'}\mathcal{L}(B,\Phi)(\partial_{s'}\Phi) & =0\\
\partial_{\beta\phi}\mathcal{L}(B,\Phi)(\partial_{s'}B)+\partial_{\phi\phi'}\mathcal{L}(B,\Phi)(\partial_{s'}\Phi) & =I
\end{aligned}
\label{eq:FOC1}
\end{equation}
Using these equations, we may solve for the derivatives of $B(b,s)$
and $\Phi(b,s)$. Let $\mathcal{H}=-\partial_{\phi\phi'}\mathcal{L}$,
and define
\[
W=-\Big(\partial_{\beta\beta}\mathcal{L}+(\partial_{\beta\phi'}\mathcal{L})\mathcal{H}^{-1}(\partial_{\beta\phi}\mathcal{L})\Big)
\]
The objects $W$ and $\mathcal{H}$ will play important roles in the
expansion. Then we have
\begin{equation}
\begin{aligned}\partial_{b}B & =-W^{-1}\\
\partial_{b}\Phi & =-\mathcal{H}^{-1}(\partial_{\beta\phi}\mathcal{L})W^{-1}\\
\partial_{s'}B & =-W^{-1}(\partial_{\beta\phi'}\mathcal{L})\mathcal{H}^{-1}=\partial_{b}\Phi'\\
\partial_{s'}\Phi & =-\mathcal{H}^{-1}-W^{-1}\mathcal{H}^{-1}(\partial_{\beta\phi}\mathcal{L})(\partial_{\beta\phi}\mathcal{L})'\mathcal{H}^{-1}
\end{aligned}
\label{eq:dB1}
\end{equation}
We compute the terms in the asymptotic expansions for a
scalar parameter $\beta$. Since $\beta$ has finite dimension, this
does not affect the asymptotic order of the terms, but it does significantly
simplify the notation. Recall that by the FOC of the objective function
we have $\partial_{b}\mathcal{L}^{*}(0,0)=-\widehat{\beta}$ and $\partial_{b}\mathcal{L}^{*}(\mathcal{S}_{\beta},\mathcal{S})=-\beta_{0}$.
Expanding around this point gives

\begin{align*}
	\partial_{b}\mathcal{L}^{*}(0,0) & =\partial_{b}\mathcal{L}^{*}-(\partial_{bb}\mathcal{L}^{*})\mathcal{S}_{\beta}-(\partial_{bs'}\mathcal{L}^{*})\mathcal{S}\\
	& +\frac{1}{2}(\partial_{bbb}\mathcal{L}^{*})\mathcal{S}_{\beta}^{2}+(\partial_{bbs'}\mathcal{L}^{*})\mathcal{S}\mathcal{S}_{\beta}+\frac{1}{2}\mathcal{S}'(\partial_{bss'}\mathcal{L}^{*})\mathcal{S}\\
	& -\frac{1}{6}(\partial_{b^{4}}\mathcal{L}^{*}(\bar{b},\bar{s}))\mathcal{S}_{\beta}^{3}-\frac{1}{2}(\partial_{bbbs'}\mathcal{L}^{*}(\bar{b},\bar{s}))\mathcal{S}\mathcal{S}_{\beta}^{2}\\
	& -\frac{1}{2}\mathcal{S}'(\partial_{bbss'}\mathcal{L}^{*})\mathcal{S}\mathcal{S}_{\beta} -\frac{1}{6}\sum_{g}\mathcal{S}'(\partial_{bss's_{g}}\mathcal{L}^{*})\mathcal{S}\mathcal{S}_{g}\\
	& +\frac{1}{4}\mathcal{S}'(\partial_{bbbss'}\mathcal{L}^{*}(\bar{b},\bar{s}))\mathcal{S}\mathcal{S}_{\beta}^{2}+\frac{1}{6}\sum_{g}\mathcal{S}'(\partial_{bbss's_{g}}\mathcal{L}^{*}(\bar{b},\bar{s}))\mathcal{S}\mathcal{S}_{g}\mathcal{S}_{\beta}\\
	& +\frac{1}{24}\sum_{f,g}\mathcal{S}'(\partial_{bss's_{f}s_{g}}\mathcal{L}^{*})\mathcal{S}\mathcal{S}_{f}\mathcal{S}_{g}-\frac{1}{24}\sum_{f,g}\mathcal{S}'(\partial_{bbss's_{f}s_{g}}\mathcal{L}^{*}(\bar{b},\bar{s}))\mathcal{S}\mathcal{S}_{f}\mathcal{S}_{g}\mathcal{S}_{\beta}\\
	& -\frac{1}{120}\sum_{e,f,g}\mathcal{S}'(\partial_{bss's_{e}s_{f}s_{g}}\mathcal{L}^{*}(\bar{b},\bar{s}))\mathcal{S}\mathcal{S}_{e}\mathcal{S}_{f}\mathcal{S}_{g}\\
\end{align*}

where $(\bar{b},\bar{s})$ are intermediate values between $(0,0)$
and $(\mathcal{S}_{\beta},\mathcal{S})$. The next section gives expressions for each of these derivatives in terms of the derivatives of the original objective function $\mathcal{L}$.

\subsection{\label{sec:expressions_ind} Detailed expansion terms}

This section contains expressions for each of the derivatives of $\mathcal{L}^*$  used in the expansion in terms of derivatives of the objective function $\mathcal{L}$. Since the expressions contain a large number of terms, in order to simplify notation we first define some new terms and give the expressions in terms of derivatives of these terms. Subsequent subsections then give expressions for each of the derivatives of these compound terms.

Let $\mathcal{E}^{s}=\partial_{\beta^{s}\phi\phi'}\mathcal{L}$, $\mathcal{F}^{s,t}=(\partial_{\beta^{s}\phi'}\mathcal{L})\mathcal{H}^{-1}(\partial_{\beta^{t}\phi}\mathcal{L})$
and $\mathcal{F}=\mathcal{F}^{1,1}$. Also, let 
\begin{align*}
	\mathcal{F}^{s,(r),t} & =(\partial_{\beta^{s}\phi'}\mathcal{L})\mathcal{H}^{-1}\mathcal{E}^{r}\mathcal{H}^{-1}(\partial_{\beta^{t}\phi}\mathcal{L})\\
	\mathcal{F}^{s,(r_{1},r_{2}),t} & =(\partial_{\beta^{s}\phi'}\mathcal{L})\mathcal{H}^{-1}\mathcal{E}^{r_{1}}\mathcal{H}^{-1}\mathcal{E}^{r_{2}}\mathcal{H}^{-1}(\partial_{\beta^{t}\phi}\mathcal{L})
\end{align*}
Finally, we define
\begin{align*}
	\mathcal{G}^{s,t} & =\mathcal{H}^{-1}(\partial_{\beta^{s}\phi}\mathcal{L})(\partial_{\beta^{t}\phi'}\mathcal{L})\mathcal{H}^{-1}
\end{align*}
and $\mathcal{G}=\mathcal{G}^{1,1}$. We denote derivatives of these
terms with respect to $b$ or $s$ using subscripts, e.g. $\mathcal{G}_{b}=\partial\mathcal{G}/\partial b$,
$\mathcal{G}_{s_{g}}=\partial\mathcal{G}/\partial s_{g}$ where $s_{g}$
is the $g$-th element of the vector $s$, etc.

We take derivatives of the transformed likelihood, evaluated at $b^{*}=\partial_{\beta}\mathcal{L}(\beta_{0},\phi_{0})=\mathcal{S}_{\beta}$,
and $s^{*}=\partial_{\phi}\mathcal{L}(\beta_{0},\phi_{0})=\mathcal{S}$.
When evaluated at these values, the arguments are not written out,
e.g. $\partial_{b}\mathcal{L}^{*}=\partial_{b}\mathcal{L}^{*}(b^{*},s^{*})$.
Note that $B(b^{*},s^{*})=\beta_{0}$ and $\Phi(b^{*},s^{*})=\phi_{0}$,
since they satisfy the first order conditions
\begin{align*}
	\partial_{\beta}\mathcal{L}(B(b^{*},s^{*}),\Phi(b^{*},s^{*})) & =b^{*}=\partial_{\beta}\mathcal{L}(\beta_{0},\phi_{0})\\
	\partial_{\phi}\mathcal{L}(B(b^{*},s^{*}),\Phi(b^{*},s^{*})) & =s^{*}=\partial_{\phi}\mathcal{L}(\beta_{0},\phi_{0})
\end{align*}

Differentiating $\mathcal{L}^{*}(b,s)=\mathcal{L}(B(b,s),\Phi(b,s))-B(b,s)b-s'\Phi(b,s)$
with respect to $b$ and $s$ and evaluating at $(b^{*},s^{*})$ gives
\begin{align*}
	\partial_{b}\mathcal{L}^{*} & =\big((\partial_{\beta}\mathcal{L})-b\big)(\partial_{b}B)+\big((\partial_{\phi'}\mathcal{L})-s'\big)(\partial_{b}\Phi)-B(b,s)\\
	& =-\beta_{0}\\
	\partial_{s'}\mathcal{L}^{*} & =\big((\partial_{\beta}\mathcal{L})-b\big)(\partial_{s'}B)+\big((\partial_{\phi'}\mathcal{L})-s'\big)(\partial_{s'}\Phi)-\Phi(b,s)'\\
	& =-\phi_{0}
\end{align*}
Taking second derivatives of these functions, and applying the identities
in (\ref{eq:dB1})
\begin{align*}
	\partial_{bb}\mathcal{L}^{*} & =\big((\partial_{\beta\beta}\mathcal{L})(\partial_{b}B)+(\partial_{\beta\phi'}\mathcal{L})(\partial_{b}\Phi)-1\big)(\partial_{b}B)\\
	& +\big((\partial_{\beta}\mathcal{L})-b\big)(\partial_{bb}B)\\
	& +\big((\partial_{\beta\phi'}\mathcal{L})(\partial_{b}B)+(\partial_{b}\Phi)'(\partial_{\phi\phi'}\mathcal{L})\big)(\partial_{b}\Phi)\\
	& +\big((\partial_{\phi'}\mathcal{L})-s'\big)(\partial_{bb}\Phi)\\
	& -(\partial_{b}B)\\
	& =-(\partial_{b}B)\\
	& =W^{-1}\\
	\partial_{bs'}\mathcal{L}^{*} & =\big((\partial_{\beta\beta}\mathcal{L})(\partial_{s'}B)+(\partial_{\beta\phi'}\mathcal{L})(\partial_{s'}\Phi)\big)(\partial_{b}B)\\
	& +\big((\partial_{\beta}\mathcal{L})-b\big)(\partial_{bs'}B)\\
	& +\big((\partial_{\beta\phi'}\mathcal{L})(\partial_{s'}B)+(\partial_{\phi\phi'}\mathcal{L})(\partial_{s}\Phi)-I\big)(\partial_{b}\Phi)\\
	& +\big((\partial_{\phi'}\mathcal{L})-s'\big)(\partial_{bs'}\Phi)\\
	& -(\partial_{s'}B)\\
	& =-(\partial_{s'}B)\\
	& =W^{-1}(\partial_{\beta\phi'}\mathcal{L})\mathcal{H}^{-1}\\
	\partial_{ss'}\mathcal{L}^{*} & =\big((\partial_{\beta}\mathcal{L})-b\big)(\partial_{ss'}B)+\sum_{g}\big((\partial_{\phi'}\mathcal{L})-s'\big)_{g}(\partial_{bss'}\Phi_{g})-\partial_{s'}\Phi\\
	& =-\partial_{s'}\Phi\\
	& =\mathcal{H}^{-1}+W^{-1}\mathcal{G}
\end{align*}

We proceed by taking further derivatives of these expressions, substituting
derivatives of the identities in (\ref{eq:dB1}).

Third derivatives:
\begin{align*}
	\partial_{bbb}\mathcal{L}^{*} & =-W^{-2}W_{b}\\
	\partial_{bbs}\mathcal{L}^{*} & =-W^{-2}W_{s}\\
	\partial_{bss'}\mathcal{L}^{*} & =-\mathcal{H}^{-1}\mathcal{H}_{b}\mathcal{H}^{-1}-W^{-2}W_{b}\mathcal{G}+W^{-1}\mathcal{G}_{b}\\
	\partial_{ss's_{g}}\mathcal{L}^{*} & =-\mathcal{H}^{-1}\mathcal{H}_{s_{g}}\mathcal{H}^{-1}-W^{-2}W_{s_{g}}\mathcal{G}+W^{-1}\mathcal{G}_{s_{g}}
\end{align*}

Fourth derivatives:
\begin{align*}
	\partial_{b^{4}}\mathcal{L}^{*} & =2W^{-3}W_{b}^{2}-W^{-2}W_{bb}\\
	\partial_{bbbs}\mathcal{L}^{*} & =2W^{-3}W_{b}W_{s}-W^{-2}W_{bs}\\
	\partial_{bbss'}\mathcal{L}^{*} & =\partial_{bbss'}\mathcal{L}_{(1)}^{*}+\partial_{bbss'}\mathcal{L}_{(2)}^{*}\\
	\partial_{bbss'}\mathcal{L}_{(1)}^{*} & =2\mathcal{H}^{-1}\mathcal{H}_{b}\mathcal{H}^{-1}\mathcal{H}_{b}\mathcal{H}^{-1}-\mathcal{H}^{-1}\mathcal{H}_{bb}\mathcal{H}^{-1}\\
	\partial_{bbss'}\mathcal{L}_{(2)}^{*} & =2W^{-3}W_{b}^{2}\mathcal{G}-W^{-2}W_{bb}\mathcal{G}-W^{-2}W_{b}\mathcal{G}_{b}\\
	& +W^{-1}\mathcal{G}_{bb}\\
	\partial_{bss's_{g}}\mathcal{L}^{*} & =\partial_{bss's_{g}}\mathcal{L}_{(1)}^{*}+\partial_{bss's_{g}}\mathcal{L}_{(2)}^{*}\\
	\partial_{bss's_{g}}\mathcal{L}_{(1)}^{*} & =\mathcal{H}^{-1}\mathcal{H}_{b}\mathcal{H}^{-1}\mathcal{H}_{s_{g}}\mathcal{H}^{-1}+\mathcal{H}^{-1}\mathcal{H}_{s_{g}}\mathcal{H}^{-1}\mathcal{H}_{b}\mathcal{H}^{-1}\\
	& -\mathcal{H}^{-1}\mathcal{H}_{bs_{g}}\mathcal{H}^{-1}\\
	\partial_{bss's_{g}}\mathcal{L}_{(2)}^{*} & =2W^{-3}W_{b}W_{s_{g}}\mathcal{G}-W^{-2}W_{bs_{g}}\mathcal{G}-W^{-2}W_{s_{g}}\mathcal{G}_{b}\\
	& -W^{-2}W_{b}\mathcal{G}_{s_{g}}+W^{-1}\mathcal{G}_{bs_{g}}\\
	\partial_{ss's_{g}s_{h}}\mathcal{L}^{*} & =\partial_{ss's_{g}s_{h}}\mathcal{L}_{(1)}^{*}+\partial_{ss's_{g}s_{h}}\mathcal{L}_{(2)}^{*}\\
	\partial_{ss's_{g}s_{h}}\mathcal{L}_{(1)}^{*} & =\mathcal{H}^{-1}\mathcal{H}_{s_{g}}\mathcal{H}^{-1}\mathcal{H}_{s_{h}}\mathcal{H}^{-1}+\mathcal{H}^{-1}\mathcal{H}_{s_{h}}\mathcal{H}^{-1}\mathcal{H}_{s_{g}}\mathcal{H}^{-1}\\
	& -\mathcal{H}^{-1}\mathcal{H}_{s_{g}s_{h}}\mathcal{H}^{-1}\\
	\partial_{ss's_{g}s_{h}}\mathcal{L}_{(2)}^{*} & =-2W^{-3}W_{s_{h}}W_{s_{g}}\mathcal{G}+W^{-2}W_{s_{g}s_{h}}\mathcal{G}+W^{-2}W_{s_{g}}\mathcal{G}_{s_{h}}\\
	& -W^{-2}W_{s_{h}}\mathcal{G}_{s_{g}}+W^{-1}\mathcal{G}_{s_{g}s_{h}}
\end{align*}
where we split some of the terms into two parts (e.g.$\partial_{bbss'}\mathcal{L}^{*}  =\partial_{bbss'}\mathcal{L}_{(1)}^{*}+\partial_{bbss'}\mathcal{L}_{(2)}^{*}$ ), since these will be bounded separately in later sections.

Fifth derivatives:
\begin{align*}
	\partial_{bbbss'}\mathcal{L}_{(1)}^{*} & =-6\mathcal{H}^{-1}\mathcal{H}_{b}\mathcal{H}^{-1}\mathcal{H}_{b}\mathcal{H}^{-1}\mathcal{H}_{b}\mathcal{H}^{-1}\\
	& +3\mathcal{H}^{-1}\mathcal{H}_{bb}\mathcal{H}^{-1}\mathcal{H}_{b}\mathcal{H}^{-1}+3\mathcal{H}^{-1}\mathcal{H}_{b}\mathcal{H}^{-1}\mathcal{H}_{bb}\mathcal{H}^{-1}\\
	& -\mathcal{H}^{-1}\mathcal{H}_{bbb}\mathcal{H}^{-1}
\end{align*}
\begin{align*}
	\partial_{bbss's_{g}}\mathcal{L}_{(1)}^{*} & =-2\mathcal{H}^{-1}\mathcal{H}_{s_{g}}\mathcal{H}^{-1}\mathcal{H}_{b}\mathcal{H}^{-1}\mathcal{H}_{b}\mathcal{H}^{-1}\\
	& -2\mathcal{H}^{-1}\mathcal{H}_{b}\mathcal{H}^{-1}\mathcal{H}_{s_{g}}\mathcal{H}^{-1}\mathcal{H}_{b}\mathcal{H}^{-1}\\
	& -2\mathcal{H}^{-1}\mathcal{H}_{b}\mathcal{H}^{-1}\mathcal{H}_{b}\mathcal{H}^{-1}\mathcal{H}_{s_{g}}\mathcal{H}^{-1}\\
	& +2\mathcal{H}^{-1}\mathcal{H}_{bs_{g}}\mathcal{H}^{-1}\mathcal{H}_{b}\mathcal{H}^{-1}+2\mathcal{H}^{-1}\mathcal{H}_{b}\mathcal{H}^{-1}\mathcal{H}_{bs_{g}}\mathcal{H}^{-1}\\
	& +\mathcal{H}^{-1}\mathcal{H}_{s_{g}}\mathcal{H}^{-1}\mathcal{H}_{bb}\mathcal{H}^{-1}+\mathcal{H}^{-1}\mathcal{H}_{bb}\mathcal{H}^{-1}\mathcal{H}_{s_{g}}\mathcal{H}^{-1}\\
	& -\mathcal{H}^{-1}\mathcal{H}_{bbs_{g}}\mathcal{H}^{-1}
\end{align*}
\begin{align*}
	\partial_{bss's_{g}s_{h}}\mathcal{L}_{(1)}^{*} & =-\mathcal{H}^{-1}\mathcal{H}_{s_{h}}\mathcal{H}^{-1}\mathcal{H}_{b}\mathcal{H}^{-1}\mathcal{H}_{s_{g}}\mathcal{H}^{-1}\\
	& -\mathcal{H}^{-1}\mathcal{H}_{b}\mathcal{H}^{-1}\mathcal{H}_{s_{h}}\mathcal{H}^{-1}\mathcal{H}_{s_{g}}\mathcal{H}^{-1}\\
	& -\mathcal{H}^{-1}\mathcal{H}_{b}\mathcal{H}^{-1}\mathcal{H}_{s_{g}}\mathcal{H}^{-1}\mathcal{H}_{s_{h}}\mathcal{H}^{-1}\\
	& +\mathcal{H}^{-1}\mathcal{H}_{bs_{h}}\mathcal{H}^{-1}\mathcal{H}_{s_{g}}\mathcal{H}^{-1}+\mathcal{H}^{-1}\mathcal{H}_{b}\mathcal{H}^{-1}\mathcal{H}_{s_{g}s_{h}}\mathcal{H}^{-1}\\
	& -\mathcal{H}^{-1}\mathcal{H}_{s_{h}}\mathcal{H}^{-1}\mathcal{H}_{s_{g}}\mathcal{H}^{-1}\mathcal{H}_{b}\mathcal{H}^{-1}\\
	& -\mathcal{H}^{-1}\mathcal{H}_{b}\mathcal{H}^{-1}\mathcal{H}_{s_{g}}\mathcal{H}^{-1}\mathcal{H}_{s_{h}}\mathcal{H}^{-1}\\
	& -\mathcal{H}^{-1}\mathcal{H}_{s_{g}}\mathcal{H}^{-1}\mathcal{H}_{b}\mathcal{H}^{-1}\mathcal{H}_{s_{h}}\mathcal{H}^{-1}\\
	& +\mathcal{H}^{-1}\mathcal{H}_{s_{g}s_{h}}\mathcal{H}^{-1}\mathcal{H}_{b}\mathcal{H}^{-1}+\mathcal{H}^{-1}\mathcal{H}_{s_g}\mathcal{H}^{-1}\mathcal{H}_{bs_{h}}\mathcal{H}^{-1}\\
	& +\mathcal{H}^{-1}\mathcal{H}_{s_{h}}\mathcal{H}^{-1}\mathcal{H}_{bs_{g}}\mathcal{H}^{-1}+\mathcal{H}^{-1}\mathcal{H}_{bs_{g}}\mathcal{H}^{-1}\mathcal{H}_{s_{h}}\mathcal{H}^{-1}\\
	& -\mathcal{H}^{-1}\mathcal{H}_{bs_{g}s_{h}}\mathcal{H}^{-1}
\end{align*}

\begin{align*}
	\partial_{ss's_{f}s_{g}s_{h}}\mathcal{L}^{*}_{(1)} & =-\mathcal{H}^{-1}\mathcal{H}_{s_{f}}\mathcal{H}^{-1}\mathcal{H}_{s_{g}}\mathcal{H}^{-1}\mathcal{H}_{s_{h}}\mathcal{H}^{-1}
	 -\mathcal{H}^{-1}\mathcal{H}_{s_{g}}\mathcal{H}^{-1}\mathcal{H}_{s_{f}}\mathcal{H}^{-1}\mathcal{H}_{s_{h}}\mathcal{H}^{-1}\\
	& -\mathcal{H}^{-1}\mathcal{H}_{s_{g}}\mathcal{H}^{-1}\mathcal{H}_{s_{h}}\mathcal{H}^{-1}\mathcal{H}_{s_{f}}\mathcal{H}^{-1}
	 +\mathcal{H}^{-1}\mathcal{H}_{s_{f}s_{g}}\mathcal{H}^{-1}\mathcal{H}_{s_{h}}\mathcal{H}^{-1}\\
	& +\mathcal{H}^{-1}\mathcal{H}_{s_{g}}\mathcal{H}^{-1}\mathcal{H}_{s_{f}s_{h}}\mathcal{H}^{-1}
	 -\mathcal{H}^{-1}\mathcal{H}_{s_{f}}\mathcal{H}^{-1}\mathcal{H}_{s_{h}}\mathcal{H}^{-1}\mathcal{H}_{s_{g}}\mathcal{H}^{-1}\\
	& -\mathcal{H}^{-1}\mathcal{H}_{s_{h}}\mathcal{H}^{-1}\mathcal{H}_{s_{f}}\mathcal{H}^{-1}\mathcal{H}_{s_{g}}\mathcal{H}^{-1}
	 -\mathcal{H}^{-1}\mathcal{H}_{s_{h}}\mathcal{H}^{-1}\mathcal{H}_{s_{g}}\mathcal{H}^{-1}\mathcal{H}_{s_{f}}\mathcal{H}^{-1}\\
	& +\mathcal{H}^{-1}\mathcal{H}_{s_{f}s_{h}}\mathcal{H}^{-1}\mathcal{H}_{s_{g}}\mathcal{H}^{-1}
	 +\mathcal{H}^{-1}\mathcal{H}_{s_{h}}\mathcal{H}^{-1}\mathcal{H}_{s_{f}s_{g}}\mathcal{H}^{-1}\\
	& +\mathcal{H}^{-1}\mathcal{H}_{s_{f}}\mathcal{H}^{-1}\mathcal{H}_{s_{g}s_{h}}\mathcal{H}^{-1}
	 +\mathcal{H}^{-1}\mathcal{H}_{s_{g}s_{h}}\mathcal{H}^{-1}\mathcal{H}_{s_{f}}\mathcal{H}^{-1}\\
	& -\mathcal{H}^{-1}\mathcal{H}_{s_{f}s_{g}s_{h}}\mathcal{H}^{-1} \\
		\partial_{ss's_{f}s_{g}s_{h}}\mathcal{L}^{*}_{(2)} &= 6W^{-4}W_{s_{h}}W_{s_{g}}W_{s_{f}}\mathcal{G}\\
	& -2W^{-3}\big(W_{s_{f}s_{h}}W_{s_{g}}+W_{s_{h}}W_{s_{f}s_{g}}+W_{s_{f}}W_{s_{g}s_{h}}\big)\mathcal{G}\\
	& -2W^{-3}\big(W_{s_{h}}W_{s_{g}}+W_{s_{g}s_{h}}\big)\mathcal{G}_{s_{f}}\\
	& -2W^{-3}\big(W_{s_{f}}W_{s_{g}}+W_{s_{f}s_{g}}\big)\mathcal{G}_{s_{h}}-2W^{-3}\big(W_{s_{f}}W_{s_{h}}+W_{s_{f}s_{h}}\big)\mathcal{G}_{s_{g}}\\
	& +W^{-2}W_{s_{f}s_{g}s_{h}}\mathcal{G}+W^{-1}\mathcal{G}_{s_{f}s_{g}s_{h}}\\
	& +W^{-2}W_{s_{h}}\mathcal{G}_{s_{f}s_{g}}+W^{-2}W_{s_{f}}\mathcal{G}_{s_{g}s_{h}}+W^{-2}W_{s_{g}}\mathcal{G}_{s_{f}s_{h}}
\end{align*}

Sixth derivatives:
\begin{align*}
	\partial_{bss's_{f}s_{g}s_{h}}\mathcal{L}_{(1)}^{*} & =\mathcal{H}^{-1}\mathcal{H}_{s_{f}}\mathcal{H}^{-1}\mathcal{H}_{b}\mathcal{H}^{-1}\mathcal{H}_{s_{g}}\mathcal{H}^{-1}\mathcal{H}_{s_{h}}\mathcal{H}^{-1}
	 +\mathcal{H}^{-1}\mathcal{H}_{b}\mathcal{H}^{-1}\mathcal{H}_{s_{f}}\mathcal{H}^{-1}\mathcal{H}_{s_{g}}\mathcal{H}^{-1}\mathcal{H}_{s_{h}}\mathcal{H}^{-1}\\
	& +\mathcal{H}^{-1}\mathcal{H}_{b}\mathcal{H}^{-1}\mathcal{H}_{s_{g}}\mathcal{H}^{-1}\mathcal{H}_{s_{f}}\mathcal{H}^{-1}\mathcal{H}_{s_{h}}\mathcal{H}^{-1}
	 +\mathcal{H}^{-1}\mathcal{H}_{b}\mathcal{H}^{-1}\mathcal{H}_{s_{g}}\mathcal{H}^{-1}\mathcal{H}_{s_{h}}\mathcal{H}^{-1}\mathcal{H}_{s_{f}}\mathcal{H}^{-1}\\
	& -\mathcal{H}^{-1}\mathcal{H}_{bs_{f}}\mathcal{H}^{-1}\mathcal{H}_{s_{g}}\mathcal{H}^{-1}\mathcal{H}_{s_{h}}\mathcal{H}^{-1}
	 -\mathcal{H}^{-1}\mathcal{H}_{b}\mathcal{H}^{-1}\mathcal{H}_{s_{f}s_{g}}\mathcal{H}^{-1}\mathcal{H}_{s_{h}}\mathcal{H}^{-1}\\
	& -\mathcal{H}^{-1}\mathcal{H}_{b}\mathcal{H}^{-1}\mathcal{H}_{s_{g}}\mathcal{H}^{-1}\mathcal{H}_{s_{f}s_{h}}\mathcal{H}^{-1}
	 +\mathcal{H}^{-1}\mathcal{H}_{s_{f}}\mathcal{H}^{-1}\mathcal{H}_{s_{g}}\mathcal{H}^{-1}\mathcal{H}_{b}\mathcal{H}^{-1}\mathcal{H}_{s_{h}}\mathcal{H}^{-1}\\
	& +\mathcal{H}^{-1}\mathcal{H}_{s_{g}}\mathcal{H}^{-1}\mathcal{H}_{s_{f}}\mathcal{H}^{-1}\mathcal{H}_{b}\mathcal{H}^{-1}\mathcal{H}_{s_{h}}\mathcal{H}^{-1}
	 +\mathcal{H}^{-1}\mathcal{H}_{s_{g}}\mathcal{H}^{-1}\mathcal{H}_{b}\mathcal{H}^{-1}\mathcal{H}_{s_{f}}\mathcal{H}^{-1}\mathcal{H}_{s_{h}}\mathcal{H}^{-1}\\
	& +\mathcal{H}^{-1}\mathcal{H}_{s_{g}}\mathcal{H}^{-1}\mathcal{H}_{b}\mathcal{H}^{-1}\mathcal{H}_{s_{h}}\mathcal{H}^{-1}\mathcal{H}_{s_{f}}\mathcal{H}^{-1}
	 -\mathcal{H}^{-1}\mathcal{H}_{s_{g}}\mathcal{H}^{-1}\mathcal{H}_{bs_{f}}\mathcal{H}^{-1}\mathcal{H}_{s_{h}}\mathcal{H}^{-1}\\
	& -\mathcal{H}^{-1}\mathcal{H}_{s_{f}s_{g}}\mathcal{H}^{-1}\mathcal{H}_{b}\mathcal{H}^{-1}\mathcal{H}_{s_{h}}\mathcal{H}^{-1}
	 -\mathcal{H}^{-1}\mathcal{H}_{s_{g}}\mathcal{H}^{-1}\mathcal{H}_{b}\mathcal{H}^{-1}\mathcal{H}_{s_{f}s_{h}}\mathcal{H}^{-1} \\
	& +\mathcal{H}^{-1}\mathcal{H}_{s_{f}}\mathcal{H}^{-1}\mathcal{H}_{s_{g}}\mathcal{H}^{-1}\mathcal{H}_{s_{h}}\mathcal{H}^{-1}\mathcal{H}_{b}\mathcal{H}^{-1}
	 +\mathcal{H}^{-1}\mathcal{H}_{s_{g}}\mathcal{H}^{-1}\mathcal{H}_{s_{f}}\mathcal{H}^{-1}\mathcal{H}_{s_{h}}\mathcal{H}^{-1}\mathcal{H}_{b}\mathcal{H}^{-1}\\
	& +\mathcal{H}^{-1}\mathcal{H}_{s_{g}}\mathcal{H}^{-1}\mathcal{H}_{s_{h}}\mathcal{H}^{-1}\mathcal{H}_{s_{f}}\mathcal{H}^{-1}\mathcal{H}_{b}\mathcal{H}^{-1}
	 +\mathcal{H}^{-1}\mathcal{H}_{s_{g}}\mathcal{H}^{-1}\mathcal{H}_{s_{h}}\mathcal{H}^{-1}\mathcal{H}_{b}\mathcal{H}^{-1}\mathcal{H}_{s_{f}}\mathcal{H}^{-1}\\
	& -\mathcal{H}^{-1}\mathcal{H}_{s_{g}}\mathcal{H}^{-1}\mathcal{H}_{s_{h}}\mathcal{H}^{-1}\mathcal{H}_{bs_{f}}\mathcal{H}^{-1}
	 -\mathcal{H}^{-1}\mathcal{H}_{s_{f}s_{g}}\mathcal{H}^{-1}\mathcal{H}_{s_{h}}\mathcal{H}^{-1}\mathcal{H}_{b}\mathcal{H}^{-1}\\
	& -\mathcal{H}^{-1}\mathcal{H}_{s_{g}}\mathcal{H}^{-1}\mathcal{H}_{s_{f}s_{h}}\mathcal{H}^{-1}\mathcal{H}_{b}\mathcal{H}^{-1}
	 -\mathcal{H}^{-1}\mathcal{H}_{s_{f}}\mathcal{H}^{-1}\mathcal{H}_{bs_{g}}\mathcal{H}^{-1}\mathcal{H}_{s_{h}}\mathcal{H}^{-1}\\
	& -\mathcal{H}^{-1}\mathcal{H}_{bs_{g}}\mathcal{H}^{-1}\mathcal{H}_{s_{f}}\mathcal{H}^{-1}\mathcal{H}_{s_{h}}\mathcal{H}^{-1}
	 -\mathcal{H}^{-1}\mathcal{H}_{bs_{g}}\mathcal{H}^{-1}\mathcal{H}_{s_{h}}\mathcal{H}^{-1}\mathcal{H}_{s_{f}}\mathcal{H}^{-1}\\
	& +\mathcal{H}^{-1}\mathcal{H}_{bs_{g}s_{f}}\mathcal{H}^{-1}\mathcal{H}_{s_{h}}\mathcal{H}^{-1}
	 +\mathcal{H}^{-1}\mathcal{H}_{bs_{g}}\mathcal{H}^{-1}\mathcal{H}_{s_{h}s_{f}}\mathcal{H}^{-1}\\
	& -\mathcal{H}^{-1}\mathcal{H}_{s_{f}}\mathcal{H}^{-1}\mathcal{H}_{s_{g}}\mathcal{H}^{-1}\mathcal{H}_{bs_{h}}\mathcal{H}^{-1}
	 -\mathcal{H}^{-1}\mathcal{H}_{s_{g}}\mathcal{H}^{-1}\mathcal{H}_{s_{f}}\mathcal{H}^{-1}\mathcal{H}_{bs_{h}}\mathcal{H}^{-1}\\
	& -\mathcal{H}^{-1}\mathcal{H}_{s_{g}}\mathcal{H}^{-1}\mathcal{H}_{bs_{h}}\mathcal{H}^{-1}\mathcal{H}_{s_{f}}\mathcal{H}^{-1}
	 +\mathcal{H}^{-1}\mathcal{H}_{s_{g}s_{f}}\mathcal{H}^{-1}\mathcal{H}_{bs_{h}}\mathcal{H}^{-1}\\
	& +\mathcal{H}^{-1}\mathcal{H}_{s_{g}}\mathcal{H}^{-1}\mathcal{H}_{bs_{h}s_{f}}\mathcal{H}^{-1}
	 +\mathcal{H}^{-1}\mathcal{H}_{s_{f}}\mathcal{H}^{-1}\mathcal{H}_{b}\mathcal{H}^{-1}\mathcal{H}_{s_{h}}\mathcal{H}^{-1}\mathcal{H}_{s_{g}}\mathcal{H}^{-1}\\
	& +\mathcal{H}^{-1}\mathcal{H}_{b}\mathcal{H}^{-1}\mathcal{H}_{s_{f}}\mathcal{H}^{-1}\mathcal{H}_{s_{h}}\mathcal{H}^{-1}\mathcal{H}_{s_{g}}\mathcal{H}^{-1}
	 +\mathcal{H}^{-1}\mathcal{H}_{b}\mathcal{H}^{-1}\mathcal{H}_{s_{h}}\mathcal{H}^{-1}\mathcal{H}_{s_{f}}\mathcal{H}^{-1}\mathcal{H}_{s_{g}}\mathcal{H}^{-1}\\
	& +\mathcal{H}^{-1}\mathcal{H}_{b}\mathcal{H}^{-1}\mathcal{H}_{s_{h}}\mathcal{H}^{-1}\mathcal{H}_{s_{g}}\mathcal{H}^{-1}\mathcal{H}_{s_{f}}\mathcal{H}^{-1}
	 -\mathcal{H}^{-1}\mathcal{H}_{bs_{f}}\mathcal{H}^{-1}\mathcal{H}_{s_{h}}\mathcal{H}^{-1}\mathcal{H}_{s_{g}}\mathcal{H}^{-1}\\
	& -\mathcal{H}^{-1}\mathcal{H}_{b}\mathcal{H}^{-1}\mathcal{H}_{s_{f}s_{h}}\mathcal{H}^{-1}\mathcal{H}_{s_{g}}\mathcal{H}^{-1}
	 -\mathcal{H}^{-1}\mathcal{H}_{b}\mathcal{H}^{-1}\mathcal{H}_{s_{h}}\mathcal{H}^{-1}\mathcal{H}_{s_{f}s_{g}}\mathcal{H}^{-1}\\
	& +\mathcal{H}^{-1}\mathcal{H}_{s_{f}}\mathcal{H}^{-1}\mathcal{H}_{s_{h}}\mathcal{H}^{-1}\mathcal{H}_{b}\mathcal{H}^{-1}\mathcal{H}_{s_{g}}\mathcal{H}^{-1}
	 +\mathcal{H}^{-1}\mathcal{H}_{s_{h}}\mathcal{H}^{-1}\mathcal{H}_{s_{f}}\mathcal{H}^{-1}\mathcal{H}_{b}\mathcal{H}^{-1}\mathcal{H}_{s_{g}}\mathcal{H}^{-1}\\
	& +\mathcal{H}^{-1}\mathcal{H}_{s_{h}}\mathcal{H}^{-1}\mathcal{H}_{b}\mathcal{H}^{-1}\mathcal{H}_{s_{f}}\mathcal{H}^{-1}\mathcal{H}_{s_{g}}\mathcal{H}^{-1}
	 +\mathcal{H}^{-1}\mathcal{H}_{s_{h}}\mathcal{H}^{-1}\mathcal{H}_{b}\mathcal{H}^{-1}\mathcal{H}_{s_{g}}\mathcal{H}^{-1}\mathcal{H}_{s_{f}}\mathcal{H}^{-1}\\
	& -\mathcal{H}^{-1}\mathcal{H}_{s_{h}}\mathcal{H}^{-1}\mathcal{H}_{bs_{f}}\mathcal{H}^{-1}\mathcal{H}_{s_{g}}\mathcal{H}^{-1}
	 -\mathcal{H}^{-1}\mathcal{H}_{s_{f}s_{h}}\mathcal{H}^{-1}\mathcal{H}_{b}\mathcal{H}^{-1}\mathcal{H}_{s_{g}}\mathcal{H}^{-1}\\
	& -\mathcal{H}^{-1}\mathcal{H}_{s_{h}}\mathcal{H}^{-1}\mathcal{H}_{b}\mathcal{H}^{-1}\mathcal{H}_{s_{f}s_{g}}\mathcal{H}^{-1}
	 +\mathcal{H}^{-1}\mathcal{H}_{s_{f}}\mathcal{H}^{-1}\mathcal{H}_{s_{h}}\mathcal{H}^{-1}\mathcal{H}_{s_{g}}\mathcal{H}^{-1}\mathcal{H}_{b}\mathcal{H}^{-1}\\
	& +\mathcal{H}^{-1}\mathcal{H}_{s_{h}}\mathcal{H}^{-1}\mathcal{H}_{s_{f}}\mathcal{H}^{-1}\mathcal{H}_{s_{g}}\mathcal{H}^{-1}\mathcal{H}_{b}\mathcal{H}^{-1}
	 +\mathcal{H}^{-1}\mathcal{H}_{s_{h}}\mathcal{H}^{-1}\mathcal{H}_{s_{g}}\mathcal{H}^{-1}\mathcal{H}_{s_{f}}\mathcal{H}^{-1}\mathcal{H}_{b}\mathcal{H}^{-1}\\
	& +\mathcal{H}^{-1}\mathcal{H}_{s_{h}}\mathcal{H}^{-1}\mathcal{H}_{s_{g}}\mathcal{H}^{-1}\mathcal{H}_{b}\mathcal{H}^{-1}\mathcal{H}_{s_{f}}\mathcal{H}^{-1}
	 -\mathcal{H}^{-1}\mathcal{H}_{s_{h}}\mathcal{H}^{-1}\mathcal{H}_{s_{g}}\mathcal{H}^{-1}\mathcal{H}_{bs_{f}}\mathcal{H}^{-1}\\
	& -\mathcal{H}^{-1}\mathcal{H}_{s_{f}s_{h}}\mathcal{H}^{-1}\mathcal{H}_{s_{g}}\mathcal{H}^{-1}\mathcal{H}_{b}\mathcal{H}^{-1}
	 -\mathcal{H}^{-1}\mathcal{H}_{s_{h}}\mathcal{H}^{-1}\mathcal{H}_{s_{f}s_{g}}\mathcal{H}^{-1}\mathcal{H}_{b}\mathcal{H}^{-1}\\
	& -\mathcal{H}^{-1}\mathcal{H}_{s_{f}}\mathcal{H}^{-1}\mathcal{H}_{bs_{h}}\mathcal{H}^{-1}\mathcal{H}_{s_{g}}\mathcal{H}^{-1}
	 -\mathcal{H}^{-1}\mathcal{H}_{bs_{h}}\mathcal{H}^{-1}\mathcal{H}_{s_{f}}\mathcal{H}^{-1}\mathcal{H}_{s_{g}}\mathcal{H}^{-1}\\
	& -\mathcal{H}^{-1}\mathcal{H}_{bs_{h}}\mathcal{H}^{-1}\mathcal{H}_{s_{g}}\mathcal{H}^{-1}\mathcal{H}_{s_{f}}\mathcal{H}^{-1}
	 +\mathcal{H}^{-1}\mathcal{H}_{bs_{h}s_{f}}\mathcal{H}^{-1}\mathcal{H}_{s_{g}}\mathcal{H}^{-1}\\
	& +\mathcal{H}^{-1}\mathcal{H}_{bs_{h}}\mathcal{H}^{-1}\mathcal{H}_{s_{g}s_{f}}\mathcal{H}^{-1}
	 -\mathcal{H}^{-1}\mathcal{H}_{s_{f}}\mathcal{H}^{-1}\mathcal{H}_{s_{h}}\mathcal{H}^{-1}\mathcal{H}_{bs_{g}}\mathcal{H}^{-1}\\
	& -\mathcal{H}^{-1}\mathcal{H}_{s_{h}}\mathcal{H}^{-1}\mathcal{H}_{s_{f}}\mathcal{H}^{-1}\mathcal{H}_{bs_{g}}\mathcal{H}^{-1}
	 -\mathcal{H}^{-1}\mathcal{H}_{s_{h}}\mathcal{H}^{-1}\mathcal{H}_{bs_{g}}\mathcal{H}^{-1}\mathcal{H}_{s_{f}}\mathcal{H}^{-1}\\
	& +\mathcal{H}^{-1}\mathcal{H}_{s_{h}s_{f}}\mathcal{H}^{-1}\mathcal{H}_{bs_{g}}\mathcal{H}^{-1}
	 +\mathcal{H}^{-1}\mathcal{H}_{s_{h}}\mathcal{H}^{-1}\mathcal{H}_{bs_{g}s_{f}}\mathcal{H}^{-1}\\
	& -\mathcal{H}^{-1}\mathcal{H}_{s_{f}}\mathcal{H}^{-1}\mathcal{H}_{b}\mathcal{H}^{-1}\mathcal{H}_{s_{g}s_{h}}\mathcal{H}^{-1}
	 -\mathcal{H}^{-1}\mathcal{H}_{b}\mathcal{H}^{-1}\mathcal{H}_{s_{f}}\mathcal{H}^{-1}\mathcal{H}_{s_{g}s_{h}}\mathcal{H}^{-1}\\
	& -\mathcal{H}^{-1}\mathcal{H}_{b}\mathcal{H}^{-1}\mathcal{H}_{s_{g}s_{h}}\mathcal{H}^{-1}\mathcal{H}_{s_{f}}\mathcal{H}^{-1}
	 +\mathcal{H}^{-1}\mathcal{H}_{bs_{f}}\mathcal{H}^{-1}\mathcal{H}_{s_{g}s_{h}}\mathcal{H}^{-1}\\
	& +\mathcal{H}^{-1}\mathcal{H}_{b}\mathcal{H}^{-1}\mathcal{H}_{s_{f}s_{g}s_{h}}\mathcal{H}^{-1}
	 -\mathcal{H}^{-1}\mathcal{H}_{s_{f}}\mathcal{H}^{-1}\mathcal{H}_{s_{g}s_{h}}\mathcal{H}^{-1}\mathcal{H}_{b}\mathcal{H}^{-1}\\
	& -\mathcal{H}^{-1}\mathcal{H}_{s_{g}s_{h}}\mathcal{H}^{-1}\mathcal{H}_{s_{f}}\mathcal{H}^{-1}\mathcal{H}_{b}\mathcal{H}^{-1}
	 -\mathcal{H}^{-1}\mathcal{H}_{s_{g}s_{h}}\mathcal{H}^{-1}\mathcal{H}_{b}\mathcal{H}^{-1}\mathcal{H}_{s_{f}}\mathcal{H}^{-1}\\
	& +\mathcal{H}^{-1}\mathcal{H}_{s_{f}s_{g}s_{h}}\mathcal{H}^{-1}\mathcal{H}_{b}\mathcal{H}^{-1}
	 +\mathcal{H}^{-1}\mathcal{H}_{s_{g}s_{h}}\mathcal{H}^{-1}\mathcal{H}_{s_{f}b}\mathcal{H}^{-1}\\
	& +\mathcal{H}^{-1}\mathcal{H}_{s_{f}}\mathcal{H}^{-1}\mathcal{H}_{bs_{g}s_{h}}\mathcal{H}^{-1}
	 +\mathcal{H}^{-1}\mathcal{H}_{bs_{g}s_{h}}\mathcal{H}^{-1}\mathcal{H}_{s_{f}}\mathcal{H}^{-1}\\
	& -\mathcal{H}^{-1}\mathcal{H}_{bs_{f}s_{g}s_{h}}\mathcal{H}^{-1}
\end{align*}

\subsection{Expressions for $\mathcal{H}$ terms}
In order to simplify some of the derivative expression we introduce some new terms.
\begin{align*}
	\mathcal{P}^{(r,s)} & =\sum_{f}(\partial_{\beta^{r}\phi\phi'\phi_{f}}\mathcal{L})\big[(\partial_{\beta^{s}\phi'}\mathcal{L})\mathcal{H}^{-1}\big]_{f}\\
	\mathcal{P}_{(g)}^{r} & =\sum_{f}(\partial_{\beta^{r}\phi\phi'\phi_{f}}\mathcal{L})\big[\mathcal{H}^{-1}\big]_{fg}\\{}
	[\mathcal{R}^{s}]_{g} & =\big[(\partial_{\beta^{s}\phi'}\mathcal{L})\mathcal{H}^{-1}\big]_{g}
\end{align*}

By definition we have that $\mathcal{H}=-\partial_{\phi\phi'}\mathcal{L}$.
Differentiating gives
\[
\mathcal{H}_{b}=W^{-1}\big(\partial_{\beta\phi\phi'}\mathcal{L}+\sum_{f}(\partial_{\phi\phi'\phi_{f}}\mathcal{L})\big[\mathcal{H}^{-1}(\partial_{\beta\phi}\mathcal{L})\big]_{f}\big)
\]
\begin{align*}
\mathcal{H}_{s_{g}} & =-(\partial_{\beta\phi\phi'}\mathcal{L})(\partial_{s_{g}}B)-\sum_{f}(\partial_{\phi\phi'\phi_{f}}\mathcal{L})(\partial_{s}\Phi)_{fg}\\
 & =W^{-1}\big[(\partial_{\beta\phi'}\mathcal{L})\mathcal{H}^{-1}\big]_{g}(\partial_{\beta\phi\phi'}\mathcal{L})\\
 & +\sum_{f}(\partial_{\phi\phi'\phi_{f}}\mathcal{L})[\mathcal{H}^{-1}]_{fg}\\
 & +W^{-1}\sum_{f}(\partial_{\phi\phi'\phi_{f}}\mathcal{L})\big[(\partial_{\beta\phi'}\mathcal{L})\mathcal{H}^{-1}\big]_{f}\big[(\partial_{\beta\phi'}\mathcal{L})\mathcal{H}^{-1}\big]_{g}
\end{align*}

Differentiating a second time gives
\begin{align*}
\mathcal{H}_{bb} & =-W^{-1}W_{b}\mathcal{H}_{b}+W^{-1}\mathcal{E}_{b}^{1}\\
 & -W^{-2}\sum_{f}(\partial_{\beta\phi\phi'\phi_{f}}\mathcal{L})\big[\mathcal{H}^{-1}(\partial_{\beta\phi}\mathcal{L})\big]_{f}\\
 & -W^{-2}\sum_{e,f}(\partial_{\phi\phi'\phi_{e}\phi_{f}}\mathcal{L})\big[\mathcal{H}^{-1}(\partial_{\beta\phi}\mathcal{L})\big]_{e}\big[\mathcal{H}^{-1}(\partial_{\beta\phi}\mathcal{L})\big]_{f}\\
 & -W^{-1}\sum_{f}(\partial_{\phi\phi'\phi_{f}}\mathcal{L})\big[\mathcal{H}^{-1}\mathcal{H}_{b}\mathcal{H}^{-1}(\partial_{\beta\phi}\mathcal{L})\big]_{f}\\
 & -W^{-2}\sum_{f}(\partial_{\phi\phi'\phi_{f}}\mathcal{L})\big[\mathcal{H}^{-1}(\partial_{\beta\beta\phi}\mathcal{L})\big]_{f}\\
 & -W^{-2}\sum_{f}(\partial_{\phi\phi'\phi_{f}}\mathcal{L})\big[\mathcal{H}^{-1}\mathcal{E}^{1}\mathcal{H}^{-1}(\partial_{\beta\phi}\mathcal{L})\big]_{f}
\end{align*}
\begin{align*}
\mathcal{H}_{bs_{g}} & =-W^{-2}W_{s_{g}}\mathcal{E}^{1}+W^{-1}\mathcal{E}_{s_{g}}^{1}\\
 & -W^{-2}W_{s_{g}}\sum_{f}(\partial_{\phi\phi'\phi_{f}}\mathcal{L})\big[\mathcal{H}^{-1}(\partial_{\beta\phi}\mathcal{L})\big]_{f}\\
 & -W^{-2}\sum_{f}(\partial_{\beta\phi\phi'\phi_{f}}\mathcal{L})\big[\mathcal{H}^{-1}(\partial_{\beta\phi}\mathcal{L})\big]_{f}\big[\mathcal{H}^{-1}(\partial_{\beta\phi}\mathcal{L})\big]_{g}\\
 & +W^{-1}\sum_{e,f}(\partial_{\phi\phi'\phi_{e}\phi_{f}}\mathcal{L})\big[\mathcal{H}^{-1}(\partial_{\beta\phi}\mathcal{L})\big]_{f}\big[\mathcal{H}^{-1}\big]_{eg}\\
 & -W^{-2}\sum_{e,f}(\partial_{\phi\phi'\phi_{e}\phi_{f}}\mathcal{L})\big[\mathcal{H}^{-1}(\partial_{\beta\phi}\mathcal{L})\big]_{e}\big[\mathcal{H}^{-1}(\partial_{\beta\phi}\mathcal{L})\big]_{f}\big[\mathcal{H}^{-1}(\partial_{\beta\phi}\mathcal{L})\big]_{g}\\
 & -W^{-1}\sum_{f}(\partial_{\phi\phi'\phi_{f}}\mathcal{L})\big[\mathcal{H}^{-1}\mathcal{H}_{s_{g}}\mathcal{H}^{-1}(\partial_{\beta\phi}\mathcal{L})\big]_{f}\\
 & -W^{-2}\sum_{f}(\partial_{\phi\phi'\phi_{f}}\mathcal{L})\big[\mathcal{H}^{-1}(\partial_{\beta\beta\phi}\mathcal{L})\big]_{f}\big[\mathcal{H}^{-1}(\partial_{\beta\phi}\mathcal{L})\big]_{g}\\
 & +W^{-1}\sum_{f}(\partial_{\phi\phi'\phi_{f}}\mathcal{L})\big[\mathcal{H}^{-1}\mathcal{E}^{1}\mathcal{H}^{-1}\big]_{fg}\\
 & -W^{-2}\sum_{f}(\partial_{\phi\phi'\phi_{f}}\mathcal{L})\big[\mathcal{H}^{-1}\mathcal{E}^{1}\mathcal{H}^{-1}(\partial_{\beta\phi}\mathcal{L})\big]_{f}\big[\mathcal{H}^{-1}(\partial_{\beta\phi}\mathcal{L})\big]_{g}
\end{align*}

\begin{align*}
\mathcal{H}_{s_{g}s_{h}} & =-W^{-2}W_{s_{h}}\big[(\partial_{\beta\phi'}\mathcal{L})\mathcal{H}^{-1}\big]_{g}(\partial_{\beta\phi\phi'}\mathcal{L})\\
 & -W^{-2}\big[(\partial_{\beta\beta\phi'}\mathcal{L})\mathcal{H}^{-1}\big]_{g}\big[(\partial_{\beta\phi'}\mathcal{L})\mathcal{H}^{-1}\big]_{h}(\partial_{\beta\phi\phi'}\mathcal{L})\\
 & -W^{-1}\big[\mathcal{H}^{-1}(\partial_{\beta\phi\phi'}\mathcal{L})\mathcal{H}^{-1}\big]_{hg}(\partial_{\beta\phi\phi'}\mathcal{L})\\
 & -W^{-2}\big[(\partial_{\beta\phi}\mathcal{L})'\mathcal{H}^{-1}\mathcal{E}\mathcal{H}^{-1}\big]_{g}\big[(\partial_{\beta\phi'}\mathcal{L})\mathcal{H}^{-1}\big]_{h}(\partial_{\beta\phi\phi'}\mathcal{L})\\
 & -W^{-1}\big[(\partial_{\beta\phi'}\mathcal{L})\mathcal{H}^{-1}\mathcal{H}_{s_{h}}\mathcal{H}^{-1}\big]_{g}(\partial_{\beta\phi\phi'}\mathcal{L})\\
 & -W^{-1}\big[(\partial_{\beta\phi'}\mathcal{L})\mathcal{H}^{-1}\big]_{g}\big[(\partial_{\beta\phi'}\mathcal{L})\mathcal{H}^{-1}\big]_{h}(\partial_{\beta\beta\phi\phi'}\mathcal{L})\\
 & -W^{-1}\big[(\partial_{\beta\phi'}\mathcal{L})\mathcal{H}^{-1}\big]_{g}\sum_{f}(\partial_{\beta\phi\phi'\phi_{f}}\mathcal{L})\big[\mathcal{H}^{-1}\big]_{fh}\\
 & -W^{-2}\big[(\partial_{\beta\phi'}\mathcal{L})\mathcal{H}^{-1}\big]_{g}\sum_{f}(\partial_{\beta\phi\phi'\phi_{f}}\mathcal{L})\big[\mathcal{H}^{-1}(\partial_{\beta\phi}\mathcal{L})\big]_{f}\big[(\partial_{\beta\phi}\mathcal{L})'\mathcal{H}^{-1}\big]_{h}\\
 & -W^{-1}\sum_{f}(\partial_{\beta\phi\phi'\phi_{f}}\mathcal{L})[\mathcal{H}^{-1}]_{fg}\big[(\partial_{\beta\phi}\mathcal{L})'\mathcal{H}^{-1}\big]_{h}\\
 & -\sum_{e,f}(\partial_{\phi\phi'\phi_{e}\phi_{f}}\mathcal{L})[\mathcal{H}^{-1}]_{fg}[\mathcal{H}^{-1}]_{eh}\\
 & -W^{-1}\sum_{e,f}(\partial_{\phi\phi'\phi_{e}\phi_{f}}\mathcal{L})[\mathcal{H}^{-1}]_{fg}\big[(\partial_{\beta\phi}\mathcal{L})'\mathcal{H}^{-1}\big]_{e}\big[(\partial_{\beta\phi}\mathcal{L})'\mathcal{H}^{-1}\big]_{h}\\
 & -\sum_{f}(\partial_{\phi\phi'\phi_{f}}\mathcal{L})[\mathcal{H}^{-1}\mathcal{H}_{s_{h}}\mathcal{H}^{-1}]_{fg}\\
 & -W^{-2}W_{s_{h}}\sum_{f}(\partial_{\phi\phi'\phi_{f}}\mathcal{L})\big[(\partial_{\beta\phi'}\mathcal{L})\mathcal{H}^{-1}\big]_{f}\big[(\partial_{\beta\phi'}\mathcal{L})\mathcal{H}^{-1}\big]_{g}\\
 & -W^{-2}\sum_{f}(\partial_{\beta\phi\phi'\phi_{f}}\mathcal{L})\big[(\partial_{\beta\phi'}\mathcal{L})\mathcal{H}^{-1}\big]_{f}\big[(\partial_{\beta\phi'}\mathcal{L})\mathcal{H}^{-1}\big]_{g}\\
 & -W^{-1}\sum_{e,f}(\partial_{\phi\phi'\phi_{e}\phi_{f}}\mathcal{L})[\mathcal{H}^{-1}]_{eh}\big[(\partial_{\beta\phi'}\mathcal{L})\mathcal{H}^{-1}\big]_{f}\big[(\partial_{\beta\phi'}\mathcal{L})\mathcal{H}^{-1}\big]_{g}\\
 & -W^{-2}\sum_{e,f}(\partial_{\phi\phi'\phi_{e}\phi_{f}}\mathcal{L})\big[(\partial_{\beta\phi'}\mathcal{L})\mathcal{H}^{-1}\big]_{e}\big[(\partial_{\beta\phi'}\mathcal{L})\mathcal{H}^{-1}\big]_{f}\big[(\partial_{\beta\phi'}\mathcal{L})\mathcal{H}^{-1}\big]_{g}\big[(\partial_{\beta\phi'}\mathcal{L})\mathcal{H}^{-1}\big]_{h}\\
 & -W^{-2}\sum_{f}(\partial_{\phi\phi'\phi_{f}}\mathcal{L})\big[(\partial_{\beta\beta\phi'}\mathcal{L})\mathcal{H}^{-1}\big]_{f}\big[(\partial_{\beta\phi'}\mathcal{L})\mathcal{H}^{-1}\big]_{g}\\
 & -W^{-1}\sum_{f}(\partial_{\phi\phi'\phi_{f}}\mathcal{L})\big[\mathcal{H}^{-1}(\partial_{\beta\phi\phi'}\mathcal{L})\mathcal{H}^{-1}\big]_{fh}\big[(\partial_{\beta\phi'}\mathcal{L})\mathcal{H}^{-1}\big]_{g}\\
 & -W^{-2}\sum_{f}(\partial_{\phi\phi'\phi_{f}}\mathcal{L})\big[(\partial_{\beta\phi'}\mathcal{L})\mathcal{H}^{-1}\mathcal{E}\mathcal{H}^{-1}\big]_{f}\big[(\partial_{\beta\phi'}\mathcal{L})\mathcal{H}^{-1}\big]_{g}\big[(\partial_{\beta\phi'}\mathcal{L})\mathcal{H}^{-1}\big]_{h}\\
 & -W^{-1}\sum_{f}(\partial_{\phi\phi'\phi_{f}}\mathcal{L})\big[(\partial_{\beta\phi'}\mathcal{L})\mathcal{H}^{-1}\mathcal{H}_{s_{h}}\mathcal{H}^{-1}\big]_{f}\big[(\partial_{\beta\phi'}\mathcal{L})\mathcal{H}^{-1}\big]_{g}\\
 & -W^{-2}\sum_{f}(\partial_{\phi\phi'\phi_{f}}\mathcal{L})\big[(\partial_{\beta\phi'}\mathcal{L})\mathcal{H}^{-1}\big]_{f}\big[(\partial_{\beta\beta\phi'}\mathcal{L})\mathcal{H}^{-1}\big]_{g}\\
 & -W^{-1}\sum_{f}(\partial_{\phi\phi'\phi_{f}}\mathcal{L})\big[(\partial_{\beta\phi'}\mathcal{L})\mathcal{H}^{-1}\big]_{f}\big[\mathcal{H}^{-1}(\partial_{\beta\phi\phi'}\mathcal{L})\mathcal{H}^{-1}\big]_{gh}\\
 & -W^{-2}\sum_{f}(\partial_{\phi\phi'\phi_{f}}\mathcal{L})\big[(\partial_{\beta\phi'}\mathcal{L})\mathcal{H}^{-1}\big]_{f}\big[(\partial_{\beta\phi'}\mathcal{L})\mathcal{H}^{-1}\mathcal{E}\mathcal{H}^{-1}\big]_{g}\big[(\partial_{\beta\phi'}\mathcal{L})\mathcal{H}^{-1}\big]_{h}\\
 & -W^{-1}\sum_{f}(\partial_{\phi\phi'\phi_{f}}\mathcal{L})\big[(\partial_{\beta\phi'}\mathcal{L})\mathcal{H}^{-1}\big]_{f}\big[(\partial_{\beta\phi'}\mathcal{L})\mathcal{H}^{-1}\mathcal{H}_{s_{h}}\mathcal{H}^{-1}\big]_{g}
\end{align*}
The third derivatives are

\begin{align*}
\mathcal{H}_{bbb} & =W^{-2}W_{b}^{2}\mathcal{H}_{b}-W^{-1}W_{bb}\mathcal{H}_{b}-W^{-1}W_{b}\mathcal{H}_{bb}\\
 & +W^{-2}W_{b}\mathcal{E}_{b}-W^{-1}\mathcal{E}_{bb}\\
 & +2W^{-3}W_{b}\sum_{f}(\partial_{\beta\phi\phi'\phi_{f}}\mathcal{L})\big[\mathcal{H}^{-1}(\partial_{\beta\phi}\mathcal{L})\big]_{f}\\
 & -W^{-3}\sum_{f}(\partial_{\beta\beta\phi\phi'\phi_{f}}\mathcal{L})\big[\mathcal{H}^{-1}(\partial_{\beta\phi}\mathcal{L})\big]_{f}\\
 & -W^{-3}\sum_{e,f}(\partial_{\beta\phi\phi'\phi_{e}\phi_{f}}\mathcal{L})\big[\mathcal{H}^{-1}(\partial_{\beta\phi}\mathcal{L})\big]_{e}\big[\mathcal{H}^{-1}(\partial_{\beta\phi}\mathcal{L})\big]_{f}\\
 & +W^{-2}\sum_{f}(\partial_{\beta\phi\phi'\phi_{f}}\mathcal{L})\big[\mathcal{H}^{-1}\mathcal{H}_{b}\mathcal{H}^{-1}(\partial_{\beta\phi}\mathcal{L})\big]_{f}\\
 & -W^{-3}\sum_{f}(\partial_{\beta\phi\phi'\phi_{f}}\mathcal{L})\big[\mathcal{H}^{-1}(\partial_{\beta\beta\phi}\mathcal{L})\big]_{f}\\
 & -W^{-3}\sum_{f}(\partial_{\beta\phi\phi'\phi_{f}}\mathcal{L})\big[\mathcal{H}^{-1}\mathcal{E}\mathcal{H}^{-1}(\partial_{\beta\phi}\mathcal{L})\big]_{f}\\
\\
 & +2W^{-3}W_{b}\sum_{e,f}(\partial_{\phi\phi'\phi_{e}\phi_{f}}\mathcal{L})\big[\mathcal{H}^{-1}(\partial_{\beta\phi}\mathcal{L})\big]_{e}\big[\mathcal{H}^{-1}(\partial_{\beta\phi}\mathcal{L})\big]_{f}\\
 & -W^{-3}\sum_{e,f}(\partial_{\beta\phi\phi'\phi_{e}\phi_{f}}\mathcal{L})\big[\mathcal{H}^{-1}(\partial_{\beta\phi}\mathcal{L})\big]_{e}\big[\mathcal{H}^{-1}(\partial_{\beta\phi}\mathcal{L})\big]_{f}\\
 & -W^{-3}\sum_{d,e,f}(\partial_{\phi\phi'\phi_{d}\phi_{e}\phi_{f}}\mathcal{L})\big[\mathcal{H}^{-1}(\partial_{\beta\phi}\mathcal{L})\big]_{d}\big[\mathcal{H}^{-1}(\partial_{\beta\phi}\mathcal{L})\big]_{e}\big[\mathcal{H}^{-1}(\partial_{\beta\phi}\mathcal{L})\big]_{f}\\
 & +2W^{-2}\sum_{e,f}(\partial_{\phi\phi'\phi_{e}\phi_{f}}\mathcal{L})\big[\mathcal{H}^{-1}\mathcal{H}_{b}\mathcal{H}^{-1}(\partial_{\beta\phi}\mathcal{L})\big]_{e}\big[\mathcal{H}^{-1}(\partial_{\beta\phi}\mathcal{L})\big]_{f}\\
 & -2W^{-3}\sum_{e,f}(\partial_{\phi\phi'\phi_{e}\phi_{f}}\mathcal{L})\big[\mathcal{H}^{-1}(\partial_{\beta\beta\phi}\mathcal{L})\big]_{e}\big[\mathcal{H}^{-1}(\partial_{\beta\phi}\mathcal{L})\big]_{f}\\
 & -2W^{-3}\sum_{e,f}(\partial_{\phi\phi'\phi_{e}\phi_{f}}\mathcal{L})\big[\mathcal{H}^{-1}\mathcal{E}\mathcal{H}^{-1}(\partial_{\beta\phi}\mathcal{L})\big]_{e}\big[\mathcal{H}^{-1}(\partial_{\beta\phi}\mathcal{L})\big]_{f}\\
\\
 & -W^{-1}W_{b}\sum_{f}(\partial_{\phi\phi'\phi_{f}}\mathcal{L})\big[\mathcal{H}^{-1}\mathcal{H}_{b}\mathcal{H}^{-1}(\partial_{\beta\phi}\mathcal{L})\big]_{f}\\
 & +W^{-2}\sum_{f}(\partial_{\beta\phi\phi'\phi_{f}}\mathcal{L})\big[\mathcal{H}^{-1}\mathcal{H}_{b}\mathcal{H}^{-1}(\partial_{\beta\phi}\mathcal{L})\big]_{f}\\
 & +W^{-2}\sum_{e,f}(\partial_{\phi\phi'\phi_{e}\phi_{f}}\mathcal{L})\big[\mathcal{H}^{-1}(\partial_{\beta\phi}\mathcal{L})\big]_{e}\big[\mathcal{H}^{-1}\mathcal{H}_{b}\mathcal{H}^{-1}(\partial_{\beta\phi}\mathcal{L})\big]_{f}\\
 & -2W^{-1}\sum_{f}(\partial_{\phi\phi'\phi_{f}}\mathcal{L})\big[\mathcal{H}^{-1}\mathcal{H}_{b}\mathcal{H}^{-1}\mathcal{H}_{b}\mathcal{H}^{-1}(\partial_{\beta\phi}\mathcal{L})\big]_{f}\\
 & +W^{-1}\sum_{f}(\partial_{\phi\phi'\phi_{f}}\mathcal{L})\big[\mathcal{H}^{-1}\mathcal{H}_{bb}\mathcal{H}^{-1}(\partial_{\beta\phi}\mathcal{L})\big]_{f}\\
 & +W^{-2}\sum_{f}(\partial_{\phi\phi'\phi_{f}}\mathcal{L})\big[\mathcal{H}^{-1}\mathcal{H}_{b}\mathcal{H}^{-1}(\partial_{\beta\beta\phi}\mathcal{L})\big]_{f}\\
 & +W^{-2}\sum_{f}(\partial_{\phi\phi'\phi_{f}}\mathcal{L})\big[\mathcal{H}^{-1}\mathcal{H}_{b}\mathcal{H}^{-1}\mathcal{E}\mathcal{H}^{-1}(\partial_{\beta\phi}\mathcal{L})\big]_{f}\\
\\
 & +2W^{-3}W_{b}\sum_{f}(\partial_{\phi\phi'\phi_{f}}\mathcal{L})\big[\mathcal{H}^{-1}(\partial_{\beta\beta\phi}\mathcal{L})\big]_{f}\\
 & -W^{-3}\sum_{f}(\partial_{\beta\phi\phi'\phi_{f}}\mathcal{L})\big[\mathcal{H}^{-1}(\partial_{\beta\beta\phi}\mathcal{L})\big]_{f}\\
 & -W^{-3}\sum_{e,f}(\partial_{\phi\phi'\phi_{e}\phi_{f}}\mathcal{L})\big[\mathcal{H}^{-1}(\partial_{\beta\phi}\mathcal{L})\big]_{e}\big[\mathcal{H}^{-1}(\partial_{\beta\beta\phi}\mathcal{L})\big]_{f}\\
 & +W^{-2}\sum_{f}(\partial_{\phi\phi'\phi_{f}}\mathcal{L})\big[\mathcal{H}^{-1}\mathcal{H}_{b}\mathcal{H}^{-1}(\partial_{\beta\beta\phi}\mathcal{L})\big]_{f}\\
 & -W^{-3}\sum_{f}(\partial_{\phi\phi'\phi_{f}}\mathcal{L})\big[\mathcal{H}^{-1}(\partial_{\beta\beta\beta\phi}\mathcal{L})\big]_{f}\\
 & -W^{-3}\sum_{f}(\partial_{\phi\phi'\phi_{f}}\mathcal{L})\big[\mathcal{H}^{-1}\mathcal{E}^{2}\mathcal{H}^{-1}(\partial_{\beta\phi}\mathcal{L})\big]_{f}\\
\\
 & +2W^{-3}W_{b}\sum_{f}(\partial_{\phi\phi'\phi_{f}}\mathcal{L})\big[\mathcal{H}^{-1}\mathcal{E}\mathcal{H}^{-1}(\partial_{\beta\phi}\mathcal{L})\big]_{f}\\
 & -W^{-3}\sum_{f}(\partial_{\beta\phi\phi'\phi_{f}}\mathcal{L})\big[\mathcal{H}^{-1}\mathcal{E}\mathcal{H}^{-1}(\partial_{\beta\phi}\mathcal{L})\big]_{f}\\
 & -W^{-3}\sum_{e,f}(\partial_{\phi\phi'\phi_{e}\phi_{f}}\mathcal{L})\big[\mathcal{H}^{-1}(\partial_{\beta\phi}\mathcal{L})\big]_{e}\big[\mathcal{H}^{-1}\mathcal{E}\mathcal{H}^{-1}(\partial_{\beta\phi}\mathcal{L})\big]_{f}\\
 & +W^{-2}\sum_{f}(\partial_{\phi\phi'\phi_{f}}\mathcal{L})\big[\mathcal{H}^{-1}\mathcal{H}_{b}\mathcal{H}^{-1}\mathcal{E}\mathcal{H}^{-1}(\partial_{\beta\phi}\mathcal{L})\big]_{f}\\
 & +W^{-2}\sum_{f}(\partial_{\phi\phi'\phi_{f}}\mathcal{L})\big[\mathcal{H}^{-1}\mathcal{E}\mathcal{H}^{-1}\mathcal{H}_{b}\mathcal{H}^{-1}(\partial_{\beta\phi}\mathcal{L})\big]_{f}\\
 & -W^{-2}\sum_{f}(\partial_{\phi\phi'\phi_{f}}\mathcal{L})\big[\mathcal{H}^{-1}\mathcal{E}_{b}\mathcal{H}^{-1}(\partial_{\beta\phi}\mathcal{L})\big]_{f}\\
 & -W^{-3}\sum_{f}(\partial_{\phi\phi'\phi_{f}}\mathcal{L})\big[\mathcal{H}^{-1}\mathcal{E}\mathcal{H}^{-1}(\partial_{\beta\beta\phi}\mathcal{L})\big]_{f}\\
 & -W^{-3}\sum_{f}(\partial_{\phi\phi'\phi_{f}}\mathcal{L})\big[\mathcal{H}^{-1}\mathcal{E}\mathcal{H}^{-1}\mathcal{E}\mathcal{H}^{-1}(\partial_{\beta\phi}\mathcal{L})\big]_{f}
\end{align*}

\begin{align*}
\mathcal{H}_{bbs_{g}} & =-2W^{-3}W_{b}W_{s_{g}}\mathcal{E}^{1}+W^{-2}W_{bs_{g}}\mathcal{E}^{1}+W^{-2}W_{s_{g}}\mathcal{E}_{b}^{1}\\
 & +W^{-2}W_{b}\mathcal{E}_{s_{g}}^{1}-W^{-1}\mathcal{E}_{bs_{g}}^{1}\\
 & -2W^{-3}W_{b}W_{s_{g}}\sum_{f}(\partial_{\phi\phi'\phi_{f}}\mathcal{L})\big[\mathcal{H}^{-1}(\partial_{\beta\phi}\mathcal{L})\big]_{f}\\
 & +W^{-2}W_{bs_{g}}\sum_{f}(\partial_{\phi\phi'\phi_{f}}\mathcal{L})\big[\mathcal{H}^{-1}(\partial_{\beta\phi}\mathcal{L})\big]_{f}\\
 & -W^{-3}W_{s_{g}}\sum_{f}(\partial_{\beta\phi\phi'\phi_{f}}\mathcal{L})\big[\mathcal{H}^{-1}(\partial_{\beta\phi}\mathcal{L})\big]_{f}\\
 & -W^{-3}W_{s_{g}}\sum_{e,f}(\partial_{\phi\phi'\phi_{e}\phi_{f}}\mathcal{L})\big[\mathcal{H}^{-1}(\partial_{\beta\phi}\mathcal{L})\big]_{e}\big[\mathcal{H}^{-1}(\partial_{\beta\phi}\mathcal{L})\big]_{f}\\
 & -W^{-2}W_{s_{g}}\sum_{f}(\partial_{\phi\phi'\phi_{f}}\mathcal{L})\big[\mathcal{H}^{-1}\mathcal{H}_{b}\mathcal{H}^{-1}(\partial_{\beta\phi}\mathcal{L})\big]_{f}\\
 & -W^{-3}W_{s_{g}}\sum_{f}(\partial_{\phi\phi'\phi_{f}}\mathcal{L})\big[\mathcal{H}^{-1}(\partial_{\beta\beta\phi}\mathcal{L})\big]_{f}\\
 & -W^{-3}W_{s_{g}}\sum_{f}(\partial_{\phi\phi'\phi_{f}}\mathcal{L})\big[\mathcal{H}^{-1}\mathcal{E}^{1}\mathcal{H}^{-1}(\partial_{\beta\phi}\mathcal{L})\big]_{f}\\
 & +W^{-3}W_{b}\sum_{f}(\partial_{\beta\phi\phi'\phi_{f}}\mathcal{L})\big[\mathcal{H}^{-1}(\partial_{\beta\phi}\mathcal{L})\big]_{f}\big[\mathcal{H}^{-1}(\partial_{\beta\phi}\mathcal{L})\big]_{g}\\
 & +W^{-3}\sum_{f}(\partial_{\beta^{2}\phi\phi'\phi_{f}}\mathcal{L})\big[\mathcal{H}^{-1}(\partial_{\beta\phi}\mathcal{L})\big]_{f}\big[\mathcal{H}^{-1}(\partial_{\beta\phi}\mathcal{L})\big]_{g}\\
 & +W^{-3}\sum_{e,f}(\partial_{\beta\phi\phi'\phi_{e}\phi_{f}}\mathcal{L})\big[\mathcal{H}^{-1}(\partial_{\beta\phi}\mathcal{L})\big]_{e}\big[\mathcal{H}^{-1}(\partial_{\beta\phi}\mathcal{L})\big]_{f}\big[\mathcal{H}^{-1}(\partial_{\beta\phi}\mathcal{L})\big]_{g}\\
 & +W^{-2}\sum_{f}(\partial_{\beta\phi\phi'\phi_{f}}\mathcal{L})\big[\mathcal{H}^{-1}\mathcal{H}_{b}\mathcal{H}^{-1}(\partial_{\beta\phi}\mathcal{L})\big]_{f}\big[\mathcal{H}^{-1}(\partial_{\beta\phi}\mathcal{L})\big]_{g}\\
 & +W^{-3}\sum_{f}(\partial_{\beta\phi\phi'\phi_{f}}\mathcal{L})\big[\mathcal{H}^{-1}(\partial_{\beta\beta\phi}\mathcal{L})\big]_{f}\big[\mathcal{H}^{-1}(\partial_{\beta\phi}\mathcal{L})\big]_{g}\\
 & +W^{-3}\sum_{f}(\partial_{\beta\phi\phi'\phi_{f}}\mathcal{L})\big[\mathcal{H}^{-1}(\partial_{\beta\phi\phi'}\mathcal{L})\mathcal{H}^{-1}(\partial_{\beta\phi}\mathcal{L})\big]_{f}\big[\mathcal{H}^{-1}(\partial_{\beta\phi}\mathcal{L})\big]_{g}\\
 & +W^{-2}\sum_{f}(\partial_{\beta\phi\phi'\phi_{f}}\mathcal{L})\big[\mathcal{H}^{-1}(\partial_{\beta\phi}\mathcal{L})\big]_{f}\big[\mathcal{H}^{-1}\mathcal{H}_{b}\mathcal{H}^{-1}(\partial_{\beta\phi}\mathcal{L})\big]_{g}\\
 & +W^{-3}\sum_{f}(\partial_{\beta\phi\phi'\phi_{f}}\mathcal{L})\big[\mathcal{H}^{-1}(\partial_{\beta\phi}\mathcal{L})\big]_{f}\big[\mathcal{H}^{-1}(\partial_{\beta\beta\phi}\mathcal{L})\big]_{g}\\
 & +W^{-3}\sum_{f}(\partial_{\beta\phi\phi'\phi_{f}}\mathcal{L})\big[\mathcal{H}^{-1}(\partial_{\beta\phi}\mathcal{L})\big]_{f}\big[\mathcal{H}^{-1}(\partial_{\beta\phi\phi'}\mathcal{L})\mathcal{H}^{-1}(\partial_{\beta\phi}\mathcal{L})\big]_{g}\\
 & -W^{-2}W_{b}\sum_{e,f}(\partial_{\phi\phi'\phi_{e}\phi_{f}}\mathcal{L})\big[\mathcal{H}^{-1}(\partial_{\beta\phi}\mathcal{L})\big]_{f}\big[\mathcal{H}^{-1}\big]_{eg}\\
 & -W^{-2}\sum_{e,f}(\partial_{\beta\phi\phi'\phi_{e}\phi_{f}}\mathcal{L})\big[\mathcal{H}^{-1}(\partial_{\beta\phi}\mathcal{L})\big]_{f}\big[\mathcal{H}^{-1}\big]_{eg}\\
 & -W^{-2}\sum_{d,e,f}(\partial_{\phi\phi'\phi_{d}\phi_{e}\phi_{f}}\mathcal{L})\big[\mathcal{H}^{-1}(\partial_{\beta\phi}\mathcal{L})\big]_{d}\big[\mathcal{H}^{-1}(\partial_{\beta\phi}\mathcal{L})\big]_{f}\big[\mathcal{H}^{-1}\big]_{eg}\\
 & -W^{-1}\sum_{e,f}(\partial_{\phi\phi'\phi_{e}\phi_{f}}\mathcal{L})\big[\mathcal{H}^{-1}\mathcal{H}_{b}\mathcal{H}^{-1}(\partial_{\beta\phi}\mathcal{L})\big]_{f}\big[\mathcal{H}^{-1}\big]_{eg}\\
 & -W^{-2}\sum_{e,f}(\partial_{\phi\phi'\phi_{e}\phi_{f}}\mathcal{L})\big[\mathcal{H}^{-1}(\partial_{\beta\beta\phi}\mathcal{L})\big]_{f}\big[\mathcal{H}^{-1}\big]_{eg}\\
 & -W^{-2}\sum_{e,f}(\partial_{\phi\phi'\phi_{e}\phi_{f}}\mathcal{L})\big[\mathcal{H}^{-1}(\partial_{\beta\phi\phi'}\mathcal{L})\mathcal{H}^{-1}(\partial_{\beta\phi}\mathcal{L})\big]_{f}\big[\mathcal{H}^{-1}\big]_{eg}\\
 & -W^{-1}\sum_{e,f}(\partial_{\phi\phi'\phi_{e}\phi_{f}}\mathcal{L})\big[\mathcal{H}^{-1}(\partial_{\beta\phi}\mathcal{L})\big]_{f}\big[\mathcal{H}^{-1}\mathcal{H}_{b}\mathcal{H}^{-1}\big]_{eg}\\
 & +2W^{-3}W_{b}\sum_{e,f}(\partial_{\phi\phi'\phi_{e}\phi_{f}}\mathcal{L})\big[\mathcal{H}^{-1}(\partial_{\beta\phi}\mathcal{L})\big]_{e}\big[\mathcal{H}^{-1}(\partial_{\beta\phi}\mathcal{L})\big]_{f}\big[\mathcal{H}^{-1}(\partial_{\beta\phi}\mathcal{L})\big]_{g}\\
 & +W^{-3}\sum_{e,f}(\partial_{\beta\phi\phi'\phi_{e}\phi_{f}}\mathcal{L})\big[\mathcal{H}^{-1}(\partial_{\beta\phi}\mathcal{L})\big]_{e}\big[\mathcal{H}^{-1}(\partial_{\beta\phi}\mathcal{L})\big]_{f}\big[\mathcal{H}^{-1}(\partial_{\beta\phi}\mathcal{L})\big]_{g}\\
 & +W^{-3}\sum_{d,e,f}(\partial_{\phi\phi'\phi_{d}\phi_{e}\phi_{f}}\mathcal{L})\big[\mathcal{H}^{-1}(\partial_{\beta\phi}\mathcal{L})\big]_{d}\big[\mathcal{H}^{-1}(\partial_{\beta\phi}\mathcal{L})\big]_{e}\big[\mathcal{H}^{-1}(\partial_{\beta\phi}\mathcal{L})\big]_{f}\big[\mathcal{H}^{-1}(\partial_{\beta\phi}\mathcal{L})\big]_{g}\\
 & +2W^{-2}\sum_{e,f}(\partial_{\phi\phi'\phi_{e}\phi_{f}}\mathcal{L})\big[\mathcal{H}^{-1}\mathcal{H}_{b}\mathcal{H}^{-1}(\partial_{\beta\phi}\mathcal{L})\big]_{e}\big[\mathcal{H}^{-1}(\partial_{\beta\phi}\mathcal{L})\big]_{f}\big[\mathcal{H}^{-1}(\partial_{\beta\phi}\mathcal{L})\big]_{g}\\
 & +2W^{-3}\sum_{e,f}(\partial_{\phi\phi'\phi_{e}\phi_{f}}\mathcal{L})\big[\mathcal{H}^{-1}(\partial_{\beta\beta\phi}\mathcal{L})\big]_{e}\big[\mathcal{H}^{-1}(\partial_{\beta\phi}\mathcal{L})\big]_{f}\big[\mathcal{H}^{-1}(\partial_{\beta\phi}\mathcal{L})\big]_{g}\\
 & +2W^{-3}\sum_{e,f}(\partial_{\phi\phi'\phi_{e}\phi_{f}}\mathcal{L})\big[\mathcal{H}^{-1}(\partial_{\beta\phi\phi'}\mathcal{L})\mathcal{H}^{-1}(\partial_{\beta\phi}\mathcal{L})\big]_{e}\big[\mathcal{H}^{-1}(\partial_{\beta\phi}\mathcal{L})\big]_{f}\big[\mathcal{H}^{-1}(\partial_{\beta\phi}\mathcal{L})\big]_{g}\\
 & +W^{-2}\sum_{e,f}(\partial_{\phi\phi'\phi_{e}\phi_{f}}\mathcal{L})\big[\mathcal{H}^{-1}(\partial_{\beta\phi}\mathcal{L})\big]_{e}\big[\mathcal{H}^{-1}(\partial_{\beta\phi}\mathcal{L})\big]_{f}\big[\mathcal{H}^{-1}\mathcal{H}_{b}\mathcal{H}^{-1}(\partial_{\beta\phi}\mathcal{L})\big]_{g}\\
 & +W^{-3}\sum_{e,f}(\partial_{\phi\phi'\phi_{e}\phi_{f}}\mathcal{L})\big[\mathcal{H}^{-1}(\partial_{\beta\phi}\mathcal{L})\big]_{e}\big[\mathcal{H}^{-1}(\partial_{\beta\phi}\mathcal{L})\big]_{f}\big[\mathcal{H}^{-1}(\partial_{\beta\beta\phi}\mathcal{L})\big]_{g}\\
 & +W^{-3}\sum_{e,f}(\partial_{\phi\phi'\phi_{e}\phi_{f}}\mathcal{L})\big[\mathcal{H}^{-1}(\partial_{\beta\phi}\mathcal{L})\big]_{e}\big[\mathcal{H}^{-1}(\partial_{\beta\phi}\mathcal{L})\big]_{f}\big[\mathcal{H}^{-1}(\partial_{\beta\phi\phi'}\mathcal{L})\mathcal{H}^{-1}(\partial_{\beta\phi}\mathcal{L})\big]_{g}\\
 & -W^{-2}W_{b}\sum_{f}(\partial_{\phi\phi'\phi_{f}}\mathcal{L})\big[\mathcal{H}^{-1}\mathcal{H}_{b}\mathcal{H}^{-1}(\partial_{\beta\phi}\mathcal{L})\big]_{f}\\
 & -W^{-2}\sum_{f}(\partial_{\beta\phi\phi'\phi_{f}}\mathcal{L})\big[\mathcal{H}^{-1}\mathcal{H}_{b}\mathcal{H}^{-1}(\partial_{\beta\phi}\mathcal{L})\big]_{f}\\
 & -W^{-2}\sum_{e,f}(\partial_{\phi\phi'\phi_{e}\phi_{f}}\mathcal{L})\big[\mathcal{H}^{-1}(\partial_{\beta\phi}\mathcal{L})\big]_{e}\big[\mathcal{H}^{-1}\mathcal{H}_{b}\mathcal{H}^{-1}(\partial_{\beta\phi}\mathcal{L})\big]_{f}\\
 & -2W^{-1}\sum_{f}(\partial_{\phi\phi'\phi_{f}}\mathcal{L})\big[\mathcal{H}^{-1}\mathcal{H}_{b}\mathcal{H}^{-1}\mathcal{H}_{b}\mathcal{H}^{-1}(\partial_{\beta\phi}\mathcal{L})\big]_{f}\\
 & +W^{-1}\sum_{f}(\partial_{\phi\phi'\phi_{f}}\mathcal{L})\big[\mathcal{H}^{-1}\mathcal{H}_{bb}\mathcal{H}^{-1}(\partial_{\beta\phi}\mathcal{L})\big]_{f}\\
 & -W^{-2}\sum_{f}(\partial_{\phi\phi'\phi_{f}}\mathcal{L})\big[\mathcal{H}^{-1}\mathcal{H}_{b}\mathcal{H}^{-1}(\partial_{\beta\beta\phi}\mathcal{L})\big]_{f}\\
 & -W^{-2}\sum_{f}(\partial_{\phi\phi'\phi_{f}}\mathcal{L})\big[\mathcal{H}^{-1}\mathcal{H}_{b}\mathcal{H}^{-1}(\partial_{\beta\phi\phi'}\mathcal{L})\mathcal{H}^{-1}(\partial_{\beta\phi}\mathcal{L})\big]_{f}\\
 & +W^{-3}W_{b}\sum_{f}(\partial_{\phi\phi'\phi_{f}}\mathcal{L})\big[\mathcal{H}^{-1}(\partial_{\beta\phi}\mathcal{L})\big]_{f}\big[\mathcal{H}^{-1}(\partial_{\beta\phi}\mathcal{L})\big]_{g}\\
 & +W^{-3}\sum_{f}(\partial_{\beta\phi\phi'\phi_{f}}\mathcal{L})\big[\mathcal{H}^{-1}(\partial_{\beta\phi}\mathcal{L})\big]_{f}\big[\mathcal{H}^{-1}(\partial_{\beta\phi}\mathcal{L})\big]_{g}\\
 & +W^{-3}\sum_{e,f}(\partial_{\phi\phi'\phi_{e}\phi_{f}}\mathcal{L})\big[\mathcal{H}^{-1}(\partial_{\beta\phi}\mathcal{L})\big]_{e}\big[\mathcal{H}^{-1}(\partial_{\beta\phi}\mathcal{L})\big]_{f}\big[\mathcal{H}^{-1}(\partial_{\beta\phi}\mathcal{L})\big]_{g}\\
 & +W^{-2}\sum_{f}(\partial_{\phi\phi'\phi_{f}}\mathcal{L})\big[\mathcal{H}^{-1}\mathcal{H}_{b}\mathcal{H}^{-1}(\partial_{\beta\phi}\mathcal{L})\big]_{f}\big[\mathcal{H}^{-1}(\partial_{\beta\phi}\mathcal{L})\big]_{g}\\
 & +W^{-3}\sum_{f}(\partial_{\phi\phi'\phi_{f}}\mathcal{L})\big[\mathcal{H}^{-1}(\partial_{\beta\beta\phi}\mathcal{L})\big]_{f}\big[\mathcal{H}^{-1}(\partial_{\beta\phi}\mathcal{L})\big]_{g}\\
 & +W^{-3}\sum_{f}(\partial_{\phi\phi'\phi_{f}}\mathcal{L})\big[\mathcal{H}^{-1}(\partial_{\beta\phi\phi'}\mathcal{L})\mathcal{H}^{-1}(\partial_{\beta\phi}\mathcal{L})\big]_{f}\big[\mathcal{H}^{-1}(\partial_{\beta\phi}\mathcal{L})\big]_{g}\\
 & +W^{-2}\sum_{f}(\partial_{\phi\phi'\phi_{f}}\mathcal{L})\big[\mathcal{H}^{-1}(\partial_{\beta\phi}\mathcal{L})\big]_{f}\big[\mathcal{H}^{-1}\mathcal{H}_{b}\mathcal{H}^{-1}(\partial_{\beta\phi}\mathcal{L})\big]_{g}\\
 & +W^{-3}\sum_{f}(\partial_{\phi\phi'\phi_{f}}\mathcal{L})\big[\mathcal{H}^{-1}(\partial_{\beta\phi}\mathcal{L})\big]_{f}\big[\mathcal{H}^{-1}(\partial_{\beta\beta\phi}\mathcal{L})\big]_{g}\\
 & +W^{-3}\sum_{f}(\partial_{\phi\phi'\phi_{f}}\mathcal{L})\big[\mathcal{H}^{-1}(\partial_{\beta\phi}\mathcal{L})\big]_{f}\big[\mathcal{H}^{-1}(\partial_{\beta\phi\phi'}\mathcal{L})\mathcal{H}^{-1}(\partial_{\beta\phi}\mathcal{L})\big]_{g}\\
 & -W^{-2}W_{b}\sum_{f}(\partial_{\phi\phi'\phi_{f}}\mathcal{L})\big[\mathcal{H}^{-1}\mathcal{E}^{1}\mathcal{H}^{-1}\big]_{fg}\\
 & -W^{-2}\sum_{f}(\partial_{\beta\phi\phi'\phi_{f}}\mathcal{L})\big[\mathcal{H}^{-1}\mathcal{E}^{1}\mathcal{H}^{-1}\big]_{fg}\\
 & -W^{-2}\sum_{e,f}(\partial_{\phi\phi'\phi_{e}\phi_{f}}\mathcal{L})\big[\mathcal{H}^{-1}(\partial_{\beta\phi}\mathcal{L})\big]_{e}\big[\mathcal{H}^{-1}\mathcal{E}^{1}\mathcal{H}^{-1}\big]_{fg}\\
 & -W^{-1}\sum_{f}(\partial_{\phi\phi'\phi_{f}}\mathcal{L})\big[\mathcal{H}^{-1}\mathcal{H}_{b}\mathcal{H}^{-1}\mathcal{E}^{1}\mathcal{H}^{-1}\big]_{fg}\\
 & -W^{-1}\sum_{f}(\partial_{\phi\phi'\phi_{f}}\mathcal{L})\big[\mathcal{H}^{-1}\mathcal{E}^{1}\mathcal{H}^{-1}\mathcal{H}_{b}\mathcal{H}^{-1}\big]_{fg}\\
 & +W^{-1}\sum_{f}(\partial_{\phi\phi'\phi_{f}}\mathcal{L})\big[\mathcal{H}^{-1}\mathcal{E}_{b}^{1}\mathcal{H}^{-1}\big]_{fg}\\
 & +2W^{-3}W_{b}\sum_{f}(\partial_{\phi\phi'\phi_{f}}\mathcal{L})\big[\mathcal{H}^{-1}\mathcal{E}^{1}\mathcal{H}^{-1}(\partial_{\beta\phi}\mathcal{L})\big]_{f}\big[\mathcal{H}^{-1}(\partial_{\beta\phi}\mathcal{L})\big]_{g}\\
 & +W^{-3}\sum_{f}(\partial_{\beta\phi\phi'\phi_{f}}\mathcal{L})\big[\mathcal{H}^{-1}\mathcal{E}^{1}\mathcal{H}^{-1}(\partial_{\beta\phi}\mathcal{L})\big]_{f}\big[\mathcal{H}^{-1}(\partial_{\beta\phi}\mathcal{L})\big]_{g}\\
 & +W^{-3}\sum_{e,f}(\partial_{\phi\phi'\phi_{e}\phi_{f}}\mathcal{L})\big[\mathcal{H}^{-1}(\partial_{\beta\phi}\mathcal{L})\big]_{e}\big[\mathcal{H}^{-1}\mathcal{E}^{1}\mathcal{H}^{-1}(\partial_{\beta\phi}\mathcal{L})\big]_{f}\big[\mathcal{H}^{-1}(\partial_{\beta\phi}\mathcal{L})\big]_{g}\\
 & +W^{-2}\sum_{f}(\partial_{\phi\phi'\phi_{f}}\mathcal{L})\big[\mathcal{H}^{-1}\mathcal{H}_{b}\mathcal{H}^{-1}\mathcal{E}^{1}\mathcal{H}^{-1}(\partial_{\beta\phi}\mathcal{L})\big]_{f}\big[\mathcal{H}^{-1}(\partial_{\beta\phi}\mathcal{L})\big]_{g}\\
 & -W^{-2}\sum_{f}(\partial_{\phi\phi'\phi_{f}}\mathcal{L})\big[\mathcal{H}^{-1}\mathcal{E}_{b}^{1}\mathcal{H}^{-1}(\partial_{\beta\phi}\mathcal{L})\big]_{f}\big[\mathcal{H}^{-1}(\partial_{\beta\phi}\mathcal{L})\big]_{g}\\
 & +W^{-2}\sum_{f}(\partial_{\phi\phi'\phi_{f}}\mathcal{L})\big[\mathcal{H}^{-1}\mathcal{E}^{1}\mathcal{H}^{-1}\mathcal{H}_{b}\mathcal{H}^{-1}(\partial_{\beta\phi}\mathcal{L})\big]_{f}\big[\mathcal{H}^{-1}(\partial_{\beta\phi}\mathcal{L})\big]_{g}\\
 & +W^{-3}\sum_{f}(\partial_{\phi\phi'\phi_{f}}\mathcal{L})\big[\mathcal{H}^{-1}\mathcal{E}^{1}\mathcal{H}^{-1}(\partial_{\beta\beta\phi}\mathcal{L})\big]_{f}\big[\mathcal{H}^{-1}(\partial_{\beta\phi}\mathcal{L})\big]_{g}\\
 & +W^{-3}\sum_{f}(\partial_{\phi\phi'\phi_{f}}\mathcal{L})\big[\mathcal{H}^{-1}\mathcal{E}^{1}\mathcal{H}^{-1}\mathcal{E}^{1}\mathcal{H}^{-1}(\partial_{\beta\phi}\mathcal{L})\big]_{f}\big[\mathcal{H}^{-1}(\partial_{\beta\phi}\mathcal{L})\big]_{g}\\
 & +W^{-2}\sum_{f}(\partial_{\phi\phi'\phi_{f}}\mathcal{L})\big[\mathcal{H}^{-1}\mathcal{E}^{1}\mathcal{H}^{-1}(\partial_{\beta\phi}\mathcal{L})\big]_{f}\big[\mathcal{H}^{-1}\mathcal{H}_{b}\mathcal{H}^{-1}(\partial_{\beta\phi}\mathcal{L})\big]_{g}\\
 & +W^{-3}\sum_{f}(\partial_{\phi\phi'\phi_{f}}\mathcal{L})\big[\mathcal{H}^{-1}\mathcal{E}^{1}\mathcal{H}^{-1}(\partial_{\beta\phi}\mathcal{L})\big]_{f}\big[\mathcal{H}^{-1}(\partial_{\beta\beta\phi}\mathcal{L})\big]_{g}\\
 & +W^{-3}\sum_{f}(\partial_{\phi\phi'\phi_{f}}\mathcal{L})\big[\mathcal{H}^{-1}\mathcal{E}^{1}\mathcal{H}^{-1}(\partial_{\beta\phi}\mathcal{L})\big]_{f}\big[\mathcal{H}^{-1}\mathcal{E}^{1}\mathcal{H}^{-1}(\partial_{\beta\beta\phi}\mathcal{L})\big]_{g}
\end{align*}

\begin{align*}
	\mathcal{H}_{bs_{g}s_{h}} & =2W^{-3}W_{b}W_{s_{h}}[\mathcal{R}^{1}]_{g}(\partial_{\beta\phi\phi'}\mathcal{L})-W^{-2}W_{bs_{h}}[\mathcal{R}^{1}]_{g}(\partial_{\beta\phi\phi'}\mathcal{L})\\
	& -W^{-2}W_{s_{h}}[\mathcal{R}_{b}^{1}]_{g}(\partial_{\beta\phi\phi'}\mathcal{L})+W^{-3}W_{s_{h}}[\mathcal{R}^{1}]_{g}(\partial_{\beta^{2}\phi\phi'}\mathcal{L})\\
	& +W^{-3}W_{s_{h}}[\mathcal{R}^{1}]_{g}\mathcal{P}^{(1,1)}-W^{-2}W_{b}[\mathcal{R}_{s_{h}}^{1}]_{g}(\partial_{\beta\phi\phi'}\mathcal{L})\\
	& +W^{-1}[\mathcal{R}_{bs_{h}}^{1}]_{g}(\partial_{\beta\phi\phi'}\mathcal{L})-W^{-2}[\mathcal{R}_{s_{h}}^{1}]_{g}(\partial_{\beta^{2}\phi\phi'}\mathcal{L})-W^{-3}[\mathcal{R}_{s_{h}}^{1}]_{g}\mathcal{P}^{(1,1)}\\
	& +2W^{-3}W_{b}[\mathcal{R}^{1}]_{g}(\partial_{\beta^{2}\phi\phi'}\mathcal{L})[\mathcal{R}^{1}]_{h}-W^{-2}[\mathcal{R}_{b}^{1}]_{g}(\partial_{\beta^{2}\phi\phi'}\mathcal{L})[\mathcal{R}^{1}]_{h}\\
	& +W^{-3}[\mathcal{R}^{1}]_{g}(\partial_{\beta^{3}\phi\phi'}\mathcal{L})[\mathcal{R}^{1}]_{h}+W^{-3}[\mathcal{R}^{1}]_{g}\mathcal{P}^{(2,1)}[\mathcal{R}^{1}]_{h}\\
	& -W^{-2}[\mathcal{R}^{1}]_{g}(\partial_{\beta^{2}\phi\phi'}\mathcal{L})[\mathcal{R}_{b}^{1}]_{h}+W^{-2}W_{b}[\mathcal{R}^{1}]_{g}\mathcal{P}_{(h)}^{1}\\
	& -W^{-1}[\mathcal{R}_{b}^{1}]_{g}\mathcal{P}_{(h)}^{1}-W^{-1}[\mathcal{R}^{1}]_{g}\mathcal{P}_{(h),b}^{1}\\
	& +2W^{-3}W_{b}[\mathcal{R}^{1}]_{g}\mathcal{P}^{(1,1)}[\mathcal{R}^{1}]_{h}-W^{-2}[\mathcal{R}_{b}^{1}]_{g}\mathcal{P}^{(1,1)}[\mathcal{R}^{1}]_{h}\\
	& -W^{-2}[\mathcal{R}^{1}]_{g}\mathcal{P}_{b}^{(1,1)}[\mathcal{R}^{1}]_{h}-W^{-2}[\mathcal{R}^{1}]_{g}\mathcal{P}^{(1,1)}[\mathcal{R}_{b}^{1}]_{h}\\
	& +\mathcal{P}_{(g),bs_{h}}^{0}\\
	& +2W^{-3}W_{b}W_{s_{h}}\mathcal{P}^{(0,1)}[\mathcal{R}^{1}]_{g}-W^{-2}W_{bs_{h}}\mathcal{P}^{(0,1)}[\mathcal{R}^{1}]_{g}\\
	& -W^{-2}W_{s_{h}}\mathcal{P}_{b}^{(0,1)}[\mathcal{R}^{1}]_{g}-W^{-2}W_{s_{h}}\mathcal{P}^{(0,1)}[\mathcal{R}_{b}^{1}]_{g}\\
	& -W^{-2}W_{b}\mathcal{P}_{s_{h}}^{(0,1)}[\mathcal{R}^{1}]_{g}+W^{-1}\mathcal{P}_{bs_{h}}^{(0,1)}[\mathcal{R}^{1}]_{g}+W^{-1}\mathcal{P}_{s_{h}}^{(0,1)}[\mathcal{R}_{b}^{1}]_{g}\\
	& -W^{-2}W_{b}\mathcal{P}^{(0,1)}[\mathcal{R}_{s_{h}}^{1}]_{g}+W^{-1}\mathcal{P}_{b}^{(0,1)}[\mathcal{R}_{s_{h}}^{1}]_{g}+W^{-1}\mathcal{P}^{(0,1)}[\mathcal{R}_{bs_{h}}^{1}]_{g}
\end{align*}

\begin{align*}
	\mathcal{H}_{s_{f}s_{g}s_{h}} & =2W^{-3}W_{s_{f}}W_{s_{h}}[\mathcal{R}^{1}]_{g}(\partial_{\beta\phi\phi'}\mathcal{L})-W^{-2}W_{s_{f}s_{h}}[\mathcal{R}^{1}]_{g}(\partial_{\beta\phi\phi'}\mathcal{L})\\
	& -W^{-2}W_{s_{h}}[\mathcal{R}_{s_{f}}^{1}]_{g}(\partial_{\beta\phi\phi'}\mathcal{L})+W^{-3}W_{s_{h}}[\mathcal{R}^{1}]_{g}(\partial_{\beta^{2}\phi\phi'}\mathcal{L})[\mathcal{R}^{1}]_{f}\\
	& +W^{-2}W_{s_{h}}[\mathcal{R}^{1}]_{g}\sum_{e}(\partial_{\beta\phi\phi'\phi_{e}}\mathcal{L})[\mathcal{H}^{-1}]_{ef}\\
	& +W^{-3}W_{s_{h}}[\mathcal{R}^{1}]_{g}\sum_{e}(\partial_{\beta\phi\phi'\phi_{e}}\mathcal{L})[\mathcal{R}^{1}]_{e}[\mathcal{R}^{1}]_{f}\\
	& -W^{-2}W_{s_{f}}[\mathcal{R}_{s_{h}}^{1}]_{g}(\partial_{\beta\phi\phi'}\mathcal{L})+W^{-1}[\mathcal{R}_{s_{f}s_{h}}^{1}]_{g}(\partial_{\beta\phi\phi'}\mathcal{L})\\
	& -W^{-2}[\mathcal{R}_{s_{h}}^{1}]_{g}(\partial_{\beta^{2}\phi\phi'}\mathcal{L})[\mathcal{R}^{1}]_{f}-W^{-1}[\mathcal{R}_{s_{h}}^{1}]_{g}\sum_{e}(\partial_{\beta\phi\phi'\phi_{e}}\mathcal{L})[\mathcal{H}^{-1}]_{ef}\\
	& -W^{-2}[\mathcal{R}_{s_{h}}^{1}]_{g}\sum_{e}(\partial_{\beta\phi\phi'\phi_{e}}\mathcal{L})[\mathcal{R}^{1}]_{e}[\mathcal{R}^{1}]_{f}\\
	& +2W^{-3}W_{s_{f}}[\mathcal{R}^{1}]_{g}(\partial_{\beta^{2}\phi\phi'}\mathcal{L})[\mathcal{R}^{1}]_{h}-W^{-2}[\mathcal{R}_{s_{f}}^{1}]_{g}(\partial_{\beta^{2}\phi\phi'}\mathcal{L})[\mathcal{R}^{1}]_{h}\\
	& +W^{-3}[\mathcal{R}^{1}]_{g}(\partial_{\beta^{3}\phi\phi'}\mathcal{L})[\mathcal{R}^{1}]_{f}[\mathcal{R}^{1}]_{h}+W^{-3}[\mathcal{R}^{1}]_{g}(\partial_{\beta^{2}\phi\phi'}\mathcal{L})[\mathcal{R}^{1}]_{h}\\
	& -W^{-2}[\mathcal{R}^{1}]_{g}\sum_{e}(\partial_{\beta^{2}\phi\phi'\phi_{e}}\mathcal{L})[\mathcal{R}^{1}]_{e}[\mathcal{R}^{1}]_{f}[\mathcal{R}^{1}]_{h}\\
	& +W^{-2}W_{s_{f}}[\mathcal{R}^{1}]_{g}\mathcal{P}_{(h)}^{1}-W^{-1}[\mathcal{R}_{s_{f}}^{1}]_{g}\mathcal{P}_{(h)}^{1}-W^{-1}[\mathcal{R}^{1}]_{g}\mathcal{P}_{(h),s_{f}}^{1}\\
	& +2W^{-3}W_{s_{f}}[\mathcal{R}^{1}]_{g}\mathcal{P}^{(1,1)}[\mathcal{R}^{1}]_{h}-W^{-2}[\mathcal{R}_{s_{f}}^{1}]_{g}\mathcal{P}^{(1,1)}[\mathcal{R}^{1}]_{h}\\
	& -W^{-2}[\mathcal{R}^{1}]_{g}\mathcal{P}_{s_{f}}^{(1,1)}[\mathcal{R}^{1}]_{h}-W^{-2}[\mathcal{R}^{1}]_{g}\mathcal{P}^{(1,1)}[\mathcal{R}_{s_{f}}^{1}]_{h}\\
	& +\mathcal{P}_{(g),s_{f}s_{h}}^{0}\\
	& +2W^{-3}W_{s_{f}}W_{s_{h}}\mathcal{P}^{(0,1)}[\mathcal{R}^{1}]_{g}-W^{-2}W_{s_{f}s_{h}}\mathcal{P}^{(0,1)}[\mathcal{R}^{1}]_{g}\\
	& -W^{-2}W_{s_{h}}\mathcal{P}_{s_{f}}^{(0,1)}[\mathcal{R}^{1}]_{g}-W^{-2}W_{s_{h}}\mathcal{P}^{(0,1)}[\mathcal{R}_{s_{f}}^{1}]_{g}\\
	& -W^{-2}W_{s_{f}}\mathcal{P}_{s_{h}}^{(0,1)}[\mathcal{R}^{1}]_{g}+W^{-1}\mathcal{P}_{s_{f}s_{h}}^{(0,1)}[\mathcal{R}^{1}]_{g}+W^{-1}\mathcal{P}_{s_{h}}^{(0,1)}[\mathcal{R}_{s_{f}}^{1}]_{g}\\
	& -W^{-2}W_{s_{f}}\mathcal{P}^{(0,1)}[\mathcal{R}_{s_{h}}^{1}]_{g}+W^{-1}\mathcal{P}_{s_{f}}^{(0,1)}[\mathcal{R}_{s_{h}}^{1}]_{g}+W^{-1}\mathcal{P}^{(0,1)}[\mathcal{R}_{s_{f}s_{h}}^{1}]_{g}
\end{align*}

Fourth derivative:

\begin{align*}
	\mathcal{H}_{bs_{f}s_{g}s_{h}} & =-6W^{-4}W_{b}W_{s_{f}}W_{s_{h}}[\mathcal{R}^{1}]_{g}(\partial_{\beta\phi\phi'}\mathcal{L})+2W^{-3}W_{bs_{f}}W_{s_{h}}[\mathcal{R}^{1}]_{g}(\partial_{\beta\phi\phi'}\mathcal{L})\\
	& +2W^{-3}W_{s_{f}}W_{bs_{h}}[\mathcal{R}^{1}]_{g}(\partial_{\beta\phi\phi'}\mathcal{L})+2W^{-3}W_{s_{f}}W_{s_{h}}[\mathcal{R}_{b}^{1}]_{g}(\partial_{\beta\phi\phi'}\mathcal{L})\\
	& -2W^{-4}W_{s_{f}}W_{s_{h}}[\mathcal{R}^{1}]_{g}(\partial_{\beta^{2}\phi\phi'}\mathcal{L})-2W^{-4}W_{s_{f}}W_{s_{h}}[\mathcal{R}^{1}]_{g}\sum_{e}(\partial_{\beta\phi\phi'\phi_{e}}\mathcal{L})[\mathcal{R}^{1}]_{e}\\
	& +2W^{-3}W_{b}W_{s_{f}s_{h}}[\mathcal{R}^{1}]_{g}(\partial_{\beta\phi\phi'}\mathcal{L})-W^{-2}W_{bs_{f}s_{h}}[\mathcal{R}^{1}]_{g}(\partial_{\beta\phi\phi'}\mathcal{L})\\
	& -W^{-2}W_{s_{f}s_{h}}[\mathcal{R}_{b}^{1}]_{g}(\partial_{\beta\phi\phi'}\mathcal{L})+W^{-3}W_{s_{f}s_{h}}[\mathcal{R}^{1}]_{g}(\partial_{\beta^{2}\phi\phi'}\mathcal{L})\\
	& +W^{-3}W_{s_{f}s_{h}}[\mathcal{R}^{1}]_{g}\sum_{e}(\partial_{\beta\phi\phi'\phi_{e}}\mathcal{L})[\mathcal{R}^{1}]_{e}+2W^{-3}W_{b}W_{s_{h}}[\mathcal{R}_{s_{f}}^{1}]_{g}(\partial_{\beta\phi\phi'}\mathcal{L})\\
	& -W^{-2}W_{bs_{h}}[\mathcal{R}_{s_{f}}^{1}]_{g}(\partial_{\beta\phi\phi'}\mathcal{L})-W^{-2}W_{s_{h}}[\mathcal{R}_{bs_{f}}^{1}]_{g}(\partial_{\beta\phi\phi'}\mathcal{L})\\
	& +W^{-3}W_{s_{h}}[\mathcal{R}_{s_{f}}^{1}]_{g}(\partial_{\beta^{2}\phi\phi'}\mathcal{L})+W^{-3}W_{s_{h}}[\mathcal{R}_{s_{f}}^{1}]_{g}\sum_{e}(\partial_{\beta\phi\phi'\phi_{e}}\mathcal{L})[\mathcal{R}^{1}]_{e}\\
	& -3W^{-4}W_{b}W_{s_{h}}[\mathcal{R}^{1}]_{g}(\partial_{\beta^{2}\phi\phi'}\mathcal{L})[\mathcal{R}^{1}]_{f}+W^{-3}W_{bs_{h}}[\mathcal{R}^{1}]_{g}(\partial_{\beta^{2}\phi\phi'}\mathcal{L})[\mathcal{R}^{1}]_{f}\\
	& +W^{-3}W_{s_{h}}[\mathcal{R}_{b}^{1}]_{g}(\partial_{\beta^{2}\phi\phi'}\mathcal{L})[\mathcal{R}^{1}]_{f}-W^{-4}W_{s_{h}}[\mathcal{R}^{1}]_{g}(\partial_{\beta^{3}\phi\phi'}\mathcal{L})[\mathcal{R}^{1}]_{f}\\
	& -W^{-4}W_{s_{h}}[\mathcal{R}^{1}]_{g}\sum_{e}(\partial_{\beta^{2}\phi\phi'\phi_{e}}\mathcal{L})[\mathcal{R}^{1}]_{e}[\mathcal{R}^{1}]_{f}\\
	& +W^{-3}W_{s_{h}}[\mathcal{R}^{1}]_{g}(\partial_{\beta^{2}\phi\phi'}\mathcal{L})[\mathcal{R}_{b}^{1}]_{f}-2W^{-3}W_{b}W_{s_{h}}[\mathcal{R}^{1}]_{g}\sum_{e}(\partial_{\beta\phi\phi'\phi_{e}}\mathcal{L})[\mathcal{H}^{-1}]_{ef}\\
	& +W^{-2}W_{bs_{h}}[\mathcal{R}^{1}]_{g}\sum_{e}(\partial_{\beta\phi\phi'\phi_{e}}\mathcal{L})[\mathcal{H}^{-1}]_{ef}+W^{-2}W_{s_{h}}[\mathcal{R}_{b}^{1}]_{g}\sum_{e}(\partial_{\beta\phi\phi'\phi_{e}}\mathcal{L})[\mathcal{H}^{-1}]_{ef}\\
	& -W^{-3}W_{s_{h}}[\mathcal{R}^{1}]_{g}\sum_{e}(\partial_{\beta^{2}\phi\phi'\phi_{e}}\mathcal{L})[\mathcal{H}^{-1}]_{ef}\\
	& -W^{-3}W_{s_{h}}[\mathcal{R}^{1}]_{g}\sum_{d,e}(\partial_{\beta\phi\phi'\phi_{d}\phi_{e}}\mathcal{L})[\mathcal{H}^{-1}]_{ef}[\mathcal{R}^{1}]_{d}\\
	& -W^{-2}W_{s_{h}}[\mathcal{R}^{1}]_{g}\sum_{e}(\partial_{\beta\phi\phi'\phi_{e}}\mathcal{L})[\mathcal{H}^{-1}\mathcal{H}_{b}\mathcal{H}^{-1}]_{ef}\\
	& -3W^{-4}W_{b}W_{s_{h}}[\mathcal{R}^{1}]_{g}\sum_{e}(\partial_{\beta\phi\phi'\phi_{e}}\mathcal{L})[\mathcal{R}^{1}]_{e}[\mathcal{R}^{1}]_{f}\\
	& +W^{-3}W_{bs_{h}}[\mathcal{R}^{1}]_{g}\sum_{e}(\partial_{\beta\phi\phi'\phi_{e}}\mathcal{L})[\mathcal{R}^{1}]_{e}[\mathcal{R}^{1}]_{f}\\
	& +W^{-3}W_{s_{h}}[\mathcal{R}_{b}^{1}]_{g}\sum_{e}(\partial_{\beta\phi\phi'\phi_{e}}\mathcal{L})[\mathcal{R}^{1}]_{e}[\mathcal{R}^{1}]_{f}\\
	& -W^{-4}W_{s_{h}}[\mathcal{R}^{1}]_{g}\sum_{e}(\partial_{\beta^{2}\phi\phi'\phi_{e}}\mathcal{L})[\mathcal{R}^{1}]_{e}[\mathcal{R}^{1}]_{f}\\
	& -W^{-4}W_{s_{h}}[\mathcal{R}^{1}]_{g}\sum_{d,e}(\partial_{\beta\phi\phi'\phi_{d}\phi_{e}}\mathcal{L})[\mathcal{R}^{1}]_{d}[\mathcal{R}^{1}]_{e}[\mathcal{R}^{1}]_{f}\\
	& +W^{-3}W_{s_{h}}[\mathcal{R}^{1}]_{g}\sum_{e}(\partial_{\beta\phi\phi'\phi_{e}}\mathcal{L})[\mathcal{R}_{b}^{1}]_{e}[\mathcal{R}^{1}]_{f}\\
	& +W^{-3}W_{s_{h}}[\mathcal{R}^{1}]_{g}\sum_{e}(\partial_{\beta\phi\phi'\phi_{e}}\mathcal{L})[\mathcal{R}^{1}]_{e}[\mathcal{R}_{b}^{1}]_{f}\\
	& +2W^{-3}W_{b}W_{s_{f}}[\mathcal{R}_{s_{h}}^{1}]_{g}(\partial_{\beta\phi\phi'}\mathcal{L})-W^{-2}W_{bs_{f}}[\mathcal{R}_{s_{h}}^{1}]_{g}(\partial_{\beta\phi\phi'}\mathcal{L})\\
	& -W^{-2}W_{s_{f}}[\mathcal{R}_{bs_{h}}^{1}]_{g}(\partial_{\beta\phi\phi'}\mathcal{L})+W^{-3}W_{s_{f}}[\mathcal{R}_{s_{h}}^{1}]_{g}(\partial_{\beta^{2}\phi\phi'}\mathcal{L})\\
	& +W^{-3}W_{s_{f}}[\mathcal{R}_{s_{h}}^{1}]_{g}\sum_{e}(\partial_{\beta\phi\phi'\phi_{e}}\mathcal{L})[\mathcal{R}^{1}]_{e}-W^{-2}[\mathcal{R}_{s_{f}s_{h}}^{1}]_{g}(\partial_{\beta\phi\phi'}\mathcal{L})\\
	& +W^{-1}[\mathcal{R}_{bs_{f}s_{h}}^{1}]_{g}(\partial_{\beta\phi\phi'}\mathcal{L})-W^{-2}[\mathcal{R}_{s_{f}s_{h}}^{1}]_{g}(\partial_{\beta^{2}\phi\phi'}\mathcal{L})\\
	& -W^{-2}[\mathcal{R}_{s_{f}s_{h}}^{1}]_{g}\sum_{e}(\partial_{\beta\phi\phi'\phi_{e}}\mathcal{L})[\mathcal{R}^{1}]_{e}+2W^{-3}W_{b}[\mathcal{R}_{s_{h}}^{1}]_{g}(\partial_{\beta^{2}\phi\phi'}\mathcal{L})[\mathcal{R}^{1}]_{f}\\
	& -W^{-2}[\mathcal{R}_{bs_{h}}^{1}]_{g}(\partial_{\beta^{2}\phi\phi'}\mathcal{L})[\mathcal{R}^{1}]_{f}+W^{-3}[\mathcal{R}_{s_{h}}^{1}]_{g}(\partial_{\beta^{3}\phi\phi'}\mathcal{L})[\mathcal{R}^{1}]_{f}\\
	& +W^{-3}[\mathcal{R}_{s_{h}}^{1}]_{g}\sum_{e}(\partial_{\beta^{2}\phi\phi'\phi_{e}}\mathcal{L})[\mathcal{R}^{1}]_{e}[\mathcal{R}^{1}]_{f}-W^{-2}[\mathcal{R}_{s_{h}}^{1}]_{g}(\partial_{\beta^{2}\phi\phi'}\mathcal{L})[\mathcal{R}_{b}^{1}]_{f}\\
	& +W^{-2}W_{b}[\mathcal{R}_{s_{h}}^{1}]_{g}\sum_{e}(\partial_{\beta\phi\phi'\phi_{e}}\mathcal{L})[\mathcal{H}^{-1}]_{ef}-W^{-1}[\mathcal{R}_{bs_{h}}^{1}]_{g}\sum_{e}(\partial_{\beta\phi\phi'\phi_{e}}\mathcal{L})[\mathcal{H}^{-1}]_{ef}\\
	& +W^{-2}[\mathcal{R}_{s_{h}}^{1}]_{g}\sum_{e}(\partial_{\beta^{2}\phi\phi'\phi_{e}}\mathcal{L})[\mathcal{H}^{-1}]_{ef}\\
	& +W^{-2}[\mathcal{R}_{s_{h}}^{1}]_{g}\sum_{d,e}(\partial_{\beta\phi\phi'\phi_{d}\phi_{e}}\mathcal{L})[\mathcal{H}^{-1}]_{ef}[\mathcal{R}^{1}]_{d}\\
	& +W^{-1}[\mathcal{R}_{s_{h}}^{1}]_{g}\sum_{e}(\partial_{\beta\phi\phi'\phi_{e}}\mathcal{L})[\mathcal{H}^{-1}\mathcal{H}_{b}\mathcal{H}^{-1}]_{ef}\\
	& +2W^{-3}W_{b}[\mathcal{R}_{s_{h}}^{1}]_{g}\sum_{e}(\partial_{\beta\phi\phi'\phi_{e}}\mathcal{L})[\mathcal{R}^{1}]_{e}[\mathcal{R}^{1}]_{f}\\
	& -W^{-2}[\mathcal{R}_{bs_{h}}^{1}]_{g}\sum_{e}(\partial_{\beta\phi\phi'\phi_{e}}\mathcal{L})[\mathcal{R}^{1}]_{e}[\mathcal{R}^{1}]_{f}\\
	& +W^{-3}[\mathcal{R}_{s_{h}}^{1}]_{g}\sum_{e}(\partial_{\beta^{2}\phi\phi'\phi_{e}}\mathcal{L})[\mathcal{R}^{1}]_{e}[\mathcal{R}^{1}]_{f}\\
	& +W^{-3}[\mathcal{R}_{s_{h}}^{1}]_{g}\sum_{d,e}(\partial_{\beta\phi\phi'\phi_{d}\phi_{e}}\mathcal{L})[\mathcal{R}^{1}]_{d}[\mathcal{R}^{1}]_{e}[\mathcal{R}^{1}]_{f}\\
	& -W^{-2}[\mathcal{R}_{s_{h}}^{1}]_{g}\sum_{e}(\partial_{\beta\phi\phi'\phi_{e}}\mathcal{L})[\mathcal{R}_{b}^{1}]_{e}[\mathcal{R}^{1}]_{f}\\
	& -W^{-2}[\mathcal{R}_{s_{h}}^{1}]_{g}\sum_{e}(\partial_{\beta\phi\phi'\phi_{e}}\mathcal{L})[\mathcal{R}^{1}]_{e}[\mathcal{R}_{b}^{1}]_{f}\\
	& -6W^{-4}W_{b}W_{s_{f}}[\mathcal{R}^{1}]_{g}(\partial_{\beta^{2}\phi\phi'}\mathcal{L})[\mathcal{R}^{1}]_{h}+2W^{-3}W_{bs_{f}}[\mathcal{R}^{1}]_{g}(\partial_{\beta^{2}\phi\phi'}\mathcal{L})[\mathcal{R}^{1}]_{h}\\
	& +2W^{-3}W_{s_{f}}[\mathcal{R}_{b}^{1}]_{g}(\partial_{\beta^{2}\phi\phi'}\mathcal{L})[\mathcal{R}^{1}]_{h}-2W^{-4}W_{s_{f}}[\mathcal{R}^{1}]_{g}(\partial_{\beta^{3}\phi\phi'}\mathcal{L})[\mathcal{R}^{1}]_{h}\\
	& -2W^{-4}W_{s_{f}}[\mathcal{R}^{1}]_{g}\sum_{e}(\partial_{\beta^{2}\phi\phi'\phi_{e}}\mathcal{L})[\mathcal{R}^{1}]_{e}[\mathcal{R}^{1}]_{h}\\
	& +2W^{-3}W_{s_{f}}[\mathcal{R}^{1}]_{g}(\partial_{\beta^{2}\phi\phi'}\mathcal{L})[\mathcal{R}_{b}^{1}]_{h}+2W^{-3}W_{b}[\mathcal{R}_{s_{f}}^{1}]_{g}(\partial_{\beta^{2}\phi\phi'}\mathcal{L})[\mathcal{R}^{1}]_{h}\\
	& -W^{-2}[\mathcal{R}_{bs_{f}}^{1}]_{g}(\partial_{\beta^{2}\phi\phi'}\mathcal{L})[\mathcal{R}^{1}]_{h}+W^{-3}[\mathcal{R}_{s_{f}}^{1}]_{g}(\partial_{\beta^{3}\phi\phi'}\mathcal{L})[\mathcal{R}^{1}]_{h}\\
	& +W^{-3}[\mathcal{R}_{s_{f}}^{1}]_{g}\sum_{e}(\partial_{\beta^{2}\phi\phi'\phi_{e}}\mathcal{L})[\mathcal{R}^{1}]_{e}[\mathcal{R}^{1}]_{h}\\
	& -W^{-2}[\mathcal{R}_{s_{f}}^{1}]_{g}(\partial_{\beta^{2}\phi\phi'}\mathcal{L})[\mathcal{R}_{b}^{1}]_{h}-3W^{-4}W_{b}[\mathcal{R}^{1}]_{g}(\partial_{\beta^{3}\phi\phi'}\mathcal{L})[\mathcal{R}^{1}]_{f}[\mathcal{R}^{1}]_{h}\\
	& +W^{-3}[\mathcal{R}_{b}^{1}]_{g}(\partial_{\beta^{3}\phi\phi'}\mathcal{L})[\mathcal{R}^{1}]_{f}[\mathcal{R}^{1}]_{h}-W^{-4}[\mathcal{R}^{1}]_{g}(\partial_{\beta^{4}\phi\phi'}\mathcal{L})[\mathcal{R}^{1}]_{f}[\mathcal{R}^{1}]_{h}\\
	& -W^{-4}[\mathcal{R}^{1}]_{g}\sum_{e}(\partial_{\beta^{3}\phi\phi'\phi_{e}}\mathcal{L})[\mathcal{R}^{1}]_{e}[\mathcal{R}^{1}]_{f}[\mathcal{R}^{1}]_{h}\\
	& +W^{-3}[\mathcal{R}^{1}]_{g}(\partial_{\beta^{3}\phi\phi'}\mathcal{L})[\mathcal{R}_{b}^{1}]_{f}[\mathcal{R}^{1}]_{h}+W^{-3}[\mathcal{R}^{1}]_{g}(\partial_{\beta^{3}\phi\phi'}\mathcal{L})[\mathcal{R}^{1}]_{f}[\mathcal{R}_{b}^{1}]_{h}\\
	& -2W^{-3}W_{b}[\mathcal{R}^{1}]_{g}\sum_{e}(\partial_{\beta^{2}\phi\phi'\phi_{e}}\mathcal{L})[\mathcal{H}^{-1}]_{ef}[\mathcal{R}^{1}]_{h}\\
	& +W^{-2}[\mathcal{R}_{b}^{1}]_{g}\sum_{e}(\partial_{\beta^{2}\phi\phi'\phi_{e}}\mathcal{L})[\mathcal{H}^{-1}]_{ef}[\mathcal{R}^{1}]_{h}\\
	& -W^{-3}[\mathcal{R}^{1}]_{g}\sum_{e}(\partial_{\beta^{3}\phi\phi'\phi_{e}}\mathcal{L})[\mathcal{H}^{-1}]_{ef}[\mathcal{R}^{1}]_{h}\\
	& -W^{-3}[\mathcal{R}^{1}]_{g}\sum_{d,e}(\partial_{\beta^{2}\phi\phi'\phi_{d}\phi_{e}}\mathcal{L})[\mathcal{H}^{-1}]_{ef}[\mathcal{R}^{1}]_{d}[\mathcal{R}^{1}]_{h}\\
	& -W^{-2}[\mathcal{R}^{1}]_{g}\sum_{e}(\partial_{\beta^{2}\phi\phi'\phi_{e}}\mathcal{L})[\mathcal{H}^{-1}\mathcal{H}_{b}\mathcal{H}^{-1}]_{ef}[\mathcal{R}^{1}]_{h}\\
	& +W^{-2}[\mathcal{R}^{1}]_{g}\sum_{e}(\partial_{\beta^{2}\phi\phi'\phi_{e}}\mathcal{L})[\mathcal{H}^{-1}]_{ef}[\mathcal{R}_{b}^{1}]_{h}\\
	& -3W^{-4}W_{b}[\mathcal{R}^{1}]_{g}(\partial_{\beta^{2}\phi\phi'}\mathcal{L})[\mathcal{R}^{1}]_{h}+W^{-3}[\mathcal{R}_{b}^{1}]_{g}(\partial_{\beta^{2}\phi\phi'}\mathcal{L})[\mathcal{R}^{1}]_{h}\\
	& -W^{-4}[\mathcal{R}^{1}]_{g}(\partial_{\beta^{3}\phi\phi'}\mathcal{L})[\mathcal{R}^{1}]_{h}-W^{-4}[\mathcal{R}^{1}]_{g}\sum_{e}(\partial_{\beta^{2}\phi\phi'\phi_{e}}\mathcal{L})[\mathcal{R}^{1}]_{e}[\mathcal{R}^{1}]_{h}\\
	& +W^{-3}[\mathcal{R}^{1}]_{g}(\partial_{\beta^{2}\phi\phi'}\mathcal{L})[\mathcal{R}_{b}^{1}]_{h}\\
	& +2W^{-3}W_{b}[\mathcal{R}^{1}]_{g}\sum_{e}(\partial_{\beta^{2}\phi\phi'\phi_{e}}\mathcal{L})[\mathcal{R}^{1}]_{e}[\mathcal{R}^{1}]_{f}[\mathcal{R}^{1}]_{h}\\
	& -W^{-2}[\mathcal{R}_{b}^{1}]_{g}\sum_{e}(\partial_{\beta^{2}\phi\phi'\phi_{e}}\mathcal{L})[\mathcal{R}^{1}]_{e}[\mathcal{R}^{1}]_{f}[\mathcal{R}^{1}]_{h}\\
	& +W^{-3}[\mathcal{R}^{1}]_{g}\sum_{e}(\partial_{\beta^{3}\phi\phi'\phi_{e}}\mathcal{L})[\mathcal{R}^{1}]_{e}[\mathcal{R}^{1}]_{f}[\mathcal{R}^{1}]_{h}\\
	& +W^{-3}[\mathcal{R}^{1}]_{g}\sum_{d,e}(\partial_{\beta^{2}\phi\phi'\phi_{d}\phi_{e}}\mathcal{L})[\mathcal{R}^{1}]_{d}[\mathcal{R}^{1}]_{e}[\mathcal{R}^{1}]_{f}[\mathcal{R}^{1}]_{h}\\
	& -W^{-2}[\mathcal{R}^{1}]_{g}\sum_{e}(\partial_{\beta^{2}\phi\phi'\phi_{e}}\mathcal{L})[\mathcal{R}_{b}^{1}]_{e}[\mathcal{R}^{1}]_{f}[\mathcal{R}^{1}]_{h}\\
	& -W^{-2}[\mathcal{R}^{1}]_{g}\sum_{e}(\partial_{\beta^{2}\phi\phi'\phi_{e}}\mathcal{L})[\mathcal{R}^{1}]_{e}[\mathcal{R}_{b}^{1}]_{f}[\mathcal{R}^{1}]_{h}\\
	& -W^{-2}[\mathcal{R}^{1}]_{g}\sum_{e}(\partial_{\beta^{2}\phi\phi'\phi_{e}}\mathcal{L})[\mathcal{R}^{1}]_{e}[\mathcal{R}^{1}]_{f}[\mathcal{R}_{b}^{1}]_{h}\\
	& -2W^{-3}W_{b}W_{s_{f}}[\mathcal{R}^{1}]_{g}\mathcal{P}_{(h)}^{1}+W^{-2}W_{bs_{f}}[\mathcal{R}^{1}]_{g}\mathcal{P}_{(h)}^{1}\\
	& +W^{-2}W_{s_{f}}[\mathcal{R}_{b}^{1}]_{g}\mathcal{P}_{(h)}^{1}+W^{-2}W_{s_{f}}[\mathcal{R}^{1}]_{g}\mathcal{P}_{(h),b}^{1}\\
	& +W^{-2}W_{b}[\mathcal{R}_{s_{f}}^{1}]_{g}\mathcal{P}_{(h)}^{1}-W^{-1}[\mathcal{R}_{bs_{f}}^{1}]_{g}\mathcal{P}_{(h)}^{1}-W^{-1}[\mathcal{R}_{s_{f}}^{1}]_{g}\mathcal{P}_{(h),b}^{1}\\
	& +W^{-2}W_{b}[\mathcal{R}^{1}]_{g}\mathcal{P}_{(h),s_{f}}^{1}-W^{-1}[\mathcal{R}_{b}^{1}]_{g}\mathcal{P}_{(h),s_{f}}^{1}-W^{-1}[\mathcal{R}^{1}]_{g}\mathcal{P}_{(h),bs_{f}}^{1}\\
	& -6W^{-4}W_{b}W_{s_{f}}[\mathcal{R}^{1}]_{g}\mathcal{P}^{(1,1)}[\mathcal{R}^{1}]_{h}+2W^{-3}W_{bs_{f}}[\mathcal{R}^{1}]_{g}\mathcal{P}^{(1,1)}[\mathcal{R}^{1}]_{h}\\
	& +2W^{-3}W_{s_{f}}[\mathcal{R}_{b}^{1}]_{g}\mathcal{P}^{(1,1)}[\mathcal{R}^{1}]_{h}+2W^{-3}W_{s_{f}}[\mathcal{R}^{1}]_{g}\mathcal{P}_{b}^{(1,1)}[\mathcal{R}^{1}]_{h}\\
	& +2W^{-3}W_{s_{f}}[\mathcal{R}^{1}]_{g}\mathcal{P}^{(1,1)}[\mathcal{R}_{b}^{1}]_{h}\\
	& +2W^{-3}W_{b}[\mathcal{R}_{s_{f}}^{1}]_{g}\mathcal{P}^{(1,1)}[\mathcal{R}^{1}]_{h}-W^{-2}[\mathcal{R}_{bs_{f}}^{1}]_{g}\mathcal{P}^{(1,1)}[\mathcal{R}^{1}]_{h}\\
	& -W^{-2}[\mathcal{R}_{s_{f}}^{1}]_{g}\mathcal{P}_{b}^{(1,1)}[\mathcal{R}^{1}]_{h}-W^{-2}[\mathcal{R}_{s_{f}}^{1}]_{g}\mathcal{P}^{(1,1)}[\mathcal{R}_{b}^{1}]_{h}\\
	& +2W^{-3}W_{b}[\mathcal{R}^{1}]_{g}\mathcal{P}_{s_{f}}^{(1,1)}[\mathcal{R}^{1}]_{h}-W^{-2}[\mathcal{R}_{b}^{1}]_{g}\mathcal{P}_{s_{f}}^{(1,1)}[\mathcal{R}^{1}]_{h}\\
	& -W^{-2}[\mathcal{R}^{1}]_{g}\mathcal{P}_{bs_{f}}^{(1,1)}[\mathcal{R}^{1}]_{h}-W^{-2}[\mathcal{R}^{1}]_{g}\mathcal{P}_{s_{f}}^{(1,1)}[\mathcal{R}_{b}^{1}]_{h}\\
	& +2W^{-3}W_{b}[\mathcal{R}^{1}]_{g}\mathcal{P}^{(1,1)}[\mathcal{R}_{s_{f}}^{1}]_{h}-W^{-2}[\mathcal{R}_{b}^{1}]_{g}\mathcal{P}^{(1,1)}[\mathcal{R}_{s_{f}}^{1}]_{h}\\
	& -W^{-2}[\mathcal{R}^{1}]_{g}\mathcal{P}_{b}^{(1,1)}[\mathcal{R}_{s_{f}}^{1}]_{h}-W^{-2}[\mathcal{R}^{1}]_{g}\mathcal{P}^{(1,1)}[\mathcal{R}_{bs_{f}}^{1}]_{h}\\
	& +\mathcal{P}_{(g),bs_{f}s_{h}}^{0}\\
	& -6W^{-4}W_{b}W_{s_{f}}W_{s_{h}}\mathcal{P}^{(0,1)}[\mathcal{R}^{1}]_{g}+2W^{-3}W_{bs_{f}}W_{s_{h}}\mathcal{P}^{(0,1)}[\mathcal{R}^{1}]_{g}\\
	& +2W^{-3}W_{s_{f}}W_{bs_{h}}\mathcal{P}^{(0,1)}[\mathcal{R}^{1}]_{g}+2W^{-3}W_{s_{f}}W_{s_{h}}\mathcal{P}_{b}^{(0,1)}[\mathcal{R}^{1}]_{g}\\
	& +2W^{-3}W_{s_{f}}W_{s_{h}}\mathcal{P}^{(0,1)}[\mathcal{R}_{b}^{1}]_{g}\\
	& +2W^{-3}W_{b}W_{s_{f}s_{h}}\mathcal{P}^{(0,1)}[\mathcal{R}^{1}]_{g}-W^{-2}W_{bs_{f}s_{h}}\mathcal{P}^{(0,1)}[\mathcal{R}^{1}]_{g}\\
	& -W^{-2}W_{s_{f}s_{h}}\mathcal{P}_{b}^{(0,1)}[\mathcal{R}^{1}]_{g}-W^{-2}W_{s_{f}s_{h}}\mathcal{P}^{(0,1)}[\mathcal{R}_{b}^{1}]_{g}\\
	& +2W^{-3}W_{b}W_{s_{h}}\mathcal{P}_{s_{f}}^{(0,1)}[\mathcal{R}^{1}]_{g}-W^{-2}W_{bs_{h}}\mathcal{P}_{s_{f}}^{(0,1)}[\mathcal{R}^{1}]_{g}\\
	& -W^{-2}W_{s_{h}}\mathcal{P}_{bs_{f}}^{(0,1)}[\mathcal{R}^{1}]_{g}-W^{-2}W_{s_{h}}\mathcal{P}_{s_{f}}^{(0,1)}[\mathcal{R}_{b}^{1}]_{g}\\
	& +2W^{-3}W_{b}W_{s_{h}}\mathcal{P}^{(0,1)}[\mathcal{R}_{s_{f}}^{1}]_{g}-W^{-2}W_{bs_{h}}\mathcal{P}^{(0,1)}[\mathcal{R}_{s_{f}}^{1}]_{g}\\
	& -W^{-2}W_{s_{h}}\mathcal{P}_{b}^{(0,1)}[\mathcal{R}_{s_{f}}^{1}]_{g}-W^{-2}W_{s_{h}}\mathcal{P}^{(0,1)}[\mathcal{R}_{bs_{f}}^{1}]_{g}\\
	& +2W^{-3}W_{b}W_{s_{f}}\mathcal{P}_{s_{h}}^{(0,1)}[\mathcal{R}^{1}]_{g}-W^{-2}W_{bs_{f}}\mathcal{P}_{s_{h}}^{(0,1)}[\mathcal{R}^{1}]_{g}\\
	& -W^{-2}W_{s_{f}}\mathcal{P}_{bs_{h}}^{(0,1)}[\mathcal{R}^{1}]_{g}-W^{-2}W_{s_{f}}\mathcal{P}_{s_{h}}^{(0,1)}[\mathcal{R}_{b}^{1}]_{g}\\
	& -W^{-2}W_{b}\mathcal{P}_{s_{f}s_{h}}^{(0,1)}[\mathcal{R}^{1}]_{g}+W^{-1}\mathcal{P}_{bs_{f}s_{h}}^{(0,1)}[\mathcal{R}^{1}]_{g}+W^{-1}\mathcal{P}_{s_{f}s_{h}}^{(0,1)}[\mathcal{R}_{b}^{1}]_{g}\\
	& -W^{-2}W_{b}\mathcal{P}_{s_{h}}^{(0,1)}[\mathcal{R}_{s_{f}}^{1}]_{g}+W^{-1}\mathcal{P}_{bs_{h}}^{(0,1)}[\mathcal{R}_{s_{f}}^{1}]_{g}+W^{-1}\mathcal{P}_{s_{h}}^{(0,1)}[\mathcal{R}_{bs_{f}}^{1}]_{g}\\
	& +2W^{-3}W_{b}W_{s_{f}}\mathcal{P}^{(0,1)}[\mathcal{R}_{s_{h}}^{1}]_{g}-W^{-2}W_{bs_{f}}\mathcal{P}^{(0,1)}[\mathcal{R}_{s_{h}}^{1}]_{g}\\
	& -W^{-2}W_{s_{f}}\mathcal{P}_{b}^{(0,1)}[\mathcal{R}_{s_{h}}^{1}]_{g}-W^{-2}W_{s_{f}}\mathcal{P}^{(0,1)}[\mathcal{R}_{bs_{h}}^{1}]_{g}\\
	& -W^{-2}W_{b}\mathcal{P}_{s_{f}}^{(0,1)}[\mathcal{R}_{s_{h}}^{1}]_{g}+W^{-1}\mathcal{P}_{bs_{f}}^{(0,1)}[\mathcal{R}_{s_{h}}^{1}]_{g}+W^{-1}\mathcal{P}_{s_{f}}^{(0,1)}[\mathcal{R}_{bs_{h}}^{1}]_{g}\\
	& -W^{-2}W_{b}\mathcal{P}^{(0,1)}[\mathcal{R}_{s_{f}s_{h}}^{1}]_{g}+W^{-1}\mathcal{P}_{b}^{(0,1)}[\mathcal{R}_{s_{f}s_{h}}^{1}]_{g}+W^{-1}\mathcal{P}^{(0,1)}[\mathcal{R}_{bs_{f}s_{h}}^{1}]_{g}
\end{align*}

\subsection{Expressions for $\mathcal{P}$ terms}
First derivatives:

\begin{align*}
	\mathcal{P}_{b}^{(r,s)} & =-W^{-1}\sum_{f}(\partial_{\beta^{r+1}\phi\phi'\phi_{f}}\mathcal{L})[\mathcal{R}^{s}]_{f}\\
	& -W^{-1}\sum_{e,f}(\partial_{\beta^{r}\phi\phi'\phi_{e}\phi_{f}}\mathcal{L})[\mathcal{R}^{s}]_{f}[\mathcal{R}^{1}]_{e}\\
	& +\sum_{f}(\partial_{\beta^{r}\phi\phi'\phi_{f}}\mathcal{L})[\mathcal{R}_{b}^{s}]_{f}\\
	\mathcal{P}_{(g),b}^{r} & =-W^{-1}\sum_{f}(\partial_{\beta^{r+1}\phi\phi'\phi_{f}}\mathcal{L})\big[\mathcal{H}^{-1}\big]_{fg}\\
	& -W^{-1}\sum_{e,f}(\partial_{\beta^{r}\phi\phi'\phi_{e}\phi_{f}}\mathcal{L})[\mathcal{R}^{1}]_{e}\big[\mathcal{H}^{-1}\big]_{fg}\\
	& -\sum_{f}(\partial_{\beta^{r}\phi\phi'\phi_{f}}\mathcal{L})\big[\mathcal{H}^{-1}\mathcal{H}_{b}\mathcal{H}^{-1}\big]_{fg}
\end{align*}

\begin{align*}
	\mathcal{P}_{s_{h}}^{(a,r)} & =-W^{-1}[\mathcal{R}^{1}]_{h}\sum_{f}(\partial_{\beta^{a+1}\phi\phi'\phi_{f}}\mathcal{L})[\mathcal{R}^{1}]_{f}\\
	& -\sum_{e,f}(\partial_{\beta^{a}\phi\phi'\phi_{e}\phi_{f}}\mathcal{L})[\mathcal{R}^{1}]_{f}[\mathcal{H}^{-1}]_{eh}\\
	& -W^{-1}\sum_{e,f}(\partial_{\beta^{a}\phi\phi'\phi_{e}\phi_{f}}\mathcal{L})[\mathcal{R}^{1}]_{f}[\mathcal{R}^{1}]_{e}[\mathcal{R}^{1}]_{h}\\
	& +\sum_{f}(\partial_{\beta^{a}\phi\phi'\phi_{f}}\mathcal{L})[\mathcal{R}_{s_{h}}^{1}]_{f}
\end{align*}

\begin{align*}
	\mathcal{P}_{(g),s_{h}}^{r} & =-W^{-1}[\mathcal{R}^{1}]_{h}\mathcal{P}_{(g)}^{r+1}\\
	& -\sum_{e,f}(\partial_{\beta^{r}\phi\phi'\phi_{e}\phi_{f}}\mathcal{L})\big[\mathcal{H}^{-1}\big]_{fg}\big[\mathcal{H}^{-1}\big]_{eh}\\
	& -W^{-1}\sum_{e,f}(\partial_{\beta^{r}\phi\phi'\phi_{e}\phi_{f}}\mathcal{L})\big[\mathcal{H}^{-1}\big]_{fg}[\mathcal{R}^{1}]_{e}[\mathcal{R}^{1}]_{h}\\
	& -\sum_{f}(\partial_{\beta^{r}\phi\phi'\phi_{f}}\mathcal{L})\big[\mathcal{H}^{-1}\mathcal{H}_{s_{h}}\mathcal{H}^{-1}\big]_{fg}
\end{align*}

Second derivatives:

\begin{align*}
	\mathcal{P}_{(g),bs_{h}}^{r} & =W^{-2}W_{b}[\mathcal{R}^{1}]_{h}\mathcal{P}_{(g)}^{r+1}-W^{-1}[\mathcal{R}_{b}^{1}]_{h}\mathcal{P}_{(g)}^{r+1}-W^{-1}[\mathcal{R}^{1}]_{h}\mathcal{P}_{(g),b}^{r+1}\\
	& +W^{-1}\sum_{e,f}(\partial_{\beta^{r+1}\phi\phi'\phi_{e}\phi_{f}}\mathcal{L})\big[\mathcal{H}^{-1}\big]_{fg}\big[\mathcal{H}^{-1}\big]_{eh}\\
	& +W^{-1}\sum_{d,e,f}(\partial_{\beta^{r}\phi\phi'\phi_{d}\phi_{e}\phi_{f}}\mathcal{L})[\mathcal{R}^{1}]_{d}\big[\mathcal{H}^{-1}\big]_{fg}\big[\mathcal{H}^{-1}\big]_{eh}\\
	& +\sum_{e,f}(\partial_{\beta^{r}\phi\phi'\phi_{e}\phi_{f}}\mathcal{L})\big[\mathcal{H}^{-1}\mathcal{H}_{b}\mathcal{H}^{-1}\big]_{fg}\big[\mathcal{H}^{-1}\big]_{eh}\\
	& +\sum_{e,f}(\partial_{\beta^{r}\phi\phi'\phi_{e}\phi_{f}}\mathcal{L})\big[\mathcal{H}^{-1}\big]_{fg}\big[\mathcal{H}^{-1}\mathcal{H}_{b}\mathcal{H}^{-1}\big]_{eh}\\
	& +W^{-2}W_{b}\sum_{e,f}(\partial_{\beta^{r}\phi\phi'\phi_{e}\phi_{f}}\mathcal{L})\big[\mathcal{H}^{-1}\big]_{fg}[\mathcal{R}^{1}]_{e}[\mathcal{R}^{1}]_{h}\\
	& +W^{-2}\sum_{e,f}(\partial_{\beta^{r+1}\phi\phi'\phi_{e}\phi_{f}}\mathcal{L})\big[\mathcal{H}^{-1}\big]_{fg}[\mathcal{R}^{1}]_{e}[\mathcal{R}^{1}]_{h}\\
	& +W^{-2}\sum_{d,e,f}(\partial_{\beta^{r}\phi\phi'\phi_{d}\phi_{e}\phi_{f}}\mathcal{L})\big[\mathcal{H}^{-1}\big]_{fg}[\mathcal{R}^{1}]_{d}[\mathcal{R}^{1}]_{e}[\mathcal{R}^{1}]_{h}\\
	& +W^{-1}\sum_{e,f}(\partial_{\beta^{r}\phi\phi'\phi_{e}\phi_{f}}\mathcal{L})\big[\mathcal{H}^{-1}\mathcal{H}_{b}\mathcal{H}^{-1}\big]_{fg}[\mathcal{R}^{1}]_{e}[\mathcal{R}^{1}]_{h}\\
	& -W^{-1}\sum_{e,f}(\partial_{\beta^{r}\phi\phi'\phi_{e}\phi_{f}}\mathcal{L})\big[\mathcal{H}^{-1}\big]_{fg}[\mathcal{R}_{b}^{1}]_{e}[\mathcal{R}^{1}]_{h}\\
	& -W^{-1}\sum_{e,f}(\partial_{\beta^{r}\phi\phi'\phi_{e}\phi_{f}}\mathcal{L})\big[\mathcal{H}^{-1}\big]_{fg}[\mathcal{R}^{1}]_{e}[\mathcal{R}_{b}^{1}]_{h}\\
	& +W^{-1}\sum_{f}(\partial_{\beta^{r+1}\phi\phi'\phi_{f}}\mathcal{L})\big[\mathcal{H}^{-1}\mathcal{H}_{s_{h}}\mathcal{H}^{-1}\big]_{fg}\\
	& +W^{-1}\sum_{e,f}(\partial_{\beta^{r}\phi\phi'\phi_{e}\phi_{f}}\mathcal{L})[\mathcal{R}^{1}]_{e}\big[\mathcal{H}^{-1}\mathcal{H}_{s_{h}}\mathcal{H}^{-1}\big]_{fg}\\
	& +\sum_{f}(\partial_{\beta^{r}\phi\phi'\phi_{f}}\mathcal{L})\big[\mathcal{H}^{-1}\mathcal{H}_{b}\mathcal{H}^{-1}\mathcal{H}_{s_{h}}\mathcal{H}^{-1}\big]_{fg}\\
	& +\sum_{f}(\partial_{\beta^{r}\phi\phi'\phi_{f}}\mathcal{L})\big[\mathcal{H}^{-1}\mathcal{H}_{s_{h}}\mathcal{H}^{-1}\mathcal{H}_{b}\mathcal{H}^{-1}\big]_{fg}\\
	& -\sum_{f}(\partial_{\beta^{r}\phi\phi'\phi_{f}}\mathcal{L})\big[\mathcal{H}^{-1}\mathcal{H}_{bs_{h}}\mathcal{H}^{-1}\big]_{fg}
\end{align*}
\begin{align*}
	\mathcal{P}_{bs_{h}}^{(a,r)} & =W^{-2}W_{b}[\mathcal{R}^{1}]_{h}\sum_{f}(\partial_{\beta^{a+1}\phi\phi'\phi_{f}}\mathcal{L})[\mathcal{R}^{1}]_{f}\\
	& -W^{-1}[\mathcal{R}_{b}^{1}]_{h}\sum_{f}(\partial_{\beta^{a+1}\phi\phi'\phi_{f}}\mathcal{L})[\mathcal{R}^{1}]_{f}\\
	& -W^{-1}[\mathcal{R}^{1}]_{h}\sum_{f}(\partial_{\beta^{a+1}\phi\phi'\phi_{f}}\mathcal{L})[\mathcal{R}_{b}^{1}]_{f}\\
	& +W^{-2}[\mathcal{R}^{1}]_{h}\sum_{f}(\partial_{\beta^{a+2}\phi\phi'\phi_{f}}\mathcal{L})[\mathcal{R}^{1}]_{f}\\
	& +W^{-2}[\mathcal{R}^{1}]_{h}\sum_{e,f}(\partial_{\beta^{a+1}\phi\phi'\phi_{e}\phi_{f}}\mathcal{L})[\mathcal{R}^{1}]_{e}[\mathcal{R}^{1}]_{f}\\
	& +W^{-1}\sum_{e,f}(\partial_{\beta^{a+1}\phi\phi'\phi_{i}\phi_{f}}\mathcal{L})[\mathcal{R}^{1}]_{f}[\mathcal{H}^{-1}]_{eh}\\
	& +W^{-1}\sum_{d,e,f}(\partial_{\beta^{a}\phi\phi'\phi_{d}\phi_{e}\phi_{f}}\mathcal{L})[\mathcal{R}^{1}]_{d}[\mathcal{R}^{1}]_{f}[\mathcal{H}^{-1}]_{eh}\\
	& -\sum_{e,f}(\partial_{\beta^{a}\phi\phi'\phi_{i}\phi_{f}}\mathcal{L})[\mathcal{R}_{b}^{1}]_{f}[\mathcal{H}^{-1}]_{eh}\\
	& +\sum_{e,f}(\partial_{\beta^{a}\phi\phi'\phi_{i}\phi_{f}}\mathcal{L})[\mathcal{R}^{1}]_{f}[\mathcal{H}^{-1}\mathcal{H}_{b}\mathcal{H}^{-1}]_{eh}\\
	& +W^{-2}W_{b}\sum_{e,f}(\partial_{\beta^{a}\phi\phi'\phi_{e}\phi_{f}}\mathcal{L})[\mathcal{R}^{1}]_{f}[\mathcal{R}^{1}]_{e}[\mathcal{R}^{1}]_{h}\\
	& +W^{-2}\sum_{e,f}(\partial_{\beta^{a+1}\phi\phi'\phi_{e}\phi_{f}}\mathcal{L})[\mathcal{R}^{1}]_{f}[\mathcal{R}^{1}]_{e}[\mathcal{R}^{1}]_{h}\\
	& +W^{-2}\sum_{d,e,f}(\partial_{\beta^{a}\phi\phi'\phi_{d}\phi_{e}\phi_{f}}\mathcal{L})[\mathcal{R}^{1}]_{d}[\mathcal{R}^{1}]_{f}[\mathcal{R}^{1}]_{e}[\mathcal{R}^{1}]_{h}\\
	& -W^{-1}\sum_{e,f}(\partial_{\beta^{a}\phi\phi'\phi_{e}\phi_{f}}\mathcal{L})[\mathcal{R}_{b}^{1}]_{f}[\mathcal{R}^{1}]_{e}[\mathcal{R}^{1}]_{h}\\
	& -W^{-1}\sum_{e,f}(\partial_{\beta^{a}\phi\phi'\phi_{e}\phi_{f}}\mathcal{L})[\mathcal{R}^{1}]_{f}[\mathcal{R}_{b}^{1}]_{e}[\mathcal{R}^{1}]_{h}\\
	& -W^{-1}\sum_{e,f}(\partial_{\beta^{a}\phi\phi'\phi_{e}\phi_{f}}\mathcal{L})[\mathcal{R}^{1}]_{f}[\mathcal{R}^{1}]_{e}[\mathcal{R}_{b}^{1}]_{h}\\
	& -W^{-1}\sum_{f}(\partial_{\beta^{a+1}\phi\phi'\phi_{f}}\mathcal{L})[\mathcal{R}_{s_{h}}^{1}]_{f}\\
	& -W^{-1}\sum_{e,f}(\partial_{\beta^{a}\phi\phi'\phi_{e}\phi_{f}}\mathcal{L})[\mathcal{R}^{1}]_{e}[\mathcal{R}_{s_{h}}^{1}]_{f}\\
	& +\sum_{f}(\partial_{\beta^{a}\phi\phi'\phi_{f}}\mathcal{L})[\mathcal{R}_{bs_{h}}^{1}]_{f}
\end{align*}

\begin{align*}
	\mathcal{P}_{s_{g}s_{h}}^{(a,r)} & =W^{-2}W_{s_{g}}[\mathcal{R}^{1}]_{h}\sum_{f}(\partial_{\beta^{a+1}\phi\phi'\phi_{f}}\mathcal{L})[\mathcal{R}^{1}]_{f}\\
	& -W^{-1}[\mathcal{R}_{s_{g}}^{1}]_{h}\sum_{f}(\partial_{\beta^{a+1}\phi\phi'\phi_{f}}\mathcal{L})[\mathcal{R}^{1}]_{f}\\
	& +W^{-2}[\mathcal{R}^{1}]_{h}\sum_{f}(\partial_{\beta^{a+2}\phi\phi'\phi_{f}}\mathcal{L})[\mathcal{R}^{1}]_{f}[\mathcal{R}^{1}]_{g}\\
	& +W^{-1}[\mathcal{R}^{1}]_{h}\sum_{e,f}(\partial_{\beta^{a+1}\phi\phi'\phi_{e}\phi_{f}}\mathcal{L})[\mathcal{R}^{1}]_{f}[\mathcal{H}^{-1}]_{eg}\\
	& +W^{-2}[\mathcal{R}^{1}]_{h}\sum_{e,f}(\partial_{\beta^{a+1}\phi\phi'\phi_{e}\phi_{f}}\mathcal{L})[\mathcal{R}^{1}]_{e}[\mathcal{R}^{1}]_{f}[\mathcal{R}^{1}]_{g}\\
	& -W^{-1}[\mathcal{R}^{1}]_{h}\sum_{f}(\partial_{\beta^{a+1}\phi\phi'\phi_{f}}\mathcal{L})[\mathcal{R}_{s_{g}}^{1}]_{f}\\
	& +W^{-1}\sum_{e,f}(\partial_{\beta^{a+1}\phi\phi'\phi_{i}\phi_{f}}\mathcal{L})[\mathcal{R}^{1}]_{f}[\mathcal{H}^{-1}]_{eh}[\mathcal{R}^{1}]_{g}\\
	& +\sum_{d,e,f}(\partial_{\beta^{a}\phi\phi'\phi_{d}\phi_{e}\phi_{f}}\mathcal{L})[\mathcal{R}^{1}]_{f}[\mathcal{H}^{-1}]_{eh}[\mathcal{H}^{-1}]_{dg}\\
	& +W^{-1}\sum_{d,e,f}(\partial_{\beta^{a}\phi\phi'\phi_{d}\phi_{e}\phi_{f}}\mathcal{L})[\mathcal{R}^{1}]_{f}[\mathcal{H}^{-1}]_{eh}[\mathcal{R}^{1}]_{d}[\mathcal{R}^{1}]_{g}\\
	& -\sum_{e,f}(\partial_{\beta^{a}\phi\phi'\phi_{i}\phi_{f}}\mathcal{L})[\mathcal{R}_{s_{g}}^{1}]_{f}[\mathcal{H}^{-1}]_{eh}\\
	& +\sum_{e,f}(\partial_{\beta^{a}\phi\phi'\phi_{i}\phi_{f}}\mathcal{L})[\mathcal{R}^{1}]_{f}[\mathcal{H}^{-1}\mathcal{H}_{s_{g}}\mathcal{H}^{-1}]_{eh}\\
	& +W^{-2}W_{s_{g}}\sum_{e,f}(\partial_{\beta^{a}\phi\phi'\phi_{e}\phi_{f}}\mathcal{L})[\mathcal{R}^{1}]_{f}[\mathcal{R}^{1}]_{e}[\mathcal{R}^{1}]_{h}\\
	& +W^{-2}\sum_{e,f}(\partial_{\beta^{a+1}\phi\phi'\phi_{e}\phi_{f}}\mathcal{L})[\mathcal{R}^{1}]_{f}[\mathcal{R}^{1}]_{e}[\mathcal{R}^{1}]_{g}[\mathcal{R}^{1}]_{h}\\
	& +W^{-1}\sum_{d,e,f}(\partial_{\beta^{a}\phi\phi'\phi_{d}\phi_{e}\phi_{f}}\mathcal{L})[\mathcal{R}^{1}]_{f}[\mathcal{R}^{1}]_{e}[\mathcal{R}^{1}]_{h}[\mathcal{H}^{-1}]_{dg}\\
	& +W^{-2}\sum_{d,e,f}(\partial_{\beta^{a}\phi\phi'\phi_{d}\phi_{e}\phi_{f}}\mathcal{L})[\mathcal{R}^{1}]_{f}[\mathcal{R}^{1}]_{d}[\mathcal{R}^{1}]_{e}[\mathcal{R}^{1}]_{g}[\mathcal{R}^{1}]_{h}\\
	& -W^{-1}\sum_{e,f}(\partial_{\beta^{a}\phi\phi'\phi_{e}\phi_{f}}\mathcal{L})[\mathcal{R}_{s_{g}}^{1}]_{f}[\mathcal{R}^{1}]_{e}[\mathcal{R}^{1}]_{h}\\
	& -W^{-1}\sum_{e,f}(\partial_{\beta^{a}\phi\phi'\phi_{e}\phi_{f}}\mathcal{L})[\mathcal{R}^{1}]_{f}[\mathcal{R}_{s_{g}}^{1}]_{e}[\mathcal{R}^{1}]_{h}\\
	& -W^{-1}\sum_{e,f}(\partial_{\beta^{a}\phi\phi'\phi_{e}\phi_{f}}\mathcal{L})[\mathcal{R}^{1}]_{f}[\mathcal{R}^{1}]_{e}[\mathcal{R}_{s_{g}}^{1}]_{h}\\
	& -W^{-1}\sum_{f}(\partial_{\beta^{a+1}\phi\phi'\phi_{f}}\mathcal{L})[\mathcal{R}_{s_{h}}^{1}]_{f}[\mathcal{R}^{1}]_{g}\\
	& -\sum_{e,f}(\partial_{\beta^{a}\phi\phi'\phi_{e}\phi_{f}}\mathcal{L})[\mathcal{R}_{s_{h}}^{1}]_{f}[\mathcal{H}^{-1}]_{eg}\\
	& W^{-1}\sum_{e,f}(\partial_{\beta^{a}\phi\phi'\phi_{e}\phi_{f}}\mathcal{L})[\mathcal{R}_{s_{h}}^{1}]_{f}[\mathcal{R}^{1}]_{e}[\mathcal{R}^{1}]_{g}\\
	& +\sum_{f}(\partial_{\beta^{a}\phi\phi'\phi_{f}}\mathcal{L})[\mathcal{R}_{s_{g}s_{h}}^{1}]_{f}
\end{align*}

\begin{align*}
	\mathcal{P}_{(g),s_{k}s_{h}}^{r} & =W^{-2}W_{s_{k}}[\mathcal{R}^{1}]_{h}\mathcal{P}_{(g)}^{r+1}-W^{-1}[\mathcal{R}_{s_{k}}^{1}]_{h}\mathcal{P}_{(g)}^{r+1}-W^{-1}[\mathcal{R}^{1}]_{h}\mathcal{P}_{(g),s_{k}}^{r+1}\\
	& +W^{-1}\sum_{e,f}(\partial_{\beta^{r+1}\phi\phi'\phi_{e}\phi_{f}}\mathcal{L})\big[\mathcal{H}^{-1}\big]_{fg}\big[\mathcal{H}^{-1}\big]_{eh}[\mathcal{R}^{1}]_{k}\\
	& +\sum_{d,e,f}(\partial_{\beta^{r}\phi\phi'\phi_{d}\phi_{e}\phi_{f}}\mathcal{L})\big[\mathcal{H}^{-1}\big]_{fg}\big[\mathcal{H}^{-1}\big]_{eh}\big[\mathcal{H}^{-1}\big]_{dk}\\
	& +W^{-1}\sum_{d,e,f}(\partial_{\beta^{r}\phi\phi'\phi_{d}\phi_{e}\phi_{f}}\mathcal{L})\big[\mathcal{H}^{-1}\big]_{fg}\big[\mathcal{H}^{-1}\big]_{eh}[\mathcal{R}^{1}]_{d}[\mathcal{R}^{1}]_{k}\\
	& +\sum_{e,f}(\partial_{\beta^{r}\phi\phi'\phi_{e}\phi_{f}}\mathcal{L})\big[\mathcal{H}^{-1}\mathcal{H}_{s_{k}}\mathcal{H}^{-1}\big]_{fg}\big[\mathcal{H}^{-1}\big]_{eh}\\
	& +\sum_{e,f}(\partial_{\beta^{r}\phi\phi'\phi_{e}\phi_{f}}\mathcal{L})\big[\mathcal{H}^{-1}\big]_{fg}\big[\mathcal{H}^{-1}\mathcal{H}_{s_{k}}\mathcal{H}^{-1}\big]_{eh}\\
	& +W^{-2}W_{s_{k}}\sum_{e,f}(\partial_{\beta^{r}\phi\phi'\phi_{e}\phi_{f}}\mathcal{L})\big[\mathcal{H}^{-1}\big]_{fg}[\mathcal{R}^{1}]_{e}[\mathcal{R}^{1}]_{h}\\
	& +W^{-2}\sum_{e,f}(\partial_{\beta^{r+1}\phi\phi'\phi_{e}\phi_{f}}\mathcal{L})\big[\mathcal{H}^{-1}\big]_{fg}[\mathcal{R}^{1}]_{e}[\mathcal{R}^{1}]_{h}[\mathcal{R}^{1}]_{k}\\
	& +W^{-1}\sum_{d,e,f}(\partial_{\beta^{r}\phi\phi'\phi_{d}\phi_{e}\phi_{f}}\mathcal{L})\big[\mathcal{H}^{-1}\big]_{fg}\big[\mathcal{H}^{-1}\big]_{dk}[\mathcal{R}^{1}]_{e}[\mathcal{R}^{1}]_{h}\\
	& +W^{-2}\sum_{d,e,f}(\partial_{\beta^{r}\phi\phi'\phi_{d}\phi_{e}\phi_{f}}\mathcal{L})\big[\mathcal{H}^{-1}\big]_{fg}[\mathcal{R}^{1}]_{e}[\mathcal{R}^{1}]_{h}[\mathcal{R}^{1}]_{d}[\mathcal{R}^{1}]_{k}\\
	& +W^{-1}\sum_{e,f}(\partial_{\beta^{r}\phi\phi'\phi_{e}\phi_{f}}\mathcal{L})\big[\mathcal{H}^{-1}\mathcal{H}_{s_{k}}\mathcal{H}^{-1}\big]_{fg}[\mathcal{R}^{1}]_{e}[\mathcal{R}^{1}]_{h}\\
	& -W^{-1}\sum_{e,f}(\partial_{\beta^{r}\phi\phi'\phi_{e}\phi_{f}}\mathcal{L})\big[\mathcal{H}^{-1}\big]_{fg}[\mathcal{R}_{s_{k}}^{1}]_{e}[\mathcal{R}^{1}]_{h}\\
	& -W^{-1}\sum_{e,f}(\partial_{\beta^{r}\phi\phi'\phi_{e}\phi_{f}}\mathcal{L})\big[\mathcal{H}^{-1}\big]_{fg}[\mathcal{R}^{1}]_{e}[\mathcal{R}_{s_{k}}^{1}]_{h}\\
	& +W^{-1}\sum_{f}(\partial_{\beta^{r+1}\phi\phi'\phi_{f}}\mathcal{L})\big[\mathcal{H}^{-1}\mathcal{H}_{s_{h}}\mathcal{H}^{-1}\big]_{fg}[\mathcal{R}^{1}]_{k}\\
	& +\sum_{e,f}(\partial_{\beta^{r}\phi\phi'\phi_{e}\phi_{f}}\mathcal{L})\big[\mathcal{H}^{-1}\mathcal{H}_{s_{h}}\mathcal{H}^{-1}\big]_{fg}\big[\mathcal{H}^{-1}\big]_{ek}\\
	& +W^{-1}\sum_{e,f}(\partial_{\beta^{r}\phi\phi'\phi_{e}\phi_{f}}\mathcal{L})\big[\mathcal{H}^{-1}\mathcal{H}_{s_{h}}\mathcal{H}^{-1}\big]_{fg}[\mathcal{R}^{1}]_{e}[\mathcal{R}^{1}]_{k}\\
	& +\sum_{f}(\partial_{\beta^{r}\phi\phi'\phi_{f}}\mathcal{L})\big[\mathcal{H}^{-1}\mathcal{H}_{s_{k}}\mathcal{H}^{-1}\mathcal{H}_{s_{h}}\mathcal{H}^{-1}\big]_{fg}\\
	& -\sum_{f}(\partial_{\beta^{r}\phi\phi'\phi_{f}}\mathcal{L})\big[\mathcal{H}^{-1}\mathcal{H}_{s_{k}s_{h}}\mathcal{H}^{-1}\big]_{fg}\\
	& +\sum_{f}(\partial_{\beta^{r}\phi\phi'\phi_{f}}\mathcal{L})\big[\mathcal{H}^{-1}\mathcal{H}_{s_{h}}\mathcal{H}^{-1}\mathcal{H}_{s_{k}}\mathcal{H}^{-1}\big]_{fg}
\end{align*}

Third derivatives:

\begin{align*}
	\mathcal{P}_{bs_{g}s_{h}}^{(a,r)} & =-2W^{-3}W_{b}W_{s_{g}}[\mathcal{R}^{1}]_{h}\sum_{f}(\partial_{\beta^{a+1}\phi\phi'\phi_{f}}\mathcal{L})[\mathcal{R}^{1}]_{f}\\
	& +W^{-2}W_{bs_{g}}[\mathcal{R}^{1}]_{h}\sum_{f}(\partial_{\beta^{a+1}\phi\phi'\phi_{f}}\mathcal{L})[\mathcal{R}^{1}]_{f}\\
	& +W^{-2}W_{s_{g}}[\mathcal{R}_{b}^{1}]_{h}\sum_{f}(\partial_{\beta^{a+1}\phi\phi'\phi_{f}}\mathcal{L})[\mathcal{R}^{1}]_{f}\\
	& -W^{-3}W_{s_{g}}[\mathcal{R}^{1}]_{h}\sum_{f}(\partial_{\beta^{a+2}\phi\phi'\phi_{f}}\mathcal{L})[\mathcal{R}^{1}]_{f}\\
	& -W^{-3}W_{s_{g}}[\mathcal{R}^{1}]_{h}\sum_{e,f}(\partial_{\beta^{a+1}\phi\phi'\phi_{e}\phi_{f}}\mathcal{L})[\mathcal{R}^{1}]_{e}[\mathcal{R}^{1}]_{f}\\
	& +W^{-2}W_{s_{g}}[\mathcal{R}^{1}]_{h}\sum_{f}(\partial_{\beta^{a+1}\phi\phi'\phi_{f}}\mathcal{L})[\mathcal{R}_{b}^{1}]_{f}\\
	& +W^{-2}W_{b}[\mathcal{R}_{s_{g}}^{1}]_{h}\sum_{f}(\partial_{\beta^{a+1}\phi\phi'\phi_{f}}\mathcal{L})[\mathcal{R}^{1}]_{f}\\
	& -W^{-1}[\mathcal{R}_{bs_{g}}^{1}]_{h}\sum_{f}(\partial_{\beta^{a+1}\phi\phi'\phi_{f}}\mathcal{L})[\mathcal{R}^{1}]_{f}\\
	& +W^{-2}[\mathcal{R}_{s_{g}}^{1}]_{h}\sum_{f}(\partial_{\beta^{a+2}\phi\phi'\phi_{f}}\mathcal{L})[\mathcal{R}^{1}]_{f}\\
	& +W^{-2}[\mathcal{R}_{s_{g}}^{1}]_{h}\sum_{e,f}(\partial_{\beta^{a+1}\phi\phi'\phi_{e}\phi_{f}}\mathcal{L})[\mathcal{R}^{1}]_{f}\\
	& -W^{-1}[\mathcal{R}_{s_{g}}^{1}]_{h}\sum_{f}(\partial_{\beta^{a+1}\phi\phi'\phi_{f}}\mathcal{L})[\mathcal{R}_{b}^{1}]_{f}\\
	& -2W^{-3}W_{b}[\mathcal{R}^{1}]_{h}\sum_{f}(\partial_{\beta^{a+2}\phi\phi'\phi_{f}}\mathcal{L})[\mathcal{R}^{1}]_{f}[\mathcal{R}^{1}]_{g}\\
	& +W^{-2}[\mathcal{R}_{b}^{1}]_{h}\sum_{f}(\partial_{\beta^{a+2}\phi\phi'\phi_{f}}\mathcal{L})[\mathcal{R}^{1}]_{f}[\mathcal{R}^{1}]_{g}\\
	& -W^{-3}[\mathcal{R}^{1}]_{h}\sum_{f}(\partial_{\beta^{a+3}\phi\phi'\phi_{f}}\mathcal{L})[\mathcal{R}^{1}]_{f}[\mathcal{R}^{1}]_{g}\\
	& -W^{-3}[\mathcal{R}^{1}]_{h}\sum_{e,f}(\partial_{\beta^{a+2}\phi\phi'\phi_{e}\phi_{f}}\mathcal{L})[\mathcal{R}^{1}]_{f}[\mathcal{R}^{1}]_{g}\\
	& +W^{-2}[\mathcal{R}^{1}]_{h}\sum_{f}(\partial_{\beta^{a+2}\phi\phi'\phi_{f}}\mathcal{L})[\mathcal{R}_{b}^{1}]_{f}[\mathcal{R}^{1}]_{g}\\
	& +W^{-2}[\mathcal{R}^{1}]_{h}\sum_{f}(\partial_{\beta^{a+2}\phi\phi'\phi_{f}}\mathcal{L})[\mathcal{R}^{1}]_{f}[\mathcal{R}_{b}^{1}]_{g}\\
	& -W^{-2}W_{b}[\mathcal{R}^{1}]_{h}\sum_{e,f}(\partial_{\beta^{a+1}\phi\phi'\phi_{e}\phi_{f}}\mathcal{L})[\mathcal{R}^{1}]_{f}[\mathcal{H}^{-1}]_{eg}\\
	& +W^{-1}[\mathcal{R}_{b}^{1}]_{h}\sum_{e,f}(\partial_{\beta^{a+1}\phi\phi'\phi_{e}\phi_{f}}\mathcal{L})[\mathcal{R}^{1}]_{f}[\mathcal{H}^{-1}]_{eg}\\
	& -W^{-2}[\mathcal{R}^{1}]_{h}\sum_{e,f}(\partial_{\beta^{a+2}\phi\phi'\phi_{e}\phi_{f}}\mathcal{L})[\mathcal{R}^{1}]_{f}[\mathcal{H}^{-1}]_{eg}\\
	& -W^{-2}[\mathcal{R}^{1}]_{h}\sum_{d,e,f}(\partial_{\beta^{a+1}\phi\phi'\phi_{d}\phi_{e}\phi_{f}}\mathcal{L})[\mathcal{R}^{1}]_{d}[\mathcal{R}^{1}]_{f}[\mathcal{H}^{-1}]_{eg}\\
	& +W^{-1}[\mathcal{R}^{1}]_{h}\sum_{e,f}(\partial_{\beta^{a+1}\phi\phi'\phi_{e}\phi_{f}}\mathcal{L})[\mathcal{R}_{b}^{1}]_{f}[\mathcal{H}^{-1}]_{eg}\\
	& -W^{-1}[\mathcal{R}^{1}]_{h}\sum_{e,f}(\partial_{\beta^{a+1}\phi\phi'\phi_{e}\phi_{f}}\mathcal{L})[\mathcal{R}^{1}]_{f}[\mathcal{H}^{-1}\mathcal{H}_{b}\mathcal{H}^{-1}]_{eg}\\
	& -2W^{-3}W_{b}[\mathcal{R}^{1}]_{h}\sum_{e,f}(\partial_{\beta^{a+1}\phi\phi'\phi_{e}\phi_{f}}\mathcal{L})[\mathcal{R}^{1}]_{e}[\mathcal{R}^{1}]_{f}[\mathcal{R}^{1}]_{g}\\
	& +W^{-2}[\mathcal{R}_{b}^{1}]_{h}\sum_{e,f}(\partial_{\beta^{a+1}\phi\phi'\phi_{e}\phi_{f}}\mathcal{L})[\mathcal{R}^{1}]_{e}[\mathcal{R}^{1}]_{f}[\mathcal{R}^{1}]_{g}\\
	& -W^{-3}[\mathcal{R}^{1}]_{h}\sum_{e,f}(\partial_{\beta^{a+2}\phi\phi'\phi_{e}\phi_{f}}\mathcal{L})[\mathcal{R}^{1}]_{e}[\mathcal{R}^{1}]_{f}[\mathcal{R}^{1}]_{g}\\
	& -W^{-3}[\mathcal{R}^{1}]_{h}\sum_{d,e,f}(\partial_{\beta^{a+1}\phi\phi'\phi_{d}\phi_{e}\phi_{f}}\mathcal{L})[\mathcal{R}^{1}]_{d}[\mathcal{R}^{1}]_{e}[\mathcal{R}^{1}]_{f}[\mathcal{R}^{1}]_{g}\\
	& +W^{-2}[\mathcal{R}^{1}]_{h}\sum_{e,f}(\partial_{\beta^{a+1}\phi\phi'\phi_{e}\phi_{f}}\mathcal{L})[\mathcal{R}_{b}^{1}]_{e}[\mathcal{R}^{1}]_{f}[\mathcal{R}^{1}]_{g}\\
	& +W^{-2}[\mathcal{R}^{1}]_{h}\sum_{e,f}(\partial_{\beta^{a+1}\phi\phi'\phi_{e}\phi_{f}}\mathcal{L})[\mathcal{R}^{1}]_{e}[\mathcal{R}_{b}^{1}]_{f}[\mathcal{R}^{1}]_{g}\\
	& +W^{-2}[\mathcal{R}^{1}]_{h}\sum_{e,f}(\partial_{\beta^{a+1}\phi\phi'\phi_{e}\phi_{f}}\mathcal{L})[\mathcal{R}^{1}]_{e}[\mathcal{R}^{1}]_{f}[\mathcal{R}_{b}^{1}]_{g}\\
	& +W^{-2}W_{b}[\mathcal{R}^{1}]_{h}\sum_{f}(\partial_{\beta^{a+1}\phi\phi'\phi_{f}}\mathcal{L})[\mathcal{R}_{s_{g}}^{1}]_{f}\\
	& -W^{-1}[\mathcal{R}_{b}^{1}]_{h}\sum_{f}(\partial_{\beta^{a+1}\phi\phi'\phi_{f}}\mathcal{L})[\mathcal{R}_{s_{g}}^{1}]_{f}\\
	& +W^{-2}[\mathcal{R}^{1}]_{h}\sum_{f}(\partial_{\beta^{a+2}\phi\phi'\phi_{f}}\mathcal{L})[\mathcal{R}_{s_{g}}^{1}]_{f}\\
	& +W^{-2}[\mathcal{R}^{1}]_{h}\sum_{e,f}(\partial_{\beta^{a+1}\phi\phi'\phi_{e}\phi_{f}}\mathcal{L})[\mathcal{R}_{s_{g}}^{1}]_{f}\\
	& -W^{-1}[\mathcal{R}^{1}]_{h}\sum_{f}(\partial_{\beta^{a+1}\phi\phi'\phi_{f}}\mathcal{L})[\mathcal{R}_{bs_{g}}^{1}]_{f}\\
	& -W^{-2}W_{b}\sum_{e,f}(\partial_{\beta^{a+1}\phi\phi'\phi_{e}\phi_{f}}\mathcal{L})[\mathcal{R}^{1}]_{f}[\mathcal{H}^{-1}]_{eh}[\mathcal{R}^{1}]_{g}\\
	& -W^{-2}\sum_{e,f}(\partial_{\beta^{a+2}\phi\phi'\phi_{e}\phi_{f}}\mathcal{L})[\mathcal{R}^{1}]_{f}[\mathcal{H}^{-1}]_{eh}[\mathcal{R}^{1}]_{g}\\
	& -W^{-2}\sum_{d,e,f}(\partial_{\beta^{a+1}\phi\phi'\phi_{d}\phi_{e}\phi_{f}}\mathcal{L})[\mathcal{R}^{1}]_{d}[\mathcal{R}^{1}]_{f}[\mathcal{H}^{-1}]_{eh}[\mathcal{R}^{1}]_{g}\\
	& +W^{-1}\sum_{e,f}(\partial_{\beta^{a+1}\phi\phi'\phi_{e}\phi_{f}}\mathcal{L})[\mathcal{R}_{b}^{1}]_{f}[\mathcal{H}^{-1}]_{eh}[\mathcal{R}^{1}]_{g}\\
	& -W^{-1}\sum_{e,f}(\partial_{\beta^{a+1}\phi\phi'\phi_{e}\phi_{f}}\mathcal{L})[\mathcal{R}^{1}]_{f}[\mathcal{H}^{-1}\mathcal{H}_{b}\mathcal{H}^{-1}]_{eh}[\mathcal{R}^{1}]_{g}\\
	& +W^{-1}\sum_{e,f}(\partial_{\beta^{a+1}\phi\phi'\phi_{e}\phi_{f}}\mathcal{L})[\mathcal{R}^{1}]_{f}[\mathcal{H}^{-1}]_{eh}[\mathcal{R}_{b}^{1}]_{g}\\
	& -W^{-1}\sum_{d,e,f}(\partial_{\beta^{a+1}\phi\phi'\phi_{d}\phi_{e}\phi_{f}}\mathcal{L})[\mathcal{R}^{1}]_{f}[\mathcal{H}^{-1}]_{eh}[\mathcal{H}^{-1}]_{dg}\\
	& -W^{-1}\sum_{c,d,e,f}(\partial_{\beta^{a}\phi\phi'\phi_{c}\phi_{d}\phi_{e}\phi_{f}}\mathcal{L})[\mathcal{R}^{1}]_{c}[\mathcal{R}^{1}]_{f}[\mathcal{H}^{-1}]_{eh}[\mathcal{H}^{-1}]_{dg}\\
	& +\sum_{d,e,f}(\partial_{\beta^{a}\phi\phi'\phi_{d}\phi_{e}\phi_{f}}\mathcal{L})[\mathcal{R}_{b}^{1}]_{f}[\mathcal{H}^{-1}]_{eh}[\mathcal{H}^{-1}]_{dg}\\
	& -\sum_{d,e,f}(\partial_{\beta^{a}\phi\phi'\phi_{d}\phi_{e}\phi_{f}}\mathcal{L})[\mathcal{R}^{1}]_{f}[\mathcal{H}^{-1}\mathcal{H}_{b}\mathcal{H}^{-1}]_{eh}[\mathcal{H}^{-1}]_{dg}\\
	& -\sum_{d,e,f}(\partial_{\beta^{a}\phi\phi'\phi_{d}\phi_{e}\phi_{f}}\mathcal{L})[\mathcal{R}^{1}]_{f}[\mathcal{H}^{-1}]_{eh}[\mathcal{H}^{-1}\mathcal{H}_{b}\mathcal{H}^{-1}]_{dg}\\
	& -W^{-2}W_{b}\sum_{d,e,f}(\partial_{\beta^{a}\phi\phi'\phi_{d}\phi_{e}\phi_{f}}\mathcal{L})[\mathcal{R}^{1}]_{f}[\mathcal{H}^{-1}]_{eh}[\mathcal{R}^{1}]_{d}[\mathcal{R}^{1}]_{g}\\
	& -W^{-2}\sum_{d,e,f}(\partial_{\beta^{a+1}\phi\phi'\phi_{d}\phi_{e}\phi_{f}}\mathcal{L})[\mathcal{R}^{1}]_{f}[\mathcal{H}^{-1}]_{eh}[\mathcal{R}^{1}]_{d}[\mathcal{R}^{1}]_{g}\\
	& -W^{-2}\sum_{c,d,e,f}(\partial_{\beta^{a}\phi\phi'\phi_{c}\phi_{d}\phi_{e}\phi_{f}}\mathcal{L})[\mathcal{R}^{1}]_{c}[\mathcal{R}^{1}]_{f}[\mathcal{H}^{-1}]_{eh}[\mathcal{R}^{1}]_{d}[\mathcal{R}^{1}]_{g}\\
	& +W^{-1}\sum_{d,e,f}(\partial_{\beta^{a}\phi\phi'\phi_{d}\phi_{e}\phi_{f}}\mathcal{L})[\mathcal{R}_{b}^{1}]_{f}[\mathcal{H}^{-1}]_{eh}[\mathcal{R}^{1}]_{d}[\mathcal{R}^{1}]_{g}\\
	& -W^{-1}\sum_{d,e,f}(\partial_{\beta^{a}\phi\phi'\phi_{d}\phi_{e}\phi_{f}}\mathcal{L})[\mathcal{R}^{1}]_{f}[\mathcal{H}^{-1}\mathcal{H}_{b}\mathcal{H}^{-1}]_{eh}[\mathcal{R}^{1}]_{d}[\mathcal{R}^{1}]_{g}\\
	& +W^{-1}\sum_{d,e,f}(\partial_{\beta^{a}\phi\phi'\phi_{d}\phi_{e}\phi_{f}}\mathcal{L})[\mathcal{R}^{1}]_{f}[\mathcal{H}^{-1}]_{eh}[\mathcal{R}_{b}^{1}]_{d}[\mathcal{R}^{1}]_{g}\\
	& +W^{-1}\sum_{d,e,f}(\partial_{\beta^{a}\phi\phi'\phi_{d}\phi_{e}\phi_{f}}\mathcal{L})[\mathcal{R}^{1}]_{f}[\mathcal{H}^{-1}]_{eh}[\mathcal{R}^{1}]_{d}[\mathcal{R}_{b}^{1}]_{g}\\
	& +W^{-1}\sum_{e,f}(\partial_{\beta^{a+1}\phi\phi'\phi_{e}\phi_{f}}\mathcal{L})[\mathcal{R}_{s_{g}}^{1}]_{f}[\mathcal{H}^{-1}]_{eh}\\
	& +W^{-1}\sum_{d,e,f}(\partial_{\beta^{a}\phi\phi'\phi_{d}\phi_{e}\phi_{f}}\mathcal{L})[\mathcal{R}^{1}]_{d}[\mathcal{R}_{s_{g}}^{1}]_{f}[\mathcal{H}^{-1}]_{eh}\\
	& -\sum_{e,f}(\partial_{\beta^{a}\phi\phi'\phi_{e}\phi_{f}}\mathcal{L})[\mathcal{R}_{bs_{g}}^{1}]_{f}[\mathcal{H}^{-1}]_{eh}\\
	& +\sum_{e,f}(\partial_{\beta^{a}\phi\phi'\phi_{e}\phi_{f}}\mathcal{L})[\mathcal{R}_{s_{g}}^{1}]_{f}[\mathcal{H}^{-1}\mathcal{H}_{b}\mathcal{H}^{-1}]_{eh}\\
	& -W^{-1}\sum_{e,f}(\partial_{\beta^{a+1}\phi\phi'\phi_{e}\phi_{f}}\mathcal{L})[\mathcal{R}^{1}]_{f}[\mathcal{H}^{-1}\mathcal{H}_{s_{g}}\mathcal{H}^{-1}]_{eh}\\
	& -W^{-1}\sum_{d,e,f}(\partial_{\beta^{a}\phi\phi'\phi_{d}\phi_{e}\phi_{f}}\mathcal{L})[\mathcal{R}^{1}]_{d}[\mathcal{R}^{1}]_{f}[\mathcal{H}^{-1}\mathcal{H}_{s_{g}}\mathcal{H}^{-1}]_{eh}\\
	& +\sum_{e,f}(\partial_{\beta^{a}\phi\phi'\phi_{e}\phi_{f}}\mathcal{L})[\mathcal{R}_{b}^{1}]_{f}[\mathcal{H}^{-1}\mathcal{H}_{s_{g}}\mathcal{H}^{-1}]_{eh}\\
	& -\sum_{e,f}(\partial_{\beta^{a}\phi\phi'\phi_{e}\phi_{f}}\mathcal{L})[\mathcal{R}^{1}]_{f}[\mathcal{H}^{-1}\mathcal{H}_{b}\mathcal{H}^{-1}\mathcal{H}_{s_{g}}\mathcal{H}^{-1}]_{eh}\\
	& +\sum_{e,f}(\partial_{\beta^{a}\phi\phi'\phi_{e}\phi_{f}}\mathcal{L})[\mathcal{R}^{1}]_{f}[\mathcal{H}^{-1}\mathcal{H}_{bs_{g}}\mathcal{H}^{-1}]_{eh}\\
	& -\sum_{e,f}(\partial_{\beta^{a}\phi\phi'\phi_{e}\phi_{f}}\mathcal{L})[\mathcal{R}^{1}]_{f}[\mathcal{H}^{-1}\mathcal{H}_{s_{g}}\mathcal{H}^{-1}\mathcal{H}_{b}\mathcal{H}^{-1}]_{eh}\\
	& -2W^{-3}W_{b}W_{s_{g}}\sum_{e,f}(\partial_{\beta^{a}\phi\phi'\phi_{e}\phi_{f}}\mathcal{L})[\mathcal{R}^{1}]_{f}[\mathcal{R}^{1}]_{e}[\mathcal{R}^{1}]_{h}\\
	& +W^{-2}W_{bs_{g}}\sum_{e,f}(\partial_{\beta^{a}\phi\phi'\phi_{e}\phi_{f}}\mathcal{L})[\mathcal{R}^{1}]_{f}[\mathcal{R}^{1}]_{e}[\mathcal{R}^{1}]_{h}\\
	& -W^{-3}W_{s_{g}}\sum_{e,f}(\partial_{\beta^{a+1}\phi\phi'\phi_{e}\phi_{f}}\mathcal{L})[\mathcal{R}^{1}]_{f}[\mathcal{R}^{1}]_{e}[\mathcal{R}^{1}]_{h}\\
	& -W^{-3}W_{s_{g}}\sum_{d,e,f}(\partial_{\beta^{a}\phi\phi'\phi_{d}\phi_{e}\phi_{f}}\mathcal{L})[\mathcal{R}^{1}]_{d}[\mathcal{R}^{1}]_{f}[\mathcal{R}^{1}]_{e}[\mathcal{R}^{1}]_{h}\\
	& +W^{-2}W_{s_{g}}\sum_{e,f}(\partial_{\beta^{a}\phi\phi'\phi_{e}\phi_{f}}\mathcal{L})[\mathcal{R}_{b}^{1}]_{f}[\mathcal{R}^{1}]_{e}[\mathcal{R}^{1}]_{h}\\
	& +W^{-2}W_{s_{g}}\sum_{e,f}(\partial_{\beta^{a}\phi\phi'\phi_{e}\phi_{f}}\mathcal{L})[\mathcal{R}^{1}]_{f}[\mathcal{R}_{b}^{1}]_{e}[\mathcal{R}^{1}]_{h}\\
	& +W^{-2}W_{s_{g}}\sum_{e,f}(\partial_{\beta^{a}\phi\phi'\phi_{e}\phi_{f}}\mathcal{L})[\mathcal{R}^{1}]_{f}[\mathcal{R}^{1}]_{e}[\mathcal{R}_{b}^{1}]_{h}\\
	& -2W^{-3}W_{b}\sum_{e,f}(\partial_{\beta^{a+1}\phi\phi'\phi_{e}\phi_{f}}\mathcal{L})[\mathcal{R}^{1}]_{f}[\mathcal{R}^{1}]_{e}[\mathcal{R}^{1}]_{g}[\mathcal{R}^{1}]_{h}\\
	& -W^{-3}\sum_{e,f}(\partial_{\beta^{a+2}\phi\phi'\phi_{e}\phi_{f}}\mathcal{L})[\mathcal{R}^{1}]_{f}[\mathcal{R}^{1}]_{e}[\mathcal{R}^{1}]_{g}[\mathcal{R}^{1}]_{h}\\
	& -W^{-3}\sum_{d,e,f}(\partial_{\beta^{a+1}\phi\phi'\phi_{d}\phi_{e}\phi_{f}}\mathcal{L})[\mathcal{R}^{1}]_{d}[\mathcal{R}^{1}]_{f}[\mathcal{R}^{1}]_{e}[\mathcal{R}^{1}]_{g}[\mathcal{R}^{1}]_{h}\\
	& +W^{-2}\sum_{e,f}(\partial_{\beta^{a+1}\phi\phi'\phi_{e}\phi_{f}}\mathcal{L})[\mathcal{R}_{b}^{1}]_{f}[\mathcal{R}^{1}]_{e}[\mathcal{R}^{1}]_{g}[\mathcal{R}^{1}]_{h}\\
	& +W^{-2}\sum_{e,f}(\partial_{\beta^{a+1}\phi\phi'\phi_{e}\phi_{f}}\mathcal{L})[\mathcal{R}^{1}]_{f}[\mathcal{R}_{b}^{1}]_{e}[\mathcal{R}^{1}]_{g}[\mathcal{R}^{1}]_{h}\\
	& +W^{-2}\sum_{e,f}(\partial_{\beta^{a+1}\phi\phi'\phi_{e}\phi_{f}}\mathcal{L})[\mathcal{R}^{1}]_{f}[\mathcal{R}^{1}]_{e}[\mathcal{R}_{b}^{1}]_{g}[\mathcal{R}^{1}]_{h}\\
	& +W^{-2}\sum_{e,f}(\partial_{\beta^{a+1}\phi\phi'\phi_{e}\phi_{f}}\mathcal{L})[\mathcal{R}^{1}]_{f}[\mathcal{R}^{1}]_{e}[\mathcal{R}^{1}]_{g}[\mathcal{R}_{b}^{1}]_{h}\\
	& -W^{-2}W_{b}\sum_{d,e,f}(\partial_{\beta^{a}\phi\phi'\phi_{d}\phi_{e}\phi_{f}}\mathcal{L})[\mathcal{R}^{1}]_{f}[\mathcal{R}^{1}]_{e}[\mathcal{R}^{1}]_{h}[\mathcal{H}^{-1}]_{dg}\\
	& -W^{-2}\sum_{d,e,f}(\partial_{\beta^{a+1}\phi\phi'\phi_{d}\phi_{e}\phi_{f}}\mathcal{L})[\mathcal{R}^{1}]_{f}[\mathcal{R}^{1}]_{e}[\mathcal{R}^{1}]_{h}[\mathcal{H}^{-1}]_{dg}\\
	& -W^{-2}\sum_{c,d,e,f}(\partial_{\beta^{a}\phi\phi'\phi_{c}\phi_{d}\phi_{e}\phi_{f}}\mathcal{L})[\mathcal{R}^{1}]_{c}[\mathcal{R}^{1}]_{f}[\mathcal{R}^{1}]_{e}[\mathcal{R}^{1}]_{h}[\mathcal{H}^{-1}]_{dg}\\
	& +W^{-1}\sum_{d,e,f}(\partial_{\beta^{a}\phi\phi'\phi_{d}\phi_{e}\phi_{f}}\mathcal{L})[\mathcal{R}_{b}^{1}]_{f}[\mathcal{R}^{1}]_{e}[\mathcal{R}^{1}]_{h}[\mathcal{H}^{-1}]_{dg}\\
	& +W^{-1}\sum_{d,e,f}(\partial_{\beta^{a}\phi\phi'\phi_{d}\phi_{e}\phi_{f}}\mathcal{L})[\mathcal{R}^{1}]_{f}[\mathcal{R}_{b}^{1}]_{e}[\mathcal{R}^{1}]_{h}[\mathcal{H}^{-1}]_{dg}\\
	& +W^{-1}\sum_{d,e,f}(\partial_{\beta^{a}\phi\phi'\phi_{d}\phi_{e}\phi_{f}}\mathcal{L})[\mathcal{R}^{1}]_{f}[\mathcal{R}^{1}]_{e}[\mathcal{R}_{b}^{1}]_{h}[\mathcal{H}^{-1}]_{dg}\\
	& -W^{-1}\sum_{d,e,f}(\partial_{\beta^{a}\phi\phi'\phi_{d}\phi_{e}\phi_{f}}\mathcal{L})[\mathcal{R}^{1}]_{f}[\mathcal{R}^{1}]_{e}[\mathcal{R}^{1}]_{h}[\mathcal{H}^{-1}\mathcal{H}_{b}\mathcal{H}^{-1}]_{dg}\\
	& -2W^{-3}W_{b}\sum_{d,e,f}(\partial_{\beta^{a}\phi\phi'\phi_{d}\phi_{e}\phi_{f}}\mathcal{L})[\mathcal{R}^{1}]_{f}[\mathcal{R}^{1}]_{d}[\mathcal{R}^{1}]_{e}[\mathcal{R}^{1}]_{g}[\mathcal{R}^{1}]_{h}\\
	& -W^{-3}\sum_{d,e,f}(\partial_{\beta^{a+1}\phi\phi'\phi_{d}\phi_{e}\phi_{f}}\mathcal{L})[\mathcal{R}^{1}]_{f}[\mathcal{R}^{1}]_{d}[\mathcal{R}^{1}]_{e}[\mathcal{R}^{1}]_{g}[\mathcal{R}^{1}]_{h}\\
	& -W^{-3}\sum_{c,d,e,f}(\partial_{\beta^{a}\phi\phi'\phi_{c}\phi_{d}\phi_{e}\phi_{f}}\mathcal{L})[\mathcal{R}^{1}]_{c}[\mathcal{R}^{1}]_{f}[\mathcal{R}^{1}]_{d}[\mathcal{R}^{1}]_{e}[\mathcal{R}^{1}]_{g}[\mathcal{R}^{1}]_{h}\\
	& +W^{-2}\sum_{d,e,f}(\partial_{\beta^{a}\phi\phi'\phi_{d}\phi_{e}\phi_{f}}\mathcal{L})[\mathcal{R}_{b}^{1}]_{f}[\mathcal{R}^{1}]_{d}[\mathcal{R}^{1}]_{e}[\mathcal{R}^{1}]_{g}[\mathcal{R}^{1}]_{h}\\
	& +W^{-2}\sum_{d,e,f}(\partial_{\beta^{a}\phi\phi'\phi_{d}\phi_{e}\phi_{f}}\mathcal{L})[\mathcal{R}^{1}]_{f}[\mathcal{R}_{b}^{1}]_{d}[\mathcal{R}^{1}]_{e}[\mathcal{R}^{1}]_{g}[\mathcal{R}^{1}]_{h}\\
	& +W^{-2}\sum_{d,e,f}(\partial_{\beta^{a}\phi\phi'\phi_{d}\phi_{e}\phi_{f}}\mathcal{L})[\mathcal{R}^{1}]_{f}[\mathcal{R}^{1}]_{d}[\mathcal{R}_{b}^{1}]_{e}[\mathcal{R}^{1}]_{g}[\mathcal{R}^{1}]_{h}\\
	& +W^{-2}\sum_{d,e,f}(\partial_{\beta^{a}\phi\phi'\phi_{d}\phi_{e}\phi_{f}}\mathcal{L})[\mathcal{R}^{1}]_{f}[\mathcal{R}^{1}]_{d}[\mathcal{R}^{1}]_{e}[\mathcal{R}_{b}^{1}]_{g}[\mathcal{R}^{1}]_{h}\\
	& +W^{-2}\sum_{d,e,f}(\partial_{\beta^{a}\phi\phi'\phi_{d}\phi_{e}\phi_{f}}\mathcal{L})[\mathcal{R}^{1}]_{f}[\mathcal{R}^{1}]_{d}[\mathcal{R}^{1}]_{e}[\mathcal{R}^{1}]_{g}[\mathcal{R}_{b}^{1}]_{h}\\
	& +W^{-2}W_{b}\sum_{e,f}(\partial_{\beta^{a}\phi\phi'\phi_{e}\phi_{f}}\mathcal{L})[\mathcal{R}_{s_{g}}^{1}]_{f}[\mathcal{R}^{1}]_{e}[\mathcal{R}^{1}]_{h}\\
	& +W^{-2}\sum_{e,f}(\partial_{\beta^{a+1}\phi\phi'\phi_{e}\phi_{f}}\mathcal{L})[\mathcal{R}_{s_{g}}^{1}]_{f}[\mathcal{R}^{1}]_{e}[\mathcal{R}^{1}]_{h}\\
	& +W^{-2}\sum_{d,e,f}(\partial_{\beta^{a}\phi\phi'\phi_{d}\phi_{e}\phi_{f}}\mathcal{L})[\mathcal{R}^{1}]_{d}[\mathcal{R}_{s_{g}}^{1}]_{f}[\mathcal{R}^{1}]_{e}[\mathcal{R}^{1}]_{h}\\
	& -W^{-1}\sum_{e,f}(\partial_{\beta^{a}\phi\phi'\phi_{e}\phi_{f}}\mathcal{L})[\mathcal{R}_{bs_{g}}^{1}]_{f}[\mathcal{R}^{1}]_{e}[\mathcal{R}^{1}]_{h}\\
	& -W^{-1}\sum_{e,f}(\partial_{\beta^{a}\phi\phi'\phi_{e}\phi_{f}}\mathcal{L})[\mathcal{R}_{s_{g}}^{1}]_{f}[\mathcal{R}_{b}^{1}]_{e}[\mathcal{R}^{1}]_{h}\\
	& -W^{-1}\sum_{e,f}(\partial_{\beta^{a}\phi\phi'\phi_{e}\phi_{f}}\mathcal{L})[\mathcal{R}_{s_{g}}^{1}]_{f}[\mathcal{R}^{1}]_{e}[\mathcal{R}_{b}^{1}]_{h}\\
	& +W^{-2}W_{b}\sum_{e,f}(\partial_{\beta^{a}\phi\phi'\phi_{e}\phi_{f}}\mathcal{L})[\mathcal{R}^{1}]_{f}[\mathcal{R}_{s_{g}}^{1}]_{e}[\mathcal{R}^{1}]_{h}\\
	& +W^{-2}\sum_{e,f}(\partial_{\beta^{a+1}\phi\phi'\phi_{e}\phi_{f}}\mathcal{L})[\mathcal{R}^{1}]_{f}[\mathcal{R}_{s_{g}}^{1}]_{e}[\mathcal{R}^{1}]_{h}\\
	& +W^{-2}\sum_{d,e,f}(\partial_{\beta^{a}\phi\phi'\phi_{d}\phi_{e}\phi_{f}}\mathcal{L})[\mathcal{R}^{1}]_{d}[\mathcal{R}^{1}]_{f}[\mathcal{R}_{s_{g}}^{1}]_{e}[\mathcal{R}^{1}]_{h}\\
	& -W^{-1}\sum_{e,f}(\partial_{\beta^{a}\phi\phi'\phi_{e}\phi_{f}}\mathcal{L})[\mathcal{R}_{b}^{1}]_{f}[\mathcal{R}_{s_{g}}^{1}]_{e}[\mathcal{R}^{1}]_{h}\\
	& -W^{-1}\sum_{e,f}(\partial_{\beta^{a}\phi\phi'\phi_{e}\phi_{f}}\mathcal{L})[\mathcal{R}^{1}]_{f}[\mathcal{R}_{bs_{g}}^{1}]_{e}[\mathcal{R}^{1}]_{h}\\
	& -W^{-1}\sum_{e,f}(\partial_{\beta^{a}\phi\phi'\phi_{e}\phi_{f}}\mathcal{L})[\mathcal{R}^{1}]_{f}[\mathcal{R}_{s_{g}}^{1}]_{e}[\mathcal{R}_{b}^{1}]_{h}\\
	& +W^{-2}W_{b}\sum_{e,f}(\partial_{\beta^{a}\phi\phi'\phi_{e}\phi_{f}}\mathcal{L})[\mathcal{R}^{1}]_{f}[\mathcal{R}^{1}]_{e}[\mathcal{R}_{s_{g}}^{1}]_{h}\\
	& +W^{-2}\sum_{e,f}(\partial_{\beta^{a+1}\phi\phi'\phi_{e}\phi_{f}}\mathcal{L})[\mathcal{R}^{1}]_{f}[\mathcal{R}^{1}]_{e}[\mathcal{R}_{s_{g}}^{1}]_{h}\\
	& +W^{-2}\sum_{d,e,f}(\partial_{\beta^{a}\phi\phi'\phi_{d}\phi_{e}\phi_{f}}\mathcal{L})[\mathcal{R}^{1}]_{d}[\mathcal{R}^{1}]_{f}[\mathcal{R}^{1}]_{e}[\mathcal{R}_{s_{g}}^{1}]_{h}\\
	& -W^{-1}\sum_{e,f}(\partial_{\beta^{a}\phi\phi'\phi_{e}\phi_{f}}\mathcal{L})[\mathcal{R}_{b}^{1}]_{f}[\mathcal{R}^{1}]_{e}[\mathcal{R}_{s_{g}}^{1}]_{h}\\
	& -W^{-1}\sum_{e,f}(\partial_{\beta^{a}\phi\phi'\phi_{e}\phi_{f}}\mathcal{L})[\mathcal{R}^{1}]_{f}[\mathcal{R}_{b}^{1}]_{e}[\mathcal{R}_{s_{g}}^{1}]_{h}\\
	& -W^{-1}\sum_{e,f}(\partial_{\beta^{a}\phi\phi'\phi_{e}\phi_{f}}\mathcal{L})[\mathcal{R}^{1}]_{f}[\mathcal{R}^{1}]_{e}[\mathcal{R}_{bs_{g}}^{1}]_{h}\\
	& +W^{-2}W_{b}\sum_{f}(\partial_{\beta^{a+1}\phi\phi'\phi_{f}}\mathcal{L})[\mathcal{R}_{s_{h}}^{1}]_{f}[\mathcal{R}^{1}]_{g}\\
	& +W^{-2}\sum_{f}(\partial_{\beta^{a+2}\phi\phi'\phi_{f}}\mathcal{L})[\mathcal{R}_{s_{h}}^{1}]_{f}[\mathcal{R}^{1}]_{g}\\
	& +W^{-2}\sum_{e,f}(\partial_{\beta^{a+1}\phi\phi'\phi_{e}\phi_{f}}\mathcal{L})[\mathcal{R}^{1}]_{e}[\mathcal{R}_{s_{h}}^{1}]_{f}[\mathcal{R}^{1}]_{g}\\
	& -W^{-1}\sum_{f}(\partial_{\beta^{a+1}\phi\phi'\phi_{f}}\mathcal{L})[\mathcal{R}_{bs_{h}}^{1}]_{f}[\mathcal{R}^{1}]_{g}\\
	& -W^{-1}\sum_{f}(\partial_{\beta^{a+1}\phi\phi'\phi_{f}}\mathcal{L})[\mathcal{R}_{s_{h}}^{1}]_{f}[\mathcal{R}_{b}^{1}]_{g}\\
	& +W^{-1}\sum_{e,f}(\partial_{\beta^{a+1}\phi\phi'\phi_{e}\phi_{f}}\mathcal{L})[\mathcal{R}_{s_{h}}^{1}]_{f}[\mathcal{H}^{-1}]_{eg}\\
	& +W^{-1}\sum_{d,e,f}(\partial_{\beta^{a}\phi\phi'\phi_{d}\phi_{e}\phi_{f}}\mathcal{L})[\mathcal{R}^{1}]_{d}[\mathcal{R}_{s_{h}}^{1}]_{f}[\mathcal{H}^{-1}]_{eg}\\
	& -\sum_{e,f}(\partial_{\beta^{a}\phi\phi'\phi_{e}\phi_{f}}\mathcal{L})[\mathcal{R}_{bs_{h}}^{1}]_{f}[\mathcal{H}^{-1}]_{eg}\\
	& +\sum_{e,f}(\partial_{\beta^{a}\phi\phi'\phi_{e}\phi_{f}}\mathcal{L})[\mathcal{R}_{s_{h}}^{1}]_{f}[\mathcal{H}^{-1}\mathcal{H}_{b}\mathcal{H}^{-1}]_{eg}\\
	& -W^{-2}W_{b}\sum_{e,f}(\partial_{\beta^{a}\phi\phi'\phi_{e}\phi_{f}}\mathcal{L})[\mathcal{R}_{s_{h}}^{1}]_{f}[\mathcal{R}^{1}]_{e}[\mathcal{R}^{1}]_{g}\\
	& -W^{-2}\sum_{e,f}(\partial_{\beta^{a+1}\phi\phi'\phi_{e}\phi_{f}}\mathcal{L})[\mathcal{R}_{s_{h}}^{1}]_{f}[\mathcal{R}^{1}]_{e}[\mathcal{R}^{1}]_{g}\\
	& -W^{-2}\sum_{d,e,f}(\partial_{\beta^{a}\phi\phi'\phi_{d}\phi_{e}\phi_{f}}\mathcal{L})[\mathcal{R}^{1}]_{d}[\mathcal{R}_{s_{h}}^{1}]_{f}[\mathcal{R}^{1}]_{e}[\mathcal{R}^{1}]_{g}\\
	& +W^{-1}\sum_{e,f}(\partial_{\beta^{a}\phi\phi'\phi_{e}\phi_{f}}\mathcal{L})[\mathcal{R}_{bs_{h}}^{1}]_{f}[\mathcal{R}^{1}]_{e}[\mathcal{R}^{1}]_{g}\\
	& +W^{-1}\sum_{e,f}(\partial_{\beta^{a}\phi\phi'\phi_{e}\phi_{f}}\mathcal{L})[\mathcal{R}_{s_{h}}^{1}]_{f}[\mathcal{R}_{b}^{1}]_{e}[\mathcal{R}^{1}]_{g}\\
	& +W^{-1}\sum_{e,f}(\partial_{\beta^{a}\phi\phi'\phi_{e}\phi_{f}}\mathcal{L})[\mathcal{R}_{s_{h}}^{1}]_{f}[\mathcal{R}^{1}]_{e}[\mathcal{R}_{b}^{1}]_{g}\\
	& -W^{-1}\sum_{f}(\partial_{\beta^{a+1}\phi\phi'\phi_{f}}\mathcal{L})[\mathcal{R}_{s_{g}s_{h}}^{1}]_{f}\\
	& -W^{-1}\sum_{e,f}(\partial_{\beta^{a}\phi\phi'\phi_{e}\phi_{f}}\mathcal{L})[\mathcal{R}^{1}]_{e}[\mathcal{R}_{s_{g}s_{h}}^{1}]_{f}\\
	& +\sum_{f}(\partial_{\beta^{a}\phi\phi'\phi_{f}}\mathcal{L})[\mathcal{R}_{bs_{g}s_{h}}^{1}]_{f}
\end{align*}

\begin{align*}
	\mathcal{P}_{(g),bs_{k}s_{h}}^{r} & =-2W^{-3}W_{b}W_{s_{k}}[\mathcal{R}^{1}]_{h}\mathcal{P}_{(g)}^{r+1}+W^{-2}W_{bs_{k}}[\mathcal{R}^{1}]_{h}\mathcal{P}_{(g)}^{r+1}\\
	& +W^{-2}W_{s_{k}}[\mathcal{R}_{b}^{1}]_{h}\mathcal{P}_{(g)}^{r+1}+W^{-2}W_{s_{k}}[\mathcal{R}^{1}]_{h}\mathcal{P}_{(g),b}^{r+1}\\
	& +W^{-2}W_{b}[\mathcal{R}_{s_{k}}^{1}]_{h}\mathcal{P}_{(g)}^{r+1}-W^{-1}[\mathcal{R}_{bs_{k}}^{1}]_{h}\mathcal{P}_{(g)}^{r+1}-W^{-1}[\mathcal{R}_{s_{k}}^{1}]_{h}\mathcal{P}_{(g),b}^{r+1}\\
	& +W^{-2}W_{b}[\mathcal{R}^{1}]_{h}\mathcal{P}_{(g),s_{k}}^{r+1}-W^{-1}[\mathcal{R}_{b}^{1}]_{h}\mathcal{P}_{(g),s_{k}}^{r+1}-W^{-1}[\mathcal{R}^{1}]_{h}\mathcal{P}_{(g),bs_{k}}^{r+1}\\
	& -W^{-2}W_{b}\sum_{e,f}(\partial_{\beta^{r+1}\phi\phi'\phi_{e}\phi_{f}}\mathcal{L})\big[\mathcal{H}^{-1}\big]_{fg}\big[\mathcal{H}^{-1}\big]_{eh}[\mathcal{R}^{1}]_{k}\\
	& -W^{-2}\sum_{e,f}(\partial_{\beta^{r+2}\phi\phi'\phi_{e}\phi_{f}}\mathcal{L})\big[\mathcal{H}^{-1}\big]_{fg}\big[\mathcal{H}^{-1}\big]_{eh}[\mathcal{R}^{1}]_{k}\\
	& -W^{-2}\sum_{d,e,f}(\partial_{\beta^{r+1}\phi\phi'\phi_{d}\phi_{e}\phi_{f}}\mathcal{L})\big[\mathcal{H}^{-1}\big]_{fg}\big[\mathcal{H}^{-1}\big]_{eh}[\mathcal{R}^{1}]_{d}[\mathcal{R}^{1}]_{k}\\
	& -W^{-1}\sum_{e,f}(\partial_{\beta^{r+1}\phi\phi'\phi_{e}\phi_{f}}\mathcal{L})\big[\mathcal{H}^{-1}\mathcal{H}_{b}\mathcal{H}^{-1}\big]_{fg}\big[\mathcal{H}^{-1}\big]_{eh}[\mathcal{R}^{1}]_{k}\\
	& -W^{-1}\sum_{e,f}(\partial_{\beta^{r+1}\phi\phi'\phi_{e}\phi_{f}}\mathcal{L})\big[\mathcal{H}^{-1}\big]_{fg}\big[\mathcal{H}^{-1}\mathcal{H}_{b}\mathcal{H}^{-1}\big]_{eh}[\mathcal{R}^{1}]_{k}\\
	& +W^{-1}\sum_{e,f}(\partial_{\beta^{r+1}\phi\phi'\phi_{e}\phi_{f}}\mathcal{L})\big[\mathcal{H}^{-1}\big]_{fg}\big[\mathcal{H}^{-1}\big]_{eh}[\mathcal{R}_{b}^{1}]_{k}\\
	& -W^{-1}\sum_{d,e,f}(\partial_{\beta^{r+1}\phi\phi'\phi_{d}\phi_{e}\phi_{f}}\mathcal{L})\big[\mathcal{H}^{-1}\big]_{fg}\big[\mathcal{H}^{-1}\big]_{eh}\big[\mathcal{H}^{-1}\big]_{dk}\\
	& -W^{-1}\sum_{c,d,e,f}(\partial_{\beta^{r}\phi\phi'\phi_{c}\phi_{d}\phi_{e}\phi_{f}}\mathcal{L})[\mathcal{R}^{1}]_{c}\big[\mathcal{H}^{-1}\big]_{fg}\big[\mathcal{H}^{-1}\big]_{eh}\big[\mathcal{H}^{-1}\big]_{dk}\\
	& -\sum_{d,e,f}(\partial_{\beta^{r}\phi\phi'\phi_{d}\phi_{e}\phi_{f}}\mathcal{L})\big[\mathcal{H}^{-1}\mathcal{H}_{b}\mathcal{H}^{-1}\big]_{fg}\big[\mathcal{H}^{-1}\big]_{eh}\big[\mathcal{H}^{-1}\big]_{dk}\\
	& -\sum_{d,e,f}(\partial_{\beta^{r}\phi\phi'\phi_{d}\phi_{e}\phi_{f}}\mathcal{L})\big[\mathcal{H}^{-1}\big]_{fg}\big[\mathcal{H}^{-1}\mathcal{H}_{b}\mathcal{H}^{-1}\big]_{eh}\big[\mathcal{H}^{-1}\big]_{dk}\\
	& -\sum_{d,e,f}(\partial_{\beta^{r}\phi\phi'\phi_{d}\phi_{e}\phi_{f}}\mathcal{L})\big[\mathcal{H}^{-1}\big]_{fg}\big[\mathcal{H}^{-1}\big]_{eh}\big[\mathcal{H}^{-1}\mathcal{H}_{b}\mathcal{H}^{-1}\big]_{dk}\\
	& -W^{-2}W_{b}\sum_{d,e,f}(\partial_{\beta^{r}\phi\phi'\phi_{d}\phi_{e}\phi_{f}}\mathcal{L})\big[\mathcal{H}^{-1}\big]_{fg}\big[\mathcal{H}^{-1}\big]_{eh}[\mathcal{R}^{1}]_{d}[\mathcal{R}^{1}]_{k}\\
	& -W^{-2}\sum_{d,e,f}(\partial_{\beta^{r+1}\phi\phi'\phi_{d}\phi_{e}\phi_{f}}\mathcal{L})\big[\mathcal{H}^{-1}\big]_{fg}\big[\mathcal{H}^{-1}\big]_{eh}[\mathcal{R}^{1}]_{d}[\mathcal{R}^{1}]_{k}\\
	& -W^{-2}\sum_{c,d,e,f}(\partial_{\beta^{r}\phi\phi'\phi_{c}\phi_{d}\phi_{e}\phi_{f}}\mathcal{L})\big[\mathcal{H}^{-1}\big]_{fg}\big[\mathcal{H}^{-1}\big]_{eh}[\mathcal{R}^{1}]_{c}[\mathcal{R}^{1}]_{d}[\mathcal{R}^{1}]_{k}\\
	& -W^{-1}\sum_{d,e,f}(\partial_{\beta^{r}\phi\phi'\phi_{d}\phi_{e}\phi_{f}}\mathcal{L})\big[\mathcal{H}^{-1}\mathcal{H}_{b}\mathcal{H}^{-1}\big]_{fg}\big[\mathcal{H}^{-1}\big]_{eh}[\mathcal{R}^{1}]_{d}[\mathcal{R}^{1}]_{k}\\
	& -W^{-1}\sum_{d,e,f}(\partial_{\beta^{r}\phi\phi'\phi_{d}\phi_{e}\phi_{f}}\mathcal{L})\big[\mathcal{H}^{-1}\big]_{fg}\big[\mathcal{H}^{-1}\mathcal{H}_{b}\mathcal{H}^{-1}\big]_{eh}[\mathcal{R}^{1}]_{d}[\mathcal{R}^{1}]_{k}\\
	& +W^{-1}\sum_{d,e,f}(\partial_{\beta^{r}\phi\phi'\phi_{d}\phi_{e}\phi_{f}}\mathcal{L})\big[\mathcal{H}^{-1}\big]_{fg}\big[\mathcal{H}^{-1}\big]_{eh}[\mathcal{R}_{b}^{1}]_{d}[\mathcal{R}^{1}]_{k}\\
	& +W^{-1}\sum_{d,e,f}(\partial_{\beta^{r}\phi\phi'\phi_{d}\phi_{e}\phi_{f}}\mathcal{L})\big[\mathcal{H}^{-1}\big]_{fg}\big[\mathcal{H}^{-1}\big]_{eh}[\mathcal{R}^{1}]_{d}[\mathcal{R}_{b}^{1}]_{k}\\
	& -W^{-1}\sum_{e,f}(\partial_{\beta^{r+1}\phi\phi'\phi_{e}\phi_{f}}\mathcal{L})\big[\mathcal{H}^{-1}\mathcal{H}_{s_{k}}\mathcal{H}^{-1}\big]_{fg}\big[\mathcal{H}^{-1}\big]_{eh}\\
	& -W^{-1}\sum_{d,e,f}(\partial_{\beta^{r}\phi\phi'\phi_{d}\phi_{e}\phi_{f}}\mathcal{L})\big[\mathcal{H}^{-1}\mathcal{H}_{s_{k}}\mathcal{H}^{-1}\big]_{fg}\big[\mathcal{H}^{-1}\big]_{eh}[\mathcal{R}^{1}]_{d}\\
	& -\sum_{e,f}(\partial_{\beta^{r}\phi\phi'\phi_{e}\phi_{f}}\mathcal{L})\big[\mathcal{H}^{-1}\mathcal{H}_{b}\mathcal{H}^{-1}\mathcal{H}_{s_{k}}\mathcal{H}^{-1}\big]_{fg}\big[\mathcal{H}^{-1}\big]_{eh}\\
	& +\sum_{e,f}(\partial_{\beta^{r}\phi\phi'\phi_{e}\phi_{f}}\mathcal{L})\big[\mathcal{H}^{-1}\mathcal{H}_{bs_{k}}\mathcal{H}^{-1}\big]_{fg}\big[\mathcal{H}^{-1}\big]_{eh}\\
	& -\sum_{e,f}(\partial_{\beta^{r}\phi\phi'\phi_{e}\phi_{f}}\mathcal{L})\big[\mathcal{H}^{-1}\mathcal{H}_{s_{k}}\mathcal{H}^{-1}\mathcal{H}_{b}\mathcal{H}^{-1}\big]_{fg}\big[\mathcal{H}^{-1}\big]_{eh}\\
	& -\sum_{e,f}(\partial_{\beta^{r}\phi\phi'\phi_{e}\phi_{f}}\mathcal{L})\big[\mathcal{H}^{-1}\mathcal{H}_{s_{k}}\mathcal{H}^{-1}\big]_{fg}\big[\mathcal{H}^{-1}\mathcal{H}_{b}\mathcal{H}^{-1}\big]_{eh}\\
	& -W^{-1}\sum_{e,f}(\partial_{\beta^{r+1}\phi\phi'\phi_{e}\phi_{f}}\mathcal{L})\big[\mathcal{H}^{-1}\big]_{fg}\big[\mathcal{H}^{-1}\mathcal{H}_{s_{k}}\mathcal{H}^{-1}\big]_{eh}\\
	& -W^{-1}\sum_{d,e,f}(\partial_{\beta^{r}\phi\phi'\phi_{d}\phi_{e}\phi_{f}}\mathcal{L})\big[\mathcal{H}^{-1}\big]_{fg}\big[\mathcal{H}^{-1}\mathcal{H}_{s_{k}}\mathcal{H}^{-1}\big]_{eh}[\mathcal{R}^{1}]_{d}\\
	& -\sum_{e,f}(\partial_{\beta^{r}\phi\phi'\phi_{e}\phi_{f}}\mathcal{L})\big[\mathcal{H}^{-1}\mathcal{H}_{b}\mathcal{H}^{-1}\big]_{fg}\big[\mathcal{H}^{-1}\mathcal{H}_{s_{k}}\mathcal{H}^{-1}\big]_{eh}\\
	& -\sum_{e,f}(\partial_{\beta^{r}\phi\phi'\phi_{e}\phi_{f}}\mathcal{L})\big[\mathcal{H}^{-1}\big]_{fg}\big[\mathcal{H}^{-1}\mathcal{H}_{b}\mathcal{H}^{-1}\mathcal{H}_{s_{k}}\mathcal{H}^{-1}\big]_{eh}\\
	& +\sum_{e,f}(\partial_{\beta^{r}\phi\phi'\phi_{e}\phi_{f}}\mathcal{L})\big[\mathcal{H}^{-1}\big]_{fg}\big[\mathcal{H}^{-1}\mathcal{H}_{bs_{k}}\mathcal{H}^{-1}\big]_{eh}\\
	& -\sum_{e,f}(\partial_{\beta^{r}\phi\phi'\phi_{e}\phi_{f}}\mathcal{L})\big[\mathcal{H}^{-1}\big]_{fg}\big[\mathcal{H}^{-1}\mathcal{H}_{s_{k}}\mathcal{H}^{-1}\mathcal{H}_{b}\mathcal{H}^{-1}\big]_{eh}\\
	& -2W^{-3}W_{b}W_{s_{k}}\sum_{e,f}(\partial_{\beta^{r}\phi\phi'\phi_{e}\phi_{f}}\mathcal{L})\big[\mathcal{H}^{-1}\big]_{fg}[\mathcal{R}^{1}]_{e}[\mathcal{R}^{1}]_{h}\\
	& +W^{-2}W_{bs_{k}}\sum_{e,f}(\partial_{\beta^{r}\phi\phi'\phi_{e}\phi_{f}}\mathcal{L})\big[\mathcal{H}^{-1}\big]_{fg}[\mathcal{R}^{1}]_{e}[\mathcal{R}^{1}]_{h}\\
	& -W^{-3}W_{s_{k}}\sum_{e,f}(\partial_{\beta^{r+1}\phi\phi'\phi_{e}\phi_{f}}\mathcal{L})\big[\mathcal{H}^{-1}\big]_{fg}[\mathcal{R}^{1}]_{e}[\mathcal{R}^{1}]_{h}\\
	& -W^{-3}W_{s_{k}}\sum_{d,e,f}(\partial_{\beta^{r}\phi\phi'\phi_{d}\phi_{e}\phi_{f}}\mathcal{L})\big[\mathcal{H}^{-1}\big]_{fg}[\mathcal{R}^{1}]_{d}[\mathcal{R}^{1}]_{e}[\mathcal{R}^{1}]_{h}\\
	& -W^{-2}W_{s_{k}}\sum_{e,f}(\partial_{\beta^{r}\phi\phi'\phi_{e}\phi_{f}}\mathcal{L})\big[\mathcal{H}^{-1}\mathcal{H}_{b}\mathcal{H}^{-1}\big]_{fg}[\mathcal{R}^{1}]_{e}[\mathcal{R}^{1}]_{h}\\
	& +W^{-2}W_{s_{k}}\sum_{e,f}(\partial_{\beta^{r}\phi\phi'\phi_{e}\phi_{f}}\mathcal{L})\big[\mathcal{H}^{-1}\big]_{fg}[\mathcal{R}_{b}^{1}]_{e}[\mathcal{R}^{1}]_{h}\\
	& +W^{-2}W_{s_{k}}\sum_{e,f}(\partial_{\beta^{r}\phi\phi'\phi_{e}\phi_{f}}\mathcal{L})\big[\mathcal{H}^{-1}\big]_{fg}[\mathcal{R}^{1}]_{e}[\mathcal{R}_{b}^{1}]_{h}\\
	& -2W^{-3}W_{b}\sum_{e,f}(\partial_{\beta^{r+1}\phi\phi'\phi_{e}\phi_{f}}\mathcal{L})\big[\mathcal{H}^{-1}\big]_{fg}[\mathcal{R}^{1}]_{e}[\mathcal{R}^{1}]_{h}[\mathcal{R}^{1}]_{k}\\
	& -W^{-3}\sum_{e,f}(\partial_{\beta^{r+2}\phi\phi'\phi_{e}\phi_{f}}\mathcal{L})\big[\mathcal{H}^{-1}\big]_{fg}[\mathcal{R}^{1}]_{e}[\mathcal{R}^{1}]_{h}[\mathcal{R}^{1}]_{k}\\
	& -W^{-3}\sum_{d,e,f}(\partial_{\beta^{r+1}\phi\phi'\phi_{d}\phi_{e}\phi_{f}}\mathcal{L})\big[\mathcal{H}^{-1}\big]_{fg}[\mathcal{R}^{1}]_{d}[\mathcal{R}^{1}]_{e}[\mathcal{R}^{1}]_{h}[\mathcal{R}^{1}]_{k}\\
	& -W^{-2}\sum_{e,f}(\partial_{\beta^{r+1}\phi\phi'\phi_{e}\phi_{f}}\mathcal{L})\big[\mathcal{H}^{-1}\mathcal{H}_{b}\mathcal{H}^{-1}\big]_{fg}[\mathcal{R}^{1}]_{e}[\mathcal{R}^{1}]_{h}[\mathcal{R}^{1}]_{k}\\
	& +W^{-2}\sum_{e,f}(\partial_{\beta^{r+1}\phi\phi'\phi_{e}\phi_{f}}\mathcal{L})\big[\mathcal{H}^{-1}\big]_{fg}[\mathcal{R}_{b}^{1}]_{e}[\mathcal{R}^{1}]_{h}[\mathcal{R}^{1}]_{k}\\
	& +W^{-2}\sum_{e,f}(\partial_{\beta^{r+1}\phi\phi'\phi_{e}\phi_{f}}\mathcal{L})\big[\mathcal{H}^{-1}\big]_{fg}[\mathcal{R}^{1}]_{e}[\mathcal{R}_{b}^{1}]_{h}[\mathcal{R}^{1}]_{k}\\
	& +W^{-2}\sum_{e,f}(\partial_{\beta^{r+1}\phi\phi'\phi_{e}\phi_{f}}\mathcal{L})\big[\mathcal{H}^{-1}\big]_{fg}[\mathcal{R}^{1}]_{e}[\mathcal{R}^{1}]_{h}[\mathcal{R}_{b}^{1}]_{k}\\
	& -W^{-2}W_{b}\sum_{d,e,f}(\partial_{\beta^{r}\phi\phi'\phi_{d}\phi_{e}\phi_{f}}\mathcal{L})\big[\mathcal{H}^{-1}\big]_{fg}\big[\mathcal{H}^{-1}\big]_{dk}[\mathcal{R}^{1}]_{e}[\mathcal{R}^{1}]_{h}\\
	& -W^{-2}\sum_{d,e,f}(\partial_{\beta^{r+1}\phi\phi'\phi_{d}\phi_{e}\phi_{f}}\mathcal{L})\big[\mathcal{H}^{-1}\big]_{fg}\big[\mathcal{H}^{-1}\big]_{dk}[\mathcal{R}^{1}]_{e}[\mathcal{R}^{1}]_{h}\\
	& -W^{-2}\sum_{c,d,e,f}(\partial_{\beta^{r}\phi\phi'\phi_{c}\phi_{d}\phi_{e}\phi_{f}}\mathcal{L})\big[\mathcal{H}^{-1}\big]_{fg}\big[\mathcal{H}^{-1}\big]_{dk}[\mathcal{R}^{1}]_{c}[\mathcal{R}^{1}]_{e}[\mathcal{R}^{1}]_{h}\\
	& -W^{-1}\sum_{d,e,f}(\partial_{\beta^{r}\phi\phi'\phi_{d}\phi_{e}\phi_{f}}\mathcal{L})\big[\mathcal{H}^{-1}\mathcal{H}_{b}\mathcal{H}^{-1}\big]_{fg}\big[\mathcal{H}^{-1}\big]_{dk}[\mathcal{R}^{1}]_{e}[\mathcal{R}^{1}]_{h}\\
	& -W^{-1}\sum_{d,e,f}(\partial_{\beta^{r}\phi\phi'\phi_{d}\phi_{e}\phi_{f}}\mathcal{L})\big[\mathcal{H}^{-1}\big]_{fg}\big[\mathcal{H}^{-1}\mathcal{H}_{b}\mathcal{H}^{-1}\big]_{dk}[\mathcal{R}^{1}]_{e}[\mathcal{R}^{1}]_{h}\\
	& +W^{-1}\sum_{d,e,f}(\partial_{\beta^{r}\phi\phi'\phi_{d}\phi_{e}\phi_{f}}\mathcal{L})\big[\mathcal{H}^{-1}\big]_{fg}\big[\mathcal{H}^{-1}\big]_{dk}[\mathcal{R}_{b}^{1}]_{e}[\mathcal{R}^{1}]_{h}\\
	& +W^{-1}\sum_{d,e,f}(\partial_{\beta^{r}\phi\phi'\phi_{d}\phi_{e}\phi_{f}}\mathcal{L})\big[\mathcal{H}^{-1}\big]_{fg}\big[\mathcal{H}^{-1}\big]_{dk}[\mathcal{R}^{1}]_{e}[\mathcal{R}_{b}^{1}]_{h}\\
	& -2W^{-3}W_{b}\sum_{d,e,f}(\partial_{\beta^{r}\phi\phi'\phi_{d}\phi_{e}\phi_{f}}\mathcal{L})\big[\mathcal{H}^{-1}\big]_{fg}[\mathcal{R}^{1}]_{e}[\mathcal{R}^{1}]_{h}[\mathcal{R}^{1}]_{d}[\mathcal{R}^{1}]_{k}\\
	& -W^{-2}\sum_{d,e,f}(\partial_{\beta^{r+1}\phi\phi'\phi_{d}\phi_{e}\phi_{f}}\mathcal{L})\big[\mathcal{H}^{-1}\big]_{fg}[\mathcal{R}^{1}]_{e}[\mathcal{R}^{1}]_{h}[\mathcal{R}^{1}]_{d}[\mathcal{R}^{1}]_{k}\\
	& -W^{-2}\sum_{c,d,e,f}(\partial_{\beta^{r}\phi\phi'\phi_{c}\phi_{d}\phi_{e}\phi_{f}}\mathcal{L})\big[\mathcal{H}^{-1}\big]_{fg}[\mathcal{R}^{1}]_{c}[\mathcal{R}^{1}]_{e}[\mathcal{R}^{1}]_{h}[\mathcal{R}^{1}]_{d}[\mathcal{R}^{1}]_{k}\\
	& -W^{-2}\sum_{d,e,f}(\partial_{\beta^{r}\phi\phi'\phi_{d}\phi_{e}\phi_{f}}\mathcal{L})\big[\mathcal{H}^{-1}\mathcal{H}_{b}\mathcal{H}^{-1}\big]_{fg}[\mathcal{R}^{1}]_{e}[\mathcal{R}^{1}]_{h}[\mathcal{R}^{1}]_{d}[\mathcal{R}^{1}]_{k}\\
	& +W^{-2}\sum_{d,e,f}(\partial_{\beta^{r}\phi\phi'\phi_{d}\phi_{e}\phi_{f}}\mathcal{L})\big[\mathcal{H}^{-1}\big]_{fg}[\mathcal{R}_{b}^{1}]_{e}[\mathcal{R}^{1}]_{h}[\mathcal{R}^{1}]_{d}[\mathcal{R}^{1}]_{k}\\
	& +W^{-2}\sum_{d,e,f}(\partial_{\beta^{r}\phi\phi'\phi_{d}\phi_{e}\phi_{f}}\mathcal{L})\big[\mathcal{H}^{-1}\big]_{fg}[\mathcal{R}^{1}]_{e}[\mathcal{R}_{b}^{1}]_{h}[\mathcal{R}^{1}]_{d}[\mathcal{R}^{1}]_{k}\\
	& +W^{-2}\sum_{d,e,f}(\partial_{\beta^{r}\phi\phi'\phi_{d}\phi_{e}\phi_{f}}\mathcal{L})\big[\mathcal{H}^{-1}\big]_{fg}[\mathcal{R}^{1}]_{e}[\mathcal{R}^{1}]_{h}[\mathcal{R}_{b}^{1}]_{d}[\mathcal{R}^{1}]_{k}\\
	& +W^{-2}\sum_{d,e,f}(\partial_{\beta^{r}\phi\phi'\phi_{d}\phi_{e}\phi_{f}}\mathcal{L})\big[\mathcal{H}^{-1}\big]_{fg}[\mathcal{R}^{1}]_{e}[\mathcal{R}^{1}]_{h}[\mathcal{R}^{1}]_{d}[\mathcal{R}_{b}^{1}]_{k}\\
	& -W^{-2}W_{b}\sum_{e,f}(\partial_{\beta^{r}\phi\phi'\phi_{e}\phi_{f}}\mathcal{L})\big[\mathcal{H}^{-1}\mathcal{H}_{s_{k}}\mathcal{H}^{-1}\big]_{fg}[\mathcal{R}^{1}]_{e}[\mathcal{R}^{1}]_{h}\\
	& -W^{-2}\sum_{e,f}(\partial_{\beta^{r+1}\phi\phi'\phi_{e}\phi_{f}}\mathcal{L})\big[\mathcal{H}^{-1}\mathcal{H}_{s_{k}}\mathcal{H}^{-1}\big]_{fg}[\mathcal{R}^{1}]_{e}[\mathcal{R}^{1}]_{h}\\
	& -W^{-2}\sum_{d,e,f}(\partial_{\beta^{r}\phi\phi'\phi_{d}\phi_{e}\phi_{f}}\mathcal{L})\big[\mathcal{H}^{-1}\mathcal{H}_{s_{k}}\mathcal{H}^{-1}\big]_{fg}[\mathcal{R}^{1}]_{d}[\mathcal{R}^{1}]_{e}[\mathcal{R}^{1}]_{h}\\
	& -W^{-1}\sum_{e,f}(\partial_{\beta^{r}\phi\phi'\phi_{e}\phi_{f}}\mathcal{L})\big[\mathcal{H}^{-1}\mathcal{H}_{b}\mathcal{H}^{-1}\mathcal{H}_{s_{k}}\mathcal{H}^{-1}\big]_{fg}[\mathcal{R}^{1}]_{e}[\mathcal{R}^{1}]_{h}\\
	& +W^{-1}\sum_{e,f}(\partial_{\beta^{r}\phi\phi'\phi_{e}\phi_{f}}\mathcal{L})\big[\mathcal{H}^{-1}\mathcal{H}_{bs_{k}}\mathcal{H}^{-1}\big]_{fg}[\mathcal{R}^{1}]_{e}[\mathcal{R}^{1}]_{h}\\
	& -W^{-1}\sum_{e,f}(\partial_{\beta^{r}\phi\phi'\phi_{e}\phi_{f}}\mathcal{L})\big[\mathcal{H}^{-1}\mathcal{H}_{s_{k}}\mathcal{H}^{-1}\mathcal{H}_{b}\mathcal{H}^{-1}\big]_{fg}[\mathcal{R}^{1}]_{e}[\mathcal{R}^{1}]_{h}\\
	& +W^{-1}\sum_{e,f}(\partial_{\beta^{r}\phi\phi'\phi_{e}\phi_{f}}\mathcal{L})\big[\mathcal{H}^{-1}\mathcal{H}_{s_{k}}\mathcal{H}^{-1}\big]_{fg}[\mathcal{R}_{b}^{1}]_{e}[\mathcal{R}^{1}]_{h}\\
	& +W^{-1}\sum_{e,f}(\partial_{\beta^{r}\phi\phi'\phi_{e}\phi_{f}}\mathcal{L})\big[\mathcal{H}^{-1}\mathcal{H}_{s_{k}}\mathcal{H}^{-1}\big]_{fg}[\mathcal{R}^{1}]_{e}[\mathcal{R}_{b}^{1}]_{h}\\
	& +W^{-2}W_{b}\sum_{e,f}(\partial_{\beta^{r}\phi\phi'\phi_{e}\phi_{f}}\mathcal{L})\big[\mathcal{H}^{-1}\big]_{fg}[\mathcal{R}_{s_{k}}^{1}]_{e}[\mathcal{R}^{1}]_{h}\\
	& +W^{-2}\sum_{e,f}(\partial_{\beta^{r+1}\phi\phi'\phi_{e}\phi_{f}}\mathcal{L})\big[\mathcal{H}^{-1}\big]_{fg}[\mathcal{R}_{s_{k}}^{1}]_{e}[\mathcal{R}^{1}]_{h}\\
	& +W^{-2}\sum_{d,e,f}(\partial_{\beta^{r}\phi\phi'\phi_{d}\phi_{e}\phi_{f}}\mathcal{L})\big[\mathcal{H}^{-1}\big]_{fg}[\mathcal{R}_{s_{k}}^{1}]_{e}[\mathcal{R}^{1}]_{d}[\mathcal{R}^{1}]_{h}\\
	& +W^{-1}\sum_{e,f}(\partial_{\beta^{r}\phi\phi'\phi_{e}\phi_{f}}\mathcal{L})\big[\mathcal{H}^{-1}\mathcal{H}_{b}\mathcal{H}^{-1}\big]_{fg}[\mathcal{R}_{s_{k}}^{1}]_{e}[\mathcal{R}^{1}]_{h}\\
	& -W^{-1}\sum_{e,f}(\partial_{\beta^{r}\phi\phi'\phi_{e}\phi_{f}}\mathcal{L})\big[\mathcal{H}^{-1}\big]_{fg}[\mathcal{R}_{bs_{k}}^{1}]_{e}[\mathcal{R}^{1}]_{h}\\
	& -W^{-1}\sum_{e,f}(\partial_{\beta^{r}\phi\phi'\phi_{e}\phi_{f}}\mathcal{L})\big[\mathcal{H}^{-1}\big]_{fg}[\mathcal{R}_{s_{k}}^{1}]_{e}[\mathcal{R}_{b}^{1}]_{h}\\
	& +W^{-2}W_{b}\sum_{e,f}(\partial_{\beta^{r}\phi\phi'\phi_{e}\phi_{f}}\mathcal{L})\big[\mathcal{H}^{-1}\big]_{fg}[\mathcal{R}^{1}]_{e}[\mathcal{R}_{s_{k}}^{1}]_{h}\\
	& +W^{-2}\sum_{e,f}(\partial_{\beta^{r+1}\phi\phi'\phi_{e}\phi_{f}}\mathcal{L})\big[\mathcal{H}^{-1}\big]_{fg}[\mathcal{R}^{1}]_{e}[\mathcal{R}_{s_{k}}^{1}]_{h}\\
	& +W^{-2}\sum_{d,e,f}(\partial_{\beta^{r}\phi\phi'\phi_{d}\phi_{e}\phi_{f}}\mathcal{L})\big[\mathcal{H}^{-1}\big]_{fg}[\mathcal{R}^{1}]_{d}[\mathcal{R}^{1}]_{e}[\mathcal{R}_{s_{k}}^{1}]_{h}\\
	& +W^{-1}\sum_{e,f}(\partial_{\beta^{r}\phi\phi'\phi_{e}\phi_{f}}\mathcal{L})\big[\mathcal{H}^{-1}\mathcal{H}_{b}\mathcal{H}^{-1}\big]_{fg}[\mathcal{R}^{1}]_{e}[\mathcal{R}_{s_{k}}^{1}]_{h}\\
	& -W^{-1}\sum_{e,f}(\partial_{\beta^{r}\phi\phi'\phi_{e}\phi_{f}}\mathcal{L})\big[\mathcal{H}^{-1}\big]_{fg}[\mathcal{R}_{b}^{1}]_{e}[\mathcal{R}_{s_{k}}^{1}]_{h}\\
	& -W^{-1}\sum_{e,f}(\partial_{\beta^{r}\phi\phi'\phi_{e}\phi_{f}}\mathcal{L})\big[\mathcal{H}^{-1}\big]_{fg}[\mathcal{R}^{1}]_{e}[\mathcal{R}_{bs_{k}}^{1}]_{h}\\
	& -W^{-2}W_{b}\sum_{f}(\partial_{\beta^{r+1}\phi\phi'\phi_{f}}\mathcal{L})\big[\mathcal{H}^{-1}\mathcal{H}_{s_{h}}\mathcal{H}^{-1}\big]_{fg}[\mathcal{R}^{1}]_{k}\\
	& -W^{-2}\sum_{f}(\partial_{\beta^{r+2}\phi\phi'\phi_{f}}\mathcal{L})\big[\mathcal{H}^{-1}\mathcal{H}_{s_{h}}\mathcal{H}^{-1}\big]_{fg}[\mathcal{R}^{1}]_{k}\\
	& -W^{-2}\sum_{e,f}(\partial_{\beta^{r+1}\phi\phi'\phi_{e}\phi_{f}}\mathcal{L})\big[\mathcal{H}^{-1}\mathcal{H}_{s_{h}}\mathcal{H}^{-1}\big]_{fg}[\mathcal{R}^{1}]_{e}[\mathcal{R}^{1}]_{k}\\
	& -W^{-1}\sum_{f}(\partial_{\beta^{r+1}\phi\phi'\phi_{f}}\mathcal{L})\big[\mathcal{H}^{-1}\mathcal{H}_{b}\mathcal{H}^{-1}\mathcal{H}_{s_{h}}\mathcal{H}^{-1}\big]_{fg}[\mathcal{R}^{1}]_{k}\\
	& +W^{-1}\sum_{f}(\partial_{\beta^{r+1}\phi\phi'\phi_{f}}\mathcal{L})\big[\mathcal{H}^{-1}\mathcal{H}_{bs_{h}}\mathcal{H}^{-1}\big]_{fg}[\mathcal{R}^{1}]_{k}\\
	& -W^{-1}\sum_{f}(\partial_{\beta^{r+1}\phi\phi'\phi_{f}}\mathcal{L})\big[\mathcal{H}^{-1}\mathcal{H}_{s_{h}}\mathcal{H}^{-1}\mathcal{H}_{b}\mathcal{H}^{-1}\big]_{fg}[\mathcal{R}^{1}]_{k}\\
	& +W^{-1}\sum_{f}(\partial_{\beta^{r+1}\phi\phi'\phi_{f}}\mathcal{L})\big[\mathcal{H}^{-1}\mathcal{H}_{s_{h}}\mathcal{H}^{-1}\big]_{fg}[\mathcal{R}_{b}^{1}]_{k}\\
	& -W^{-1}\sum_{e,f}(\partial_{\beta^{r+1}\phi\phi'\phi_{e}\phi_{f}}\mathcal{L})\big[\mathcal{H}^{-1}\mathcal{H}_{s_{h}}\mathcal{H}^{-1}\big]_{fg}\big[\mathcal{H}^{-1}\big]_{ek}\\
	& -W^{-1}\sum_{d,e,f}(\partial_{\beta^{r}\phi\phi'\phi_{d}\phi_{e}\phi_{f}}\mathcal{L})\big[\mathcal{H}^{-1}\mathcal{H}_{s_{h}}\mathcal{H}^{-1}\big]_{fg}\big[\mathcal{H}^{-1}\big]_{ek}[\mathcal{R}^{1}]_{d}\\
	& -\sum_{e,f}(\partial_{\beta^{r}\phi\phi'\phi_{e}\phi_{f}}\mathcal{L})\big[\mathcal{H}^{-1}\mathcal{H}_{b}\mathcal{H}^{-1}\mathcal{H}_{s_{h}}\mathcal{H}^{-1}\big]_{fg}\big[\mathcal{H}^{-1}\big]_{ek}\\
	& +\sum_{e,f}(\partial_{\beta^{r}\phi\phi'\phi_{e}\phi_{f}}\mathcal{L})\big[\mathcal{H}^{-1}\mathcal{H}_{bs_{h}}\mathcal{H}^{-1}\big]_{fg}\big[\mathcal{H}^{-1}\big]_{ek}\\
	& -\sum_{e,f}(\partial_{\beta^{r}\phi\phi'\phi_{e}\phi_{f}}\mathcal{L})\big[\mathcal{H}^{-1}\mathcal{H}_{s_{h}}\mathcal{H}^{-1}\mathcal{H}_{b}\mathcal{H}^{-1}\big]_{fg}\big[\mathcal{H}^{-1}\big]_{ek}\\
	& -\sum_{e,f}(\partial_{\beta^{r}\phi\phi'\phi_{e}\phi_{f}}\mathcal{L})\big[\mathcal{H}^{-1}\mathcal{H}_{s_{h}}\mathcal{H}^{-1}\big]_{fg}\big[\mathcal{H}^{-1}\mathcal{H}_{b}\mathcal{H}^{-1}\big]_{ek}\\
	& -W^{-2}W_{b}\sum_{e,f}(\partial_{\beta^{r}\phi\phi'\phi_{e}\phi_{f}}\mathcal{L})\big[\mathcal{H}^{-1}\mathcal{H}_{s_{h}}\mathcal{H}^{-1}\big]_{fg}[\mathcal{R}^{1}]_{e}[\mathcal{R}^{1}]_{k}\\
	& -W^{-2}\sum_{e,f}(\partial_{\beta^{r+1}\phi\phi'\phi_{e}\phi_{f}}\mathcal{L})\big[\mathcal{H}^{-1}\mathcal{H}_{s_{h}}\mathcal{H}^{-1}\big]_{fg}[\mathcal{R}^{1}]_{e}[\mathcal{R}^{1}]_{k}\\
	& -W^{-2}\sum_{d,e,f}(\partial_{\beta^{r}\phi\phi'\phi_{d}\phi_{e}\phi_{f}}\mathcal{L})\big[\mathcal{H}^{-1}\mathcal{H}_{s_{h}}\mathcal{H}^{-1}\big]_{fg}[\mathcal{R}^{1}]_{d}[\mathcal{R}^{1}]_{e}[\mathcal{R}^{1}]_{k}\\
	& -W^{-1}\sum_{e,f}(\partial_{\beta^{r}\phi\phi'\phi_{e}\phi_{f}}\mathcal{L})\big[\mathcal{H}^{-1}\mathcal{H}_{b}\mathcal{H}^{-1}\mathcal{H}_{s_{h}}\mathcal{H}^{-1}\big]_{fg}[\mathcal{R}^{1}]_{e}[\mathcal{R}^{1}]_{k}\\
	& +W^{-1}\sum_{e,f}(\partial_{\beta^{r}\phi\phi'\phi_{e}\phi_{f}}\mathcal{L})\big[\mathcal{H}^{-1}\mathcal{H}_{bs_{h}}\mathcal{H}^{-1}\big]_{fg}[\mathcal{R}^{1}]_{e}[\mathcal{R}^{1}]_{k}\\
	& -W^{-1}\sum_{e,f}(\partial_{\beta^{r}\phi\phi'\phi_{e}\phi_{f}}\mathcal{L})\big[\mathcal{H}^{-1}\mathcal{H}_{s_{h}}\mathcal{H}^{-1}\mathcal{H}_{b}\mathcal{H}^{-1}\big]_{fg}[\mathcal{R}^{1}]_{e}[\mathcal{R}^{1}]_{k}\\
	& +W^{-1}\sum_{e,f}(\partial_{\beta^{r}\phi\phi'\phi_{e}\phi_{f}}\mathcal{L})\big[\mathcal{H}^{-1}\mathcal{H}_{s_{h}}\mathcal{H}^{-1}\big]_{fg}[\mathcal{R}_{b}^{1}]_{e}[\mathcal{R}^{1}]_{k}\\
	& +W^{-1}\sum_{e,f}(\partial_{\beta^{r}\phi\phi'\phi_{e}\phi_{f}}\mathcal{L})\big[\mathcal{H}^{-1}\mathcal{H}_{s_{h}}\mathcal{H}^{-1}\big]_{fg}[\mathcal{R}^{1}]_{e}[\mathcal{R}_{b}^{1}]_{k}\\
	& -W^{-1}\sum_{f}(\partial_{\beta^{r+1}\phi\phi'\phi_{f}}\mathcal{L})\big[\mathcal{H}^{-1}\mathcal{H}_{s_{k}}\mathcal{H}^{-1}\mathcal{H}_{s_{h}}\mathcal{H}^{-1}\big]_{fg}\\
	& -W^{-1}\sum_{e,f}(\partial_{\beta^{r}\phi\phi'\phi_{e}\phi_{f}}\mathcal{L})\big[\mathcal{H}^{-1}\mathcal{H}_{s_{k}}\mathcal{H}^{-1}\mathcal{H}_{s_{h}}\mathcal{H}^{-1}\big]_{fg}[\mathcal{R}^{1}]_{e}\\
	& -\sum_{f}(\partial_{\beta^{r}\phi\phi'\phi_{f}}\mathcal{L})\big[\mathcal{H}^{-1}\mathcal{H}_{b}\mathcal{H}^{-1}\mathcal{H}_{s_{k}}\mathcal{H}^{-1}\mathcal{H}_{s_{h}}\mathcal{H}^{-1}\big]_{fg}\\
	& +\sum_{f}(\partial_{\beta^{r}\phi\phi'\phi_{f}}\mathcal{L})\big[\mathcal{H}^{-1}\mathcal{H}_{bs_{k}}\mathcal{H}^{-1}\mathcal{H}_{s_{h}}\mathcal{H}^{-1}\big]_{fg}\\
	& -\sum_{f}(\partial_{\beta^{r}\phi\phi'\phi_{f}}\mathcal{L})\big[\mathcal{H}^{-1}\mathcal{H}_{s_{k}}\mathcal{H}^{-1}\mathcal{H}_{b}\mathcal{H}^{-1}\mathcal{H}_{s_{h}}\mathcal{H}^{-1}\big]_{fg}\\
	& +\sum_{f}(\partial_{\beta^{r}\phi\phi'\phi_{f}}\mathcal{L})\big[\mathcal{H}^{-1}\mathcal{H}_{s_{k}}\mathcal{H}^{-1}\mathcal{H}_{bs_{h}}\mathcal{H}^{-1}\big]_{fg}\\
	& -\sum_{f}(\partial_{\beta^{r}\phi\phi'\phi_{f}}\mathcal{L})\big[\mathcal{H}^{-1}\mathcal{H}_{s_{k}}\mathcal{H}^{-1}\mathcal{H}_{s_{h}}\mathcal{H}^{-1}\mathcal{H}_{b}\mathcal{H}^{-1}\big]_{fg}\\
	& +W^{-1}\sum_{f}(\partial_{\beta^{r+1}\phi\phi'\phi_{f}}\mathcal{L})\big[\mathcal{H}^{-1}\mathcal{H}_{s_{k}s_{h}}\mathcal{H}^{-1}\big]_{fg}\\
	& +W^{-1}\sum_{e,f}(\partial_{\beta^{r}\phi\phi'\phi_{e}\phi_{f}}\mathcal{L})\big[\mathcal{H}^{-1}\mathcal{H}_{s_{k}s_{h}}\mathcal{H}^{-1}\big]_{fg}[\mathcal{R}^{1}]_{e}\\
	& +\sum_{f}(\partial_{\beta^{r}\phi\phi'\phi_{f}}\mathcal{L})\big[\mathcal{H}^{-1}\mathcal{H}_{b}\mathcal{H}^{-1}\mathcal{H}_{s_{k}s_{h}}\mathcal{H}^{-1}\big]_{fg}\\
	& -\sum_{f}(\partial_{\beta^{r}\phi\phi'\phi_{f}}\mathcal{L})\big[\mathcal{H}^{-1}\mathcal{H}_{bs_{k}s_{h}}\mathcal{H}^{-1}\big]_{fg}\\
	& +\sum_{f}(\partial_{\beta^{r}\phi\phi'\phi_{f}}\mathcal{L})\big[\mathcal{H}^{-1}\mathcal{H}_{s_{k}s_{h}}\mathcal{H}^{-1}\mathcal{H}_{b}\mathcal{H}^{-1}\big]_{fg}\\
	& -W^{-1}\sum_{f}(\partial_{\beta^{r+1}\phi\phi'\phi_{f}}\mathcal{L})\big[\mathcal{H}^{-1}\mathcal{H}_{s_{k}}\mathcal{H}^{-1}\mathcal{H}_{s_{h}}\mathcal{H}^{-1}\big]_{fg}\\
	& -W^{-1}\sum_{e,f}(\partial_{\beta^{r}\phi\phi'\phi_{e}\phi_{f}}\mathcal{L})\big[\mathcal{H}^{-1}\mathcal{H}_{s_{h}}\mathcal{H}^{-1}\mathcal{H}_{s_{k}}\mathcal{H}^{-1}\big]_{fg}[\mathcal{R}^{1}]_{e}\\
	& -\sum_{f}(\partial_{\beta^{r}\phi\phi'\phi_{f}}\mathcal{L})\big[\mathcal{H}^{-1}\mathcal{H}_{b}\mathcal{H}^{-1}\mathcal{H}_{s_{h}}\mathcal{H}^{-1}\mathcal{H}_{s_{k}}\mathcal{H}^{-1}\big]_{fg}\\
	& +\sum_{f}(\partial_{\beta^{r}\phi\phi'\phi_{f}}\mathcal{L})\big[\mathcal{H}^{-1}\mathcal{H}_{bs_{h}}\mathcal{H}^{-1}\mathcal{H}_{s_{k}}\mathcal{H}^{-1}\big]_{fg}\\
	& -\sum_{f}(\partial_{\beta^{r}\phi\phi'\phi_{f}}\mathcal{L})\big[\mathcal{H}^{-1}\mathcal{H}_{s_{h}}\mathcal{H}^{-1}\mathcal{H}_{b}\mathcal{H}^{-1}\mathcal{H}_{s_{k}}\mathcal{H}^{-1}\big]_{fg}\\
	& +\sum_{f}(\partial_{\beta^{r}\phi\phi'\phi_{f}}\mathcal{L})\big[\mathcal{H}^{-1}\mathcal{H}_{s_{h}}\mathcal{H}^{-1}\mathcal{H}_{bs_{k}}\mathcal{H}^{-1}\big]_{fg}\\
	& -\sum_{f}(\partial_{\beta^{r}\phi\phi'\phi_{f}}\mathcal{L})\big[\mathcal{H}^{-1}\mathcal{H}_{s_{h}}\mathcal{H}^{-1}\mathcal{H}_{s_{k}}\mathcal{H}^{-1}\mathcal{H}_{b}\mathcal{H}^{-1}\big]_{fg}
\end{align*}

\subsection{Expressions for $\mathcal{R}$ terms}

First derivatives:

\begin{align*}
	\mathcal{R}_{s_{h}}^{a} & =-W^{-1}[\mathcal{R}^{1}]_{h}\mathcal{R}^{a+1}\\
	& -[\mathcal{H}^{-1}]_{h,\cdot}(\partial_{\beta^{a}\phi\phi'}\mathcal{L})\mathcal{H}^{-1}\\
	& -W^{-1}[\mathcal{R}^{1}]_{h}\mathcal{R}^{1}(\partial_{\beta^{a}\phi\phi'}\mathcal{L})\mathcal{H}^{-1}\\
	& -\mathcal{R}^{a}\mathcal{H}_{s_{h}}\mathcal{H}^{-1}
\end{align*}

\begin{align*}
	\mathcal{R}_{b}^{a} & =-W^{-1}\mathcal{R}^{a+1}\\
	& -W^{-1}\mathcal{R}_{}^{1}{}'(\partial_{\beta^{a}\phi\phi'}\mathcal{L})\mathcal{H}^{-1}\\
	& -\mathcal{R}^{a+1}{}'\mathcal{H}_{b}\mathcal{H}^{-1}
\end{align*}

Second derivatives:

\begin{align*}
	\mathcal{R}_{bs_{h}}^{a} & =W^{-2}W_{b}[\mathcal{R}^{1}]_{h}\mathcal{R}^{a+1}-W^{-1}[\mathcal{R}_{b}^{1}]_{h}\mathcal{R}^{a+1}\\
	& +W^{-2}[\mathcal{R}^{1}]_{h}\mathcal{R}_{b}^{a+1}\\
	& +[\mathcal{H}^{-1}\mathcal{H}_{b}\mathcal{H}^{-1}]_{h,\cdot}(\partial_{\beta^{a}\phi\phi'}\mathcal{L})\mathcal{H}^{-1}\\
	& +W^{-1}[\mathcal{H}^{-1}]_{h,\cdot}(\partial_{\beta^{a+1}\phi\phi'}\mathcal{L})\mathcal{H}^{-1}+W^{-1}[\mathcal{H}^{-1}]_{h,\cdot}\mathcal{P}^{a+1,1}\mathcal{H}^{-1}\\
	& +[\mathcal{H}^{-1}]_{h,\cdot}(\partial_{\beta^{a}\phi\phi'}\mathcal{L})\mathcal{H}^{-1}\mathcal{H}_{b}\mathcal{H}^{-1}\\
	& +W^{-2}W_{b}[\mathcal{R}^{1}]_{h}\mathcal{R}^{1}(\partial_{\beta^{a}\phi\phi'}\mathcal{L})\mathcal{H}^{-1}-W^{-1}[\mathcal{R}_{b}^{1}]_{h}\mathcal{R}^{1}(\partial_{\beta^{a}\phi\phi'}\mathcal{L})\mathcal{H}^{-1}\\
	& -W^{-1}[\mathcal{R}^{1}]_{h}\mathcal{R}_{b}^{1}(\partial_{\beta^{a}\phi\phi'}\mathcal{L})\mathcal{H}^{-1}\\
	& +W^{-2}[\mathcal{R}^{1}]_{h}\mathcal{R}^{1}(\partial_{\beta^{a+1}\phi\phi'}\mathcal{L})\mathcal{H}^{-1}+W^{-2}[\mathcal{R}^{1}]_{h}\mathcal{R}^{1}\mathcal{P}^{a+1,1}\mathcal{H}^{-1}\\
	& +W^{-1}[\mathcal{R}^{1}]_{h}\mathcal{R}^{1}(\partial_{\beta^{a}\phi\phi'}\mathcal{L})\mathcal{H}^{-1}\mathcal{H}_{b}\mathcal{H}^{-1}\\
	& -\mathcal{R}_{b}^{a}\mathcal{H}_{s_{h}}\mathcal{H}^{-1}-\mathcal{R}^{a}\mathcal{H}_{bs_{h}}\mathcal{H}^{-1}\\
	& +\mathcal{R}^{a}\mathcal{H}_{s_{h}}\mathcal{H}^{-1}\mathcal{H}_{b}\mathcal{H}^{-1}
\end{align*}

\begin{align*}
	\mathcal{R}_{s_{g}s_{h}}^{a} & =W^{-2}W_{s_{g}}[\mathcal{R}^{1}]_{h}\mathcal{R}^{a+1}\\
	& -W^{-1}[\mathcal{R}_{s_{g}}^{1}]_{h}\mathcal{R}^{a+1}\\
	& -W^{-1}[\mathcal{R}^{1}]_{h}\mathcal{R}_{s_{g}}^{a+1}\\
	& +[\mathcal{H}^{-1}\mathcal{H}_{s_{g}}^{-1}\mathcal{H}^{-1}]_{h,\cdot}(\partial_{\beta^{a}\phi\phi'}\mathcal{L})\mathcal{H}^{-1}\\
	& +W^{-1}[\mathcal{H}^{-1}]_{h,\cdot}(\partial_{\beta^{a+1}\phi\phi'}\mathcal{L})\mathcal{H}^{-1}[\mathcal{R}^{1}]_{g}\\
	& +[\mathcal{H}^{-1}]_{h,\cdot}\sum_{e}(\partial_{\beta^{a}\phi\phi'\phi_{e}}\mathcal{L})\mathcal{H}^{-1}[\mathcal{H}^{-1}]_{eg}\\
	& +W^{-1}[\mathcal{H}^{-1}]_{h,\cdot}\sum_{e}(\partial_{\beta^{a}\phi\phi'\phi_{e}}\mathcal{L})\mathcal{H}^{-1}[\mathcal{R}^{1}]_{e}[\mathcal{R}^{1}]_{g}\\
	& +[\mathcal{H}^{-1}]_{h,\cdot}(\partial_{\beta^{a}\phi\phi'}\mathcal{L})\mathcal{H}^{-1}\mathcal{H}_{s_{g}}^{-1}\mathcal{H}^{-1}\\
	& +W^{-2}W_{s_{g}}[\mathcal{R}^{1}]_{h}\mathcal{R}^{1}(\partial_{\beta^{a}\phi\phi'}\mathcal{L})\mathcal{H}^{-1}\\
	& -W^{-1}[\mathcal{R}_{s_{g}}^{1}]_{h}\mathcal{R}^{1}(\partial_{\beta^{a}\phi\phi'}\mathcal{L})\mathcal{H}^{-1}\\
	& -W^{-1}[\mathcal{R}^{1}]_{h}\mathcal{R}_{s_{g}}^{1}(\partial_{\beta^{a}\phi\phi'}\mathcal{L})\mathcal{H}^{-1}\\
	& +W^{-2}[\mathcal{R}^{1}]_{h}\mathcal{R}^{1}(\partial_{\beta^{a+1}\phi\phi'}\mathcal{L})\mathcal{H}^{-1}[\mathcal{R}^{1}]_{g}\\
	& +W^{-1}[\mathcal{R}^{1}]_{h}\mathcal{R}^{1}\sum_{e}(\partial_{\beta^{a}\phi\phi'\phi_{e}}\mathcal{L})\mathcal{H}^{-1}[\mathcal{H}^{-1}]_{eg}\\
	& +W^{-2}[\mathcal{R}^{1}]_{h}\mathcal{R}^{1}\sum_{e}(\partial_{\beta^{a}\phi\phi'\phi_{e}}\mathcal{L})\mathcal{H}^{-1}[\mathcal{R}^{1}]_{e}[\mathcal{R}^{1}]_{g}\\
	& +W^{-1}[\mathcal{R}^{1}]_{h}\mathcal{R}^{1}(\partial_{\beta^{a}\phi\phi'}\mathcal{L})\mathcal{H}^{-1}\mathcal{H}_{s_{g}}^ {}\mathcal{H}^{-1}\\
	& -\mathcal{R}_{s_{g}}^{a}\mathcal{H}_{s_{h}}\mathcal{H}^{-1}\\
	& -\mathcal{R}^{a}\mathcal{H}_{s_{g}s_{h}}\mathcal{H}^{-1}\\
	& +\mathcal{R}^{a}\mathcal{H}_{s_{h}}\mathcal{H}^{-1}\mathcal{H}_{s_{g}}^{-1}\mathcal{H}^{-1}
\end{align*}

Third derivatives:

\begin{align*}
	\mathcal{R}_{bs_{g}s_{h}}^{a} & =-2W^{-3}W_{b}W_{s_{g}}[\mathcal{R}^{1}]_{h}\mathcal{R}^{a+1}+W^{-2}W_{bs_{g}}[\mathcal{R}^{1}]_{h}\mathcal{R}^{a+1}\\
	& +W^{-2}W_{s_{g}}[\mathcal{R}_{b}^{1}]_{h}\mathcal{R}^{a+1}+W^{-2}W_{s_{g}}[\mathcal{R}^{1}]_{h}\mathcal{R}_{b}^{a+1}\\
	& +W^{-2}W_{b}[\mathcal{R}_{s_{g}}^{1}]_{h}\mathcal{R}^{a+1}-W^{-1}[\mathcal{R}_{bs_{g}}^{1}]_{h}\mathcal{R}^{a+1}\\
	& -W^{-1}[\mathcal{R}_{s_{g}}^{1}]_{h}\mathcal{R}_{b}^{a+1}+W^{-2}W_{b}[\mathcal{R}^{1}]_{h}\mathcal{R}_{s_{g}}^{a+1}\\
	& -W^{-1}[\mathcal{R}_{b}^{1}]_{h}\mathcal{R}_{s_{g}}^{a+1}-W^{-1}[\mathcal{R}^{1}]_{h}\mathcal{R}_{bs_{g}}^{a+1}\\
	& -[\mathcal{H}^{-1}\mathcal{H}_{b}^{-1}\mathcal{H}^{-1}\mathcal{H}_{s_{g}}^{-1}\mathcal{H}^{-1}]_{h,\cdot}(\partial_{\beta^{a}\phi\phi'}\mathcal{L})\mathcal{H}^{-1}\\
	& +[\mathcal{H}^{-1}\mathcal{H}_{bs_{g}}^{-1}\mathcal{H}^{-1}]_{h,\cdot}(\partial_{\beta^{a}\phi\phi'}\mathcal{L})\mathcal{H}^{-1}\\
	& -[\mathcal{H}^{-1}\mathcal{H}_{s_{g}}^{-1}\mathcal{H}^{-1}\mathcal{H}_{b}^{-1}\mathcal{H}^{-1}]_{h,\cdot}(\partial_{\beta^{a}\phi\phi'}\mathcal{L})\mathcal{H}^{-1}\\
	& -W^{-1}[\mathcal{H}^{-1}\mathcal{H}_{s_{g}}^{-1}\mathcal{H}^{-1}]_{h,\cdot}(\partial_{\beta^{a+1}\phi\phi'}\mathcal{L})\mathcal{H}^{-1}\\
	& -W^{-1}[\mathcal{H}^{-1}\mathcal{H}_{s_{g}}^{-1}\mathcal{H}^{-1}]_{h,\cdot}\sum_{f}(\partial_{\beta^{a}\phi\phi'\phi_{f}}\mathcal{L})\mathcal{H}^{-1}[\mathcal{R}^{1}]_{f}\\
	& -[\mathcal{H}^{-1}\mathcal{H}_{s_{g}}^{-1}\mathcal{H}^{-1}]_{h,\cdot}(\partial_{\beta^{a}\phi\phi'}\mathcal{L})\mathcal{H}^{-1}\mathcal{H}_{b}^{-1}\mathcal{H}^{-1}\\
	& -W^{-2}W_{b}[\mathcal{H}^{-1}]_{h,\cdot}(\partial_{\beta^{a+1}\phi\phi'}\mathcal{L})\mathcal{H}^{-1}[\mathcal{R}^{1}]_{g}\\
	& -W^{-1}[\mathcal{H}^{-1}\mathcal{H}_{b}^{-1}\mathcal{H}^{-1}]_{h,\cdot}(\partial_{\beta^{a+1}\phi\phi'}\mathcal{L})\mathcal{H}^{-1}[\mathcal{R}^{1}]_{g}\\
	& -W^{-2}[\mathcal{H}^{-1}]_{h,\cdot}(\partial_{\beta^{a+2}\phi\phi'}\mathcal{L})\mathcal{H}^{-1}[\mathcal{R}^{1}]_{g}\\
	& -W^{-2}[\mathcal{H}^{-1}]_{h,\cdot}\sum_{f}(\partial_{\beta^{a+1}\phi\phi'\phi_{f}}\mathcal{L})\mathcal{H}^{-1}[\mathcal{R}^{1}]_{f}[\mathcal{R}^{1}]_{g}\\
	& -W^{-1}[\mathcal{H}^{-1}]_{h,\cdot}(\partial_{\beta^{a+1}\phi\phi'}\mathcal{L})\mathcal{H}^{-1}\mathcal{H}_{b}^{-1}\mathcal{H}^{-1}[\mathcal{R}^{1}]_{g}\\
	& +W^{-1}[\mathcal{H}^{-1}]_{h,\cdot}(\partial_{\beta^{a+1}\phi\phi'}\mathcal{L})\mathcal{H}^{-1}[\mathcal{R}_{b}^{1}]_{g}\\
	& -[\mathcal{H}^{-1}\mathcal{H}_{b}^{-1}\mathcal{H}^{-1}]_{h,\cdot}\sum_{e}(\partial_{\beta^{a}\phi\phi'\phi_{e}}\mathcal{L})\mathcal{H}^{-1}[\mathcal{H}^{-1}]_{eg}\\
	& -W^{-1}[\mathcal{H}^{-1}]_{h,\cdot}\sum_{e}(\partial_{\beta^{a+1}\phi\phi'\phi_{e}}\mathcal{L})\mathcal{H}^{-1}[\mathcal{H}^{-1}]_{eg}\\
	& -W^{-1}[\mathcal{H}^{-1}]_{h,\cdot}\sum_{e,f}(\partial_{\beta^{a}\phi\phi'\phi_{e}\phi_{f}}\mathcal{L})\mathcal{H}^{-1}[\mathcal{H}^{-1}]_{eg}[\mathcal{R}^{1}]_{f}\\
	& -[\mathcal{H}^{-1}]_{h,\cdot}\sum_{e}(\partial_{\beta^{a}\phi\phi'\phi_{e}}\mathcal{L})\mathcal{H}^{-1}\mathcal{H}_{b}^{-1}\mathcal{H}^{-1}[\mathcal{H}^{-1}]_{eg}\\
	& -[\mathcal{H}^{-1}]_{h,\cdot}\sum_{e}(\partial_{\beta^{a}\phi\phi'\phi_{e}}\mathcal{L})\mathcal{H}^{-1}[\mathcal{H}^{-1}\mathcal{H}_{b}^{-1}\mathcal{H}^{-1}]_{eg}\\
	& -W^{-2}W_{b}[\mathcal{H}^{-1}]_{h,\cdot}\sum_{e}(\partial_{\beta^{a}\phi\phi'\phi_{e}}\mathcal{L})\mathcal{H}^{-1}[\mathcal{R}^{1}]_{e}[\mathcal{R}^{1}]_{g}\\
	& -W^{-1}[\mathcal{H}^{-1}\mathcal{H}_{b}^{-1}\mathcal{H}^{-1}]_{h,\cdot}\sum_{e}(\partial_{\beta^{a}\phi\phi'\phi_{e}}\mathcal{L})\mathcal{H}^{-1}[\mathcal{R}^{1}]_{e}[\mathcal{R}^{1}]_{g}\\
	& -W^{-2}[\mathcal{H}^{-1}]_{h,\cdot}\sum_{e}(\partial_{\beta^{a+1}\phi\phi'\phi_{e}}\mathcal{L})\mathcal{H}^{-1}[\mathcal{R}^{1}]_{e}[\mathcal{R}^{1}]_{g}\\
	& -W^{-2}[\mathcal{H}^{-1}]_{h,\cdot}\sum_{e,f}(\partial_{\beta^{a}\phi\phi'\phi_{e}\phi_{f}}\mathcal{L})\mathcal{H}^{-1}[\mathcal{R}^{1}]_{e}[\mathcal{R}^{1}]_{f}[\mathcal{R}^{1}]_{g}\\
	& -W^{-1}[\mathcal{H}^{-1}]_{h,\cdot}\sum_{e}(\partial_{\beta^{a}\phi\phi'\phi_{e}}\mathcal{L})\mathcal{H}^{-1}\mathcal{H}_{b}^{-1}\mathcal{H}^{-1}[\mathcal{R}^{1}]_{e}[\mathcal{R}^{1}]_{g}\\
	& +W^{-1}[\mathcal{H}^{-1}]_{h,\cdot}\sum_{e}(\partial_{\beta^{a}\phi\phi'\phi_{e}}\mathcal{L})\mathcal{H}^{-1}[\mathcal{R}_{b}^{1}]_{e}[\mathcal{R}^{1}]_{g}\\
	& +W^{-1}[\mathcal{H}^{-1}]_{h,\cdot}\sum_{e}(\partial_{\beta^{a}\phi\phi'\phi_{e}}\mathcal{L})\mathcal{H}^{-1}[\mathcal{R}^{1}]_{e}[\mathcal{R}_{b}^{1}]_{g}\\
	& -[\mathcal{H}^{-1}\mathcal{H}_{b}^{-1}\mathcal{H}^{-1}]_{h,\cdot}(\partial_{\beta^{a}\phi\phi'}\mathcal{L})\mathcal{H}^{-1}\mathcal{H}_{s_{g}}^{-1}\mathcal{H}^{-1}\\
	& -W^{-1}[\mathcal{H}^{-1}]_{h,\cdot}(\partial_{\beta^{a+1}\phi\phi'}\mathcal{L})\mathcal{H}^{-1}\mathcal{H}_{s_{g}}^ {}\mathcal{H}^{-1}\\
	& -W^{-1}[\mathcal{H}^{-1}]_{h,\cdot}\sum_{f}(\partial_{\beta^{a}\phi\phi'\phi_{f}}\mathcal{L})\mathcal{H}^{-1}\mathcal{H}_{s_{g}}^ {}\mathcal{H}^{-1}[\mathcal{R}^{1}]_{f}\\
	& -[\mathcal{H}^{-1}]_{h,\cdot}(\partial_{\beta^{a}\phi\phi'}\mathcal{L})\mathcal{H}^{-1}\mathcal{H}_{b}^{-1}\mathcal{H}^{-1}\mathcal{H}_{s_{g}}^ {}\mathcal{H}^{-1}\\
	& +[\mathcal{H}^{-1}]_{h,\cdot}(\partial_{\beta^{a}\phi\phi'}\mathcal{L})\mathcal{H}^{-1}\mathcal{H}_{bs_{g}}^ {}\mathcal{H}^{-1}\\
	& -[\mathcal{H}^{-1}]_{h,\cdot}(\partial_{\beta^{a}\phi\phi'}\mathcal{L})\mathcal{H}^{-1}\mathcal{H}_{s_{g}}^ {}\mathcal{H}^{-1}\mathcal{H}_{b}^{-1}\mathcal{H}^{-1}\\
	& -2W^{-3}W_{b}W_{s_{g}}[\mathcal{R}^{1}]_{h}\mathcal{R}^{1}(\partial_{\beta^{a}\phi\phi'}\mathcal{L})\mathcal{H}^{-1}\\
	& +W^{-2}W_{bs_{g}}[\mathcal{R}^{1}]_{h}\mathcal{R}^{1}(\partial_{\beta^{a}\phi\phi'}\mathcal{L})\mathcal{H}^{-1}\\
	& +W^{-2}W_{s_{g}}[\mathcal{R}_{b}^{1}]_{h}\mathcal{R}^{1}(\partial_{\beta^{a}\phi\phi'}\mathcal{L})\mathcal{H}^{-1}\\
	& +W^{-2}W_{s_{g}}[\mathcal{R}^{1}]_{h}\mathcal{R}_{b}^{1}(\partial_{\beta^{a}\phi\phi'}\mathcal{L})\mathcal{H}^{-1}\\
	& -W^{-3}W_{s_{g}}[\mathcal{R}^{1}]_{h}\mathcal{R}^{1}(\partial_{\beta^{a+1}\phi\phi'}\mathcal{L})\mathcal{H}^{-1}\\
	& -W^{-3}W_{s_{g}}[\mathcal{R}^{1}]_{h}\mathcal{R}^{1}\sum_{f}(\partial_{\beta^{a}\phi\phi'\phi_{f}}\mathcal{L})\mathcal{H}^{-1}[\mathcal{R}^{1}]_{f}\\
	& -W^{-2}W_{s_{g}}[\mathcal{R}^{1}]_{h}\mathcal{R}^{1}(\partial_{\beta^{a}\phi\phi'}\mathcal{L})\mathcal{H}^{-1}\mathcal{H}_{b}^{-1}\mathcal{H}^{-1}\\
	& +W^{-2}W_{b}[\mathcal{R}_{s_{g}}^{1}]_{h}\mathcal{R}^{1}(\partial_{\beta^{a}\phi\phi'}\mathcal{L})\mathcal{H}^{-1}\\
	& -W^{-1}[\mathcal{R}_{bs_{g}}^{1}]_{h}\mathcal{R}^{1}(\partial_{\beta^{a}\phi\phi'}\mathcal{L})\mathcal{H}^{-1}\\
	& -W^{-1}[\mathcal{R}_{s_{g}}^{1}]_{h}\mathcal{R}_{b}^{1}(\partial_{\beta^{a}\phi\phi'}\mathcal{L})\mathcal{H}^{-1}\\
	& +W^{-2}[\mathcal{R}_{s_{g}}^{1}]_{h}\mathcal{R}^{1}(\partial_{\beta^{a+1}\phi\phi'}\mathcal{L})\mathcal{H}^{-1}\\
	& +W^{-2}[\mathcal{R}_{s_{g}}^{1}]_{h}\mathcal{R}^{1}\sum_{f}(\partial_{\beta^{a}\phi\phi'\phi_{f}}\mathcal{L})\mathcal{H}^{-1}[\mathcal{R}^{1}]_{f}\\
	& +W^{-1}[\mathcal{R}_{s_{g}}^{1}]_{h}\mathcal{R}^{1}(\partial_{\beta^{a}\phi\phi'}\mathcal{L})\mathcal{H}^{-1}\mathcal{H}_{b}^{-1}\mathcal{H}^{-1}\\
	& +W^{-2}W_{b}[\mathcal{R}^{1}]_{h}\mathcal{R}_{s_{g}}^{1}(\partial_{\beta^{a}\phi\phi'}\mathcal{L})\mathcal{H}^{-1}\\
	& -W^{-1}[\mathcal{R}_{b}^{1}]_{h}\mathcal{R}_{s_{g}}^{1}(\partial_{\beta^{a}\phi\phi'}\mathcal{L})\mathcal{H}^{-1}\\
	& -W^{-1}[\mathcal{R}^{1}]_{h}\mathcal{R}_{bs_{g}}^{1}(\partial_{\beta^{a}\phi\phi'}\mathcal{L})\mathcal{H}^{-1}\\
	& +W^{-2}[\mathcal{R}^{1}]_{h}\mathcal{R}_{s_{g}}^{1}(\partial_{\beta^{a+1}\phi\phi'}\mathcal{L})\mathcal{H}^{-1}\\
	& +W^{-2}[\mathcal{R}^{1}]_{h}\mathcal{R}_{s_{g}}^{1}\sum_{f}(\partial_{\beta^{a}\phi\phi'\phi_{f}}\mathcal{L})\mathcal{H}^{-1}[\mathcal{R}^{1}]_{f}\\
	& +W^{-1}[\mathcal{R}^{1}]_{h}\mathcal{R}_{s_{g}}^{1}(\partial_{\beta^{a}\phi\phi'}\mathcal{L})\mathcal{H}^{-1}\mathcal{H}_{b}^{-1}\mathcal{H}^{-1}\\
	& -2W^{-3}W_{b}[\mathcal{R}^{1}]_{h}\mathcal{R}^{1}(\partial_{\beta^{a+1}\phi\phi'}\mathcal{L})\mathcal{H}^{-1}[\mathcal{R}^{1}]_{g}\\
	& +W^{-2}[\mathcal{R}_{b}^{1}]_{h}\mathcal{R}^{1}(\partial_{\beta^{a+1}\phi\phi'}\mathcal{L})\mathcal{H}^{-1}[\mathcal{R}^{1}]_{g}\\
	& +W^{-2}[\mathcal{R}^{1}]_{h}\mathcal{R}_{b}^{1}(\partial_{\beta^{a+1}\phi\phi'}\mathcal{L})\mathcal{H}^{-1}[\mathcal{R}^{1}]_{g}\\
	& -W^{-3}[\mathcal{R}^{1}]_{h}\mathcal{R}^{1}(\partial_{\beta^{a+2}\phi\phi'}\mathcal{L})\mathcal{H}^{-1}[\mathcal{R}^{1}]_{g}\\
	& -W^{-3}[\mathcal{R}^{1}]_{h}\mathcal{R}^{1}\sum_{f}(\partial_{\beta^{a+1}\phi\phi'\phi_{f}}\mathcal{L})\mathcal{H}^{-1}[\mathcal{R}^{1}]_{f}[\mathcal{R}^{1}]_{g}\\
	& -W^{-2}[\mathcal{R}^{1}]_{h}\mathcal{R}^{1}(\partial_{\beta^{a+1}\phi\phi'}\mathcal{L})\mathcal{H}^{-1}\mathcal{H}_{b}^{-1}\mathcal{H}^{-1}[\mathcal{R}^{1}]_{g}\\
	& +W^{-2}[\mathcal{R}^{1}]_{h}\mathcal{R}^{1}(\partial_{\beta^{a+1}\phi\phi'}\mathcal{L})\mathcal{H}^{-1}[\mathcal{R}_{b}^{1}]_{g}\\
	& -W_{b}^{-2}[\mathcal{R}^{1}]_{h}\mathcal{R}^{1}\sum_{e}(\partial_{\beta^{a}\phi\phi'\phi_{e}}\mathcal{L})\mathcal{H}^{-1}[\mathcal{H}^{-1}]_{eg}\\
	& +W^{-1}[\mathcal{R}_{b}^{1}]_{h}\mathcal{R}^{1}\sum_{e}(\partial_{\beta^{a}\phi\phi'\phi_{e}}\mathcal{L})\mathcal{H}^{-1}[\mathcal{H}^{-1}]_{eg}\\
	& +W^{-1}[\mathcal{R}^{1}]_{h}\mathcal{R}_{b}^{1}\sum_{e}(\partial_{\beta^{a}\phi\phi'\phi_{e}}\mathcal{L})\mathcal{H}^{-1}[\mathcal{H}^{-1}]_{eg}\\
	& -W^{-2}[\mathcal{R}^{1}]_{h}\mathcal{R}^{1}\sum_{e}(\partial_{\beta^{a+1}\phi\phi'\phi_{e}}\mathcal{L})\mathcal{H}^{-1}[\mathcal{H}^{-1}]_{eg}\\
	& -W^{-2}[\mathcal{R}^{1}]_{h}\mathcal{R}^{1}\sum_{e,f}(\partial_{\beta^{a}\phi\phi'\phi_{e}\phi_{f}}\mathcal{L})\mathcal{H}^{-1}[\mathcal{H}^{-1}]_{eg}[\mathcal{R}^{1}]_{f}\\
	& -W^{-1}[\mathcal{R}^{1}]_{h}\mathcal{R}^{1}\sum_{e}(\partial_{\beta^{a}\phi\phi'\phi_{e}}\mathcal{L})\mathcal{H}^{-1}\mathcal{H}_{b}^{-1}\mathcal{H}^{-1}[\mathcal{H}^{-1}]_{eg}\\
	& -W^{-1}[\mathcal{R}^{1}]_{h}\mathcal{R}^{1}\sum_{e}(\partial_{\beta^{a}\phi\phi'\phi_{e}}\mathcal{L})\mathcal{H}^{-1}[\mathcal{H}^{-1}\mathcal{H}_{b}^{-1}\mathcal{H}^{-1}]_{eg}\\
	& -2W^{-3}W_{b}[\mathcal{R}^{1}]_{h}\mathcal{R}^{1}\sum_{e}(\partial_{\beta^{a}\phi\phi'\phi_{e}}\mathcal{L})\mathcal{H}^{-1}[\mathcal{R}^{1}]_{e}[\mathcal{R}^{1}]_{g}\\
	& +W^{-2}[\mathcal{R}_{b}^{1}]_{h}\mathcal{R}^{1}\sum_{e}(\partial_{\beta^{a}\phi\phi'\phi_{e}}\mathcal{L})\mathcal{H}^{-1}[\mathcal{R}^{1}]_{e}[\mathcal{R}^{1}]_{g}\\
	& +W^{-2}[\mathcal{R}^{1}]_{h}\mathcal{R}_{b}^{1}\sum_{e}(\partial_{\beta^{a}\phi\phi'\phi_{e}}\mathcal{L})\mathcal{H}^{-1}[\mathcal{R}^{1}]_{e}[\mathcal{R}^{1}]_{g}\\
	& -W^{-3}[\mathcal{R}^{1}]_{h}\mathcal{R}^{1}\sum_{e}(\partial_{\beta^{a+1}\phi\phi'\phi_{e}}\mathcal{L})\mathcal{H}^{-1}[\mathcal{R}^{1}]_{e}[\mathcal{R}^{1}]_{g}\\
	& -W^{-3}[\mathcal{R}^{1}]_{h}\mathcal{R}^{1}\sum_{e,f}(\partial_{\beta^{a}\phi\phi'\phi_{e}\phi_{f}}\mathcal{L})\mathcal{H}^{-1}[\mathcal{R}^{1}]_{e}[\mathcal{R}^{1}]_{f}[\mathcal{R}^{1}]_{g}\\
	& -W^{-2}[\mathcal{R}^{1}]_{h}\mathcal{R}^{1}\sum_{e}(\partial_{\beta^{a}\phi\phi'\phi_{e}}\mathcal{L})\mathcal{H}^{-1}\mathcal{H}_{b}^{-1}\mathcal{H}^{-1}[\mathcal{R}^{1}]_{e}[\mathcal{R}^{1}]_{g}\\
	& +W^{-2}[\mathcal{R}^{1}]_{h}\mathcal{R}^{1}\sum_{e}(\partial_{\beta^{a}\phi\phi'\phi_{e}}\mathcal{L})\mathcal{H}^{-1}[\mathcal{R}_{b}^{1}]_{e}[\mathcal{R}^{1}]_{g}\\
	& +W^{-2}[\mathcal{R}^{1}]_{h}\mathcal{R}^{1}\sum_{e}(\partial_{\beta^{a}\phi\phi'\phi_{e}}\mathcal{L})\mathcal{H}^{-1}[\mathcal{R}^{1}]_{e}[\mathcal{R}_{b}^{1}]_{g}\\
	& -W^{-2}W_{b}[\mathcal{R}^{1}]_{h}\mathcal{R}^{1}(\partial_{\beta^{a}\phi\phi'}\mathcal{L})\mathcal{H}^{-1}\mathcal{H}_{s_{g}}^{-1}\mathcal{H}^{-1}\\
	& +W^{-1}[\mathcal{R}_{b}^{1}]_{h}\mathcal{R}^{1}(\partial_{\beta^{a}\phi\phi'}\mathcal{L})\mathcal{H}^{-1}\mathcal{H}_{s_{g}}^{-1}\mathcal{H}^{-1}\\
	& +W^{-1}[\mathcal{R}^{1}]_{h}\mathcal{R}_{b}^{1}(\partial_{\beta^{a}\phi\phi'}\mathcal{L})\mathcal{H}^{-1}\mathcal{H}_{s_{g}}^{-1}\mathcal{H}^{-1}\\
	& -W^{-2}[\mathcal{R}^{1}]_{h}\mathcal{R}^{1}(\partial_{\beta^{a+1}\phi\phi'}\mathcal{L})\mathcal{H}^{-1}\mathcal{H}_{s_{g}}^{-1}\mathcal{H}^{-1}\\
	& -W^{-2}[\mathcal{R}^{1}]_{h}\mathcal{R}^{1}\sum_{f}(\partial_{\beta^{a}\phi\phi'\phi_{f}}\mathcal{L})\mathcal{H}^{-1}\mathcal{H}_{s_{g}}^{-1}\mathcal{H}^{-1}[\mathcal{R}^{1}]_{f}\\
	& -W^{-1}[\mathcal{R}^{1}]_{h}\mathcal{R}^{1}(\partial_{\beta^{a}\phi\phi'}\mathcal{L})\mathcal{H}^{-1}\mathcal{H}_{b}^{-1}\mathcal{H}^{-1}\mathcal{H}_{s_{g}}^{-1}\mathcal{H}^{-1}\\
	& +W^{-1}[\mathcal{R}^{1}]_{h}\mathcal{R}^{1}(\partial_{\beta^{a}\phi\phi'}\mathcal{L})\mathcal{H}^{-1}\mathcal{H}_{bs_{g}}^{-1}\mathcal{H}^{-1}\\
	& -W^{-1}[\mathcal{R}^{1}]_{h}\mathcal{R}^{1}(\partial_{\beta^{a}\phi\phi'}\mathcal{L})\mathcal{H}^{-1}\mathcal{H}_{s_{g}}^{-1}\mathcal{H}^{-1}\mathcal{H}_{b}^{-1}\mathcal{H}^{-1}\\
	& -\mathcal{R}_{bs_{g}}^{a}\mathcal{H}_{s_{h}}\mathcal{H}^{-1}-\mathcal{R}_{s_{g}}^{a}\mathcal{H}_{bs_{h}}\mathcal{H}^{-1}+\mathcal{R}_{s_{g}}^{a}\mathcal{H}_{s_{h}}\mathcal{H}^{-1}\mathcal{H}_{b}^{-1}\mathcal{H}^{-1}\\
	& -\mathcal{R}_{b}^{a}\mathcal{H}_{s_{g}s_{h}}\mathcal{H}^{-1}-\mathcal{R}^{a}\mathcal{H}_{bs_{g}s_{h}}\mathcal{H}^{-1}+\mathcal{R}^{a}\mathcal{H}_{s_{g}s_{h}}\mathcal{H}^{-1}\mathcal{H}_{b}^{-1}\mathcal{H}^{-1}\\
	& +\mathcal{R}_{b}^{a}\mathcal{H}_{s_{h}}\mathcal{H}^{-1}\mathcal{H}_{s_{g}}^{-1}\mathcal{H}^{-1}+\mathcal{R}^{a}\mathcal{H}_{bs_{h}}\mathcal{H}^{-1}\mathcal{H}_{s_{g}}^{-1}\mathcal{H}^{-1}\\
	& -\mathcal{R}^{a}\mathcal{H}_{s_{h}}\mathcal{H}^{-1}\mathcal{H}_{b}^{-1}\mathcal{H}^{-1}\mathcal{H}_{s_{g}}^{-1}\mathcal{H}^{-1}+\mathcal{R}^{a}\mathcal{H}_{s_{h}}\mathcal{H}^{-1}\mathcal{H}_{bs_{g}}^{-1}\mathcal{H}^{-1}\\
	& -\mathcal{R}^{a}\mathcal{H}_{s_{h}}\mathcal{H}^{-1}\mathcal{H}_{s_{g}}^{-1}\mathcal{H}^{-1}\mathcal{H}_{b}^{-1}\mathcal{H}^{-1}
\end{align*}

\subsection{Expressions for $\mathcal{F}$ terms}

Define $\mathcal{E}^{s}=\partial_{\beta^{s}\phi\phi'}\mathcal{L}$,
$\mathcal{F}^{s,t}=(\partial_{\beta^{s}\phi'}\mathcal{L})\mathcal{H}^{-1}(\partial_{\beta^{t}\phi}\mathcal{L})$,
and
\begin{align*}
\mathcal{F}^{s,(r),t} & =(\partial_{\beta^{s}\phi'}\mathcal{L})\mathcal{H}^{-1}\mathcal{E}^{r}\mathcal{H}^{-1}(\partial_{\beta^{t}\phi}\mathcal{L})\\
\mathcal{F}^{s,(r_{1},r_{2}),t} & =(\partial_{\beta^{s}\phi'}\mathcal{L})\mathcal{H}^{-1}\mathcal{E}^{r_{1}}\mathcal{H}^{-1}\mathcal{E}^{r_{2}}\mathcal{H}^{-1}(\partial_{\beta^{t}\phi}\mathcal{L})
\end{align*}
 Then, differentiating gives

(1)
\begin{align*}
\mathcal{F}_{b}^{s,t} & =-W^{-1}\big((\partial_{\beta^{s+1}\phi'}\mathcal{L})+(\partial_{\beta\phi'}\mathcal{L})\mathcal{H}^{-1}\mathcal{E}^{s}\big)\mathcal{H}^{-1}(\partial_{\beta^{t}\phi}\mathcal{L})\\
 & \quad-W^{-1}(\partial_{\beta^{s}\phi'}\mathcal{L})\mathcal{H}^{-1}\big((\partial_{\beta^{t+1}\phi}\mathcal{L})+\mathcal{E}^{t}\mathcal{H}^{-1}(\partial_{\beta\phi}\mathcal{L})\big)\\
 & \quad-(\partial_{\beta^{s}\phi'}\mathcal{L})\mathcal{H}^{-1}\mathcal{H}_{b}\mathcal{H}^{-1}(\partial_{\beta^{t}\phi}\mathcal{L})\\
 & =-W^{-1}\Big(\mathcal{F}^{s+1,t}+\mathcal{F}^{1,(s),t}+\mathcal{F}^{s,t+1}+\mathcal{F}^{s,(t),1}\Big)\\
 & \quad-(\partial_{\beta^{s}\phi'}\mathcal{L})\mathcal{H}^{-1}\mathcal{H}_{b}\mathcal{H}^{-1}(\partial_{\beta^{t}\phi}\mathcal{L})
\end{align*}

(2)
\begin{align*}
\mathcal{F}_{s}^{s,t} & =-W^{-1}\big(\mathcal{H}^{-1}(\partial_{\beta\phi}\mathcal{L})(\partial_{\beta^{s+1}\phi'}\mathcal{L})+\mathcal{G}(\partial_{\beta^{s}\phi\phi'}\mathcal{L})\big)\mathcal{H}^{-1}(\partial_{\beta^{t}\phi}\mathcal{L})\\
 & \quad-\mathcal{H}^{-1}(\partial_{\beta^{s}\phi\phi'}\mathcal{L})\mathcal{H}^{-1}(\partial_{\beta^{t}\phi}\mathcal{L})\\
 & \quad-W^{-1}(\partial_{\beta^{s}\phi'}\mathcal{L})\mathcal{H}^{-1}\big((\partial_{\beta^{t+1}\phi}\mathcal{L})(\partial_{\beta\phi'}\mathcal{L})\mathcal{H}^{-1}+(\partial_{\beta^{t}\phi\phi'}\mathcal{L})\mathcal{G}\big)\\
 & \quad-(\partial_{\beta^{s}\phi'}\mathcal{L})\mathcal{H}^{-1}(\partial_{\beta^{t}\phi\phi'}\mathcal{L})\mathcal{H}^{-1}\\
 & \quad-(\partial_{\beta^{s}\phi'}\mathcal{L})\mathcal{H}^{-1}\mathcal{H}_{s}\mathcal{H}^{-1}(\partial_{\beta^{t}\phi}\mathcal{L})\\
 & =-W^{-1}\big(\mathcal{G}^{1,s+1}+\mathcal{G}\mathcal{E}^{s}\mathcal{H}^{-1}\big)(\partial_{\beta^{t}\phi}\mathcal{L})-\mathcal{H}^{-1}\mathcal{E}^{s}\mathcal{H}^{-1}(\partial_{\beta^{t}\phi}\mathcal{L})\\
 & \quad-W^{-1}(\partial_{\beta^{s}\phi}\mathcal{L})\big(\mathcal{G}^{t+1,1}+\mathcal{H}^{-1}\mathcal{E}^{t}\mathcal{G}\big)-(\partial_{\beta^{s}\phi'}\mathcal{L})\mathcal{H}^{-1}\mathcal{E}^{t}\mathcal{H}^{-1}\\
 & \quad-(\partial_{\beta^{s}\phi'}\mathcal{L})\mathcal{H}^{-1}\mathcal{H}_{s}\mathcal{H}^{-1}(\partial_{\beta^{t}\phi}\mathcal{L})
\end{align*}

Differentiating a second time gives

(1)
\begin{align*}
\mathcal{F}_{bb}^{s,t} & =W^{-1}W_{b}\mathcal{F}_{b}^{s,t}\\
 & -W^{-1}\Big(\mathcal{F}_{b}^{s+1,t}+\mathcal{F}_{b}^{1,(s),t}+\mathcal{F}_{b}^{s,t+1}+\mathcal{F}_{b}^{s,(t),1}\Big)\\
 & +W^{-1}(\partial_{\beta^{s+1}\phi'}\mathcal{L})\mathcal{H}^{-1}\mathcal{H}_{b}\mathcal{H}^{-1}(\partial_{\beta^{t}\phi}\mathcal{L})\\
 & +W^{-1}(\partial_{\beta\phi'}\mathcal{L})\mathcal{H}^{-1}\mathcal{E}^{s}\mathcal{H}^{-1}\mathcal{H}_{b}\mathcal{H}^{-1}(\partial_{\beta^{t}\phi}\mathcal{L})\\
 & -2(\partial_{\beta^{s}\phi'}\mathcal{L})\mathcal{H}^{-1}\mathcal{H}_{b}\mathcal{H}^{-1}\mathcal{H}_{b}\mathcal{H}^{-1}(\partial_{\beta^{t}\phi}\mathcal{L})\\
 & -(\partial_{\beta^{s}\phi'}\mathcal{L})\mathcal{H}^{-1}\mathcal{H}_{bb}\mathcal{H}^{-1}(\partial_{\beta^{t}\phi}\mathcal{L})\\
 & +W^{-1}(\partial_{\beta^{s}\phi'}\mathcal{L})\mathcal{H}^{-1}\mathcal{H}_{b}\mathcal{H}^{-1}(\partial_{\beta^{t+1}\phi}\mathcal{L})\\
 & +W^{-1}(\partial_{\beta^{s}\phi'}\mathcal{L})\mathcal{H}^{-1}\mathcal{H}_{b}\mathcal{H}^{-1}\mathcal{E}^{t}\mathcal{H}^{-1}(\partial_{\beta\phi}\mathcal{L})
\end{align*}

(2)
\begin{align*}
\mathcal{F}_{bs}^{s,t} & =-W^{-1}\big(\mathcal{G}^{1,s+1}+\mathcal{G}\mathcal{E}^{s}\mathcal{H}^{-1}\big)(\partial_{\beta^{t}\phi}\mathcal{L})\\
 & -\mathcal{H}^{-1}\mathcal{E}^{s}\mathcal{H}^{-1}(\partial_{\beta^{t}\phi}\mathcal{L})\\
 & \quad-W^{-1}(\partial_{\beta^{s}\phi}\mathcal{L})\big(\mathcal{G}^{t+1,1}+\mathcal{H}^{-1}\mathcal{E}^{t}\mathcal{G}\big)\\
 & -(\partial_{\beta^{s}\phi'}\mathcal{L})\mathcal{H}^{-1}\mathcal{E}^{t}\mathcal{H}^{-1}\\
 & \quad-(\partial_{\beta^{s}\phi'}\mathcal{L})\mathcal{H}^{-1}\mathcal{H}_{s}\mathcal{H}^{-1}(\partial_{\beta^{t}\phi}\mathcal{L})\\
\end{align*}

(3)

\begin{align*}\mathcal{F}_{s_{f}s_{g}} & =W^{-2}W_{s_{f}}(\mathcal{F}^{2,1}+\mathcal{F}^{1,2}+2\mathcal{F}^{1,(1),1})[\mathcal{R}^{1}]_{g}\\
	& -W^{-1}(\mathcal{F}_{s_{f}}^{2,1}+\mathcal{F}_{s_{f}}^{1,2}+2\mathcal{F}_{s_{f}}^{1,(1),1})[\mathcal{R}^{1}]_{g}\\
	& +W^{-1}(\mathcal{F}^{2,1}+\mathcal{F}^{1,2}+2\mathcal{F}^{1,(1),1})\big[\mathcal{H}^{-1}\mathcal{H}_{s_{f}}\mathcal{R}^{1}\big]_{g}\\
	& +W^{-2}(\mathcal{F}^{2,1}+\mathcal{F}^{1,2}+2\mathcal{F}^{1,(1),1})[\mathcal{R}^{2}]_{g}[\mathcal{R}^{1}]_{f}\\
	& +W^{-2}(\mathcal{F}^{2,1}+\mathcal{F}^{1,2}+2\mathcal{F}^{1,(1),1})\big[\mathcal{H}^{-1}\mathcal{E}\mathcal{R}^{1}\big]_{g}[\mathcal{R}^{1}]_{f}\\
	& +W^{-1}(\mathcal{F}^{2,1}+\mathcal{F}^{1,2}+2\mathcal{F}^{1,(1),1})\big[\mathcal{H}^{-1}\mathcal{E}\mathcal{H}^{-1}\big]_{fg}\\
	& +2[\mathcal{H}^{-1}\mathcal{H}_{s_{f}}\mathcal{H}^{-1}]_{\cdot,g}\mathcal{E}\mathcal{R}^{1}-2[\mathcal{H}^{-1}]_{\cdot,g}\mathcal{E}_{s_{f}}\mathcal{R}^{1}\\
	& +2[\mathcal{H}^{-1}]_{\cdot,g}\mathcal{E}\mathcal{H}^{-1}\mathcal{H}_{s_{f}}\mathcal{R}^{1}+2W^{-1}[\mathcal{H}^{-1}]_{\cdot,g}\mathcal{E}\mathcal{R}^{2}[\mathcal{R}^{1}]_{f}\\
	& +2W^{-1}[\mathcal{H}^{-1}]_{\cdot,g}\mathcal{E}\mathcal{H}^{-1}\mathcal{E}\mathcal{R}^{1}[\mathcal{R}^{1}]_{f}+2[\mathcal{H}^{-1}\mathcal{E}\mathcal{H}^{-1}\mathcal{E}\mathcal{H}^{-1}]_{fg}\\
	& +2W^{-1}(\mathcal{R}^{2})'\mathcal{H}_{s_{g}}\mathcal{R}^{1}[\mathcal{R}^{1}]_{f}+2W^{-1}(\mathcal{R}^{1})'\mathcal{E}\mathcal{H}^{-1}\mathcal{H}_{s_{g}}\mathcal{R}^{1}[\mathcal{R}^{1}]_{f}\\
	& +2[\mathcal{H}^{-1}]_{\cdot,g}\mathcal{E}\mathcal{H}^{-1}\mathcal{H}_{s_{g}}\mathcal{R}^{1}+2(\mathcal{R}^{1})'\mathcal{H}_{s_{f}}\mathcal{H}^{-1}\mathcal{H}_{s_{g}}\mathcal{R}^{1}\\
	& -(\mathcal{R}^{1})'\mathcal{H}_{s_{f}s_{g}}\mathcal{R}^{1}
\end{align*}

Third derivative:
\begin{align*}\mathcal{F}_{bs_{f}s_{g}} & =-2W^{-3}W_{b}W_{s_{f}}(\mathcal{F}^{2,1}+\mathcal{F}^{1,2}+2\mathcal{F}^{1,(1),1})[\mathcal{R}^{1}]_{g}\\
	& +W^{-2}W_{bs_{f}}(\mathcal{F}^{2,1}+\mathcal{F}^{1,2}+2\mathcal{F}^{1,(1),1})[\mathcal{R}^{1}]_{g}\\
	& +W^{-2}W_{s_{f}}(\mathcal{F}_{b}^{2,1}+\mathcal{F}_{b}^{1,2}+2\mathcal{F}_{b}^{1,(1),1})[\mathcal{R}^{1}]_{g}\\
	& +W^{-2}W_{s_{f}}(\mathcal{F}^{2,1}+\mathcal{F}^{1,2}+2\mathcal{F}^{1,(1),1})[\mathcal{R}_{b}^{1}]_{g}\\
	& +W^{-2}W_{b}(\mathcal{F}_{s_{f}}^{2,1}+\mathcal{F}_{s_{f}}^{1,2}+2\mathcal{F}_{s_{f}}^{1,(1),1})[\mathcal{R}^{1}]_{g}\\
	& -W^{-1}(\mathcal{F}_{bs_{f}}^{2,1}+\mathcal{F}_{bs_{f}}^{1,2}+2\mathcal{F}_{bs_{f}}^{1,(1),1})[\mathcal{R}^{1}]_{g}\\
	& -W^{-1}(\mathcal{F}_{s_{f}}^{2,1}+\mathcal{F}_{s_{f}}^{1,2}+2\mathcal{F}_{s_{f}}^{1,(1),1})[\mathcal{R}_{b}^{1}]_{g}\\
	& -W^{-2}W_{b}(\mathcal{F}^{2,1}+\mathcal{F}^{1,2}+2\mathcal{F}^{1,(1),1})\big[\mathcal{H}^{-1}\mathcal{H}_{s_{f}}\mathcal{R}^{1}\big]_{g}\\
	& +W^{-1}(\mathcal{F}_{b}^{2,1}+\mathcal{F}_{b}^{1,2}+2\mathcal{F}_{b}^{1,(1),1})\big[\mathcal{H}^{-1}\mathcal{H}_{s_{f}}\mathcal{R}^{1}\big]_{g}\\
	& -W^{-1}(\mathcal{F}^{2,1}+\mathcal{F}^{1,2}+2\mathcal{F}^{1,(1),1})\big[\mathcal{H}^{-1}\mathcal{H}_{b}\mathcal{H}^{-1}\mathcal{H}_{s_{f}}\mathcal{R}^{1}\big]_{g}\\
	& +W^{-1}(\mathcal{F}^{2,1}+\mathcal{F}^{1,2}+2\mathcal{F}^{1,(1),1})\big[\mathcal{H}^{-1}\mathcal{H}_{bs_{f}}\mathcal{R}^{1}\big]_{g}\\
	& +W^{-1}(\mathcal{F}^{2,1}+\mathcal{F}^{1,2}+2\mathcal{F}^{1,(1),1})\big[\mathcal{H}^{-1}\mathcal{H}_{s_{f}}\mathcal{R}_{b}^{1}\big]_{g}\\
	& -2W^{-3}W_{b}(\mathcal{F}^{2,1}+\mathcal{F}^{1,2}+2\mathcal{F}^{1,(1),1})[\mathcal{R}^{2}]_{g}[\mathcal{R}^{1}]_{f}\\
	& +W^{-2}(\mathcal{F}_{b}^{2,1}+\mathcal{F}_{b}^{1,2}+2\mathcal{F}_{b}^{1,(1),1})[\mathcal{R}^{2}]_{g}[\mathcal{R}^{1}]_{f}\\
	& +W^{-2}(\mathcal{F}^{2,1}+\mathcal{F}^{1,2}+2\mathcal{F}^{1,(1),1})[\mathcal{R}_{b}^{2}]_{g}[\mathcal{R}^{1}]_{f}\\
	& +W^{-2}(\mathcal{F}^{2,1}+\mathcal{F}^{1,2}+2\mathcal{F}^{1,(1),1})[\mathcal{R}^{2}]_{g}[\mathcal{R}_{b}^{1}]_{f}\\
	& -2W^{-3}W_{b}(\mathcal{F}^{2,1}+\mathcal{F}^{1,2}+2\mathcal{F}^{1,(1),1})\big[\mathcal{H}^{-1}\mathcal{E}\mathcal{R}^{1}\big]_{g}[\mathcal{R}^{1}]_{f}\\
	& +W^{-2}(\mathcal{F}_{b}^{2,1}+\mathcal{F}_{b}^{1,2}+2\mathcal{F}_{b}^{1,(1),1})\big[\mathcal{H}^{-1}\mathcal{E}\mathcal{R}^{1}\big]_{g}[\mathcal{R}^{1}]_{f}\\
	& -W^{-2}(\mathcal{F}^{2,1}+\mathcal{F}^{1,2}+2\mathcal{F}^{1,(1),1})\big[\mathcal{H}^{-1}\mathcal{H}_{b}\mathcal{H}^{-1}\mathcal{E}\mathcal{R}^{1}\big]_{g}[\mathcal{R}^{1}]_{f}\\
	& +W^{-2}(\mathcal{F}^{2,1}+\mathcal{F}^{1,2}+2\mathcal{F}^{1,(1),1})\big[\mathcal{H}^{-1}\mathcal{E}_{b}\mathcal{R}^{1}\big]_{g}[\mathcal{R}^{1}]_{f}\\
	& +W^{-2}(\mathcal{F}^{2,1}+\mathcal{F}^{1,2}+2\mathcal{F}^{1,(1),1})\big[\mathcal{H}^{-1}\mathcal{E}\mathcal{R}_{b}^{1}\big]_{g}[\mathcal{R}^{1}]_{f}\\
	& +W^{-2}(\mathcal{F}^{2,1}+\mathcal{F}^{1,2}+2\mathcal{F}^{1,(1),1})\big[\mathcal{H}^{-1}\mathcal{E}\mathcal{R}^{1}\big]_{g}[\mathcal{R}_{b}^{1}]_{f}\\
	& -W^{-2}W_{b}(\mathcal{F}^{2,1}+\mathcal{F}^{1,2}+2\mathcal{F}^{1,(1),1})\big[\mathcal{H}^{-1}\mathcal{E}\mathcal{H}^{-1}\big]_{fg}\\
	& +W^{-1}(\mathcal{F}_{b}^{2,1}+\mathcal{F}_{b}^{1,2}+2\mathcal{F}_{b}^{1,(1),1})\big[\mathcal{H}^{-1}\mathcal{E}\mathcal{H}^{-1}\big]_{fg}\\
	& -W^{-1}(\mathcal{F}^{2,1}+\mathcal{F}^{1,2}+2\mathcal{F}_{}^{1,(1),1})\big[\mathcal{H}^{-1}\mathcal{H}_{b}\mathcal{H}^{-1}\mathcal{E}\mathcal{H}^{-1}\big]_{fg}\\
	& +W^{-1}(\mathcal{F}^{2,1}+\mathcal{F}^{1,2}+2\mathcal{F}^{1,(1),1})\big[\mathcal{H}^{-1}\mathcal{E}_{b}\mathcal{H}^{-1}\big]_{fg}\\
	& -W^{-1}(\mathcal{F}^{2,1}+\mathcal{F}^{1,2}+2\mathcal{F}^{1,(1),1})\big[\mathcal{H}^{-1}\mathcal{E}\mathcal{H}^{-1}\mathcal{H}_{b}\mathcal{H}^{-1}\big]_{fg}\\
	& -2[\mathcal{H}^{-1}\mathcal{H}_{b}\mathcal{H}^{-1}\mathcal{H}_{s_{f}}\mathcal{H}^{-1}]_{\cdot,g}\mathcal{E}\mathcal{R}^{1}\\
	& +2[\mathcal{H}^{-1}\mathcal{H}_{bs_{f}}\mathcal{H}^{-1}]_{\cdot,g}\mathcal{E}\mathcal{R}^{1}\\
	& -2[\mathcal{H}^{-1}\mathcal{H}_{s_{f}}\mathcal{H}^{-1}\mathcal{H}_{b}\mathcal{H}^{-1}]_{\cdot,g}\mathcal{E}\mathcal{R}^{1}\\
	& +2[\mathcal{H}^{-1}\mathcal{H}_{s_{f}}\mathcal{H}^{-1}]_{\cdot,g}(\mathcal{E}_{b}\mathcal{R}^{1}+\mathcal{E}\mathcal{R}_{b}^{1})\\
	& +2[\mathcal{H}^{-1}\mathcal{H}_{b}\mathcal{H}^{-1}]_{\cdot,g}\mathcal{E}_{s_{f}}\mathcal{R}^{1}\\
	& -2[\mathcal{H}^{-1}]_{\cdot,g}(\mathcal{E}_{bs_{f}}\mathcal{R}^{1}+\mathcal{E}_{s_{f}}\mathcal{R}_{b}^{1})\\
	& -2[\mathcal{H}^{-1}\mathcal{H}_{b}\mathcal{H}^{-1}]_{\cdot,g}\mathcal{E}\mathcal{H}^{-1}\mathcal{H}_{s_{f}}\mathcal{R}^{1}\\
	& +2[\mathcal{H}^{-1}]_{\cdot,g}\mathcal{E}_{b}\mathcal{H}^{-1}\mathcal{H}_{s_{f}}\mathcal{R}^{1}\\
	& -2[\mathcal{H}^{-1}]_{\cdot,g}\mathcal{E}\mathcal{H}^{-1}\mathcal{H}_{b}\mathcal{H}^{-1}\mathcal{H}_{s_{f}}\mathcal{R}^{1}\\
	& +2[\mathcal{H}^{-1}]_{\cdot,g}\mathcal{E}\mathcal{H}^{-1}(\mathcal{H}_{bs_{f}}\mathcal{R}^{1}+\mathcal{H}_{s_{f}}\mathcal{R}_{b}^{1})\\
	& -2W^{-2}W_{b}[\mathcal{H}^{-1}]_{\cdot,g}\mathcal{E}\mathcal{R}^{2}[\mathcal{R}^{1}]_{f}\\
	& -2W^{-1}[\mathcal{H}^{-1}\mathcal{H}_{b}\mathcal{H}^{-1}]_{\cdot,g}\mathcal{E}\mathcal{R}^{2}[\mathcal{R}^{1}]_{f}\\
	& +2W^{-1}[\mathcal{H}^{-1}]_{\cdot,g}(\mathcal{E}_{b}\mathcal{R}^{2}[\mathcal{R}^{1}]_{f}+\mathcal{E}\mathcal{R}_{b}^{2}[\mathcal{R}^{1}]_{f}+\mathcal{E}\mathcal{R}^{2}[\mathcal{R}_{b}^{1}]_{f})\\
	& -2W^{-2}W_{b}[\mathcal{H}^{-1}]_{\cdot,g}\mathcal{E}\mathcal{H}^{-1}\mathcal{E}\mathcal{R}^{1}[\mathcal{R}^{1}]_{f}\\
	& +2W^{-1}[\mathcal{H}^{-1}\mathcal{H}_{b}\mathcal{H}^{-1}]_{\cdot,g}\mathcal{E}\mathcal{H}^{-1}\mathcal{E}\mathcal{R}^{1}[\mathcal{R}^{1}]_{f}\\
	& +2W^{-1}[\mathcal{H}^{-1}]_{\cdot,g}\mathcal{E}_{b}\mathcal{H}^{-1}\mathcal{E}\mathcal{R}^{1}[\mathcal{R}^{1}]_{f}\\
	& +2W^{-1}[\mathcal{H}^{-1}]_{\cdot,g}\mathcal{E}\mathcal{H}^{-1}\mathcal{H}_{b}\mathcal{H}^{-1}\mathcal{E}\mathcal{R}^{1}[\mathcal{R}^{1}]_{f}\\
	& +2W^{-1}[\mathcal{H}^{-1}]_{\cdot,g}\mathcal{E}\mathcal{H}^{-1}(\mathcal{E}_{b}\mathcal{R}^{1}[\mathcal{R}^{1}]_{f}+\mathcal{E}\mathcal{R}_{b}^{1}[\mathcal{R}^{1}]_{f}+\mathcal{E}\mathcal{R}^{1}[\mathcal{R}_{b}^{1}]_{f})\\
	& -2[\mathcal{H}^{-1}\mathcal{H}_{b}\mathcal{H}^{-1}\mathcal{E}\mathcal{H}^{-1}\mathcal{E}\mathcal{H}^{-1}]_{fg}-2[\mathcal{H}^{-1}\mathcal{E}\mathcal{H}^{-1}\mathcal{H}_{b}\mathcal{H}^{-1}\mathcal{E}\mathcal{H}^{-1}]_{fg}\\
	& -2[\mathcal{H}^{-1}\mathcal{E}\mathcal{H}^{-1}\mathcal{E}\mathcal{H}^{-1}\mathcal{H}_{b}\mathcal{H}^{-1}]_{fg}+2[\mathcal{H}^{-1}\mathcal{E}_{b}\mathcal{H}^{-1}\mathcal{E}\mathcal{H}^{-1}]_{fg}\\
	& +2[\mathcal{H}^{-1}\mathcal{E}\mathcal{H}^{-1}\mathcal{E}_{b}\mathcal{H}^{-1}]_{fg}-2W^{-2}W_{b}(\mathcal{R}^{2})'\mathcal{H}_{s_{g}}\mathcal{R}^{1}[\mathcal{R}^{1}]_{f}\\
	& +2W^{-1}(\mathcal{R}_{b}^{2})'\mathcal{H}_{s_{g}}\mathcal{R}^{1}[\mathcal{R}^{1}]_{f}+2W^{-1}(\mathcal{R}^{2})'\mathcal{H}_{bs_{g}}\mathcal{R}^{1}[\mathcal{R}^{1}]_{f}\\
	& +2W^{-1}(\mathcal{R}^{2})'\mathcal{H}_{s_{g}}\mathcal{R}_{b}^{1}[\mathcal{R}^{1}]_{f}+2W^{-1}(\mathcal{R}^{2})'\mathcal{H}_{s_{g}}\mathcal{R}^{1}[\mathcal{R}_{b}^{1}]_{f}\\
	& -2W^{-2}W_{b}(\mathcal{R}^{1})'\mathcal{E}\mathcal{H}^{-1}\mathcal{H}_{s_{g}}\mathcal{R}^{1}[\mathcal{R}^{1}]_{f}\\
	& +2W^{-1}(\mathcal{R}_{b}^{1})'\mathcal{E}\mathcal{H}^{-1}\mathcal{H}_{s_{g}}\mathcal{R}^{1}[\mathcal{R}^{1}]_{f}\\
	& +2W^{-1}(\mathcal{R}^{1})'\mathcal{E}_{b}\mathcal{H}^{-1}\mathcal{H}_{s_{g}}\mathcal{R}^{1}[\mathcal{R}^{1}]_{f}\\
	& -2W^{-1}(\mathcal{R}^{1})'\mathcal{E}\mathcal{H}^{-1}\mathcal{H}_{b}\mathcal{H}^{-1}\mathcal{H}_{s_{g}}\mathcal{R}^{1}[\mathcal{R}^{1}]_{f}\\
	& +2W^{-1}(\mathcal{R}^{1})'\mathcal{E}\mathcal{H}^{-1}\mathcal{H}_{bs_{g}}\mathcal{R}^{1}[\mathcal{R}^{1}]_{f}\\
	& +2W^{-1}(\mathcal{R}^{1})'\mathcal{E}\mathcal{H}^{-1}\mathcal{H}_{s_{g}}\mathcal{R}_{b}^{1}[\mathcal{R}^{1}]_{f}\\
	& +2W^{-1}(\mathcal{R}^{1})'\mathcal{E}\mathcal{H}^{-1}\mathcal{H}_{s_{g}}\mathcal{R}^{1}[\mathcal{R}_{b}^{1}]_{f}\\
	& -2[\mathcal{H}^{-1}\mathcal{H}_{b}\mathcal{H}^{-1}]_{\cdot,g}\mathcal{E}\mathcal{H}^{-1}\mathcal{H}_{s_{g}}\mathcal{R}^{1}\\
	& +2[\mathcal{H}^{-1}]_{\cdot,g}\mathcal{E}_{b}\mathcal{H}^{-1}\mathcal{H}_{s_{g}}\mathcal{R}^{1}-2[\mathcal{H}^{-1}]_{\cdot,g}\mathcal{E}\mathcal{H}^{-1}\mathcal{H}_{b}\mathcal{H}^{-1}\mathcal{H}_{s_{g}}\mathcal{R}^{1}\\
	& +2[\mathcal{H}^{-1}]_{\cdot,g}\mathcal{E}\mathcal{H}^{-1}\mathcal{H}_{bs_{g}}\mathcal{R}^{1}+2[\mathcal{H}^{-1}]_{\cdot,g}\mathcal{E}\mathcal{H}^{-1}\mathcal{H}_{s_{g}}\mathcal{R}_{b}^{1}\\
	& +2(\mathcal{R}_{b}^{1})'\mathcal{H}_{s_{f}}\mathcal{H}^{-1}\mathcal{H}_{s_{g}}\mathcal{R}^{1}+2(\mathcal{R}^{1})'\mathcal{H}_{bs_{f}}\mathcal{H}^{-1}\mathcal{H}_{s_{g}}\mathcal{R}^{1}\\
	& -2(\mathcal{R}^{1})'\mathcal{H}_{s_{f}}\mathcal{H}^{-1}\mathcal{H}_{b}\mathcal{H}^{-1}\mathcal{H}_{s_{g}}\mathcal{R}^{1}+2(\mathcal{R}^{1})'\mathcal{H}_{s_{f}}\mathcal{H}^{-1}\mathcal{H}_{bs_{g}}\mathcal{R}^{1}\\
	& +2(\mathcal{R}^{1})'\mathcal{H}_{s_{f}}\mathcal{H}^{-1}\mathcal{H}_{s_{g}}\mathcal{R}_{b}^{1}-(\mathcal{R}_{b}^{1})'\mathcal{H}_{s_{f}s_{g}}\mathcal{R}^{1}\\
	& -(\mathcal{R}^{1})'\mathcal{H}_{bs_{f}s_{g}}\mathcal{R}^{1}-(\mathcal{R}^{1})'\mathcal{H}_{s_{f}s_{g}}\mathcal{R}_{b}^{1}
\end{align*}

\subsection{Expressions for $\mathcal{G}$ terms}

Differentiating $\mathcal{G}=\mathcal{H}^{-1}(\partial_{\beta\phi}\mathcal{L})(\partial_{\beta\phi'}\mathcal{L})\mathcal{H}^{-1}$
with respect to $b$ and $s$ gives

\begin{align*}
\mathcal{G}_{b} & =\mathcal{H}^{-1}\mathcal{H}_{b}\mathcal{G}+\mathcal{G}\mathcal{H}_{b}\mathcal{H}^{-1}\\
 & +W^{-1}\mathcal{G}^{2,1}+W^{-1}\mathcal{H}^{-1}\mathcal{E}\mathcal{G}\\
 & +W^{-1}\mathcal{G}^{1,2}+W^{-1}\mathcal{G}\mathcal{E}\mathcal{H}^{-1}
\end{align*}

\begin{align*}
\mathcal{G}_{s_{g}} & =-\mathcal{H}^{-1}\mathcal{H}_{s_{g}}\mathcal{G}-\mathcal{G}\mathcal{H}_{s_{g}}\mathcal{H}^{-1}\\
 & -W^{-1}(\mathcal{G}^{2,1}+\mathcal{G}^{1,2})\big[\mathcal{H}^{-1}(\partial_{\beta\phi}\mathcal{L})\big]_{g}\\
 & +\mathcal{H}^{-1}\mathcal{E}^{1}\big[\mathcal{H}^{-1}\big]_{g}(\partial_{\beta^{t}\phi'}\mathcal{L})\mathcal{H}^{-1}\\
 & +\mathcal{H}^{-1}(\partial_{\beta^{s}\phi}\mathcal{L})\big[\mathcal{H}^{-1}\big]_{g}'\mathcal{E}^{1}\mathcal{H}^{-1}\\
 & -W^{-1}\mathcal{H}^{-1}\mathcal{E}^{1}\mathcal{G}\big[\mathcal{H}^{-1}(\partial_{\beta\phi}\mathcal{L})\big]_{g}\\
 & -W^{-1}\mathcal{G}\mathcal{E}^{1}\mathcal{H}^{-1}\big[\mathcal{H}^{-1}(\partial_{\beta\phi}\mathcal{L})\big]_{g}
\end{align*}
Differentiating a second time gives
\begin{align*}
\mathcal{G}_{bb} & =-\mathcal{H}^{-1}\mathcal{H}_{b}\mathcal{H}^{-1}\mathcal{H}_{b}\mathcal{G}+\mathcal{H}^{-1}\mathcal{H}_{bb}\mathcal{G}+\mathcal{H}^{-1}\mathcal{H}_{b}\mathcal{G}_{b}\\
 & +\mathcal{G}_{b}\mathcal{H}_{b}\mathcal{H}^{-1}+\mathcal{G}\mathcal{H}_{bb}\mathcal{H}^{-1}-\mathcal{G}\mathcal{H}_{b}\mathcal{H}^{-1}\mathcal{H}_{b}\mathcal{H}^{-1}\\
 & -W^{-2}W_{b}\mathcal{G}^{2,1}+W^{-1}\mathcal{G}_{b}^{2,1}-W^{-2}W_{b}\mathcal{G}^{1,2}+W^{-1}\mathcal{G}_{b}^{1,2}\\
 & -W^{-2}W_{b}\mathcal{H}^{-1}\mathcal{E}\mathcal{G}-W^{-2}W_{b}\mathcal{G}\mathcal{E}\mathcal{H}^{-1}\\
 & -W^{-1}\mathcal{H}^{-1}\mathcal{H}_{b}\mathcal{H}^{-1}\mathcal{E}\mathcal{G}-W^{-1}\mathcal{G}\mathcal{E}\mathcal{H}^{-1}\mathcal{H}_{b}\mathcal{H}^{-1}\\
 & +W^{-1}\mathcal{H}^{-1}\mathcal{E}_{b}\mathcal{G}\\
 & +W^{-1}\mathcal{G}\mathcal{E}_{b}\mathcal{H}^{-1}\\
 & +W^{-1}\mathcal{H}^{-1}\mathcal{E}\mathcal{G}_{b}+W^{-1}\mathcal{G}_{b}\mathcal{E}\mathcal{H}^{-1}
\end{align*}

\begin{align*}
\mathcal{G}_{bs_{g}} & =\mathcal{H}^{-1}\mathcal{H}_{s_{g}}\mathcal{H}^{-1}\mathcal{H}_{b}\mathcal{G}+\mathcal{G}\mathcal{H}_{b}\mathcal{H}^{-1}\mathcal{H}_{s_{g}}\mathcal{H}^{-1}\\
 & -\mathcal{H}^{-1}\mathcal{H}_{bs_{g}}\mathcal{G}-\mathcal{G}\mathcal{H}_{b_{s_{g}}}\mathcal{H}^{-1}\\
 & -\mathcal{H}^{-1}\mathcal{H}_{b}\mathcal{G}_{s_{g}}-\mathcal{G}_{s_{g}}\mathcal{H}_{b}\mathcal{H}^{-1}\\
 & +W^{-2}W_{s_{g}}\mathcal{G}^{2,1}-W^{-1}\mathcal{G}_{s_{g}}^{2,1}\\
 & +W^{-2}W_{s_{g}}\mathcal{G}^{1,2}-W^{-1}\mathcal{G}_{s_{g}}^{1,2}\\
 & +W^{-2}W_{s_{g}}(\mathcal{H}^{-1}\mathcal{E}\mathcal{G}+\mathcal{G}\mathcal{E}\mathcal{H}^{-1})\\
 & +W^{-1}\mathcal{H}^{-1}\mathcal{H}_{s_{g}}\mathcal{H}^{-1}\mathcal{E}\mathcal{G}+W^{-1}\mathcal{G}\mathcal{E}\mathcal{H}^{-1}\mathcal{H}_{s_{g}}\mathcal{H}^{-1}\\
 & -W^{-1}\mathcal{H}^{-1}\mathcal{E}_{s_{g}}\mathcal{G}-W^{-1}\mathcal{H}^{-1}\mathcal{E}\mathcal{G}_{s_{g}}\\
 & -W^{-1}\mathcal{G}\mathcal{E}_{s_{g}}\mathcal{H}^{-1}-W^{-1}\mathcal{G}_{s_{g}}\mathcal{E}\mathcal{H}^{-1}
\end{align*}

\begin{align*}
\mathcal{G}_{s_{f}s_{g}} & =\mathcal{H}^{-1}\mathcal{H}_{s_{f}}\mathcal{H}^{-1}\mathcal{H}_{s_{g}}\mathcal{G}+\mathcal{G}\mathcal{H}_{s_{g}}\mathcal{H}^{-1}\mathcal{H}_{s_{f}}\mathcal{H}^{-1}\\
 & -\mathcal{H}^{-1}\mathcal{H}_{s_{f}s_{g}}\mathcal{G}-\mathcal{G}\mathcal{H}_{s_{f}s_{g}}\mathcal{H}^{-1}\\
 & -\mathcal{H}^{-1}\mathcal{H}_{s_{g}}\mathcal{G}_{s_{f}}-\mathcal{G}_{s_{f}}\mathcal{H}_{s_{g}}\mathcal{H}^{-1}\\
 & +W^{-2}W_{s_{f}}(\mathcal{G}^{2,1}+\mathcal{G}^{1,2})\big[\mathcal{H}^{-1}(\partial_{\beta\phi}\mathcal{L})\big]_{g}\\
 & -W^{-1}(\mathcal{G}_{s_{f}}^{2,1}+\mathcal{G}_{s_{f}}^{1,2})\big[\mathcal{H}^{-1}(\partial_{\beta\phi}\mathcal{L})\big]_{g}\\
 & +W^{-1}(\mathcal{G}^{2,1}+\mathcal{G}^{1,2})\big[\mathcal{H}^{-1}\mathcal{H}_{s_{f}}\mathcal{H}^{-1}(\partial_{\beta\phi}\mathcal{L})\big]_{g}\\
 & +W^{-2}(\mathcal{G}^{2,1}+\mathcal{G}^{1,2})\big[\mathcal{H}^{-1}(\partial_{\beta\beta\phi}\mathcal{L})\big]_{g}\big[\mathcal{H}^{-1}(\partial_{\beta\phi}\mathcal{L})\big]_{f}\\
 & +W^{-1}(\mathcal{G}^{2,1}+\mathcal{G}^{1,2})\big[\mathcal{H}^{-1}\mathcal{E}\mathcal{H}^{-1}\big]_{gf}\\
 & +W^{-2}(\mathcal{G}^{2,1}+\mathcal{G}^{1,2})\big[\mathcal{H}^{-1}\mathcal{E}\mathcal{H}^{-1}(\partial_{\beta\phi}\mathcal{L})\big]_{g}\big[\mathcal{H}^{-1}(\partial_{\beta\phi}\mathcal{L})\big]_{f}\\
 & -\mathcal{H}^{-1}\mathcal{H}_{s_{f}}\mathcal{H}^{-1}\mathcal{E}\big[\mathcal{H}^{-1}\big]_{g,\cdot}(\partial_{\beta\phi'}\mathcal{L})\mathcal{H}^{-1}\\
 & +\mathcal{H}^{-1}\mathcal{E}_{s_{f}}\big[\mathcal{H}^{-1}\big]_{g,\cdot}(\partial_{\beta\phi'}\mathcal{L})\mathcal{H}^{-1}\\
 & -\mathcal{H}^{-1}\mathcal{E}\big[\mathcal{H}^{-1}\mathcal{H}_{s_{f}}\mathcal{H}^{-1}\big]_{g,\cdot}(\partial_{\beta\phi'}\mathcal{L})\mathcal{H}^{-1}\\
 & -W^{-1}\mathcal{H}^{-1}\mathcal{E}^{1}\big[\mathcal{H}^{-1}\big]_{g}(\partial_{\beta\beta\phi'}\mathcal{L})\mathcal{H}^{-1}\big[\mathcal{H}^{-1}(\partial_{\beta\phi}\mathcal{L})\big]_{f}\\
 & -\mathcal{H}^{-1}\mathcal{E}^{1}\big[\mathcal{H}^{-1}\big]_{g}\big[\mathcal{H}^{-1}\mathcal{E}\mathcal{H}^{-1}\big]_{\cdot,f}\\
 & -W^{-1}\mathcal{H}^{-1}\mathcal{E}^{1}\big[\mathcal{H}^{-1}\big]_{g}(\partial_{\beta\phi'}\mathcal{L})\mathcal{H}^{-1}\mathcal{E}\mathcal{H}^{-1}\big[\mathcal{H}^{-1}(\partial_{\beta\phi}\mathcal{L})\big]_{f}\\
 & -\mathcal{H}^{-1}\mathcal{E}^{1}\big[\mathcal{H}^{-1}\big]_{g}(\partial_{\beta\phi'}\mathcal{L})\mathcal{H}^{-1}\mathcal{H}_{s_{f}}\mathcal{H}^{-1}\\
 & -\mathcal{H}^{-1}\mathcal{H}_{s_{f}}\mathcal{H}^{-1}(\partial_{\beta\phi}\mathcal{L})\big[\mathcal{H}^{-1}\big]_{g}'\mathcal{E}\mathcal{H}^{-1}\\
 & -W^{-1}\mathcal{H}^{-1}(\partial_{\beta\beta\phi}\mathcal{L})\big[\mathcal{H}^{-1}\big]_{g}'\mathcal{E}\mathcal{H}^{-1}\big[\mathcal{H}^{-1}(\partial_{\beta\phi}\mathcal{L})\big]_{f}\\
 & -\big[\mathcal{H}^{-1}\mathcal{E}\mathcal{H}^{-1}\big]_{f,\cdot}\big[\mathcal{H}^{-1}\big]_{g,\cdot}'\mathcal{E}\mathcal{H}^{-1}\\
 & -W^{-1}\mathcal{H}^{-1}\mathcal{E}\mathcal{H}^{-1}(\partial_{\beta\phi}\mathcal{L})\big[\mathcal{H}^{-1}\big]_{g}'\mathcal{E}\mathcal{H}^{-1}\big[\mathcal{H}^{-1}(\partial_{\beta\phi}\mathcal{L})\big]_{f}\\
 & -\mathcal{H}^{-1}(\partial_{\beta^{s}\phi}\mathcal{L})\big[\mathcal{H}^{-1}\mathcal{H}_{s_{f}}\mathcal{H}^{-1}\big]_{g}'\mathcal{E}\mathcal{H}^{-1}\\
 & +\mathcal{H}^{-1}(\partial_{\beta^{s}\phi}\mathcal{L})\big[\mathcal{H}^{-1}\big]_{g}'\mathcal{E}_{s_{f}}\mathcal{H}^{-1}\\
 & -\mathcal{H}^{-1}(\partial_{\beta^{s}\phi}\mathcal{L})\big[\mathcal{H}^{-1}\big]_{g}'\mathcal{E}^{1}\mathcal{H}^{-1}\mathcal{H}_{s_{f}}\mathcal{H}^{-1}\\
 & +W^{-2}W_{s_{g}}\mathcal{H}^{-1}\mathcal{E}\mathcal{G}\big[\mathcal{H}^{-1}(\partial_{\beta\phi}\mathcal{L})\big]_{g}\\
 & +W^{-1}\mathcal{H}^{-1}\mathcal{H}_{s_{f}}\mathcal{H}^{-1}\mathcal{E}\mathcal{G}\big[\mathcal{H}^{-1}(\partial_{\beta\phi}\mathcal{L})\big]_{g}\\
 & -W^{-1}\mathcal{H}^{-1}\mathcal{E}_{s_{f}}\mathcal{G}\big[\mathcal{H}^{-1}(\partial_{\beta\phi}\mathcal{L})\big]_{g}\\
 & -W^{-1}\mathcal{H}^{-1}\mathcal{E}\mathcal{G}_{s_{f}}\big[\mathcal{H}^{-1}(\partial_{\beta\phi}\mathcal{L})\big]_{g}\\
 & +W^{-1}\mathcal{H}^{-1}\mathcal{E}\mathcal{G}\big[\mathcal{H}^{-1}\mathcal{H}_{s_{f}}\mathcal{H}^{-1}(\partial_{\beta\phi}\mathcal{L})\big]_{g}\\
 & +W^{-2}\mathcal{H}^{-1}\mathcal{E}\mathcal{G}\big[\mathcal{H}^{-1}(\partial_{\beta\beta\phi}\mathcal{L})\big]_{g}\big[\mathcal{H}^{-1}(\partial_{\beta\phi}\mathcal{L})\big]_{f}\\
 & +W^{-1}\mathcal{H}^{-1}\mathcal{E}\mathcal{G}\big[\mathcal{H}^{-1}(\partial_{\beta\phi\phi'}\mathcal{L})\mathcal{H}^{-1}\big]_{g}\\
 & +W^{-2}\mathcal{H}^{-1}\mathcal{E}\mathcal{G}\big[\mathcal{H}^{-1}\mathcal{E}\mathcal{H}^{-1}(\partial_{\beta\phi}\mathcal{L})\big]_{g}\big[\mathcal{H}^{-1}(\partial_{\beta\phi}\mathcal{L})\big]_{f}\\
 & +W^{-2}W_{s_{f}}\mathcal{G}\mathcal{E}\mathcal{H}^{-1}\big[\mathcal{H}^{-1}(\partial_{\beta\phi}\mathcal{L})\big]_{g}\\
 & -W^{-1}\mathcal{G}_{s_{f}}\mathcal{E}\mathcal{H}^{-1}\big[\mathcal{H}^{-1}(\partial_{\beta\phi}\mathcal{L})\big]_{g}\\
 & -W^{-1}\mathcal{G}\mathcal{E}_{s_{f}}\mathcal{H}^{-1}\big[\mathcal{H}^{-1}(\partial_{\beta\phi}\mathcal{L})\big]_{g}\\
 & +W^{-1}\mathcal{G}\mathcal{E}\mathcal{H}^{-1}\mathcal{H}_{s_{f}}\mathcal{H}^{-1}\big[\mathcal{H}^{-1}(\partial_{\beta\phi}\mathcal{L})\big]_{g}\\
 & +W^{-1}\mathcal{G}\mathcal{E}\mathcal{H}^{-1}\big[\mathcal{H}^{-1}\mathcal{H}_{s_{f}}\mathcal{H}^{-1}(\partial_{\beta\phi}\mathcal{L})\big]_{g}\\
 & +W^{-2}\mathcal{G}\mathcal{E}\mathcal{H}^{-1}\big[\mathcal{H}^{-1}(\partial_{\beta\beta\phi}\mathcal{L})\big]_{g}\big[\mathcal{H}^{-1}(\partial_{\beta\phi}\mathcal{L})\big]_{f}\\
 & +W^{-1}\mathcal{G}\mathcal{E}\mathcal{H}^{-1}\big[\mathcal{H}^{-1}(\partial_{\beta\phi\phi'}\mathcal{L})\mathcal{H}^{-1}\big]_{g}\\
 & +W^{-2}\mathcal{G}\mathcal{E}\mathcal{H}^{-1}\big[\mathcal{H}^{-1}\mathcal{E}\mathcal{H}^{-1}(\partial_{\beta\phi}\mathcal{L})\big]_{g}\big[\mathcal{H}^{-1}(\partial_{\beta\phi}\mathcal{L})\big]_{f}
\end{align*}

Finally, the third derivatives are
\begin{align*}
\mathcal{G}_{bbb} & =2\mathcal{H}^{-1}\mathcal{H}_{b}\mathcal{H}^{-1}\mathcal{H}_{b}\mathcal{H}^{-1}\mathcal{H}_{b}\mathcal{G}-2\mathcal{H}^{-1}\mathcal{H}_{b}\mathcal{H}^{-1}\mathcal{H}_{b}\mathcal{G}_{b}\\
 & -\mathcal{H}^{-1}\mathcal{H}_{bb}\mathcal{H}^{-1}\mathcal{H}_{b}\mathcal{G}-2\mathcal{H}^{-1}\mathcal{H}_{b}\mathcal{H}^{-1}\mathcal{H}_{bb}\mathcal{G}\\
 & +\mathcal{H}^{-1}\mathcal{H}_{bbb}\mathcal{G}+2\mathcal{H}^{-1}\mathcal{H}_{bb}\mathcal{G}_{b}\\
 & +\mathcal{H}^{-1}\mathcal{H}_{b}\mathcal{G}_{bb}\\
 & +\mathcal{G}_{bb}\mathcal{H}_{b}\mathcal{H}^{-1}+2\mathcal{G}_{b}\mathcal{H}_{bb}\mathcal{H}^{-1}-2\mathcal{G}_{b}\mathcal{H}_{b}\mathcal{H}^{-1}\mathcal{H}_{b}\mathcal{H}^{-1}\\
 & +\mathcal{G}\mathcal{H}_{bbb}\mathcal{H}^{-1}-2\mathcal{G}\mathcal{H}_{bb}\mathcal{H}^{-1}\mathcal{H}_{b}\mathcal{H}^{-1}-\mathcal{G}\mathcal{H}_{b}\mathcal{H}^{-1}\mathcal{H}_{bb}\mathcal{H}^{-1}\\
 & +2\mathcal{G}\mathcal{H}_{b}\mathcal{H}^{-1}\mathcal{H}_{b}\mathcal{H}^{-1}\mathcal{H}_{b}\mathcal{H}^{-1}\\
 & +2W^{-3}W_{b}^{2}(\mathcal{G}^{2,1}+\mathcal{G}^{1,2})-W^{-2}W_{bb}(\mathcal{G}^{2,1}+\mathcal{G}^{1,2})-W^{-2}W_{b}(\mathcal{G}_{b}^{2,1}+\mathcal{G}_{b}^{1,2})\\
 & -W^{-2}W_{b}(\mathcal{G}_{b}^{2,1}+\mathcal{G}_{b}^{1,2})+W^{-1}(\mathcal{G}_{bb}^{2,1}+\mathcal{G}_{bb}^{1,2})\\
 & +2W^{-3}W_{b}^{2}\mathcal{H}^{-1}\mathcal{E}\mathcal{G}-W^{-2}W_{bb}\mathcal{H}^{-1}\mathcal{E}\mathcal{G}+W^{-2}W_{b}\mathcal{H}^{-1}\mathcal{H}_{b}\mathcal{H}^{-1}\mathcal{E}\mathcal{G}\\
 & -W^{-2}W_{b}\mathcal{H}^{-1}\mathcal{E}_{b}\mathcal{G}-W^{-2}W_{b}\mathcal{H}^{-1}\mathcal{E}\mathcal{G}_{b}\\
 & +2W^{-3}W_{b}^{2}\mathcal{G}\mathcal{E}\mathcal{H}^{-1}-W^{-2}W_{bb}\mathcal{G}\mathcal{E}\mathcal{H}^{-1}-W^{-2}W_{b}\mathcal{G}_{b}\mathcal{E}\mathcal{H}^{-1}\\
 & -W^{-2}W_{b}\mathcal{G}\mathcal{E}_{b}\mathcal{H}^{-1}+W^{-2}W_{b}\mathcal{G}\mathcal{E}\mathcal{H}^{-1}\mathcal{H}_{b}\mathcal{H}^{-1}\\
 & +W^{-2}W_{b}\mathcal{H}^{-1}\mathcal{H}_{b}\mathcal{H}^{-1}\mathcal{E}\mathcal{G}+2W^{-1}\mathcal{H}^{-1}\mathcal{H}_{b}\mathcal{H}^{-1}\mathcal{H}_{b}\mathcal{H}^{-1}\mathcal{E}\mathcal{G}\\
 & -W^{-1}\mathcal{H}^{-1}\mathcal{H}_{bb}\mathcal{H}^{-1}\mathcal{E}\mathcal{G}-W^{-1}\mathcal{H}^{-1}\mathcal{H}_{b}\mathcal{H}^{-1}\mathcal{E}_{b}\mathcal{G}-W^{-1}\mathcal{H}^{-1}\mathcal{H}_{b}\mathcal{H}^{-1}\mathcal{E}\mathcal{G}_{b}\\
 & +W^{-2}W_{b}\mathcal{G}\mathcal{E}\mathcal{H}^{-1}\mathcal{H}_{b}\mathcal{H}^{-1}-W^{-1}\mathcal{G}_{b}\mathcal{E}\mathcal{H}^{-1}\mathcal{H}_{b}\mathcal{H}^{-1}-W^{-1}\mathcal{G}\mathcal{E}_{b}\mathcal{H}^{-1}\mathcal{H}_{b}\mathcal{H}^{-1}\\
 & -W^{-1}\mathcal{G}\mathcal{E}\mathcal{H}^{-1}\mathcal{H}_{bb}\mathcal{H}^{-1}+W^{-1}\mathcal{G}\mathcal{E}\mathcal{H}^{-1}\mathcal{H}_{b}\mathcal{H}^{-1}\mathcal{H}_{b}\mathcal{H}^{-1}\\
 & -W^{-2}W_{b}\mathcal{H}^{-1}(\mathcal{E}_{b}\mathcal{G}+\mathcal{E}\mathcal{G}_{b})-W^{-1}\mathcal{H}^{-1}\mathcal{H}_{b}\mathcal{H}^{-1}(\mathcal{E}_{b}\mathcal{G}+\mathcal{E}\mathcal{G}_{b})\\
 & +W^{-1}\mathcal{H}^{-1}(\mathcal{E}_{bb}\mathcal{G}+\mathcal{E}\mathcal{G}_{bb})\\
 & -W^{-2}W_{b}(\mathcal{E}_{b}\mathcal{G}+\mathcal{E}\mathcal{G}_{b})\mathcal{H}^{-1}+W^{-1}(\mathcal{E}_{bb}\mathcal{G}+\mathcal{E}\mathcal{G}_{bb})\mathcal{H}^{-1}\\
 & -W^{-1}(\mathcal{E}_{b}\mathcal{G}+\mathcal{E}\mathcal{G}_{b})\mathcal{H}^{-1}\mathcal{H}_{b}\mathcal{H}^{-1}
\end{align*}

\subsection{Expressions for $W$ terms}

We begin by expressing the first three derivatives of $W$ in terms
of $\mathcal{F}$. To do this, first note that by the definition of
$\mathcal{F}$

\begin{align*}
W & =\partial_{\beta\beta}\mathcal{L}+(\partial_{\beta\phi'}\mathcal{L})\mathcal{H}^{-1}(\partial_{\beta\phi}\mathcal{L})\\
 & =(\partial_{\beta\beta}\mathcal{L})+\mathcal{F}
\end{align*}

Differentiating with respect to $b$ and $s$ gives
\begin{align*}
W_{b} & =-W^{-1}(\partial_{\beta\beta\beta}\mathcal{L})-W^{-1}\mathcal{F}^{2,1}+\mathcal{F}_{b}
\end{align*}

\begin{align*}
W_{s} & =-W^{-1}(\partial_{\beta\beta\beta}\mathcal{L})\mathcal{H}^{-1}(\partial_{\beta\phi}\mathcal{L})-\Big[\mathcal{H}^{-1}+W^{-1}\mathcal{G}\big](\partial_{\beta\beta\phi}\mathcal{L})+\mathcal{F}_{s}
\end{align*}

Differentiating a second time gives

(1)

\begin{align*}
W_{bb} & =W^{-2}W_{b}(\partial_{\beta\beta\beta}\mathcal{L})+W^{-2}(\partial_{\beta^{4}}\mathcal{L})+W^{-2}\mathcal{F}^{3,1}\\
 & +W^{-2}W_{b}\mathcal{F}^{2,1}-W^{-1}\mathcal{F}_{b}^{2,1}\\
 & +\mathcal{F}_{bb}
\end{align*}

(2)

\begin{align*}
W_{bs} & =W^{-2}W_{b}(\partial_{\beta\beta\beta}\mathcal{L})\mathcal{H}^{-1}(\partial_{\beta\phi}\mathcal{L})+W^{-2}\big[(\partial_{\beta^{4}}\mathcal{L})+\mathcal{F}^{3,1}\big]\mathcal{H}^{-1}(\partial_{\beta\phi}\mathcal{L})\\
 & +W^{-1}(\partial_{\beta\beta\beta}\mathcal{L})\mathcal{H}^{-1}\mathcal{H}_{b}\mathcal{H}^{-1}(\partial_{\beta\phi}\mathcal{L})+W^{-2}(\partial_{\beta\beta\beta}\mathcal{L})\mathcal{H}^{-1}\big[(\partial_{\beta\beta\phi}\mathcal{L})+(\partial_{\beta\phi\phi'}\mathcal{L})\mathcal{H}^{-1}(\partial_{\beta\phi}\mathcal{L})\big]\\
 & +W^{-1}\big[\mathcal{H}^{-1}+W^{-1}\mathcal{G}\big]\big[(\partial_{\beta\beta\beta\phi}\mathcal{L})+(\partial_{\beta\beta\phi\phi'}\mathcal{L})\mathcal{H}^{-1}(\partial_{\beta\phi}\mathcal{L})\big]\\
 & +\Big[\mathcal{H}^{-1}\mathcal{H}_{b}\mathcal{H}^{-1}+W^{-2}W_{b}\mathcal{G}-W^{-1}\mathcal{G}_{b}\big](\partial_{\beta\beta\phi}\mathcal{L})\\
 & +\mathcal{F}_{sb}
\end{align*}

(3)

\begin{align*}W_{ss_{g}} & =W^{-2}W_{s_{g}}(\partial_{\beta\beta\beta}\mathcal{L})\mathcal{R}^{1}+W^{-2}(\partial_{\beta^{4}}\mathcal{L})[\mathcal{R}^{1}]_{g}\mathcal{R}^{1}\\
	& +W^{-1}(\partial_{\beta^{3}\phi'}\mathcal{L})[\mathcal{H}^{-1}]_{\cdot,g}\mathcal{R}^{1}+W^{-2}(\partial_{\beta^{3}\phi'}\mathcal{L})\mathcal{R}^{1}[\mathcal{R}^{1}]_{g}\mathcal{R}^{1}\\
	& -W^{-1}(\partial_{\beta\beta\beta}\mathcal{L})\mathcal{H}^{-1}\mathcal{R}_{b}^{1}+\mathcal{H}^{-1}\mathcal{H}_{s_{g}}\mathcal{R}^{2}\\
	& +W^{-2}W_{s_{g}}\mathcal{G}(\partial_{\beta\beta\phi}\mathcal{L})-W^{-1}\mathcal{G}_{b}(\partial_{\beta\beta\phi}\mathcal{L})\\
	& +W^{-1}\Big[\mathcal{H}^{-1}+W^{-1}\mathcal{G}\big](\partial_{\beta^{3}\phi}\mathcal{L})[\mathcal{R}^{1}]_{g}\\
	& +\Big[\mathcal{H}^{-1}+W^{-1}\mathcal{G}\big](\partial_{\beta\beta\phi\phi'}\mathcal{L})[\mathcal{H}^{-1}]_{\cdot,g}\\
	& +W^{-1}\Big[\mathcal{H}^{-1}+W^{-1}\mathcal{G}\big](\partial_{\beta\beta\phi\phi'}\mathcal{L})\mathcal{R}^{1}[\mathcal{R}^{1}]_{g}\\
	& +\mathcal{F}_{ss_{g}}
\end{align*}

Third derivative:
\begin{align*}W_{bss_{g}} & =-2W^{-3}W_{b}W_{s_{g}}(\partial_{\beta\beta\beta}\mathcal{L})\mathcal{R}^{1}+W^{-2}W_{bs_{g}}(\partial_{\beta\beta\beta}\mathcal{L})\mathcal{R}^{1}\\
	& -W^{-3}W_{s_{g}}(\partial_{\beta^{4}}\mathcal{L})\mathcal{R}^{1}-W^{-3}W_{s_{g}}\mathcal{R}^{1}(\partial_{\beta^{3}\phi'}\mathcal{L})\mathcal{R}^{1}\\
	& +W^{-2}W_{s_{g}}(\partial_{\beta\beta\beta}\mathcal{L})\mathcal{R}_{b}^{1}\\
	& -2W^{-3}W_{b}(\partial_{\beta^{4}}\mathcal{L})[\mathcal{R}^{1}]_{g}\mathcal{R}^{1}\\
	& -W^{-3}(\partial_{\beta^{5}}\mathcal{L})[\mathcal{R}^{1}]_{g}\mathcal{R}^{1}-W^{-3}\mathcal{R}^{1}(\partial_{\beta^{4}\phi'}\mathcal{L})\mathcal{R}^{1}[\mathcal{R}^{1}]_{g}\\
	& +W^{-2}(\partial_{\beta^{4}}\mathcal{L})[\mathcal{R}_{b}^{1}]_{g}\mathcal{R}^{1}+W^{-2}(\partial_{\beta^{4}}\mathcal{L})[\mathcal{R}_{b}^{1}]_{g}\mathcal{R}^{1}\\
	& -W^{-2}W_{b}(\partial_{\beta^{3}\phi'}\mathcal{L})[\mathcal{H}^{-1}]_{\cdot,g}\mathcal{R}^{1}\\
	& -W^{-2}(\partial_{\beta^{4}\phi'}\mathcal{L})[\mathcal{H}^{-1}]_{\cdot,g}\mathcal{R}^{1}-W^{-2}(\mathcal{R}^{1})'(\partial_{\beta^{3}\phi\phi'}\mathcal{L})[\mathcal{H}^{-1}]_{\cdot,g}\mathcal{R}^{1}\\
	& -W^{-1}(\partial_{\beta^{3}\phi'}\mathcal{L})[\mathcal{H}^{-1}\mathcal{H}_{b}\mathcal{H}^{-1}]_{\cdot,g}\mathcal{R}^{1}+W^{-1}(\partial_{\beta^{3}\phi'}\mathcal{L})[\mathcal{H}^{-1}]_{\cdot,g}\mathcal{R}_{b}^{1}\\
	& -2W^{-3}W_{b}(\partial_{\beta^{3}\phi'}\mathcal{L})\mathcal{R}^{1}[\mathcal{R}^{1}]_{g}\mathcal{R}^{1}\\
	& -W^{-3}(\partial_{\beta^{4}\phi'}\mathcal{L})\mathcal{R}^{1}[\mathcal{R}^{1}]_{g}\mathcal{R}^{1}-W^{-3}(\mathcal{R}^{1})'(\partial_{\beta^{3}\phi\phi'}\mathcal{L})\mathcal{R}^{1}[\mathcal{R}^{1}]_{g}\mathcal{R}^{1}\\
	& +W^{-2}(\partial_{\beta^{3}\phi'}\mathcal{L})\mathcal{R}_{b}^{1}[\mathcal{R}^{1}]_{g}\mathcal{R}^{1}+W^{-2}(\partial_{\beta^{3}\phi'}\mathcal{L})\mathcal{R}^{1}[\mathcal{R}_{b}^{1}]_{g}\mathcal{R}^{1}\\
	& +W^{-2}(\partial_{\beta^{3}\phi'}\mathcal{L})\mathcal{R}^{1}[\mathcal{R}^{1}]_{g}\mathcal{R}_{b}^{1}\\
	& +W^{-2}W_{b}(\partial_{\beta\beta\beta}\mathcal{L})\mathcal{H}^{-1}\mathcal{R}_{b}^{1}\\
	& +W^{-2}(\partial_{\beta^{4}}\mathcal{L})\mathcal{H}^{-1}\mathcal{R}_{b}^{1}+W^{-2}(\partial_{\beta^{3}\phi'}\mathcal{L})\mathcal{R}^{1}\mathcal{H}^{-1}\mathcal{R}_{b}^{1}\\
	& +W^{-1}(\partial_{\beta\beta\beta}\mathcal{L})\mathcal{H}^{-1}\mathcal{H}_{b}\mathcal{H}^{-1}\mathcal{R}_{b}^{1}-W^{-1}(\partial_{\beta\beta\beta}\mathcal{L})\mathcal{H}^{-1}\mathcal{R}_{bb}^{1}\\
	& -\mathcal{H}^{-1}\mathcal{H}_{b}\mathcal{H}^{-1}\mathcal{H}_{s_{g}}\mathcal{R}^{2}+\mathcal{H}^{-1}\mathcal{H}_{bs_{g}}\mathcal{R}^{2}\\
	& -\mathcal{H}^{-1}\mathcal{H}_{s_{g}}\mathcal{R}_{b}^{2}\\
	& -2W^{-3}W_{b}W_{s_{g}}\mathcal{G}(\partial_{\beta\beta\phi}\mathcal{L})+W^{-2}W_{bs_{g}}\mathcal{G}(\partial_{\beta\beta\phi}\mathcal{L})\\
	& +W^{-2}W_{s_{g}}\mathcal{G}_{b}(\partial_{\beta\beta\phi}\mathcal{L})\\
	& -W^{-3}W_{s_{g}}\mathcal{G}(\partial_{\beta^{3}\phi}\mathcal{L})-W^{-3}W_{s_{g}}\mathcal{G}(\partial_{\beta\beta\phi\phi'}\mathcal{L})\mathcal{R}^{1}\\
	& +W^{-2}W_{b}\mathcal{G}_{b}(\partial_{\beta\beta\phi}\mathcal{L})-W^{-1}\mathcal{G}_{bb}(\partial_{\beta\beta\phi}\mathcal{L})\\
	& +W^{-2}\mathcal{G}_{b}(\partial_{\beta^{3}\phi}\mathcal{L})+W^{-2}\mathcal{G}_{b}(\partial_{\beta\beta\phi\phi'}\mathcal{L})\mathcal{R}^{1}\\
	& -W^{-2}W_{b}\Big[\mathcal{H}^{-1}+W^{-1}\mathcal{G}\big](\partial_{\beta^{3}\phi}\mathcal{L})[\mathcal{R}^{1}]_{g}\\
	& -W^{-1}\Big[\mathcal{H}^{-1}\mathcal{H}_{b}\mathcal{H}^{-1}+W^{-2}W_{b}\mathcal{G}-W^{-1}\mathcal{G}_{b}\big](\partial_{\beta^{3}\phi}\mathcal{L})[\mathcal{R}^{1}]_{g}\\
	& -W^{-2}\Big[\mathcal{H}^{-1}+W^{-1}\mathcal{G}\big](\partial_{\beta^{4}\phi}\mathcal{L})[\mathcal{R}^{1}]_{g}\\
	& -W^{-2}\Big[\mathcal{H}^{-1}+W^{-1}\mathcal{G}\big](\partial_{\beta^{3}\phi\phi'}\mathcal{L})\mathcal{R}^{1}[\mathcal{R}^{1}]_{g}\\
	& +W^{-1}\Big[\mathcal{H}^{-1}+W^{-1}\mathcal{G}\big](\partial_{\beta^{3}\phi}\mathcal{L})[\mathcal{R}_{b}^{1}]_{g}\\
	& -\Big[\mathcal{H}^{-1}\mathcal{H}_{b}\mathcal{H}^{-1}+W^{-2}W_{b}\mathcal{G}-W^{-1}\mathcal{G}_{b}\big](\partial_{\beta\beta\phi\phi'}\mathcal{L})[\mathcal{H}^{-1}]_{\cdot,g}\\
	& -W^{-1}\Big[\mathcal{H}^{-1}+W^{-1}\mathcal{G}\big](\partial_{\beta^{3}\phi\phi'}\mathcal{L})[\mathcal{H}^{-1}]_{\cdot,g}\\
	& -W^{-1}\Big[\mathcal{H}^{-1}+W^{-1}\mathcal{G}\big]\sum_{f}(\partial_{\beta^{2}\phi\phi'\phi_{f}}\mathcal{L})[\mathcal{H}^{-1}]_{\cdot,g}[\mathcal{R}^{1}]_{f}\\
	& -\Big[\mathcal{H}^{-1}+W^{-1}\mathcal{G}\big](\partial_{\beta\beta\phi\phi'}\mathcal{L})[\mathcal{H}^{-1}\mathcal{H}_{b}\mathcal{H}^{-1}]_{\cdot,g}\\
	& -W^{-2}W_{b}\Big[\mathcal{H}^{-1}+W^{-1}\mathcal{G}\big](\partial_{\beta\beta\phi\phi'}\mathcal{L})\mathcal{R}^{1}[\mathcal{R}^{1}]_{g}\\
	& -W^{-1}\Big[\mathcal{H}^{-1}\mathcal{H}_{b}\mathcal{H}^{-1}+W^{-2}W_{b}\mathcal{G}-W^{-1}\mathcal{G}_{b}\big](\partial_{\beta\beta\phi\phi'}\mathcal{L})\mathcal{R}^{1}[\mathcal{R}^{1}]_{g}\\
	& -W^{-2}\Big[\mathcal{H}^{-1}+W^{-1}\mathcal{G}\big](\partial_{\beta^{3}\phi\phi'}\mathcal{L})\mathcal{R}^{1}[\mathcal{R}^{1}]_{g}\\
	& -W^{-2}\Big[\mathcal{H}^{-1}+W^{-1}\mathcal{G}\big]\sum_{f}(\partial_{\beta^{2}\phi\phi'\phi_{f}}\mathcal{L})[\mathcal{R}^{1}]_{f}\mathcal{R}^{1}[\mathcal{R}^{1}]_{g}\\
	& +W^{-1}\Big[\mathcal{H}^{-1}+W^{-1}\mathcal{G}\big](\partial_{\beta\beta\phi\phi'}\mathcal{L})\mathcal{R}_{b}^{1}[\mathcal{R}^{1}]_{g}\\
	& +W^{-1}\Big[\mathcal{H}^{-1}+W^{-1}\mathcal{G}\big](\partial_{\beta\beta\phi\phi'}\mathcal{L})\mathcal{R}^{1}[\mathcal{R}_{b}^{1}]_{g}\\
	& +\mathcal{F}_{bss_{g}}
\end{align*}

\subsection{\label{sec:beta_expansion} Expansion for $\widehat{\beta}$}

In Section \ref{sec:full_expansion} we provide bounds for each of the above terms that
by Assumptions 1 and 2 can be shown to hold uniformly over a neighborhood
of the truth. Using these results, we will be able to show that
\begin{align*}
	\partial_{b}\mathcal{L}^{*}(0,0) & =\partial_{b}\mathcal{L}^{*}-(\partial_{bb}\mathcal{L}^{*})\mathcal{S}_{\beta}-(\partial_{bs'}\mathcal{L}^{*})\mathcal{S}\\
	& +\frac{1}{2}(\partial_{bbb}\mathcal{L}^{*})\mathcal{S}_{\beta}^{2}+(\partial_{bbs'}\mathcal{L}^{*})\mathcal{S}\mathcal{S}_{\beta}+\frac{1}{2}\mathcal{S}'(\partial_{bss'}\mathcal{L}^{*})\mathcal{S}\\
	& -\frac{1}{2}\mathcal{S}'(\partial_{bbss'}\mathcal{L}_{(1)}^{*})\mathcal{S}\mathcal{S}_{\beta}-\frac{1}{6}\sum_{g}\mathcal{S}'(\partial_{bss's_{g}}\mathcal{L}_{(1)}^{*})\mathcal{S}\mathcal{S}_{g}\\
	& +\frac{1}{4}\mathcal{S}'(\partial_{bbbss'}\mathcal{L}_{(1)}^{*}(\bar{b},\bar{s}))\mathcal{S}\mathcal{S}_{\beta}^{2}+\frac{1}{24}\sum_{f,g}\mathcal{S}'(\partial_{bss's_{f}s_{g}}\mathcal{L}_{(2)}^{*})\mathcal{S}\mathcal{S}_{f}\mathcal{S}_{g}\\
	& +o_{p}(N^{-2}) \\
	&=\partial_{b}\mathcal{L}^{*}-(\partial_{bb}\mathcal{L}^{*})\mathcal{S}_{\beta}-(\partial_{bs'}\mathcal{L}^{*})\mathcal{S}+\frac{1}{2}\mathcal{S}'(\partial_{bss'}\mathcal{L}^{*})\mathcal{S}+o_{p}(1)
\end{align*}
where $(\bar{b},\bar{s})$ are intermediate values between $(0,0)$
and $(\mathcal{S}_{\beta},\mathcal{S})$, which
under Assumptions 1 and 2 must lie in $\mathcal{B}(r_\beta,\beta_0)\times\mathcal{B}_q(r_\phi,\phi_0)$ for sufficiently large $N$ (see Corollary B.3 in FW16).

Hence, we have the first-order expansion

\begin{align*}
	\partial_{b}\mathcal{L}^{*}(0,0) & =W_{N}^{-1}\mathcal{S}_{\beta}+W_{N}^{-1}(\partial_{\beta\phi'}\mathcal{L})\mathcal{H}^{-1}\mathcal{S}+N\frac{1}{2}\mathcal{S}'\mathcal{H}^{-1}\mathcal{H}_{b}\mathcal{H}^{-1}\mathcal{S}+o_{p}(1)\\
	NW_{N}(\widehat{\beta}-\beta) & =(\partial_{\beta}\mathcal{L})+(\partial_{\beta\phi'}\mathcal{L})\mathcal{H}^{-1}\mathcal{S}+\frac{1}{2}\mathcal{S}'\mathcal{H}^{-1}(\partial_{\beta\phi\phi'}\mathcal{L})\mathcal{H}^{-1}\mathcal{S}\\
	& \quad+\frac{1}{2}\mathcal{S}'\mathcal{H}^{-1}\sum_{f}(\partial_{\phi\phi'\phi_{f}}\mathcal{L})\big[\mathcal{H}^{-1}(\partial_{\beta\phi}\mathcal{L})\big]_{f}\mathcal{H}^{-1}\mathcal{S}+o_{p}(1)
\end{align*}

\section{\label{sec:basic_lemmas}Lemmas for bounding individual components}

In order to provide bounds on the components in the asymptotic expansion above, we first demonstrate some basic results that are used frequently. We begin with a result from FW16 (Lemma S.6 in that paper). Through the proofs we apply the theorem with $q=16$ which is justified by Assumption 2.

\begin{lem}
\label{lem:S6} Let Assumptions 1 and
2 hold for the dyadic network M-estimator described in the main paper.
Let $\mathcal{B}(r_{\beta},\beta_{0})$ an\textup{d} $\mathcal{B}_{q}(r_{\phi},\phi_{0})$
be neighborhoods of the true parameter values, with $r_{\beta}=o(N^{-1/q-2\epsilon})$
and $r_{\phi}=o(N^{-2\epsilon})$, and let $4<q\leq16$ and $\epsilon=\frac{1}{2q}$.
The following conditions hold.

(i) for $s\leq3$
\begin{align*}
\frac{1}{N-1}\sum_{i}\sum_{j\ne i}\partial_{\beta}\ell_{ij} & =O_{p}(1),\quad\frac{1}{N-1}\sum_{i}\sum_{j\ne i}\partial_{\beta^{3}}\tilde{\ell}_{ij}=O_{p}(1)\\
\sup_{\beta\in\mathcal{B}(r_{\beta},\beta_{0})}\sup_{\phi\in\mathcal{B}_{q}(r_{\phi},\phi_{0})} & \frac{1}{N(N-1)}\sum_{i}\sum_{j\ne i}\partial_{\beta^{s}}\ell_{ij} = O_p(1)
\end{align*}

(ii) for $s\leq3$
\[
\sup_{\beta\in\mathcal{B}(r_{\beta},\beta_{0})}\sup_{\phi\in\mathcal{B}_{q}(r_{\phi},\phi_{0})}\frac{1}{N}\sum_{i}\lvert\frac{1}{N-1}\sum_{j\ne i}\partial_{\beta^{s}\pi}\ell_{ij}\rvert^{q}=O_{p}(1)
\]

(iii) for $s\leq3$, $t\leq5$ and $s+t\leq6$
\[
\sup_{\beta\in\mathcal{B}(r_{\beta},\beta_{0})}\sup_{\phi\in\mathcal{B}_{q}(r_{\phi},\phi_{0})}\max_{i}\frac{1}{N-1}\sum_{j\ne i}\lvert\partial_{\beta^{s}\pi^{t}}\ell_{ij}\rvert^{q}=O_{p}(N^{2\epsilon})
\]

(iv) 

\[
\frac{1}{N}\sum_{i}\lvert\frac{1}{\sqrt{N-1}}\sum_{j\ne i}\partial_{\pi}\tilde{\ell}_{ij}\rvert^{q}=O_{p}(1),\quad\frac{1}{N}\sum_{i}\lvert\frac{1}{\sqrt{N-1}}\sum_{j\ne i}\partial_{\beta\pi}\tilde{\ell}_{ij}\rvert^{2}=O_{p}(1)
\]

(v) for $s\leq3$, $t\leq5$ and $s+t\leq6$
\begin{align*}
\max_{i}\bar{E}\big[\partial_{\beta^{s}\pi^{t}}\tilde{\ell}_{ij}^{q}\big] & \leq C,\quad\max_{i}\bar{E}\Big[\Big(\frac{1}{\sqrt{N-1}}\sum_{j\ne i}\partial_{\beta^{s}\pi^{t}}\tilde{\ell}_{ij}\Big)^{q}\Big]\leq C
\end{align*}
\end{lem}

\subsection{Bounds for $\bar{\mathcal{H}}^{-1}$ and $\overline{W}_{N}^{-1}$}

Here we provide important results related to $\bar{\mathcal{H}}^{-1}$ and $\overline{W}_{N}^{-1}$. The objective function for the estimator is given by
\[
\mathcal{L}_{N}(\beta,\phi_{N})=\frac{1}{N-1}\sum_{i}\sum_{j\ne i}\ell(Z_{ij},\beta,\alpha_{i},\gamma_{j})-\frac{1}{N}b(v_{N}'\phi_{N})^{2}/2
\]
 As in FW16 and Dzemski 2019, we may write
\[
\bar{\mathcal{H}}=\begin{bmatrix}\bar{\mathcal{H}}_{\alpha\alpha}^{*} & \bar{\mathcal{H}}_{\alpha\gamma}^{*}\\
\bar{\mathcal{H}}_{\gamma\alpha}^{*} & \bar{\mathcal{H}}_{\gamma\gamma}^{*}
\end{bmatrix}+\frac{1}{N}bv_{N}v_{N}'
\]
where $\bar{\mathcal{H}}_{\alpha\alpha}^{*}$ and $\bar{\mathcal{H}}_{\gamma\gamma}^{*}$
are the diagonal matrices with elements
\begin{align*}
\big(\bar{\mathcal{H}}_{\alpha\alpha}^{*}\big)_{ii} & =-\frac{1}{N-1}\sum_{j\ne i}\bar{E}[\partial_{\pi^{2}}\ell_{ij}]\\
\big(\bar{\mathcal{H}}_{\gamma\gamma}^{*}\big)_{ii} & =-\frac{1}{N-1}\sum_{j\ne i}\bar{E}[\partial_{\pi^{2}}\ell_{ji}]
\end{align*}
and $\bar{\mathcal{H}}_{\alpha\gamma}^{*}=(\bar{\mathcal{H}}_{\gamma\alpha}^{*})'$
has off-diagonal entries $\big(\bar{\mathcal{H}}_{\alpha\gamma}^{*}\big)_{ij}=-\bar{E}[\partial_{\pi^{2}}\ell_{ij}]/(N-1)$
and zeroes in diagonal entries. Following Lemma A.1 in D19 and Lemma
D.1 in FW16, we may prove the following approximation result.
\begin{lem}
\label{lem:H_approx}Under Assumptions 1 and 2
\begin{align*}
\lVert\bar{\mathcal{H}}^{-1}-D^{-1}\rVert_{max} & =O_{p}(N^{-1})\\
\lVert\bar{\mathcal{H}}^{-1}\rVert_{q} & =O_{p}(1)
\end{align*}
where $D=diag\big(\bar{\mathcal{H}}_{\alpha\alpha}^{*},\bar{\mathcal{H}}_{\gamma\gamma}^{*}\big)$.
\end{lem}
The proof of this Lemma is identical to the versions in D19 and FW16.
We next bound some deviations between the sample average values and
conditional expectations of the Hessian matrix.
\begin{lem}
\label{lem:H_tilde}Let Assumptions 1 and 2 hold. Then
\begin{align*}
\lVert\mathcal{H}-\bar{\mathcal{H}}\rVert & =O_{p}(N^{-\frac{1}{2}+2\epsilon})\\
\lVert\mathcal{H}_{b}-\bar{\mathcal{H}}_{b}\rVert & =O_{p}(N^{-\frac{3}{2}+6\epsilon})
\end{align*}
\end{lem}
\begin{proof}
The proof of the first result follows the equivalent result in FW16.
We have
\[
\lVert\mathcal{H}-\bar{\mathcal{H}}\rVert\leq\lVert\partial_{\alpha\alpha}\mathcal{L}-\partial_{\alpha\alpha}\bar{\mathcal{L}}\rVert+2\lVert\partial_{\alpha\gamma}\mathcal{L}-\partial_{\alpha\gamma}\bar{\mathcal{L}}\rVert+\lVert\partial_{\gamma\gamma}\mathcal{L}-\partial_{\gamma\gamma}\bar{\mathcal{L}}\rVert
\]
To bound the first term, note that from Lemma \ref{lem:S6}
\begin{align*}
\bar{E}\lVert\partial_{\alpha\alpha}\mathcal{L}-\partial_{\alpha\alpha}\bar{\mathcal{L}}\rVert^{q} & =\bar{E}\Big[\max_{i}\Big(\frac{1}{N-1}\sum_{j\ne i}\partial_{\pi^{2}}\tilde{\ell}_{ij}\Big)^{q}\Big]\\
 & =O_{p}(N^{1-q/2})
\end{align*}
and hence $\lVert\partial_{\alpha\alpha}\mathcal{L}-\partial_{\alpha\alpha}\bar{\mathcal{L}}\rVert=O_{p}(N^{-\frac{1}{2}+\frac{1}{q}})$
and similarly for $\lVert\partial_{\gamma\gamma}\mathcal{L}-\partial_{\gamma\gamma}\bar{\mathcal{L}}\rVert$.
The bound $\lVert\partial_{\alpha\gamma}\mathcal{L}-\partial_{\alpha\gamma}\bar{\mathcal{L}}\rVert=O_{p}(N^{-\frac{1}{2}+\frac{1}{q}})$
follows from Lemma \ref{lem:S6}. This gives $\lVert\mathcal{H}-\bar{\mathcal{H}}\rVert=O_{p}(N^{-\frac{1}{2}+\frac{1}{q}})=O_{p}(N^{-\frac{1}{2}+2\epsilon})$.
For the second result, recall that by definition we have
\[
W\mathcal{H}_{b}=\partial_{\beta\phi\phi'}\mathcal{L}+\sum_{f}(\partial_{\phi\phi'\phi_{f}}\mathcal{L})\big[\mathcal{H}^{-1}(\partial_{\beta\phi}\mathcal{L})\big]_{f}
\]
The bound $\lVert\partial_{\beta\phi\phi'}\tilde{\mathcal{L}}\rVert=O_{p}(N^{-\frac{1}{2}+2\epsilon})$
follows identically to the first result. For the second part
\begin{align*}
\lVert & \sum_{f}(\partial_{\phi\phi'\phi_{f}}\mathcal{L})\big[\mathcal{H}^{-1}(\partial_{\beta\phi}\mathcal{L})\big]_{f}-\sum_{f}(\partial_{\phi\phi'\phi_{f}}\bar{\mathcal{L}})\big[\bar{\mathcal{H}}^{-1}(\partial_{\beta\phi}\mathcal{\bar{L}})\big]_{f}\rVert\\
\leq & \rVert\sum_{f}(\partial_{\phi\phi'\phi_{f}}\tilde{\mathcal{L}})\big[\mathcal{H}^{-1}(\partial_{\beta\phi}\mathcal{L})\big]_{f}\rVert+\rVert\sum_{f}(\partial_{\phi\phi'\phi_{f}}\bar{\mathcal{L}})\big[(\mathcal{H}^{-1}-\bar{\mathcal{H}}^{-1})(\partial_{\beta\phi}\mathcal{L})\big]_{f}\rVert\\
 & +\rVert\sum_{f}(\partial_{\phi\phi'\phi_{f}}\bar{\mathcal{L}})\big[\mathcal{\bar{H}}^{-1}(\partial_{\beta\phi}\tilde{\mathcal{L}})\big]_{f}\rVert
\end{align*}
For the first term, we can write 
\begin{align*}
\rVert\sum_{f}(\partial_{\phi\phi'\phi_{f}}\tilde{\mathcal{L}})\big[\mathcal{H}^{-1}(\partial_{\beta\phi}\mathcal{L})\big]_{f}\rVert & \leq\rVert\sum_{f}(\partial_{\phi\phi'\phi_{f}}\tilde{\mathcal{L}})\big[\bar{\mathcal{H}}^{-1}(\partial_{\beta\phi}\mathcal{L})\big]_{f}\rVert+\rVert\partial_{\phi\phi\phi}\tilde{\mathcal{L}}\rVert\lVert\mathcal{H}^{-1}-\bar{\mathcal{H}}^{-1}\rVert\lVert\partial_{\beta\phi}\mathcal{L}\rVert\\
 & =\rVert\sum_{f}(\partial_{\phi\phi'\phi_{f}}\tilde{\mathcal{L}})\big[\bar{\mathcal{H}}^{-1}(\partial_{\beta\phi}\mathcal{L})\big]_{f}\rVert+O_{p}(N^{-\frac{1}{2}+4\epsilon})
\end{align*}
where the bound on $\rVert\partial_{\phi\phi\phi}\tilde{\mathcal{L}}\rVert$
follows from the same reasoning as the bound on $\rVert\partial_{\phi\phi}\tilde{\mathcal{L}}\rVert$,
applying the result on tensor norms in Lemma S.5 of FW16. Decompose
into four components based on $\phi=(\alpha',\gamma')'$ and consider
the first component, and let $\Gamma_{isjt}=\bar{\mathcal{H}}_{\alpha\alpha,i,s}^{-1}+\bar{\mathcal{H}}_{\alpha\gamma,i,t}^{-1}+\bar{\mathcal{H}}_{\alpha\alpha,j,s}^{-1}+\bar{\mathcal{H}}_{\alpha\gamma,j,t}^{-1}$
\begin{align*}
\rVert\sum_{f}(\partial_{\alpha\alpha'\phi_{f}}\tilde{\mathcal{L}})\big[\bar{\mathcal{H}}^{-1}(\partial_{\beta\phi}\mathcal{L})\big]_{f}\rVert & =\rVert\sum_{f}(\partial_{\alpha\alpha'\phi_{f}}\tilde{\mathcal{L}})\big[\bar{\mathcal{H}}^{-1}(\partial_{\beta\phi}\mathcal{L})\big]_{f}\rVert\\
 & =\rVert\frac{1}{N-1}\sum_{i}\sum_{s}\sum_{t\ne s}(\partial_{\alpha\alpha'\alpha_{i}}\tilde{\mathcal{L}})\big[\bar{\mathcal{H}}_{\alpha\alpha,i,s}^{-1}+\bar{\mathcal{H}}_{\alpha\gamma,i,t}^{-1}\big](\partial_{\beta\pi}\ell_{st})\\
 & +\frac{1}{N-1}\sum_{i}\sum_{s}\sum_{t\ne s}(\partial_{\alpha\alpha'\gamma_{i}}\tilde{\mathcal{L}})\big[\bar{\mathcal{H}}_{\gamma\alpha,i,s}^{-1}+\bar{\mathcal{H}}_{\gamma\gamma,i,t}^{-1}\big](\partial_{\beta\pi}\ell_{st})\rVert\\
 & =\max_{i}\lvert\frac{1}{(N-1)^{2}}\sum_{j\ne i}\sum_{s}\sum_{t\ne s}(\partial_{\pi^{3}}\tilde{\ell}_{ij})\Gamma_{isjt}(\partial_{\beta\pi}\ell_{st})\rvert\\
 & \leq\max_{i}\lvert\frac{1}{(N-1)^{2}}\sum_{j\ne i}\sum_{s}\sum_{t\ne s}(\partial_{\pi^{3}}\tilde{\ell}_{ij})\Gamma_{isjt}(\partial_{\beta\pi}\ell_{st})\rvert\\
\end{align*}
Since the max over $(i,s,j,t)$ of $\Gamma_{isjt}=O_{p}(1)$ when
$s=i$, $t=i$, $j=s$, or $j=t$ and is $O_{p}(N^{-1})$ otherwise,
the above sum is $O_{p}(N^{-\frac{1}{2}+2\epsilon})$ using the bounds
in Lemma \ref{lem:S6}. The same bound can be shown for remaining
components of $\rVert\sum_{f}(\partial_{\phi\phi'\phi_{f}}\tilde{\mathcal{L}})\big[\bar{\mathcal{H}}^{-1}(\partial_{\beta\phi}\mathcal{L})\big]_{f}\rVert$.
Next, from Lemmas \ref{lem:tensor_bounds} and \ref{lem:H1_approx}
\begin{align*}
\rVert\sum_{f}(\partial_{\phi\phi'\phi_{f}}\bar{\mathcal{L}})\big[(\mathcal{H}^{-1}-\bar{\mathcal{H}}^{-1})(\partial_{\beta\phi}\mathcal{L})\big]_{f}\rVert & \leq\rVert\sum_{f}(\partial_{\phi\phi'\phi_{f}}\bar{\mathcal{L}})\big[\bar{\mathcal{H}}^{-1}\tilde{\mathcal{H}}\bar{\mathcal{H}}^{-1}(\partial_{\beta\phi}\mathcal{L})\big]_{f}\rVert\\
 & +\rVert\sum_{f}(\partial_{\phi\phi'\phi_{f}}\bar{\mathcal{L}})\big[(\mathcal{H}^{-1}-\bar{\mathcal{H}}^{-1}-\bar{\mathcal{H}}^{-1}\tilde{\mathcal{H}}\bar{\mathcal{H}}^{-1})(\partial_{\beta\phi}\mathcal{L})\big]_{f}\rVert\\
 & \leq O_{p}(N^{-\frac{1}{2}+4\epsilon})+\rVert\partial_{\phi\phi\phi}\bar{\mathcal{L}}\rVert\lVert\mathcal{H}^{-1}-\bar{\mathcal{H}}^{-1}-\bar{\mathcal{H}}^{-1}\tilde{\mathcal{H}}\bar{\mathcal{H}}^{-1}\rVert\lVert\partial_{\beta\phi}\mathcal{L}\rVert\\
 & =O_{p}(N^{-\frac{1}{2}+6\epsilon})
\end{align*}

and $\rVert\sum_{f}(\partial_{\phi\phi'\phi_{f}}\bar{\mathcal{L}})\big[\mathcal{\bar{H}}^{-1}(\partial_{\beta\phi}\tilde{\mathcal{L}})\big]_{f}\rVert=O_{p}(N^{-1/2+2\epsilon})$
as in Lemma \ref{lem:tensor_bounds}. Since $\lVert W\rVert=O_{p}(N^{-1})$
and $\lVert W_{N}-\overline{W}_{N}\rVert=O_{p}(N^{-1/2+2\epsilon})$, this
gives the result.
\end{proof}
In the proofs below we will frequently replace $\mathcal{H}^{-1}$
and $W^{-1}$ by the following approximations.
\begin{lem}
\label{lem:H1_approx}Let Assumptions 1 and 2 hold. Then, for $k=0,1,2,3$
\begin{align*}
\lVert\mathcal{H}^{-1}-\sum_{j=0}^{k}(-1)^{j}(\bar{\mathcal{H}}^{-1}\tilde{\mathcal{H}})^{j}\bar{\mathcal{H}}^{-1}\rVert & =O_{p}(N^{-\frac{k+1}{2}+2(k+1)\epsilon})\\
W_N - \overline{W}_{N} &=O_{p}(N^{-\frac{1}{2}+2\epsilon}) \\
\lVert W_{N}^{-1}-\sum_{j=0}^{k}(-1)^{j}(\overline{W}_N^{-1}\tilde{W}_N)^{j}\overline{W}_N^{-1}\rVert & =O_{p}(N^{-\frac{k+1}{2}+2(k+1)\epsilon})
\end{align*}
\end{lem}
\begin{proof}
The approximations are based on the Neumann series expansion. Since
$\lVert\tilde{\mathcal{H}}\rVert=O_{p}(N^{-\frac{1}{2}+2\epsilon})=o_{p}(1)$,
the series converges with high probability for sufficiently large
$N$, and so we may write
\begin{align*}
\mathcal{H}^{-1} & =\sum_{j=0}^{\infty}(-1)^{j}(\bar{\mathcal{H}}^{-1}\tilde{\mathcal{H}})^{j}\bar{\mathcal{H}}^{-1}\\
 & =\sum_{j=0}^{k}(-1)^{j}(\bar{\mathcal{H}}^{-1}\tilde{\mathcal{H}})^{j}\bar{\mathcal{H}}^{-1}+\sum_{j=k+1}^{\infty}(-1)^{j}(\bar{\mathcal{H}}^{-1}\tilde{\mathcal{H}})^{j}\bar{\mathcal{H}}^{-1}\\
 & =\mathbf{H}_{k}+\zeta_{k+1}
\end{align*}
where $\mathbf{H}_{k}$ is the approximation for $\mathcal{H}^{-1}$
up to $k+1$ terms, and the approximation error satisfies
\begin{align*}
\lVert\zeta_{k+1}\rVert & =\lVert\sum_{j=k+1}^{\infty}(-1)^{j}(\bar{\mathcal{H}}^{-1}\tilde{\mathcal{H}})^{j}\bar{\mathcal{H}}^{-1}\rVert\\
 & \leq\lVert(\bar{\mathcal{H}}^{-1}\tilde{\mathcal{H}})^{k+1}\rVert\lVert\sum_{j=0}^{\infty}(-1)^{j}(\bar{\mathcal{H}}^{-1}\tilde{\mathcal{H}})^{j}\bar{\mathcal{H}}^{-1}\rVert\\
 & \leq\lVert\tilde{\mathcal{H}}\rVert^{k+1}\lVert\mathcal{\bar{H}}^{-1}\rVert^{k+1}\lVert\mathcal{H}^{-1}\rVert\\
 & =O_{p}(N^{-\frac{k+1}{2}-2(k+1)\epsilon})
\end{align*}

For $W$, define $\overline{W}=\partial_{\beta\beta}\bar{\mathcal{L}}+\frac{1}{N}(\partial_{\beta\phi'}\bar{\mathcal{L}})\bar{\mathcal{H}}^{-1}(\partial_{\beta\phi}\bar{\mathcal{L}})$
and
\[
\tilde{W}_{N}=\frac{1}{N}\partial_{\beta\beta}\mathcal{\tilde{L}}+\frac{1}{N}\big((\partial_{\beta\phi'}\mathcal{L})\mathcal{H}^{-1}(\partial_{\beta\phi}\mathcal{L})-(\partial_{\beta\phi'}\bar{\mathcal{L}})\bar{\mathcal{H}}^{-1}(\partial_{\beta\phi}\bar{\mathcal{L}})\big)
\]
The first term in the approximation error $\frac{1}{N}\partial_{\beta\beta}\mathcal{\tilde{L}}=O_{p}(N^{-1})$
by Lemma \ref{lem:S6}. For the remaining term, we can decompose it
as
\begin{align*}
\frac{1}{N}(\partial_{\beta\phi'}\mathcal{L}) & \mathcal{H}^{-1}(\partial_{\beta\phi}\mathcal{L})-\frac{1}{N}(\partial_{\beta\phi'}\bar{\mathcal{L}})\bar{\mathcal{H}}^{-1}(\partial_{\beta\phi}\bar{\mathcal{L}})\\
 & =\frac{1}{N}(\partial_{\beta\phi'}\tilde{\mathcal{L}})\bar{\mathcal{H}}^{-1}(\partial_{\beta\phi}\bar{\mathcal{L}})+\frac{1}{N}(\partial_{\beta\phi'}\bar{\mathcal{L}})\bar{\mathcal{H}}^{-1}(\partial_{\beta\phi}\tilde{\mathcal{L}})\\
 & +\frac{1}{N}(\partial_{\beta\phi'}\tilde{\mathcal{L}})\bar{\mathcal{H}}^{-1}(\partial_{\beta\phi}\tilde{\mathcal{L}})+\frac{1}{N}(\partial_{\beta\phi'}\mathcal{L})\big(\mathcal{H}^{-1}-\bar{\mathcal{H}}^{-1}\big)(\partial_{\beta\phi}\mathcal{L})
\end{align*}
From Assumption B.1 of FW16, we have $\frac{1}{N}\lVert(\partial_{\beta\phi'}\tilde{\mathcal{L}})\bar{\mathcal{H}}^{-1}(\partial_{\beta\phi}\bar{\mathcal{L}})\rVert\leq O_{p}(N^{-1/2})$
and also that $\frac{1}{N}\lVert(\partial_{\beta\phi'}\tilde{\mathcal{L}})\bar{\mathcal{H}}^{-1}(\partial_{\beta\phi}\tilde{\mathcal{L}})\rVert\leq O_{p}(N^{-1})$.
Also
\begin{align*}
\frac{1}{N}\lVert(\partial_{\beta\phi'}\mathcal{L})\big(\mathcal{H}^{-1}-\bar{\mathcal{H}}^{-1}\big)(\partial_{\beta\phi}\mathcal{L})\rVert & \leq\frac{1}{N}\lVert\partial_{\beta\phi'}\mathcal{L}\rVert^{2}\lVert\mathcal{H}^{-1}-\bar{\mathcal{H}}^{-1}\rVert\\
 & =O_{p}(N^{-\frac{1}{2}+2\epsilon})
\end{align*}
So we may write $\tilde{W}_{N}=O_{p}(N^{-\frac{1}{2}+2\epsilon})$,
and noting that $\overline{W}_{N}>0$ by Assumption 2, the approximation follows as for $\mathcal{H}^{-1}$.
\end{proof}
\begin{lem}
\label{lem:W_Op1}Let
\[
W_{N}=\frac{1}{N}W=\frac{1}{N}\partial_{\beta\beta}\mathcal{L}+\frac{1}{N}(\partial_{\beta\phi'}\mathcal{L})\mathcal{H}^{-1}(\partial_{\beta\phi}\mathcal{L})
\]
then, under Assumptions 1 and 2
\[
\lVert W_{N}\rVert=O_{p}(1)
\]
\end{lem}
\begin{proof}
Using the result in Lemma \ref{lem:H_tilde}, we have
\[
W_{N}=\frac{1}{N}W=\frac{1}{N}\partial_{\beta\beta}\mathcal{L}+\frac{1}{N}(\partial_{\beta\phi'}\mathcal{L})\bar{\mathcal{H}}^{-1}(\partial_{\beta\phi}\mathcal{L})
\]
The first term is $O_{p}(1)$, while the second term can be decomposed
into four parts
\begin{align*}
(\partial_{\beta\phi'}\mathcal{L})\mathcal{H}^{-1}(\partial_{\beta\phi}\mathcal{L}) & =(\partial_{\beta\alpha'}\mathcal{L})\mathcal{H}_{\alpha\alpha}^{-1}(\partial_{\beta\alpha}\mathcal{L})+(\partial_{\beta\alpha'}\mathcal{L})\mathcal{H}_{\alpha\gamma}^{-1}(\partial_{\beta\gamma}\mathcal{L})\\
 & +(\partial_{\beta\gamma'}\mathcal{L})\mathcal{H}_{\gamma\alpha}^{-1}(\partial_{\beta\alpha}\mathcal{L})+(\partial_{\beta\gamma'}\mathcal{L})\mathcal{H}_{\gamma\gamma}^{-1}(\partial_{\beta\gamma}\mathcal{L})
\end{align*}
The first part is
\begin{align*}
\frac{1}{N}(\partial_{\beta\alpha'}\mathcal{L})\mathcal{H}_{\alpha\alpha}^{-1}(\partial_{\beta\alpha}\mathcal{L}) & =\frac{1}{N(N-1)^{2}}\sum_{i,s}(\mathcal{H}_{\alpha\alpha}^{-1})_{is}\sum_{j\ne i}\sum_{t\ne s}(\partial_{\beta\pi}\ell_{ij})(\partial_{\beta\pi}\ell_{st})\\
 & =\frac{1}{N(N-1)^{2}}\sum_{i}(\mathcal{H}_{\alpha\alpha}^{-1})_{ii}\sum_{j\ne i}\sum_{t\ne i}(\partial_{\beta\pi}\ell_{ij})(\partial_{\beta\pi}\ell_{it})\\
 & +\frac{1}{N^{2}(N-1)^{2}}\sum_{i}\sum_{s\ne i}N(\mathcal{H}_{\alpha\alpha}^{-1})_{is}\sum_{j\ne i}\sum_{t\ne s}(\partial_{\beta\pi}\ell_{ij})(\partial_{\beta\pi}\ell_{st})\\
 & =O_{p}(1)
\end{align*}
where we use Lemma \ref{lem:H_approx}. Similar derivations for the
other parts give the result.
\end{proof}

\subsection{Products of matrices}
The next result demonstrates that the properties of $\mathcal{\bar{H}}^{-1}$ (a dominant diagonal and smaller off-diagonal terms) transfer to certain products of this matrix.
\begin{lem}
\label{lem:HAH_bound}Let Assumptions 1 and 2 hold, and define $A=\mathcal{\bar{H}}^{-1}M\mathcal{\bar{H}}^{-1}$,
where $M$ is a matrix of the form
\[
M=\begin{bmatrix}M_{\alpha\alpha} & M_{\alpha\gamma}\\
M_{\gamma\alpha} & M_{\gamma\gamma}
\end{bmatrix}
\]
with each block satisfying $\max_{i}\sum_{j}\vert M_{\alpha\alpha,ij}\vert=O_{p}(1)$.
Then $A_{ii}=O_{p}(1)$, and $\max_{i,j\ne i}\vert A_{ij}\vert=O_{p}(N^{-1})$.
\end{lem}
\begin{proof}
For $i\leq N$, we can write
\begin{align*}
\big[\mathcal{\bar{H}}^{-1}M\mathcal{\bar{H}}^{-1}\big]_{ij} & =(\bar{\mathcal{H}}_{\alpha\alpha}^{-1})_{i,\cdot}M_{\alpha\alpha'}(\bar{\mathcal{H}}_{\alpha\alpha}^{-1})_{\cdot,j}+(\bar{\mathcal{H}}_{\alpha\gamma}^{-1})_{i,\cdot}M_{\gamma\alpha'}(\bar{\mathcal{H}}_{\alpha\alpha}^{-1})_{\cdot,j}\\
 & \quad+(\bar{\mathcal{H}}_{\alpha\alpha}^{-1})_{i,\cdot}M_{\alpha\gamma}(\bar{\mathcal{H}}_{\gamma\alpha}^{-1})_{\cdot,j}+(\bar{\mathcal{H}}_{\alpha\gamma}^{-1})_{i,\cdot}M_{\gamma\gamma'}(\bar{\mathcal{H}}_{\gamma\alpha}^{-1})_{\cdot,j}
\end{align*}
For the first component, we have that
\begin{align*}
\max_{i,j\ne i}\lvert(\bar{\mathcal{H}}_{\alpha\alpha}^{-1})_{i,\cdot}M_{\alpha\alpha'}(\bar{\mathcal{H}}_{\alpha\alpha}^{-1})_{\cdot,j}\rvert & =\max_{i,j\ne i}\lvert\sum_{s,t}(\bar{\mathcal{H}}_{\alpha\alpha}^{-1})_{i,s}(\bar{\mathcal{H}}_{\alpha\alpha}^{-1})_{t,j}M_{\alpha\alpha,st}\rvert\\
 & \leq\max_{i,s\ne i}\vert(\bar{\mathcal{H}}_{\alpha\alpha}^{-1})_{i,s}\vert\max_{j,t\ne j}\vert(\bar{\mathcal{H}}_{\alpha\alpha}^{-1})_{t,j}\vert\max_{i,j}\sum_{s\ne i}\sum_{t\ne j}\vert M_{\alpha\alpha,st}\rvert\\
 & +\max_{i,s\ne i}\vert(\bar{\mathcal{H}}_{\alpha\alpha}^{-1})_{i,s}\vert\max_{j}\vert(\bar{\mathcal{H}}_{\alpha\alpha}^{-1})_{j,j}\vert\max_{i,j}\sum_{s\ne i}\vert M_{\alpha\alpha,sj}\rvert\\
 & +\max_{i}\vert(\bar{\mathcal{H}}_{\alpha\alpha}^{-1})_{i,i}\vert\max_{t,t\ne j}\vert(\bar{\mathcal{H}}_{\alpha\alpha}^{-1})_{t,j}\vert\max_{i,j}\sum_{t\ne j}\vert M_{\alpha\alpha,it}\rvert\\
 & +\max_{i}\vert(\bar{\mathcal{H}}_{\alpha\alpha}^{-1})_{i,i}\vert^{2}\max_{i,j}\vert M_{\alpha\alpha,ij}\vert\\
 & =1\{i=j\}\times O_{p}(1)+O_{p}(N^{-1})
\end{align*}
where the final line follows from the condition $\max_{i}\sum_{j}\vert M_{\alpha\alpha,ij}\vert=O_{p}(1)$,
$\max_{i}(\bar{\mathcal{H}}_{\alpha\alpha}^{-1})_{ii}=O_{p}(1)$,
and $\max_{i,j\ne i}(\bar{\mathcal{H}}_{\alpha\alpha}^{-1})_{ij}=O_{p}(N^{-1})$.
The next component is
\begin{align*}
\max_{i,j\ne i}\lvert(\bar{\mathcal{H}}_{\alpha\alpha}^{-1})_{i,\cdot}M_{\alpha\gamma'}(\bar{\mathcal{H}}_{\gamma\alpha}^{-1})_{\cdot,j}\rvert & =\max_{i,j\ne i}\lvert\sum_{s,t}(\bar{\mathcal{H}}_{\alpha\alpha}^{-1})_{i,s}(\bar{\mathcal{H}}_{\gamma\alpha}^{-1})_{t,j}M_{\alpha\gamma,st}\rvert\\
 & \leq\max_{i,s\ne i}\vert(\bar{\mathcal{H}}_{\alpha\alpha}^{-1})_{i,s}\vert\max_{j,t}\vert(\bar{\mathcal{H}}_{\gamma\alpha}^{-1})_{t,j}\vert\max_{i,j}\sum_{s\ne i}\sum_{t}\vert M_{\alpha\gamma,st}\rvert\\
 & +\max_{i}\vert(\bar{\mathcal{H}}_{\alpha\alpha}^{-1})_{i,i}\vert\max_{j,t}\vert(\bar{\mathcal{H}}_{\gamma\alpha}^{-1})_{t,j}\vert\max_{i,j}\sum_{t\ne j}\vert M_{\alpha\gamma,it}\rvert\\
 & =O_{p}(N^{-1})
\end{align*}

where we use $\max_{i,j}(\bar{\mathcal{H}}_{\gamma\alpha}^{-1})_{ij}=O_{p}(N^{-1}).$
Similar logic applies to the remaining components of $A$.
\end{proof}
\begin{lem}
\label{lem:HEH_bound}Let Assumptions 1 and 2 hold, and define as
either $M=\mathcal{\bar{H}}^{-1}\bar{\mathcal{E}}\mathcal{\bar{H}}^{-1}$,
or $M=W\mathcal{\bar{H}}^{-1}\bar{\mathcal{H}}_{b}\mathcal{\bar{H}}^{-1}$.
Then $M_{ii}=O_{p}(1)$, and $\max_{i,j\ne i}\vert M_{ij}\vert=O_{p}(N^{-1})$.
\end{lem}
\begin{proof}
For the first result, it remains only to show that $\bar{\mathcal{E}}$
satisfies the conditions of Lemma \ref{lem:HAH_bound}. This follows
directly from the form of $\partial_{\beta\phi\phi'}\bar{\mathcal{L}}$
and Assumption 2, which implies $\vert\partial_{\beta\pi^{2}}\bar{\ell}_{ij}\vert\leq C$
for all $i$ and $j$. For $W\bar{\mathcal{H}}_{b}$, we have shown
the result holds for $\partial_{\beta\phi\phi'}\bar{\mathcal{L}}=\bar{\mathcal{E}}$
and so it remains to demonstrate the same for the term $\sum_{f}(\partial_{\phi\phi'\phi_{f}}\bar{\mathcal{L}}_{ij})\big[\bar{\mathcal{H}}^{-1}(\partial_{\beta\phi}\bar{\mathcal{L}})\big]_{f}$.
We have
\begin{align*}
\max_{i}&\sum_{j}\vert\sum_{f}(\partial_{\alpha\alpha'\phi_{f}}\bar{\mathcal{L}}_{ij})\big[\bar{\mathcal{H}}^{-1}(\partial_{\beta\phi}\bar{\mathcal{L}})\big]_{f}\vert \\
& \leq\max_{i}\frac{1}{N-1}\vert\sum_{j\ne i}(\partial_{\pi^{3}}\bar{\ell}_{ij})\big[\bar{\mathcal{H}}^{-1}(\partial_{\beta\phi}\bar{\mathcal{L}})\big]_{i}\vert\\
 & +\max_{i}\frac{1}{N-1}\vert\sum_{j\ne i}(\partial_{\pi^{3}}\bar{\ell}_{ij})\big[\bar{\mathcal{H}}^{-1}(\partial_{\beta\phi}\bar{\mathcal{L}})\big]_{j}\vert\\
 & =\max_{i}\frac{1}{(N-1)^{2}}\vert\sum_{s}\sum_{t\ne s}\sum_{j\ne i}(\partial_{\pi^{3}}\bar{\ell}_{ij})\big([\bar{\mathcal{H}}_{\alpha\alpha,is}^{-1}](\partial_{\beta\pi}\bar{\ell}_{st})+[\bar{\mathcal{H}}_{\alpha\gamma,is}^{-1}](\partial_{\beta\pi}\bar{\ell}_{ts})\big)\vert\\
 & +\max_{i}\frac{1}{(N-1)^{2}}\vert\sum_{s}\sum_{t\ne s}\sum_{j\ne i}(\partial_{\pi^{3}}\bar{\ell}_{ij})\big([\bar{\mathcal{H}}_{\gamma\alpha,js}^{-1}](\partial_{\beta\pi}\bar{\ell}_{st})+[\bar{\mathcal{H}}_{\gamma\gamma,js}^{-1}](\partial_{\beta\pi}\bar{\ell}_{ts})\big)\vert\\
 & =\max_{i}\frac{1}{(N-1)^{2}}\vert\sum_{s}\sum_{t\ne s}\sum_{j\ne i}(\partial_{\pi^{3}}\bar{\ell}_{ij})\Gamma_{isjt}(\partial_{\beta\pi}\bar{\ell}_{st})\vert
\end{align*}
Then, using the fact that $\Gamma_{isjt}=O_{p}(1)$ only when $i=s$,
$i=t$, $j=s$ or $j=t$ and $\vert\partial_{\beta\pi}\bar{\ell}_{st}\vert\leq C$
the term is clearly $O_{p}(1)$. The same steps can be used for the
$\partial_{\alpha\gamma'\phi_{f}}\bar{\mathcal{L}}$, $\partial_{\gamma\alpha'\phi_{f}}\bar{\mathcal{L}}$,
and $\partial_{\gamma\gamma'\phi_{f}}\bar{\mathcal{L}}$ components,
and hence $W\bar{\mathcal{H}}_{b}$ satisfies the required condition.
\end{proof}

\subsection{Further terms}
The remaing parts of this section provide bounds on a series of terms that appear in the expansion. These bounds are used in later sections to bound the expansion terrms in Section \ref{SA:expansion}.
\begin{lem}
\label{lem:element1}Let Assumptions 1 and 2 hold. Then, for $s+t \leq 5$\\ \\
(i) $\lVert(\partial_{\beta^{s}\phi'}\mathcal{L})\mathcal{H}^{-1}\mathcal{S}\rVert  =O_{p}(1)$ \\\
(ii) $ \lVert(\partial_{\beta^{s}\phi'}\mathcal{L})\mathcal{H}^{-1}\mathcal{E}\mathcal{H}^{-1}\mathcal{S}\rVert =O_{p}(1)$\\
(iii) $ \lVert(\partial_{\beta^{s}\phi'}\mathcal{L})\mathcal{H}^{-1}(\partial_{\beta^{t}\phi}\mathcal{L})\rVert =O_{p}(N)$\\
(iv) $\lVert(\partial_{\beta^{s}\phi'}\mathcal{L})\mathcal{H}^{-1}\mathcal{E}\mathcal{H}^{-1}(\partial_{\beta^{t}\phi}\mathcal{L})\rVert =O_{p}(N)$
\end{lem}

\begin{proof}
Write
\begin{align*}
\lVert\mathcal{S}'\mathcal{H}^{-1}(\partial_{\beta^{s}\phi}\mathcal{L})\rVert & \leq\lVert\mathcal{S}'\mathcal{\bar{H}}^{-1}(\partial_{\beta^{s}\phi}\mathcal{L})\rVert+\lVert\mathcal{S}'(\mathcal{H}^{-1}-\mathcal{\bar{H}}^{-1})(\partial_{\beta^{s}\phi}\mathcal{L})\rVert\\
 & \leq\lVert\mathcal{S}'\mathcal{\bar{H}}^{-1}(\partial_{\beta^{s}\phi}\mathcal{L})\rVert+o_{p}(1)
\end{align*}
Decomposing this term gives
\begin{align*}
\lVert\mathcal{S}'\bar{\mathcal{H}}^{-1}(\partial_{\beta^{s}\phi}\mathcal{L})\rVert & =\lVert(\partial_{\alpha'}\mathcal{L})\bar{\mathcal{H}}_{\alpha\alpha}^{-1}(\partial_{\beta^{s}\alpha}\mathcal{L})\rVert+\lVert(\partial_{\alpha}\mathcal{L})\bar{\mathcal{H}}_{\alpha\gamma}^{-1}(\partial_{\beta^{s}\gamma}\mathcal{L})\rVert\\
 & \quad\lVert(\partial_{\gamma'}\mathcal{L})\bar{\mathcal{H}}_{\gamma\alpha}^{-1}(\partial_{\beta^{s}\alpha}\mathcal{L})\rVert+\lVert(\partial_{\gamma'}\mathcal{L})\bar{\mathcal{H}}_{\gamma\gamma}^{-1}(\partial_{\beta^{s}\gamma}\mathcal{L})\rVert
\end{align*}
the first term of which has second moment
\begin{align*}
&\bar{E}\Big[\lVert(\partial_{\alpha'}\mathcal{L})\bar{\mathcal{H}}_{\alpha\alpha}^{-1}(\partial_{\beta^{s}\alpha}\mathcal{L})\rVert^{2}\Big] \\
& =\frac{1}{(N-1)^{4}}\sum_{i,i'}\sum_{j,j'}\sum_{s\ne i}\sum_{s'\ne i'}\sum_{t\ne j}\sum_{t'\ne j'}(\bar{\mathcal{H}}_{\alpha\alpha}^{-1})_{st}(\bar{\mathcal{H}}_{\alpha\alpha}^{-1})_{s't'}\bar{E}\big[(\partial_{\pi}\ell_{is})(\partial_{\beta^{s}\pi}\ell_{jt})(\partial_{\pi}\ell_{i's'})(\partial_{\beta^{s}\pi}\ell_{j't'})\big]\\
 & =\frac{1}{(N-1)^{4}}\sum_{i,i'}\sum_{j,j'}\sum_{s\ne\{i,j\}}\sum_{s'\ne\{i',j'\}}(\bar{\mathcal{H}}_{\alpha\alpha}^{-1})_{ss}(\bar{\mathcal{H}}_{\alpha\alpha}^{-1})_{s's'}\bar{E}\big[(\partial_{\pi}\ell_{is})(\partial_{\beta^{s}\pi}\ell_{js})(\partial_{\pi}\ell_{i's'})(\partial_{\beta^{s}\pi}\ell_{j's'})\big]\\
 & +\frac{2}{(N-1)^{4}}\sum_{i,i'}\sum_{j,j'}\sum_{s\ne\{i,t\}}\sum_{s'\ne\{i',j'\}}\sum_{t\ne j}(\bar{\mathcal{H}}_{\alpha\alpha}^{-1})_{st}(\bar{\mathcal{H}}_{\alpha\alpha}^{-1})_{s's'}\bar{E}\big[(\partial_{\pi}\ell_{is})(\partial_{\beta^{s}\pi}\ell_{jt})(\partial_{\pi}\ell_{i's'})(\partial_{\beta^{s}\pi}\ell_{j's'})\big]\\
 & +\frac{2}{(N-1)^{4}}\sum_{i,i'}\sum_{j,j'}\sum_{s\ne\{i,t\}}\sum_{s'\ne\{i',t'\}}\sum_{t\ne j}\sum_{t'\ne j'}(\bar{\mathcal{H}}_{\alpha\alpha}^{-1})_{st}(\bar{\mathcal{H}}_{\alpha\alpha}^{-1})_{s't'}\bar{E}\big[(\partial_{\pi}\ell_{is})(\partial_{\beta^{s}\pi}\ell_{jt})(\partial_{\pi}\ell_{i's'})(\partial_{\beta^{s}\pi}\ell_{s't'})\big]
\end{align*}
Consider the first term above. Since $\bar{E}[\partial_{\pi}\ell_{is}]=0$,
we must have either: $(i',s')\in\{(i,s),(s,i)\}$, or $(i,s)\in\{(j,s),(s,j)\}$
and $(i',s')\in\{(j',s'),(s',j')\}$. In either case, there are at
most four unique subscripts in the summation, so that the sum is $O_{p}(1)$.
A similar argument applied to the remaining terms, combined with the
fact that $(\bar{\mathcal{H}}_{\alpha\alpha}^{-1})_{ss}=O_{p}(1)$
while $(\bar{\mathcal{H}}_{\alpha\alpha}^{-1})_{st}=O_{p}(N^{-1})$
for $s\ne t$, gives $O_{p}(1)$ for the whole term.

The second result is identical to the first, using Lemma \ref{lem:HEH_bound}
and $\lVert\tilde{\mathcal{E}}\rVert=o_{p}(1)$. For the third result,
we have $\mathcal{F}^{s,t}=(\partial_{\beta^{s}\phi'}\mathcal{L})\mathcal{H}^{-1}(\partial_{\beta^{t}\phi}\mathcal{L})$.
For $s=t=1$ we have
\begin{align*}
\lVert(\partial_{\beta^{s}\phi'}\mathcal{L})\mathcal{H}^{-1}(\partial_{\beta^{t}\phi}\mathcal{L})\rVert & \leq\lVert(\partial_{\beta^{s}\phi'}\mathcal{L})\bar{\mathcal{H}}^{-1}(\partial_{\beta^{t}\phi}\mathcal{L})\rVert+\lVert(\partial_{\beta^{s}\phi'}\mathcal{L})\big(\mathcal{H}^{-1}-\bar{\mathcal{H}}^{-1}\big)(\partial_{\beta^{t}\phi}\mathcal{L})\rVert\\
 & \leq\lVert(\partial_{\beta^{s}\phi'}\mathcal{L})\bar{\mathcal{H}}^{-1}(\partial_{\beta^{t}\phi}\mathcal{L})\rVert+\lVert\partial_{\beta^{s}\phi'}\mathcal{L}\rVert\lVert\mathcal{H}^{-1}-\bar{\mathcal{H}}^{-1}\rVert\lVert\partial_{\beta^{t}\phi}\mathcal{L}\rVert\\
 & \leq\lVert(\partial_{\beta^{s}\phi'}\mathcal{L})\bar{\mathcal{H}}^{-1}(\partial_{\beta^{t}\phi}\mathcal{L})\rVert+O_{p}(N^{-1/2+3/q})
\end{align*}
from lemma \ref{lem:S6} (iii) and lemma \ref{lem:H_tilde}. We can decompose the first term as 
\begin{align*}
\lVert(\partial_{\beta\phi'}\mathcal{L})\bar{\mathcal{H}}^{-1}(\partial_{\beta\phi}\mathcal{L})\rVert & =\lVert(\partial_{\beta\alpha'}\mathcal{L})\bar{\mathcal{H}}_{\alpha\alpha}^{-1}(\partial_{\beta\alpha}\mathcal{L})\rVert+\lVert(\partial_{\beta\alpha}\mathcal{L})\bar{\mathcal{H}}_{\alpha\gamma}^{-1}(\partial_{\beta\gamma}\mathcal{L})\rVert\\
 & \quad\lVert(\partial_{\beta\gamma'}\mathcal{L})\bar{\mathcal{H}}_{\gamma\alpha}^{-1}(\partial_{\beta\alpha}\mathcal{L})\rVert+\lVert(\partial_{\beta\gamma'}\mathcal{L})\bar{\mathcal{H}}_{\gamma\gamma}^{-1}(\partial_{\beta\gamma}\mathcal{L})\rVert
\end{align*}
the first term of which is
\begin{align*}
&\bar{E}\Big[\lVert\frac{1}{N}(\partial_{\beta\alpha'}\mathcal{L})\bar{\mathcal{H}}_{\alpha\alpha}^{-1}(\partial_{\beta\alpha}\mathcal{L})\rVert^{2}\Big] \\
& =\frac{1}{N^{2}(N-1)^{4}}\sum_{i,i'}\sum_{j,j'}\sum_{s\ne i}\sum_{s'\ne i'}\sum_{t\ne j}\sum_{t'\ne j'}(\bar{\mathcal{H}}_{\alpha\alpha}^{-1})_{st}(\bar{\mathcal{H}}_{\alpha\alpha}^{-1})_{s't'}\bar{E}\big[(\partial_{\beta\pi}\ell_{is})(\partial_{\beta\pi}\ell_{jt})(\partial_{\beta\pi}\ell_{i's'})(\partial_{\beta\pi}\ell_{j't'})\big]\\
 & =\frac{1}{N^{2}(N-1)^{4}}\sum_{i,i'}\sum_{j,j'}\sum_{s\ne\{i,j\}}\sum_{s'\ne\{i',j'\}}(\bar{\mathcal{H}}_{\alpha\alpha}^{-1})_{ss}(\bar{\mathcal{H}}_{\alpha\alpha}^{-1})_{s's'}\bar{E}\big[(\partial_{\beta\pi}\ell_{is})(\partial_{\beta\pi}\ell_{js})(\partial_{\beta\pi}\ell_{i's'})(\partial_{\beta\pi}\ell_{j's'})\big]\\
 & +\frac{2}{N^{2}(N-1)^{4}}\sum_{i,i'}\sum_{j,j'}\sum_{s\ne\{i,t\}}\sum_{s'\ne\{i',j'\}}\sum_{t\ne j}(\bar{\mathcal{H}}_{\alpha\alpha}^{-1})_{st}(\bar{\mathcal{H}}_{\alpha\alpha}^{-1})_{s's'}\bar{E}\big[(\partial_{\beta\pi}\ell_{is})(\partial_{\beta\pi}\ell_{jt})(\partial_{\beta\pi}\ell_{i's'})(\partial_{\beta\pi}\ell_{s's'})\big]\\
 & +\frac{2}{N^{2}(N-1)^{4}}\sum_{i,i'}\sum_{j,j'}\sum_{s\ne\{i,t\}}\sum_{s'\ne\{i',t'\}}\sum_{t\ne j}\sum_{t'\ne j'}(\bar{\mathcal{H}}_{\alpha\alpha}^{-1})_{st}(\bar{\mathcal{H}}_{\alpha\alpha}^{-1})_{s't'}\bar{E}\big[(\partial_{\beta\pi}\ell_{is})(\partial_{\beta\pi}\ell_{jt})(\partial_{\beta\pi}\ell_{i's'})(\partial_{\beta\pi}\ell_{s't'})\big]\\
 & =O_{p}(1)
\end{align*}
where the final line follows from the fact that $(\bar{\mathcal{H}}_{\alpha\alpha}^{-1})_{ss}=O_{p}(1)$
while $(\bar{\mathcal{H}}_{\alpha\alpha}^{-1})_{st}=O_{p}(N^{-1})$
for $s\ne t$. The remaining terms are $O_{p}(N)$ by similar reasoning.
The final result can be shown in the same way as the third result
using Lemma \ref{lem:HEH_bound}.
\end{proof}
\begin{lem}
\label{lem:tensor_bounds}Let Assumptions 1 and 2 hold. For $r+s\leq3$
and for $t_{1}=\{0,1\},t_{2}=\{1,2\}$\\ \\
(i) $\lVert\sum_{f}(\partial_{\beta^{r}\phi\phi'\phi_{f}}\mathcal{L})\big[\mathcal{H}^{-1}(\partial_{\beta^{s}\phi}\mathcal{L})\big]_{f}\rVert  =O_{p}(N^{4\epsilon})$\\
(ii) $\lVert\sum_{f,g}(\partial_{\beta^{t_{1}}\phi\phi'\phi_{f}\phi_{g}}\mathcal{L})\big[\mathcal{H}^{-1}(\partial_{\beta\phi}\mathcal{L})\big]_{f}\big[\mathcal{H}^{-1}(\partial_{\beta^{t_{2}}\phi}\mathcal{L})\big]_{g}\rVert =O_{p}(N^{6\epsilon})$\\
(iii) $\lVert\sum_{e,f,g}(\partial_{\phi\phi'\phi_{f}\phi_{g}}\mathcal{L})\big[\mathcal{H}^{-1}(\partial_{\beta\phi}\mathcal{L})\big]_{e}\big[\mathcal{H}^{-1}(\partial_{\beta\phi}\mathcal{L})\big]_{f}\big[\mathcal{H}^{-1}(\partial_{\beta\phi}\mathcal{L})\big]_{g}\rVert  =O_{p}(N^{8\epsilon}$)\\
(iv) $\lVert\sum_{f}(\partial_{\phi\phi'\phi_{f}}\mathcal{L})\big[\mathcal{H}^{-1}\mathcal{S}\big]_{f}\rVert  =O_{p}(N^{-\frac{1}{2}+2\epsilon})$\\
(v) $\lVert\sum_{e,f}(\partial_{\phi\phi'\phi_{e}\phi_{f}}\mathcal{L})\big[\mathcal{H}^{-1}(\partial_{\beta\phi}\mathcal{L})\big]_{f}\big[\mathcal{H}^{-1}\mathcal{S}\big]_{e}\rVert  =O_{p}(N^{-\frac{1}{2}+4\epsilon})$\\
(vi) $\lVert\sum_{f}(\partial_{\beta^{r}\phi\phi'\phi_{f}}\mathcal{L})\big[\mathcal{H}^{-1}\tilde{\mathcal{H}}\mathcal{H}^{-1}(\partial_{\beta^{s}\phi}\mathcal{L})\big]_{f}\rVert  =O_{p}(N^{-\frac{1}{2}+4\epsilon})$\\
(vii) $\lVert\sum_{e,f,g}(\partial_{\phi\phi'\phi_{f}\phi_{g}}\mathcal{L})\big[\mathcal{H}^{-1}(\partial_{\beta\phi}\mathcal{L})\big]_{e}\big[\mathcal{H}^{-1}(\partial_{\beta\phi}\mathcal{L})\big]_{f}\big[\mathcal{H}^{-1}\mathcal{S}\big]_{g}\rVert  =O_{p}(N^{-\frac{1}{2}+6\epsilon})$
\end{lem}

\begin{proof}
The proofs are substantially similar for each of the terms and so we
provide details for the second statement only, with $t=1$. We can
again decompose this term into
\begin{align*}
\lVert & \sum_{f,g}(\partial_{\phi\phi'\phi_{f}\phi_{g}}\mathcal{L})\big[\mathcal{H}^{-1}(\partial_{\beta\phi}\mathcal{L})\big]_{f}\big[\mathcal{H}^{-1}(\partial_{\beta\phi}\mathcal{L})\big]_{g}\rVert\\
 & \leq\lVert\sum_{f,g}(\partial_{\phi\phi'\phi_{f}\phi_{g}}\mathcal{L})\big[\mathcal{\bar{H}}^{-1}(\partial_{\beta\phi}\mathcal{L})\big]_{f}\big[\mathcal{\bar{H}}^{-1}(\partial_{\beta\phi}\mathcal{L})\big]_{g}\rVert\\
 & +\lVert\sum_{f,g}(\partial_{\phi\phi'\phi_{f}\phi_{g}}\mathcal{L})\big[\mathcal{\bar{H}}^{-1}\tilde{\mathcal{H}}\mathcal{\bar{H}}^{-1}(\partial_{\beta\phi}\mathcal{L})\big]_{f}\big[\mathcal{\bar{H}}^{-1}(\partial_{\beta\phi}\mathcal{L})\big]_{g}\rVert\\
 & +\lVert\sum_{f,g}(\partial_{\phi\phi'\phi_{f}\phi_{g}}\mathcal{L})\big[\mathcal{\bar{H}}^{-1}(\partial_{\beta\phi}\mathcal{L})\big]_{f}\big[\mathcal{\bar{H}}^{-1}\tilde{\mathcal{H}}\mathcal{\bar{H}}^{-1}(\partial_{\beta\phi}\mathcal{L})\big]_{g}\rVert\\
 & +O_{p}(N^{6\epsilon})
\end{align*}
with the bound on the remainder following from
\begin{align*}
\lVert & \sum_{f,g}(\partial_{\phi\phi'\phi_{f}\phi_{g}}\mathcal{L})\big[\big(\mathcal{H}^{-1}-\mathcal{\bar{H}}^{-1}+\mathcal{\bar{H}}^{-1}\tilde{\mathcal{H}}\mathcal{\bar{H}}^{-1}\big)(\partial_{\beta\phi}\mathcal{L})\big]_{f}\big[\mathcal{H}^{-1}(\partial_{\beta\phi}\mathcal{L})\big]_{g}\rVert\\
 & \leq\lVert\partial_{\phi^{4}}\mathcal{L}\rVert\cdot\lVert\mathcal{H}^{-1}-\mathcal{\bar{H}}^{-1}+\mathcal{\bar{H}}^{-1}\tilde{\mathcal{H}}\mathcal{\bar{H}}^{-1}\rVert\cdot\lVert\partial_{\beta\phi}\mathcal{L}\rVert^{2}\cdot\lVert\mathcal{H}^{-1}\rVert
\end{align*}
and Lemma \ref{lem:H1_approx}.

Focussing on the first term, we can decompose the matrix $\partial_{\phi\phi'\phi_{f}\phi_{g}}\mathcal{L}$
into four parts as $\partial_{\alpha\alpha'\phi_{f}\phi_{g}}\mathcal{L}$,
$\partial_{\alpha\gamma'\phi_{f}\phi_{g}}\mathcal{L}$ etc. The first
of these is
\begin{align*}
\lVert & \sum_{f,g}(\partial_{\alpha\alpha'\phi_{f}\phi_{g}}\mathcal{L})\big[\mathcal{\bar{H}}^{-1}(\partial_{\beta\phi}\mathcal{L})\big]_{f}\big[\mathcal{\bar{H}}^{-1}(\partial_{\beta\phi}\mathcal{L})\big]_{g}\rVert\\
 & \leq\lVert\sum_{f=1}^{N}(\partial_{\alpha\alpha'\alpha_{f}\alpha_{f}}\mathcal{L})\big[\mathcal{\bar{H}}_{\alpha,\cdot}^{-1}(\partial_{\beta\phi}\mathcal{L})\big]_{f}\big[\mathcal{\bar{H}}_{\alpha,\cdot}^{-1}(\partial_{\beta\phi}\mathcal{L})\big]_{f}\rVert\\
 & +2\lVert\sum_{f,g=1}^{N}(\partial_{\alpha\alpha'\alpha_{f}\gamma_{g}}\mathcal{L})\big[\mathcal{\bar{H}}_{\alpha,\cdot}^{-1}(\partial_{\beta\phi}\mathcal{L})\big]_{f}\big[\mathcal{\bar{H}}_{\gamma,\cdot}^{-1}(\partial_{\beta\phi}\mathcal{L})\big]_{g}\rVert\\
 & +\lVert\sum_{f=1}^{N}(\partial_{\alpha\alpha'\gamma_{f}\gamma_{f}}\mathcal{L})\big[\mathcal{\bar{H}}_{\gamma,\cdot}^{-1}(\partial_{\beta\phi}\mathcal{L})\big]_{f}\big[\mathcal{\bar{H}}_{\gamma,\cdot}^{-1}(\partial_{\beta\phi}\mathcal{L})\big]_{f}\rVert\\
 & =\max_{i}\lvert\frac{1}{(N-1)^{3}}\sum_{j\ne i}\partial_{\pi^{4}}\ell_{ij}\big(\sum_{k}\sum_{l\ne k}((\bar{\mathcal{H}}_{\alpha\alpha}^{-1})_{ik}+(\bar{\mathcal{H}}_{\alpha\gamma}^{-1})_{il})\partial_{\beta\pi}\ell_{kl}\big)^{2}\rvert\\
 & +\max_{i}\lvert\frac{1}{(N-1)^{3}}\sum_{j\ne i}\partial_{\pi^{4}}\ell_{ij}\big(\sum_{k}\sum_{l\ne k}((\bar{\mathcal{H}}_{\alpha\alpha}^{-1})_{ik}+(\bar{\mathcal{H}}_{\alpha\gamma}^{-1})_{il})\partial_{\beta\pi}\ell_{kl}\big)\\
 & \quad\times\big(\sum_{k}\sum_{l\ne k}((\bar{\mathcal{H}}_{\gamma\alpha}^{-1})_{jk}+(\bar{\mathcal{H}}_{\gamma\gamma}^{-1})_{jl})\partial_{\beta\pi}\ell_{kl}\big)\rvert\\
 & +\max_{i}\lvert\frac{1}{(N-1)^{3}}\sum_{j\ne i}\partial_{\pi^{4}}\ell_{ij}\big(\sum_{k}\sum_{l\ne k}((\bar{\mathcal{H}}_{\gamma\alpha}^{-1})_{jk}+(\bar{\mathcal{H}}_{\gamma\gamma}^{-1})_{jl})\partial_{\beta\pi}\ell_{kl}\big)^{2}\rvert
\end{align*}
Considering the first of this new set of terms
\begin{align*}
\max_{i} & \lvert\frac{1}{(N-1)^{3}}\sum_{j\ne i}\partial_{\pi^{4}}\ell_{ij}\big(\sum_{k}\sum_{l\ne k}((\bar{\mathcal{H}}_{\alpha\alpha}^{-1})_{ik}+(\bar{\mathcal{H}}_{\alpha\gamma}^{-1})_{il})\partial_{\beta\pi}\ell_{kl}\big)^{2}\rvert\\
 & =\max_{i}\lvert\frac{1}{(N-1)^{3}}\sum_{j\ne i}\partial_{\pi^{4}}\ell_{ij}\big(\sum_{k}\sum_{l\ne k}\sum_{s}\sum_{t\ne s}((\bar{\mathcal{H}}_{\alpha\alpha}^{-1})_{ik}(\bar{\mathcal{H}}_{\alpha\alpha}^{-1})_{is}+(\bar{\mathcal{H}}_{\alpha\alpha}^{-1})_{ik}(\bar{\mathcal{H}}_{\alpha\gamma}^{-1})_{it}\\
 & +(\bar{\mathcal{H}}_{\alpha\gamma}^{-1})_{il}(\bar{\mathcal{H}}_{\alpha\alpha}^{-1})_{is}+(\bar{\mathcal{H}}_{\alpha\gamma}^{-1})_{il}(\bar{\mathcal{H}}_{\alpha\gamma}^{-1})_{it})(\partial_{\beta\pi}\ell_{kl})(\partial_{\beta\pi}\ell_{st})\big)\rvert\\
 & =\max_{i}\lvert\frac{1}{(N-1)^{3}}\sum_{j\ne i}\partial_{\pi^{4}}\ell_{ij}\big(\sum_{l\ne i}\sum_{t\ne i}(\bar{\mathcal{H}}_{\alpha\alpha}^{-1})_{ii}(\bar{\mathcal{H}}_{\alpha\alpha}^{-1})_{ii}(\partial_{\beta\pi}\ell_{il})(\partial_{\beta\pi}\ell_{it})\big)\rvert\\
 & +2\max_{i}\lvert\frac{1}{N(N-1)^{3}}\sum_{j\ne i}\partial_{\pi^{4}}\ell_{ij}\big(\sum_{l\ne k}\sum_{s\ne i}\sum_{t\ne s}((\bar{\mathcal{H}}_{\alpha\alpha}^{-1})_{ii}(N\bar{\mathcal{H}}_{\alpha\alpha}^{-1})_{is})(\partial_{\beta\pi}\ell_{il})(\partial_{\beta\pi}\ell_{st})\big)\rvert\\
 & +\max_{i}\lvert\frac{1}{N^{2}(N-1)^{3}}\sum_{j\ne i}\partial_{\pi^{4}}\ell_{ij}\big(\sum_{k\ne i}\sum_{l\ne k}\sum_{s\ne i}\sum_{t\ne s}((N\bar{\mathcal{H}}_{\alpha\alpha}^{-1})_{ik}(N\bar{\mathcal{H}}_{\alpha\alpha}^{-1})_{is})(\partial_{\beta\pi}\ell_{kl})(\partial_{\beta\pi}\ell_{st})\big)\rvert\\
 & +2\max_{i}\lvert\frac{1}{N(N-1)^{3}}\sum_{j\ne i}\partial_{\pi^{4}}\ell_{ij}\big(\sum_{l\ne i}\sum_{s}\sum_{t\ne s}(\bar{\mathcal{H}}_{\alpha\alpha}^{-1})_{ii}(N\bar{\mathcal{H}}_{\alpha\gamma}^{-1})_{it}(\partial_{\beta\pi}\ell_{il})(\partial_{\beta\pi}\ell_{st})\big)\rvert\\
 & +2\max_{i}\lvert\frac{1}{N^{2}(N-1)^{3}}\sum_{j\ne i}\partial_{\pi^{4}}\ell_{ij}\big(\sum_{k\ne i}\sum_{l\ne k}\sum_{s}\sum_{t\ne s}(N\bar{\mathcal{H}}_{\alpha\alpha}^{-1})_{ik}(N\bar{\mathcal{H}}_{\alpha\gamma}^{-1})_{it}(\partial_{\beta\pi}\ell_{kl})(\partial_{\beta\pi}\ell_{st})\big)\rvert\\
 & +\max_{i}\lvert\frac{1}{N^{2}(N-1)^{3}}\sum_{j\ne i}\partial_{\pi^{4}}\ell_{ij}\big(\sum_{k}\sum_{l\ne k}\sum_{s}\sum_{t\ne s}(N\bar{\mathcal{H}}_{\alpha\gamma}^{-1})_{il}(N\bar{\mathcal{H}}_{\alpha\gamma}^{-1})_{it}(\partial_{\beta\pi}\ell_{kl})(\partial_{\beta\pi}\ell_{st})\big)\rvert\\
 & =O_{p}(N^{6\epsilon})
\end{align*}
where the final line applies Lemma \ref{lem:S6} to each term, e.g
\begin{align*}
\max_{i} & \lvert\frac{1}{(N-1)^{3}}\sum_{j\ne i}\partial_{\pi^{4}}\ell_{ij}\big(\sum_{l\ne i}\sum_{t\ne i}(\bar{\mathcal{H}}_{\alpha\alpha}^{-1})_{ii}(\bar{\mathcal{H}}_{\alpha\alpha}^{-1})_{ii}(\partial_{\beta\pi}\ell_{il})(\partial_{\beta\pi}\ell_{it})\big)\rvert\\
 & \leq\max_{i}\lvert(\bar{\mathcal{H}}_{\alpha\alpha}^{-1})_{ii}\rvert^{2}\times\max_{i}\lvert\frac{1}{N-1}\sum_{j\ne i}\partial_{\pi^{4}}\ell_{ij}\rvert\\
 & \times\max_{i}\lvert\frac{1}{N-1}\sum_{j\ne i}\partial_{\beta\pi}\ell_{ij}\rvert\times\max_{i}\lvert\frac{1}{N-1}\sum_{j\ne i}\partial_{\beta\pi}\ell_{ij}\rvert\\
 & =O_{p}(1)\times O_{p}(N^{2\epsilon})\times O_{p}(N^{2\epsilon})\times O_{p}(N^{2\epsilon})
\end{align*}
The remaining terms in $\partial_{\phi\phi'\phi_{f}\phi_{g}}\mathcal{L}$
can be dealt with identically. \\
The terms related to $\sum_{f,g}(\partial_{\phi\phi'\phi_{f}\phi_{g}}\mathcal{L})\big[\mathcal{\bar{H}}^{-1}\tilde{\mathcal{H}}\mathcal{\bar{H}}^{-1}(\partial_{\beta\phi}\mathcal{L})\big]_{f}\big[\mathcal{\bar{H}}^{-1}(\partial_{\beta\phi}\mathcal{L})\big]_{g}$
also have nearly identical proofs, applying Lemma \ref{lem:HEH_bound}
so that $\mathcal{\bar{H}}^{-1}\tilde{\mathcal{H}}\mathcal{\bar{H}}^{-1}$
behaves similarly to $\mathcal{\bar{H}}^{-1}$ in the proof. The proofs
for terms withe $[\mathcal{H}^{-1}\mathcal{S}]$ or $\tilde{\mathcal{H}}$
are essentially the same, but use the fact the one of the terms in
the sum is mean zero, e.g.
\begin{align*}
\bar{E}\Big[\big(\max_{i}\lvert\frac{1}{\sqrt{N-1}}\sum_{j\ne i}\partial_{\pi}\ell_{ij}\rvert\big)^{2}\Big] & \leq\sum_{i}\bar{E}\Big[\big(\lvert\frac{1}{\sqrt{N-1}}\sum_{j\ne i}\partial_{\pi}\ell_{ij}\rvert\big)^{2}\Big]\\
 & \leq CNN^{-1}=O_{p}(1)
\end{align*}
and hence $\max_{i}\lvert\frac{1}{N-1}\sum_{j\ne i}\partial_{\pi}\ell_{ij}\rvert=O_{p}(N^{-1/2})$.
\end{proof}
\begin{lem}
\label{lem:tensors2}Let Assumptions 1 and 2 hold. Then, $s+t \leq 5$\\ \\
(i) $\lVert\sum_{f}(\partial_{\phi\phi'\phi_{f}}\mathcal{L})\big[\mathcal{H}^{-1}\mathcal{H}_{b}\mathcal{H}^{-1}(\partial_{\beta\beta\phi}\mathcal{L})\big]_{f}\rVert  =O_{p}(N^{-1+8\epsilon})$\\
(ii) $(\partial_{\beta^{s}\phi'}\mathcal{L})\mathcal{H}^{-1}\mathcal{H}_{b}\mathcal{H}^{-1}\mathcal{E}^{r}\mathcal{H}^{-1}(\partial_{\beta^{t}\phi}\mathcal{L})  =O_{p}(1)$\\
(iii) $\mathcal{S}'\mathcal{H}^{-1}\sum_{f}(\partial_{\beta\phi\phi'\phi_{f}}\mathcal{L})\big[\mathcal{H}^{-1}(\partial_{\beta\phi}\mathcal{L})\big]_{f}\mathcal{H}^{-1}(\partial_{\beta\phi}\mathcal{L})  =O_{p}(1)$\\
(iv) $\mathcal{S}'\mathcal{H}^{-1}\sum_{f}(\partial_{\beta\phi\phi'\phi_{f}}\mathcal{L})\big[\mathcal{H}^{-1}(\partial_{\beta\phi}\mathcal{L})\big]_{f}\mathcal{H}^{-1}\mathcal{S}  =O_{p}(1)$\\
(v) $\mathcal{S}'\mathcal{H}^{-1}\sum_{e,f}(\partial_{\phi\phi'\phi_{e}\phi_{f}}\mathcal{L})\big[\mathcal{H}^{-1}(\partial_{\beta\phi}\mathcal{L})\big]_{e}\big[\mathcal{H}^{-1}(\partial_{\beta\phi}\mathcal{L})\big]_{f}\mathcal{H}^{-1}(\partial_{\beta\phi}\mathcal{L})  =O_{p}(N^{4\epsilon})$\\
(vi) $\lVert\sum_{g}\mathcal{S}'\mathcal{H}^{-1}\mathcal{H}_{b}\mathcal{H}^{-1}\mathcal{H}_{s_{g}}\mathcal{H}^{-1}\mathcal{S}\mathcal{S}_{g}\rVert  =O_{p}(N^{-\frac{3}{2}+4\epsilon})$\\
(vii) $ \lVert\mathcal{S}'\mathcal{H}^{-1}W^{-1}\sum_{f}(\partial_{\phi\phi'\phi_{f}}\mathcal{L})\big[\mathcal{H}^{-1}\mathcal{H}_{bb}\mathcal{H}^{-1}(\partial_{\beta\phi}\mathcal{L})\big]_{f}\mathcal{H}^{-1}\mathcal{S}\rVert  =O_{p}(N^{-3+14\epsilon})$\\
(vii) $\lVert\sum_{g,h}\mathcal{S}'\mathcal{H}^{-1}\mathcal{H}_{b}\mathcal{H}^{-1}\mathcal{H}_{s_{g}s_{h}}\mathcal{H}^{-1}\mathcal{S}\mathcal{S}_{g}\mathcal{S}_{h}\rVert  =O_{p}(N^{-2+10\epsilon})$
\end{lem}

\begin{proof}
For expression 1 we have 
\begin{align*}
\lVert\sum_{f}(\partial_{\phi\phi'\phi_{f}}\mathcal{L})\big[\mathcal{H}^{-1}\mathcal{H}_{b}\mathcal{H}^{-1}(\partial_{\beta\beta\phi}\mathcal{L})\big]_{f}\rVert & \leq\lVert\sum_{f}(\partial_{\phi\phi'\phi_{f}}\mathcal{L})\big[\mathcal{H}^{-1}\mathcal{\bar{H}}_{b}\mathcal{H}^{-1}(\partial_{\beta\beta\phi}\mathcal{L})\big]_{f}\rVert\\
 & +\lVert\sum_{f}(\partial_{\phi\phi'\phi_{f}}\mathcal{L})\big[\mathcal{H}^{-1}(\mathcal{H}_{b}-\bar{\mathcal{H}}_{b})\mathcal{H}^{-1}(\partial_{\beta\beta\phi}\mathcal{L})\big]_{f}\rVert\\
 & =\lVert\sum_{f}(\partial_{\phi\phi'\phi_{f}}\mathcal{L})\big[\mathcal{H}^{-1}\mathcal{\bar{H}}_{b}\mathcal{H}^{-1}(\partial_{\beta\beta\phi}\mathcal{L})\big]_{f}\rVert+O_{p}(N^{-1+8\epsilon})\\
 & =O_{p}(N^{-1+8\epsilon})
\end{align*}
using Lemmas \ref{lem:H_tilde} and \ref{lem:HEH_bound} combined
with Lemma \ref{lem:tensor_bounds}. Next 
\begin{align*}
\lVert(\partial_{\beta^{s}\phi'}\mathcal{L})\mathcal{H}^{-1}\mathcal{H}_{b}\mathcal{H}^{-1}\mathcal{E}^{r}\mathcal{H}^{-1}(\partial_{\beta^{t}\phi}\mathcal{L})\rVert & \leq\lVert(\partial_{\beta^{s}\phi'}\mathcal{L})\mathcal{H}^{-1}\mathcal{\bar{H}}_{b}\mathcal{H}^{-1}\mathcal{E}^{r}\mathcal{H}^{-1}(\partial_{\beta^{t}\phi}\mathcal{L})\rVert\\
 & +\lVert(\partial_{\beta^{s}\phi'}\mathcal{L})\mathcal{H}^{-1}(\mathcal{H}_{b}-\mathcal{\bar{H}}_{b})\mathcal{H}^{-1}\mathcal{E}^{r}\mathcal{H}^{-1}(\partial_{\beta^{t}\phi}\mathcal{L})\rVert\\
 & =O_{p}(1)+O_{p}(N^{-\frac{1}{2}+6\epsilon})=O_{p}(1)
\end{align*}

For the third and fourth terms, note that $\sum_{f}(\partial_{\beta\phi\phi'\phi_{f}}\mathcal{L})\big[\mathcal{H}^{-1}(\partial_{\beta\phi}\mathcal{L})\big]_{f}$
satisfies the same conditions as $W\mathcal{H}_{b}$ in the proof
in Lemmas \ref{lem:H_tilde} and \ref{lem:HEH_bound}, and hence we
can apply the results of those lemmas along with the result in Lemma
\ref{lem:tensor_bounds}. 

The fifth result can be shown similarly. First note that
\begin{align*}
\lVert & \sum_{e,f}(\partial_{\phi\phi'\phi_{e}\phi_{f}}\mathcal{L})\big[\mathcal{H}^{-1}(\partial_{\beta\phi}\mathcal{L})\big]_{e}\big[\mathcal{H}^{-1}(\partial_{\beta\phi}\mathcal{L})\big]_{f}-\sum_{e,f}(\partial_{\phi\phi'\phi_{e}\phi_{f}}\bar{\mathcal{L}})\big[\mathcal{\bar{H}}^{-1}(\partial_{\beta\phi}\bar{\mathcal{L}})\big]_{e}\big[\bar{\mathcal{H}}^{-1}(\partial_{\beta\phi}\bar{\mathcal{L}})\big]_{f}\rVert\\
 & =O_{p}(N^{-\frac{1}{2}+4\epsilon})
\end{align*}
following the proofs in Lemma \ref{lem:tensor_bounds}. Then, to apply
Lemma \ref{lem:HEH_bound}, let $i\leq N$ and so
\begin{align*}
\Big[\sum_{e,f}(\partial_{\phi\phi'\phi_{e}\phi_{f}}\bar{\mathcal{L}})\big[\mathcal{\bar{H}}^{-1}(\partial_{\beta\phi}\bar{\mathcal{L}})\big]_{e}\big[\bar{\mathcal{H}}^{-1}(\partial_{\beta\phi}\bar{\mathcal{L}})\big]_{f}\Big]_{ii} & =(\partial_{\alpha_{i}^{4}}\bar{\mathcal{L}})\big[\mathcal{\bar{H}}^{-1}(\partial_{\beta\phi}\bar{\mathcal{L}})\big]_{i}\big[\bar{\mathcal{H}}^{-1}(\partial_{\beta\phi}\bar{\mathcal{L}})\big]_{i}\\
 & +2\sum_{j\ne i}(\partial_{\alpha_{i}^{3}\gamma_{j}}\bar{\mathcal{L}})\big[\mathcal{\bar{H}}^{-1}(\partial_{\beta\phi}\bar{\mathcal{L}})\big]_{i}\big[\bar{\mathcal{H}}^{-1}(\partial_{\beta\phi}\bar{\mathcal{L}})\big]_{j}\\
 & +\sum_{j\ne i}(\partial_{\alpha_{i}^{2}\gamma_{j}^{2}}\bar{\mathcal{L}})\big[\mathcal{\bar{H}}^{-1}(\partial_{\beta\phi}\bar{\mathcal{L}})\big]_{j}\big[\bar{\mathcal{H}}^{-1}(\partial_{\beta\phi}\bar{\mathcal{L}})\big]_{j}
\end{align*}
and since $\big[\mathcal{\bar{H}}^{-1}(\partial_{\beta\phi}\bar{\mathcal{L}})\big]_{i}\leq C$,
$\partial_{\alpha_{i}^{4}}\bar{\mathcal{L}}\leq C$ and $\sum_{j\ne i}(\partial_{\pi^{4}}\bar{\mathcal{\ell}}_{ij})\leq C$
we have \\
$\max_{i}\vert\Big[\sum_{e,f}(\partial_{\phi\phi'\phi_{e}\phi_{f}}\bar{\mathcal{L}})\big[\mathcal{\bar{H}}^{-1}(\partial_{\beta\phi}\bar{\mathcal{L}})\big]_{e}\big[\bar{\mathcal{H}}^{-1}(\partial_{\beta\phi}\bar{\mathcal{L}})\big]_{f}\Big]_{ii}\vert=O_{p}(1)$
as required (other components of the matrix can be shown in the same
way). This then gives
\[
\mathcal{S}'\mathcal{H}^{-1}\sum_{e,f}(\partial_{\phi\phi'\phi_{e}\phi_{f}}\bar{\mathcal{L}})\big[\mathcal{\bar{H}}^{-1}(\partial_{\beta\phi}\bar{\mathcal{L}})\big]_{e}\big[\bar{\mathcal{H}}^{-1}(\partial_{\beta\phi}\bar{\mathcal{L}})\big]_{f}\mathcal{H}^{-1}(\partial_{\beta\phi}\mathcal{L})=O_{p}(1)
\]
 and so 
\begin{align*}
&\vert\mathcal{S}'\mathcal{H}^{-1}\sum_{e,f}(\partial_{\phi\phi'\phi_{e}\phi_{f}}\mathcal{L})\big[\mathcal{H}^{-1}(\partial_{\beta\phi}\mathcal{L})\big]_{e}\big[\mathcal{H}^{-1}(\partial_{\beta\phi}\mathcal{L})\big]_{f}\mathcal{H}^{-1}(\partial_{\beta\phi}\mathcal{L})\vert \\
& \leq\lVert\mathcal{S}'\mathcal{H}^{-1}\rVert\lVert\mathcal{H}^{-1}(\partial_{\beta\phi}\mathcal{L})\rVert O_{p}(N^{-1/2+4\epsilon})+O_{p}(1) =O_{p}(N^{4\epsilon})
\end{align*}

The 6th and 7th terms can be bounded as
\begin{align*}
\lVert\sum_{g}\mathcal{S}'\mathcal{H}^{-1}\mathcal{H}_{b}\mathcal{H}^{-1}\mathcal{H}_{s_{g}}\mathcal{H}^{-1}\mathcal{S}\mathcal{S}_{g}\rVert & \leq\lVert\sum_{g}\mathcal{S}'\bar{\mathcal{H}}^{-1}\bar{\mathcal{H}}_{b}\mathcal{\bar{H}}^{-1}\mathcal{H}_{s_{g}}\mathcal{H}^{-1}\mathcal{S}\mathcal{S}_{g}\rVert\\
 & +\lVert\sum_{g}\mathcal{S}'(\mathcal{H}^{-1}\mathcal{H}_{b}\mathcal{H}^{-1}-\bar{\mathcal{H}}^{-1}\bar{\mathcal{H}}_{b}\mathcal{\bar{H}}^{-1})\mathcal{H}_{s_{g}}\mathcal{H}^{-1}\mathcal{S}\mathcal{S}_{g}\rVert\\
 & \leq N^{1-1/q}\lVert\mathcal{H}^{-1}\rVert_{q}\lVert\bar{\mathcal{H}}^{-1}\bar{\mathcal{H}}_{b}\mathcal{\bar{H}}^{-1}\rVert_{q}\lVert\sum_{g}\mathcal{H}_{s_{g}}\mathcal{S}_{g}\rVert_{q}\lVert\mathcal{S}\rVert_{q}^{2}\\
 & +N^{1-1/q}\lVert\mathcal{H}^{-1}\mathcal{H}_{b}\mathcal{H}^{-1}-\bar{\mathcal{H}}^{-1}\bar{\mathcal{H}}_{b}\mathcal{\bar{H}}^{-1}\rVert\lVert\mathcal{H}^{-1}\rVert_{q}\lVert\sum_{g}\mathcal{H}_{s_{g}}\mathcal{S}_{g}\rVert_{q}\lVert\mathcal{S}\rVert_{q}^{2}\\
 & =O_{p}(N^{-\frac{3}{2}+4\epsilon})
\end{align*}
\begin{align*}
\lVert & \mathcal{S}'\mathcal{H}^{-1}W^{-1}\sum_{f}(\partial_{\phi\phi'\phi_{f}}\mathcal{L})\big[\mathcal{H}^{-1}\mathcal{H}_{bb}\mathcal{H}^{-1}(\partial_{\beta\phi}\mathcal{L})\big]_{f}\mathcal{H}^{-1}\mathcal{S}\rVert\\
 & \leq N^{1-1/q}\lVert W^{-1}\rVert\lVert\partial_{\phi\phi\phi}\mathcal{L}\rVert_{q}\lVert\mathcal{H}^{-1}\rVert_{q}^{4}\lVert\mathcal{H}_{bb}\rVert_{q}\lVert\mathcal{S}\rVert_{q}^{2}\lVert\partial_{\beta\phi}\mathcal{L}\rVert_{q}\\
 & =O_{p}(N^{-3+14\epsilon})
\end{align*}
Finally, for term 8 we have
\begin{align*}
\lVert\sum_{g,h}\mathcal{S}'\mathcal{H}^{-1}\mathcal{H}_{b}\mathcal{H}^{-1}\mathcal{H}_{s_{g}s_{h}}\mathcal{H}^{-1}\mathcal{S}\mathcal{S}_{g}\mathcal{S}_{h}\rVert & \leq\lVert\sum_{g,h}\mathcal{S}'\mathcal{\bar{H}}^{-1}\mathcal{\bar{H}}_{b}\mathcal{\bar{H}}^{-1}\mathcal{H}_{s_{g}s_{h}}\mathcal{H}^{-1}\mathcal{S}\mathcal{S}_{g}\mathcal{S}_{h}\rVert\\
 & +\lVert\sum_{g,h}\mathcal{S}'(\mathcal{H}^{-1}\mathcal{H}_{b}\mathcal{H}^{-1}-\mathcal{\bar{H}}^{-1}\mathcal{\bar{H}}_{b}\mathcal{\bar{H}}^{-1})\mathcal{H}_{s_{g}s_{h}}\mathcal{H}^{-1}\mathcal{S}\mathcal{S}_{g}\mathcal{S}_{h}\rVert\\
 & \leq N^{1-\frac{1}{q}}O_{p}(N^{-1})\lVert\mathcal{S}\rVert_{q}^{4}\lVert\mathcal{H}_{ss}\rVert_{q}\\
 & =O_{p}(N^{-2+\frac{3}{q}+4\epsilon})=O_{p}(N^{-2+10\epsilon})
\end{align*}
\end{proof}

\section{\label{sec:bounds_ind}Bounds on individual components}
This section provides bounds on the terms $\mathcal{P}$, $\mathcal{R}$, $\mathcal{F}$, $\mathcal{G}$, $W$, and $\mathcal{H}$, and their derivatives. These form the basis for bounding the terms in the asymptotic expansion. The bounds follow from the definition of the terms (see Section \ref{sec:expressions_ind}) as well as the bounds presented in the previous section.

\subsection{Bounds on $\mathcal{F}, \mathcal{P}, \mathcal{R}$ terms}
\begin{lem}
\label{lem:F}Let Assumptions 1 and 2 hold. Then we have, for $s,t\in\{1,2,3\}$
and $s+t\leq5$, 
\begin{align*}
\lVert\mathcal{F}^{s,t}\rVert & =O_{p}(N)\\
\lVert\mathcal{F}^{s,(r),t}\rVert & =O_{p}(N)\\
\lVert\mathcal{F}_{b}^{s,t}\rVert & =O_{p}(1)\\
\lVert\mathcal{F}_{b}^{s,(r),t}\rVert & =O_{p}(1)\\
\lVert\mathcal{S}'\mathcal{F}_{s}^{s,t}\rVert & =O_{p}(1)\\
\lVert\mathcal{F}_{bb}^{s,t}\rVert & =O_{p}(N^{-1+8\epsilon})\\
\lVert\mathcal{S}'\mathcal{F}_{ss'}\mathcal{S}\rVert & =O_{p}(1)
\end{align*}
\end{lem}
\begin{proof}
The first two results follow directly from Lemma \ref{lem:HEH_bound}.
For the first derivatives, we have
\begin{align*}
\lVert\mathcal{F}_{b}^{s,t}\rVert & \leq\lVert W^{-1}\rVert\Big(\lVert\mathcal{F}^{s+1,t}\rVert+\lVert\mathcal{F}^{1,(s),t}\rVert+\lVert\mathcal{F}^{s,t+1}\rVert+\lVert\mathcal{F}^{s,(t),1}\rVert\Big)\\
 & \quad+\lVert(\partial_{\beta^{s}\phi'}\mathcal{L})\mathcal{H}^{-1}\mathcal{H}_{b}\mathcal{H}^{-1}(\partial_{\beta^{t}\phi}\mathcal{L})\rVert\\
 & =O_{p}(1)
\end{align*}
from Lemma \ref{lem:tensors2} (the result for $\mathcal{F}_{b}^{s,(r),t}$
follows similarly). Also
\begin{align*}
\lVert\mathcal{S}'\mathcal{F}_{s}^{s,t}\rVert & \leq\lVert W^{-1}\rVert\lVert\mathcal{S}'\big(\mathcal{G}^{1,s+1}+\mathcal{G}\mathcal{E}^{s}\mathcal{H}^{-1}\big)(\partial_{\beta^{t}\phi}\mathcal{L})\rVert+\lVert\mathcal{S}'\mathcal{H}^{-1}\mathcal{E}^{s}\mathcal{H}^{-1}(\partial_{\beta^{t}\phi}\mathcal{L})\rVert\\
 & \quad+\lVert W^{-1}\rVert\lVert\mathcal{S}'\big(\mathcal{G}^{t+1,1}+\mathcal{H}^{-1}\mathcal{E}^{t}\mathcal{G}\big)(\partial_{\beta^{s}\phi}\mathcal{L})\rVert+\lVert\mathcal{S}'\mathcal{H}^{-1}\mathcal{E}^{t}\mathcal{H}^{-1}(\partial_{\beta^{s}\phi'}\mathcal{L})\rVert\\
 & \quad+\lVert(\partial_{\beta^{s}\phi'}\mathcal{L})\mathcal{H}^{-1}\sum_{g}\mathcal{H}_{s_{g}}\mathcal{S}_{g}\mathcal{H}^{-1}(\partial_{\beta^{t}\phi}\mathcal{L})\rVert\\
 & =O_{p}(1)
\end{align*}
which follows from results in Lemma \ref{lem:element1} and the bound

\begin{align*}
\lVert &(\partial_{\beta^{s}\phi'}\mathcal{L})\mathcal{H}^{-1}\sum_{g}\mathcal{H}_{s_{g}}\mathcal{S}_{g}\mathcal{H}^{-1}(\partial_{\beta^{t}\phi}\mathcal{L})\rVert \\
& \leq\lVert W^{-1}(\partial_{\beta\phi'}\mathcal{L})\mathcal{H}^{-1}\mathcal{S}\cdot(\partial_{\beta^{s}\phi'}\mathcal{L})\mathcal{H}^{-1}(\partial_{\beta\phi\phi'}\mathcal{L})\mathcal{H}^{-1}(\partial_{\beta^{t}\phi}\mathcal{L})\rVert\\
 & +\lVert(\partial_{\beta^{s}\phi'}\mathcal{L})\mathcal{H}^{-1}\sum_{f}(\partial_{\phi\phi'\phi_{f}}\mathcal{L})[\mathcal{H}^{-1}\mathcal{S}]_{f}\mathcal{H}^{-1}(\partial_{\beta^{t}\phi}\mathcal{L})\rVert\\
 & +\lVert W^{-1}(\partial_{\beta\phi'}\mathcal{L})\mathcal{H}^{-1}\mathcal{S}\cdot(\partial_{\beta^{s}\phi'}\mathcal{L})\mathcal{H}^{-1}\sum_{f}(\partial_{\phi\phi'\phi_{f}}\mathcal{L})\big[(\partial_{\beta\phi'}\mathcal{L})\mathcal{H}^{-1}\big]_{f}\mathcal{H}^{-1}(\partial_{\beta^{t}\phi}\mathcal{L})\rVert\\
 & =O_{p}(1)
\end{align*}
where we use Lemma \ref{lem:HEH_bound} on the term $\sum_{f}(\partial_{\phi\phi'\phi_{f}}\bar{\mathcal{L}})\big[(\partial_{\beta\phi'}\bar{\mathcal{L}})\bar{\mathcal{H}}^{-1}\big]_{f}$ and apply Lemma \ref{lem:element1}.

The result for $\mathcal{F}_{bb}^{s,t}$ follows similarly by applying
results in Lemmas \ref{lem:tensor_bounds} and \ref{lem:tensors2}
along with Lemma \ref{lem:HEH_bound} (note that Lemma \ref{lem:HEH_bound}
may be applied repeatedly since the product of two matrices with $O_{p}(1)$
diagonal and $O_{p}(N^{-1})$ off-diagonals has this same property).
The final result follows from applications of the same bounds to the
expression for $\mathcal{F}_{ss}$.
\end{proof}

\begin{lem}
	\label{lem:R}Let Assumptions 1 and 2 hold. Then we have, for $a\in\{1,2,3\}$, 
\begin{align*}
	\lVert\mathcal{R}^{a}\rVert_{q} & =O_{p}(N^{2\epsilon})\\
	\lVert\mathcal{R}_{b}^{a}\rVert_{q} & =O_{p}(N^{-1+6\epsilon})\\
	\lVert\mathcal{R}_{s}^{a}\rVert_{q} & =O_{p}(N^{4\epsilon})\\
	\lVert\mathcal{R}_{bs}^{a}\rVert_{q} & =O_{p}(N^{-1+8\epsilon})\\
	\lVert\mathcal{R}_{ss}^{a}\rVert_{q} & =O_{p}(N^{6\epsilon})\\
	\lVert\mathcal{R}_{bss}^{a}\rVert_{q} & =O_{p}(N^{-1+10\epsilon})
\end{align*}
\end{lem}
\begin{proof}
Application of the individual bounds shown in this section along with the Cauchy-Schwarz and triangle inequalities give the results. For example,
\begin{align*}
	\lVert\mathcal{R}_{s}^{a}\rVert_{q} & \leq\Big(\lVert W^{-1}\rVert_{q}\lVert\mathcal{R}^{1}\rVert_{q}\lVert\mathcal{R}^{a+1}\rVert_{q}\\
	& +\lVert\mathcal{E}^{a}\rVert_{q}\lVert\mathcal{H}^{-1}\rVert_{q}^{2}+\lVert W^{-1}\rVert_{q}\lVert\mathcal{R}^{1}\rVert_{q}^{2}\lVert\mathcal{E}^{a}\rVert_{q}\lVert\mathcal{H}^{-1}\rVert_{q}\\
	& +\lVert\mathcal{R}^{a}\rVert_{q}\lVert\mathcal{H}^{-1}\rVert_{q}\lVert\mathcal{H}_{s}\rVert_{q}\Big)\\
	& =O_{p}(N^{4\epsilon})
\end{align*}	
\end{proof}

\begin{lem}
	\label{lem:P}Let Assumptions 1 and 2 hold. Then we have, for $a,r\in\{1,2,3\}$, 
\begin{align*}
	\lVert\mathcal{P}^{r}\rVert_{q} & =O_{p}(N^{2\epsilon}) & \lVert\mathcal{P}^{(a,r)}\rVert_{q} & =O_{p}(N^{4\epsilon})\\
	\lVert\mathcal{P}_{b}^{r}\rVert_{q} & =O_{p}(N^{-1+6\epsilon}) & \lVert\mathcal{P}_{b}^{(a,r)}\rVert_{q} & =O_{p}(N^{-1+8\epsilon})\\
	\lVert\mathcal{P}_{s}^{r}\rVert_{q} & =O_{p}(N^{4\epsilon}) & \lVert\mathcal{P}_{s}^{(a,r)}\rVert_{q} & =O_{p}(N^{6\epsilon})\\
	\lVert\mathcal{P}_{bs_{}}^{r}\rVert_{q} & =O_{p}(N^{-1+8\epsilon}) & \lVert\mathcal{P}_{bs}^{(a,r)}\rVert & =O_{p}(N^{-1+10\epsilon})\\
	\lVert\mathcal{P}_{ss}^{r}\rVert_{q} & =O_{p}(N^{6\epsilon}) & \lVert\mathcal{P}_{ss}^{(a,r)}\rVert_{q} & =O_{p}(N^{8\epsilon})\\
	\lVert\mathcal{P}_{bss}^{r}\rVert_{q} & =O_{p}(N^{-1+10\epsilon}) & \lVert\mathcal{P}_{bss}^{(a,r)}\rVert_{q} & =O_{p}(N^{-1+12\epsilon})
\end{align*}
\begin{align*}
	\lVert\mathcal{P}_{s_{}}^{r}\rVert_{q} & \leq\lVert W^{-1}\rVert_{q}\lVert\mathcal{R}^{1}\rVert_{q}\lVert\mathcal{P}_{}^{r+1}\rVert_{q}\\
	& +\lVert\partial_{\beta^{r}\phi^{4}}\mathcal{L}\rVert_{q}\lVert\mathcal{H}^{-1}\rVert_{q}^{2}\\
	& +\lVert W^{-1}\rVert_{q}\lVert\partial_{\beta^{r}\phi^{4}}\mathcal{L}\rVert_{q}\lVert\mathcal{H}^{-1}\rVert_{q}\lVert\mathcal{R}^{1}\rVert_{q}^{2}\\
	& +\lVert\partial_{\beta^{r}\phi^{3}}\mathcal{L}\rVert_{q}\lVert\mathcal{H}^{-1}\rVert_{q}^{2}\lVert\mathcal{H}_{s}\rVert_{q}\\
	& =O_{p}(N^{4\epsilon})
\end{align*}
\end{lem}
\begin{proof}
	As above, application of the individual bounds shown in this section along with the Cauchy-Schwarz and triangle inequalities give the results. For example,
\begin{align*}
	\lVert\mathcal{P}_{s_{}}^{r}\rVert_{q} & \leq\lVert W^{-1}\rVert_{q}\lVert\mathcal{R}^{1}\rVert_{q}\lVert\mathcal{P}_{}^{r+1}\rVert_{q}\\
	& +\lVert\partial_{\beta^{r}\phi^{4}}\mathcal{L}\rVert_{q}\lVert\mathcal{H}^{-1}\rVert_{q}^{2}\\
	& +\lVert W^{-1}\rVert_{q}\lVert\partial_{\beta^{r}\phi^{4}}\mathcal{L}\rVert_{q}\lVert\mathcal{H}^{-1}\rVert_{q}\lVert\mathcal{R}^{1}\rVert_{q}^{2}\\
	& +\lVert\partial_{\beta^{r}\phi^{3}}\mathcal{L}\rVert_{q}\lVert\mathcal{H}^{-1}\rVert_{q}^{2}\lVert\mathcal{H}_{s}\rVert_{q}\\
	& =O_{p}(N^{4\epsilon})
\end{align*}
\end{proof}

\subsection{Bounds on $W$ terms}
\begin{lem}
\label{lem:W}Let Assumptions 1 and 2 hold. Then
\begin{align*}
\lVert W\rVert & =O_{p}(N) & \lVert W_{b}\rVert & =O_{p}(N^{4\epsilon})\\
\lVert\mathcal{S}'W_{s}\rVert & =O_{p}(1) & \lVert W_{s}\rVert_{q} &=O_{p}(N^{1/q})\\
\lVert W_{bb}\rVert & =O_{p}(N^{-1+8\epsilon}) & \lVert W_{bs}\rVert_{q} &=O_{p}(N^{-1+6\epsilon})\\
\lVert\mathcal{S}'W_{bs}\rVert & =O_{p}(N^{-1+4\epsilon}) & \lVert W_{ss}\rVert_{q} & =O_{p}(N^{6\epsilon})\\
\lVert W_{bss}\rVert_{q} & =O_{p}(N^{-1+12\epsilon})
\end{align*}
\end{lem}
\begin{proof}
The first result follows from the expression for $W$ and the result
in Lemma \ref{lem:F}.
\begin{align*}
\lVert W_{b}\rVert & \leq\lVert W^{-1}\rVert\lVert\partial_{\beta\beta\beta}\mathcal{L}\rVert+\lVert W^{-1}\rVert\lVert\mathcal{F}^{2,1}\rVert+\lVert\mathcal{F}_{b}\rVert\\
 & =O_{p}(N^{4\epsilon})
\end{align*}
For the next result we have
\begin{align*}
\lVert\mathcal{S}'W_{s}\rVert & \leq\lVert W^{-1}\partial_{\beta\beta\beta}\mathcal{L}\rVert\lVert\mathcal{S}'\mathcal{H}^{-1}(\partial_{\beta\phi}\mathcal{L})\rVert+\lVert\mathcal{S}'\Big[\mathcal{H}^{-1}+W^{-1}\mathcal{G}\big](\partial_{\beta\beta\phi'}\mathcal{L})\rVert+\lVert\mathcal{S}'\mathcal{F}_{s}\rVert\\
 & =O_{p}(1)
\end{align*}
Similarly, the remaining results simply follow from the definitions
of the terms and the bounds on individual components that have already
been derived.
\end{proof}

\subsection{Bounds on $\mathcal{H}$ terms}
\begin{lem}
\label{lem:H}Let Assumptions 1 and 2 hold. Then,\\ \\
(i) $\lVert\mathcal{H}_{b}\rVert  =O_{p}(N^{-1+4\epsilon})$ \\
(ii) $\lVert\mathcal{H}_{s'}\mathcal{S}\rVert  =O_{p}(N^{-\frac{1}{2}+2\epsilon})$ and $\lVert\mathcal{H}_{s}\rVert_{q}=O_{p}(N^{2\epsilon})$ \\
(iii) $\lVert\mathcal{H}_{bb}\rVert  =O_{p}(N^{-2+8\epsilon})$ and $\lVert\mathcal{S}'\mathcal{H}^{-1}\mathcal{H}_{bb}\mathcal{H}^{-1}(\partial_{\beta\phi}\mathcal{L})\rVert=O_{p}(N^{-2+4\epsilon})$\\
(iv) $\lVert\sum_{g}\mathcal{S}'\mathcal{H}^{-1}\mathcal{H}_{bs_{g}}\mathcal{H}^{-1}\mathcal{S}\mathcal{S}_{g}\rVert  =O_{p}(N^{-\frac{3}{2}+2\epsilon})$ and $\lVert\mathcal{H}_{bs}\rVert_{q}=O_{p}(N^{-1+4\epsilon})$\\
(v) $\lVert\mathcal{H}_{ss}\rVert_{q} =O_{p}(N^{4\epsilon})$\\
(vi) $\lVert\mathcal{S}'\mathcal{H}^{-1}\mathcal{H}_{bbb}\mathcal{H}^{-1}\mathcal{S}\rVert =O_{p}(N^{-3+14\epsilon})$ and $\lVert\mathcal{H}_{bbb}\rVert_{q}=O_{p}(N^{-3+12\epsilon})$\\
(vii) $\lVert\sum_{g}\mathcal{S}'\mathcal{H}^{-1}\mathcal{H}_{bbs_{g}}\mathcal{H}^{-1}\mathcal{S}\mathcal{S}_{g}\rVert  =O_{p}(N^{-5/2+10\epsilon})$\\
(viii) $\lVert\mathcal{H}_{bss}\rVert_{q}=O_{p}(N^{-1+8\epsilon})$ and $\lVert\mathcal{H}_{bbs}\rVert_{q}=O_{p}(N^{-2+8\epsilon})$\\
(ix) $\lVert\mathcal{H}_{bsss}\rVert_{q}  =O_{p}(N^{-1+12\epsilon})$ and $\lVert\mathcal{H}_{sss}\rVert_{q}=O_{p}(N^{6\epsilon})$
\end{lem}
\begin{proof}
For (i), application of Lemma \ref{lem:tensor_bounds} and the fact that $\frac{1}{N}W>0$
by assumption gives.
\begin{align*}
\lVert\mathcal{H}_{b}\rVert & \leq\lVert W^{-1}\rVert\big(\lVert\partial_{\beta\phi\phi'}\mathcal{L}\rVert+\lVert\sum_{f}(\partial_{\phi\phi'\phi_{f}}\mathcal{L})\big[\mathcal{H}^{-1}(\partial_{\beta\phi}\mathcal{L})\big]_{f}\rVert\big)\\
 & =O_{p}(N^{-1})\big(O_{p}(N^{2\epsilon})+O_{p}(N^{4\epsilon})\big)=O_{p}(N^{-1+4\epsilon})
\end{align*}

For (ii), we have 
\begin{align*}
\sum_{g}\mathcal{H}_{s_{g}}\mathcal{S}_{g} & =W^{-1}(\partial_{\beta\phi'}\mathcal{L})\mathcal{H}^{-1}\mathcal{S}(\partial_{\beta\phi\phi'}\mathcal{L})\\
 & +\sum_{f}(\partial_{\phi\phi'\phi_{f}}\mathcal{L})[\mathcal{H}^{-1}\mathcal{S}]_{f}\\
 & +W^{-1}\sum_{f}(\partial_{\phi\phi'\phi_{f}}\mathcal{L})\big[(\partial_{\beta\phi'}\mathcal{L})\mathcal{H}^{-1}\big]_{f}(\partial_{\beta\phi'}\mathcal{L})\mathcal{H}^{-1}\mathcal{S}
\end{align*}
and hence, using $\lVert W^{-1}\rVert =O_p(N^{-1})$, an applying Lemma \ref{lem:S6} (iii) and Lemma \ref{lem:element1} we get
\begin{align*}
\lVert\sum_{g}\mathcal{H}_{s_{g}}\mathcal{S}_{g}\rVert & \leq\lVert W^{-1}\rVert\lVert(\partial_{\beta\phi'}\mathcal{L})\mathcal{H}^{-1}\mathcal{S}\rVert\lVert\partial_{\beta\phi\phi'}\mathcal{L}\rVert+\lVert\sum_{f}(\partial_{\phi\phi'\phi_{f}}\mathcal{L})[\mathcal{H}^{-1}S]_{f}\rVert\\
 & +\lVert W^{-1}\rVert\lVert\sum_{f}(\partial_{\phi\phi'\phi_{f}}\mathcal{L})\big[(\partial_{\beta\phi'}\mathcal{L})\mathcal{H}^{-1}\big]_{f}\rVert\lVert(\partial_{\beta\phi'}\mathcal{L})\mathcal{H}^{-1}\mathcal{S}\rVert\\
 & =O_{p}(N^{-1+2\epsilon})+O_{p}(N^{-1/2+2\epsilon})+O_{p}(N^{-1+4\epsilon})\\
 & =O_{p}(N^{-1/2+2\epsilon})
\end{align*}
The second statement follows from Lemma \ref{lem:S6} (iii) and the definition of $\mathcal{H}_s$. 

For (iii), we have from the expression for $\mathcal{H}_{bb}$ that
\begin{align*}
	\lVert \mathcal{H}_{bb}\rVert  & \leq \lVert W^{-1}\rVert \lVert W_{b} \lVert \mathcal{H}_{b}+ \lVert W^{-1} \lVert \mathcal{E}_{b}^{1}\rVert\\
	&+ \lVert W^{-2}\rVert   \lVert  \sum_{f}(\partial_{\beta\phi\phi'\phi_{f}}\mathcal{L})\big[\mathcal{H}^{-1}(\partial_{\beta\phi}\mathcal{L})\big]_{f}\rVert\\
	& +\lVert W^{-2} \rVert \lVert \sum_{e,f}(\partial_{\phi\phi'\phi_{e}\phi_{f}}\mathcal{L})\big[\mathcal{H}^{-1}(\partial_{\beta\phi}\mathcal{L})\big]_{e}\big[\mathcal{H}^{-1}(\partial_{\beta\phi}\mathcal{L})\big]_{f}\rVert\\
	&  +\lVert W^{-1}\rVert \lVert\sum_{f}(\partial_{\phi\phi'\phi_{f}}\mathcal{L})\big[\mathcal{H}^{-1}\mathcal{H}_{b}\mathcal{H}^{-1}(\partial_{\beta\phi}\mathcal{L})\big]_{f}\rVert\\
	&  +\lVert W^{-2} \rVert\lVert\sum_{f}(\partial_{\phi\phi'\phi_{f}}\mathcal{L})\big[\mathcal{H}^{-1}(\partial_{\beta\beta\phi}\mathcal{L})\big]_{f}\rVert\\
	&  +\lVert W^{-2}\rVert \lVert\sum_{f}(\partial_{\phi\phi'\phi_{f}}\mathcal{L})\big[\mathcal{H}^{-1}\mathcal{E}^{1}\mathcal{H}^{-1}(\partial_{\beta\phi}\mathcal{L})\big]_{f}\rVert\\
	&=O_p(N^{-2+8\epsilon})
\end{align*}
where the final line comes frrom applying Lemmas \ref{lem:tensors2} and \ref{lem:tensors2} as well as \ref{lem:W}. Similarly, each remaining term can be bound using the lemmas in Section \ref{sec:basic_lemmas} and this section, the expressions given for each term and application of the Cauchy-Schwarz and triangle inequalities.
\end{proof}

\subsection{Bounds on $\mathcal{G}$ terms}
\begin{lem}
\label{lem:G}Let Assumptions 1 and 2 hold. Then,$s+t \leq 5$\\ \\
(i) $\lVert\mathcal{S}'\mathcal{G}^{s,t}\mathcal{S}\rVert  =O_{p}(1)$\\
(ii) $\lVert\mathcal{S}'\mathcal{G}_{b}^{s,t}\mathcal{S}\rVert =O_{p}(N^{-1})$\\
(iii) $\lVert\sum_{g}\mathcal{S}'\mathcal{G}_{s_{g}}^{s,t}\mathcal{S}\mathcal{S}_{g}\rVert  =O_{p}(N^{-\frac{1}{2}+4\epsilon})$\\
(iv) $\lVert\mathcal{S}'\mathcal{G}_{bb}\mathcal{S}\rVert  =O_{p}(N^{-2+4\epsilon})$\\
(v) $\lVert\sum_{g}\mathcal{S}'\mathcal{G}_{bs_{g}}\mathcal{S}\mathcal{S}_{g}\rVert  =O_{p}(N^{-\frac{3}{2}+4\epsilon})$\\
(vi) $\lVert\mathcal{S}'\mathcal{G}_{bbb}\mathcal{S}\rVert  =O_{p}(N^{-\frac{5}{2}+14\epsilon})$
\end{lem}

\begin{proof}
(i) Writing $\lVert\mathcal{S}'\mathcal{G}^{s,t}\mathcal{S}\rVert\leq\lVert\mathcal{S}'\mathcal{H}^{-1}(\partial_{\beta^s\phi}\mathcal{L})\rVert\lVert\mathcal{S}'\mathcal{H}^{-1}(\partial_{\beta^t\phi}\mathcal{L})\rVert$,
the result follows from Lemma \ref{lem:element1}. For (ii) we have
\begin{align*}
\lVert\mathcal{S}'\mathcal{G}_{b}\mathcal{S}\rVert & \leq\lVert\mathcal{S}'\mathcal{H}^{-1}\mathcal{H}_{b}\mathcal{G}\mathcal{S}\rVert+\lVert\mathcal{S}'\mathcal{G}\mathcal{H}_{b}\mathcal{H}^{-1}\mathcal{S}\rVert+\lVert W^{-1}\mathcal{S}'\mathcal{G}^{2,1}\mathcal{S}\rVert\\
 & +\lVert W^{-1}\mathcal{S}'\mathcal{H}^{-1}\mathcal{E}\mathcal{G}\mathcal{S}\rVert+\lVert W^{-1}\mathcal{S}'\mathcal{G}^{1,2}\mathcal{S}\rVert+\lVert W^{-1}\mathcal{S}'\mathcal{G}\mathcal{E}\mathcal{H}^{-1}\mathcal{S}\rVert\\
 & =O_{p}(N^{-1})
\end{align*}
where the final line applies results from Lemma \ref{lem:element1} (we show for $s=t=1$ here but results apply for $s,t$ as in the lemma),
e.g.
\begin{align*}
\rVert\mathcal{S}'\mathcal{H}^{-1}\mathcal{H}_{b}\mathcal{G}\mathcal{S}\rVert & \leq\rVert\mathcal{S}'\mathcal{H}^{-1}\mathcal{H}_{b}\mathcal{H}^{-1}(\partial_{\beta\phi'}\mathcal{L})\rVert\lVert(\partial_{\beta\phi'}\mathcal{L})\mathcal{H}^{-1}\mathcal{S}\rVert\\
 & \leq\rVert\mathcal{S}'\mathcal{H}^{-1}\mathcal{\bar{H}}_{b}\mathcal{H}^{-1}(\partial_{\beta\phi'}\mathcal{L})\rVert O_{p}(1)+o_{p}(1)\\
 & =O_{p}(N^{-1})
\end{align*}

For (iii)
\begin{align*}
\lVert\sum_{g}\mathcal{S}'\mathcal{G}_{s_{g}}\mathcal{S}\mathcal{S}_{g}\rVert & \leq\lVert\sum_{g}\mathcal{S}'\mathcal{H}^{-1}\mathcal{H}_{s_{g}}\mathcal{G}\mathcal{S}\mathcal{S}_{g}\rVert+\lVert\sum_{g}\mathcal{S}'\mathcal{G}\mathcal{H}_{s_{g}}\mathcal{H}^{-1}\mathcal{S}\mathcal{S}_{g}\rVert\\
 & +\lVert W^{-1}\mathcal{S}'(\mathcal{G}^{2,1}+\mathcal{G}^{1,2})\mathcal{S}\rVert\lVert\mathcal{S}'\mathcal{H}^{-1}(\partial_{\beta\phi}\mathcal{L})\rVert\\
 & +\lVert\mathcal{S}'\mathcal{H}^{-1}\mathcal{E}^{1}\mathcal{H}^{-1}\mathcal{S}\lVert\rVert(\partial_{\beta^{t}\phi'}\mathcal{L})\mathcal{H}^{-1}\mathcal{S}\rVert\\
 & +\lVert\mathcal{S}'\mathcal{H}^{-1}(\partial_{\beta^{s}\phi}\mathcal{L})\rVert\lVert\mathcal{S}'\mathcal{H}^{-1}\mathcal{E}^{1}\mathcal{H}^{-1}\mathcal{S}\rVert\\
 & +\lVert W^{-1}\mathcal{S}'\mathcal{H}^{-1}\mathcal{E}^{1}\mathcal{G}\mathcal{S}\rVert\lVert\mathcal{S}'\mathcal{H}^{-1}(\partial_{\beta\phi}\mathcal{L})\rVert\\
 & +\lVert W^{-1}\mathcal{S}'\mathcal{G}\mathcal{E}^{1}\mathcal{H}^{-1}\mathcal{S}\rVert\lVert\mathcal{S}'\mathcal{H}^{-1}(\partial_{\beta\phi}\mathcal{L})\rVert\\
 & =O_{p}(N^{2\epsilon})+O_{p}(N^{2\epsilon})+O_{p}(N^{-1})+O_{p}(1)\\
 & =O_{p}(N^{2\epsilon})
\end{align*}
by application of Lemma \ref{lem:element1} and the bounds on individual
components, e.g.
\begin{align*}
\lVert\sum_{g}\mathcal{S}'\mathcal{H}^{-1}\mathcal{H}_{s_{g}}\mathcal{G}\mathcal{S}\mathcal{S}_{g}\rVert & \leq N^{1-\frac{2}{q}}\lVert\mathcal{H}_{s}\rVert_{q}\lVert\mathcal{H}^{-1}\rVert_{q}\lVert\mathcal{S}\rVert_{q}^{3}\\
 & =O_{p}(N^{-\frac{1}{2}+\frac{1}{q}+2\epsilon})=O_{p}(N^{-\frac{1}{2}+4\epsilon})
\end{align*}

The remaining components can be bounded similarly using Lemmas \ref{lem:element1} to \ref{lem:tensors2}, but are not shown here due to the length of the expressions for second and third order derivatives of $\mathcal{G}$.
\end{proof}

\section{\label{sec:full_expansion} Bounds for expansion terms}
Finally, we make use of the bounds in the previous section to give bounds on the asymptotic expansion terms, which justifies the expansion in \ref{sec:beta_expansion}. All results follow through applcation of the expressions derived in Section \ref{sec:expressions_ind},
along with the bounds from Section \ref{sec:bounds_ind}, and application
of the Cauchy-Schwarz/triangle inequalities.

\subsection{Third derivative terms}

\begin{align*}
\lVert\partial_{bbb}\mathcal{L}^{*}\rVert & \leq\lVert W^{-2}\rVert\lVert W_{b}\rVert\\
 & =O_{p}(N^{-2+4\epsilon})\\
\lVert(\partial_{bbs}\mathcal{L}^{*})\mathcal{S}\rVert & \leq\lVert W^{-2}\rVert\lVert W_{s}\mathcal{S}\rVert\\
 & =O_{p}(N^{-2})\\
\lVert\mathcal{S}'(\partial_{bss'}\mathcal{L}^{*})\mathcal{S}\rVert & \leq\lVert\mathcal{S}'\mathcal{H}^{-1}\mathcal{H}_{b}\mathcal{H}^{-1}\mathcal{S}\rVert-\lVert W^{-2}\rVert\lVert W_{b}\rVert\lVert\mathcal{S}'\mathcal{G}\mathcal{S}\rVert+\lVert W^{-1}\rVert\lVert\mathcal{S}'\mathcal{G}_{b}\mathcal{S}\rVert\\
 & =O_{p}(N^{-1+4\epsilon})+O_{p}(N^{-2+4\epsilon})+O_{p}(N^{-\frac{3}{2}+4\epsilon})\\
 & =O_{p}(N^{-1+4\epsilon})\\
\lVert\sum_{g}\mathcal{S}'(\partial_{ss's_{g}}\mathcal{L}^{*})\mathcal{S}\mathcal{S}_{g}\rVert & =\lVert\sum_{g}\mathcal{S}'\mathcal{H}^{-1}\mathcal{H}_{s_{g}}\mathcal{H}^{-1}\mathcal{S}\mathcal{S}_{g}\rVert-\lVert W^{-2}\rVert\lVert\sum_{g}W_{s_{g}}\mathcal{S}_{g}\rVert\lVert\mathcal{S}'\mathcal{G}\mathcal{S}\rVert \\
&+\lVert W^{-1}\rVert\lVert\sum_{g}\mathcal{S}'\mathcal{G}_{s_{g}}\mathcal{S}\mathcal{S}_{g}\rVert\\
 & =O_{p}(N^{-1/2+2\epsilon})+O_{p}(N^{-3/2+2\epsilon})+O_{p}(N^{-1+2\epsilon})\\
 & =O_{p}(N^{-1/2+2\epsilon})
\end{align*}

\subsection{Fourth derivative terms}
Using the bounds provided in Section \ref{sec:bounds_ind} we find

\begin{align*}
\lVert\partial_{b^{4}}\mathcal{L}^{*}\rVert & \leq2\lVert W^{-3}\rVert\lVert W_{b}^{2}\rVert+\lVert W^{-2}\lVert\lVert W_{bb}\rVert\\
 & =O_{p}(N^{-3+8\epsilon})=o_{p}(N^{-2})\\
\lVert(\partial_{bbbs'}\mathcal{L}^{*})\mathcal{S}\rVert & \leq2\lVert W^{-3}\rVert\lVert W_{b}\rVert\lVert W_{s'}\mathcal{S}\rVert+\lVert W^{-2}\rVert\lVert W_{bs'}\mathcal{S}\rVert\\
 & =O_{p}(N^{-\frac{5}{2}+4\epsilon})+O_{p}(N^{-3+4\epsilon})\\
 & =O_{p}(N^{-\frac{5}{2}+4\epsilon})=o_{p}(N^{-2})\\
\lVert\mathcal{S}'(\partial_{bbss'}\mathcal{L}_{(1)}^{*})\mathcal{S}\rVert & \leq\lVert2\mathcal{S}'\mathcal{H}^{-1}\mathcal{H}_{b}\mathcal{H}^{-1}\mathcal{H}_{b}\mathcal{H}^{-1}\mathcal{S}\rVert+\lVert\mathcal{S}'\mathcal{H}^{-1}\mathcal{H}_{bb}\mathcal{H}^{-1}\mathcal{S}\rVert\\
 & =O_{p}(N^{-2})+O_{p}(N^{-2+8\epsilon})\\
\lVert\mathcal{S}'(\partial_{bbss'}\mathcal{L}_{(2)}^{*})\mathcal{S}\rVert & \leq\lVert2\mathcal{S}'W^{-3}W_{b}^{2}\mathcal{G}\mathcal{S}\rVert+\lVert\mathcal{S}'W^{-2}W_{bb}\mathcal{G}\mathcal{S}\rVert+\lVert\mathcal{S}'W^{-2}W_{b}\mathcal{G}_{b}\mathcal{S}\rVert\\
 & +\lVert\mathcal{S}'W^{-1}\mathcal{G}_{bb}\mathcal{S}\rVert\\
 & =O_{p}(N^{-3+8\epsilon})+O_{p}(N^{-3+8\epsilon})+O_{p}(N^{-3+4\epsilon})=o_{p}(N^{-2})\\
\lVert\sum_{g}\mathcal{S}'\partial_{bss's_{g}}\mathcal{L}_{(1)}^{*}\mathcal{S}\mathcal{S}_{g}\rVert & \leq\lVert\sum_{g}\mathcal{S}'\mathcal{H}^{-1}\mathcal{H}_{b}\mathcal{H}^{-1}\mathcal{H}_{s_{g}}\mathcal{H}^{-1}\mathcal{S}\mathcal{S}_{g}\rVert+\lVert\sum_{g}\mathcal{S}'\mathcal{H}^{-1}\mathcal{H}_{s_{g}}\mathcal{H}^{-1}\mathcal{H}_{b}\mathcal{H}^{-1}\mathcal{S}\mathcal{S}_{g}\rVert\\
 & +\lVert\sum_{g}\mathcal{S}'\mathcal{H}^{-1}\mathcal{H}_{bs_{g}}\mathcal{H}^{-1}\mathcal{S}\mathcal{S}_{g}\rVert\\
 & =O_{p}(N^{-\frac{3}{2}+4\epsilon})+O_{p}(N^{-\frac{3}{2}+2\epsilon})+O_{p}(N^{-\frac{3}{2}+2\epsilon})\\
\lVert\sum_{g}\mathcal{S}'\partial_{bss's_{g}}\mathcal{L}_{(2)}^{*}\mathcal{S}\mathcal{S}_{g}\rVert & \leq\lVert2\sum_{g}W^{-3}W_{b}W_{s_{g}}\mathcal{S}'\mathcal{G}\mathcal{S}\mathcal{S}_{g}\rVert+\lVert\sum_{g}W^{-2}W_{bs_{g}}\mathcal{S}'\mathcal{G}\mathcal{S}\mathcal{S}_{g}\rVert\\
 & +\lVert\sum_{g}W^{-2}W_{s_{g}}\mathcal{S}'\mathcal{G}_{b}\mathcal{S}\mathcal{S}_{g}\rVert\\
 & +\lVert\sum_{g}W^{-2}W_{b}\mathcal{S}'\mathcal{G}_{s_{g}}\mathcal{S}\mathcal{S}_{g}\rVert+\lVert\sum_{g}W^{-1}\mathcal{S}'\mathcal{G}_{bs_{g}}\mathcal{S}\mathcal{S}_{g}\rVert\\
 & =O_{p}(N^{-\frac{5}{2}+6\epsilon})+O_{p}(N^{-3+4\epsilon})+O_{p}(N^{-\frac{5}{2}+2\epsilon})\\
 & +O_{p}(N^{-\frac{5}{2}+6\epsilon})+O_{p}(N^{-\frac{5}{2}+6\epsilon})\\
 & =o_{p}(N^{-2})
\end{align*}

\subsection{Fifth derivative terms}

Application of the bounds provided in Section \ref{sec:bounds_ind}
gives

\begin{align*}
\lVert\mathcal{S}'\partial_{bbbss'}\mathcal{L}_{(1)}^{*}\mathcal{S}\rVert & \leq\lVert6\mathcal{S}'\mathcal{H}^{-1}\mathcal{H}_{b}\mathcal{H}^{-1}\mathcal{H}_{b}\mathcal{H}^{-1}\mathcal{H}_{b}\mathcal{H}^{-1}\mathcal{S}\rVert\\
 & +3\lVert\mathcal{S}'\mathcal{H}^{-1}\mathcal{H}_{bb}\mathcal{H}^{-1}\mathcal{H}_{b}\mathcal{H}^{-1}\mathcal{S}\rVert+3\lVert\mathcal{S}'\mathcal{H}^{-1}\mathcal{H}_{b}\mathcal{H}^{-1}\mathcal{H}_{bb}\mathcal{H}^{-1}\mathcal{S}\rVert\\
 & +\lVert\mathcal{S}'\mathcal{H}^{-1}\mathcal{H}_{bbb}\mathcal{H}^{-1}\mathcal{S}\rVert\\
 & =O_{p}(N^{-3})+O_{p}(N^{-3+12\epsilon})+O_{p}(N^{-3+14\epsilon})=o_{p}(N^{-2})
\end{align*}
\begin{align*}
\lVert\sum_{g}\mathcal{S}'\partial_{bbss's_{g}}\mathcal{L}_{(1)}^{*}\mathcal{S}\mathcal{S}_{g}\rVert & =6\lVert\sum_{g}\mathcal{S}'\mathcal{H}^{-1}\mathcal{H}_{s_{g}}\mathcal{H}^{-1}\mathcal{H}_{b}\mathcal{H}^{-1}\mathcal{H}_{b}\mathcal{H}^{-1}\mathcal{S}\mathcal{S}_{g}\rVert\\
 & +4\lVert\sum_{g}\mathcal{S}'\mathcal{H}^{-1}\mathcal{H}_{bs_{g}}\mathcal{H}^{-1}\mathcal{H}_{b}\mathcal{H}^{-1}\mathcal{S}\mathcal{S}_{g}\rVert\\
 & +2\lVert\sum_{g}\mathcal{S}'\mathcal{H}^{-1}\mathcal{H}_{s_{g}}\mathcal{H}^{-1}\mathcal{H}_{bb}\mathcal{H}^{-1}\mathcal{S}\mathcal{S}_{g}\rVert+\\
 & +\lVert\sum_{g}\mathcal{S}'\mathcal{H}^{-1}\mathcal{H}_{bbs_{g}}\mathcal{H}^{-1}\mathcal{S}\mathcal{S}_{g}\rVert\\
 & =O_{p}(N^{-\frac{5}{2}+10\epsilon})+O_{p}(N^{-\frac{5}{2}+6\epsilon})+O_{p}(N^{-\frac{5}{2}+10\epsilon})\\
 & =o_{p}(N^{-2})
\end{align*}
\begin{align*}
\lVert\sum_{g,h}\mathcal{S}'\partial_{bss's_{g}s_{h}}\mathcal{L}_{(1)}^{*}\mathcal{S}\mathcal{S}_{g}\mathcal{S}_{h}\rVert & \leq6\lVert\sum_{g,h}\mathcal{S}'\mathcal{H}^{-1}\mathcal{H}_{s_{h}}\mathcal{H}^{-1}\mathcal{H}_{b}\mathcal{H}^{-1}\mathcal{H}_{s_{g}}\mathcal{H}^{-1}\mathcal{S}\mathcal{S}_{g}\mathcal{S}_{h}\rVert\\
 & +4\lVert\sum_{g,h}\mathcal{S}'\mathcal{H}^{-1}\mathcal{H}_{bs_{h}}\mathcal{H}^{-1}\mathcal{H}_{s_{g}}\mathcal{H}^{-1}\mathcal{S}\mathcal{S}_{g}\mathcal{S}_{h}\rVert\\
 & +2\lVert\sum_{g,h}\mathcal{S}'\mathcal{H}^{-1}\mathcal{H}_{b}\mathcal{H}^{-1}\mathcal{H}_{s_{g}s_{h}}\mathcal{H}^{-1}\mathcal{S}\mathcal{S}_{g}\mathcal{S}_{h}\rVert\\
 & +\lVert\sum_{g,h}\mathcal{S}'\mathcal{H}^{-1}\mathcal{H}_{bs_{g}s_{h}}\mathcal{H}^{-1}\mathcal{S}\mathcal{S}_{g}\mathcal{S}_{h}\rVert\\
 & =O_{p}(N^{-2+8\epsilon})+O_{p}(N^{-2+4\epsilon})+O_{p}(N^{-2+10\epsilon})+O_{p}(N^{-2+10\epsilon})
\end{align*}

\subsection{Sixth derivative terms}

\begin{align*}
\lVert\sum_{f,g,h}\mathcal{S}'\partial_{bss's_{f}s_{g}s_{h}}\mathcal{L}_{(1)}^{*}\mathcal{S}\mathcal{S}_{f}\mathcal{S}_{g}\mathcal{S}_{h}\rVert & \leq C\Big(\lVert\sum_{f,g,h}\mathcal{S}'\mathcal{H}^{-1}\mathcal{H}_{s_{f}}\mathcal{H}^{-1}\mathcal{H}_{b}\mathcal{H}^{-1}\mathcal{H}_{s_{g}}\mathcal{H}^{-1}\mathcal{H}_{s_{h}}\mathcal{H}^{-1}\mathcal{S}\mathcal{S}_{f}\mathcal{S}_{g}\mathcal{S}_{h}\rVert\\
 & +\lVert\sum_{f,g,h}\mathcal{S}'\mathcal{H}^{-1}\mathcal{H}_{bs_{f}}\mathcal{H}^{-1}\mathcal{H}_{s_{g}}\mathcal{H}^{-1}\mathcal{H}_{s_{h}}\mathcal{H}^{-1}\mathcal{S}\mathcal{S}_{f}\mathcal{S}_{g}\mathcal{S}_{h}\rVert\\
 & +\lVert\sum_{f,g,h}\mathcal{S}'\mathcal{H}^{-1}\mathcal{H}_{b}\mathcal{H}^{-1}\mathcal{H}_{s_{f}s_{g}}\mathcal{H}^{-1}\mathcal{H}_{s_{h}}\mathcal{H}^{-1}\mathcal{S}\mathcal{S}_{f}\mathcal{S}_{g}\mathcal{S}_{h}\rVert\\
 & +\lVert\sum_{f,g,h}\mathcal{S}'\mathcal{H}^{-1}\mathcal{H}_{bs_{g}s_{f}}\mathcal{H}^{-1}\mathcal{H}_{s_{h}}\mathcal{H}^{-1}\mathcal{S}\mathcal{S}_{f}\mathcal{S}_{g}\mathcal{S}_{h}\rVert\\
 & +\lVert\sum_{f,g,h}\mathcal{S}'\mathcal{H}^{-1}\mathcal{H}_{bs_{g}}\mathcal{H}^{-1}\mathcal{H}_{s_{h}s_{f}}\mathcal{H}^{-1\mathcal{S}\mathcal{S}_{f}\mathcal{S}_{g}\mathcal{S}_{h}}\rVert\\
 & +\lVert\sum_{f,g,h}\mathcal{S}'\mathcal{H}^{-1}\mathcal{H}_{b}\mathcal{H}^{-1}\mathcal{H}_{s_{f}s_{g}s_{h}}\mathcal{H}^{-1}\mathcal{S}\mathcal{S}_{f}\mathcal{S}_{g}\mathcal{S}_{h}\rVert\\
 & +\lVert\sum_{f,g,h}\mathcal{S}'\mathcal{H}^{-1}\mathcal{H}_{s_{f}s_{g}s_{h}}\mathcal{H}^{-1}\mathcal{S}\mathcal{S}_{f}\mathcal{S}_{g}\mathcal{S}_{h}\rVert\Big)\\
 & =o_{p}(N^{-2})
\end{align*}

\section{Jackknife expansion } \label{SA:expansion_betak}
\subsection{\label{subsec:leaveout_expansion_beta}Asymptotic expansion in leave-out
samples}

In order to derive results for the jackknife estimator for $\beta$, we first establish
the asymptotic expansion for $\widehat{\beta}_{(k)}$,
i.e. the leave-out sample parameter estimate. The expansions can
be derived identically to the full sample estimators, replacing $\ell_{ij}$
in the objective function with $\frac{N-1}{N-2}\ell_{ij}1_{ij}^{k}$
where $1_{ij}^{k}$ is an indicator variable that is equal to one
whenever the observation $(i,j)$ is included in the $k$-th leave-out
sample, and is zero when that observation has been dropped. In comparing
the asymptotic expansions of the full-sample and leave-out sample
estimators, we will replace $\mathcal{H}^{-1}$ and $W_{N}^{-1}$
terms with their conditional expectations. For this purpose, we first
state a new version of Lemma \ref{lem:H_tilde} for the leave-out
sample.
\begin{lem}
\label{lem:HW_approx_leaveout}Let Assumptions 1 and 2 hold, and let
$\mathcal{H}_{(k)}=-\frac{1}{N-2}\sum_{i}\sum_{j\ne i}\partial_{\phi\phi'}\ell_{ij}1_{ij}^{k}$.
Then, for $s=0,1,2,3$, $t=2,3,4,5$ and $s+t\leq6$
\begin{align*}
\lVert\mathcal{H}_{(k)}-\bar{\mathcal{H}}\rVert & =O_{p}(N^{-\frac{1}{2}+2\epsilon})\\
\lVert\partial_{\beta^{s}\phi^{t}}\mathcal{L}_{(k)}-\partial_{\beta^{s}\phi^{t}}\bar{\mathcal{L}}\rVert & =O_{p}(N^{-\frac{1}{2}+2\epsilon})
\end{align*}
and 
\begin{align*}
\lVert\mathcal{H}_{(k)}^{-1}-\bar{\mathcal{H}}^{-1}(\tilde{\mathcal{H}}_{(k)}\bar{\mathcal{H}}^{-1})^{k}\rVert & =O_{p}(N^{-\frac{k+1}{2}+2(k+1)\epsilon})\\
\lVert W_{N,(k)}^{-1}-\bar{W}_{N}^{-1}(\tilde{W}_{N,(k)}\bar{W}_{N}^{-1})^{k}\rVert & =O_{p}(N^{-\frac{k+1}{2}+2(k+1)\epsilon})
\end{align*}
\end{lem}
\begin{proof}
Following the proof for the full-sample matrix, first decompose the
second derivative matrix as
\[
\lVert\mathcal{H}_{(k)}-\bar{\mathcal{H}}\rVert\leq\lVert\partial_{\alpha\alpha}\mathcal{L}_{(k)}-\partial_{\alpha\alpha}\bar{\mathcal{L}}\rVert+2\lVert\partial_{\alpha\gamma}\mathcal{L}_{(k)}-\partial_{\alpha\gamma}\bar{\mathcal{L}}\rVert+\lVert\partial_{\gamma\gamma}\mathcal{L}_{(k)}-\partial_{\gamma\gamma}\bar{\mathcal{L}}\rVert
\]
Define $\partial_{\pi^{2}}\tilde{\ell}_{(k),ij}=\frac{N-1}{N-2}\partial_{\pi^{2}}\ell_{ij}1_{ij}^{k}-\partial_{\pi^{2}}\bar{\ell}_{ij}$.
To bound the first term, first note that
\begin{align*}
\bar{E}\Big[\max_{i}\Big(\frac{1}{N-1}\sum_{j\ne i}\partial_{\pi^{2}}\tilde{\ell}_{(k),ij}\Big)^{q}\Big] & \leq\frac{1}{(N-1)^{q}}\sum_{i}\bar{E}\Big[\Big(\sum_{j\ne i}\partial_{\pi^{2}}\tilde{\ell}_{(k),ij}\Big)^{q}\Big]\\
 & \leq C\frac{1}{(N-1)^{q}}\sum_{i}\bar{E}\Big[\Big(\sum_{j\ne i}\partial_{\pi^{2}}\tilde{\ell}_{ij}1_{ij}^{k}\Big)^{q}\Big] \\
 &+C\frac{1}{(N-1)^{q}}\sum_{i}\bar{E}\Big[\Big(\sum_{j\ne i}\partial_{\pi^{2}}\bar{\ell}_{ij}(1-1_{ij}^{k})\Big)^{q}\Big]\\
 & \leq C\frac{1}{(N-1)^{q}}\sum_{i}\bar{E}\Big[\Big(\partial_{\pi^{2}}\bar{\ell}_{ii_{k}^{*}}\Big)^{q}\Big]+O_{p}(N^{1-q/2})\\
 & =O_{p}(N^{1-q/2})
\end{align*}
where $i_{k}^{*}$ is the receiver such that $1_{ii_{k}^{*}}^{k}=0$.
Then we can apply the same steps as in the proof of Lemma \ref{lem:H_tilde}
to give
\begin{align*}
\bar{E}\lVert\partial_{\alpha\alpha}\mathcal{L}_{(k)}-\partial_{\alpha\alpha}\bar{\mathcal{L}}\rVert^{q} & =O_{p}(N^{1-q/2})
\end{align*}
and hence $\lVert\partial_{\alpha\alpha}\mathcal{L}_{(k)}-\partial_{\alpha\alpha}\bar{\mathcal{L}}\rVert=O_{p}(N^{-\frac{1}{2}+\frac{1}{q}})$
and similarly for $\lVert\partial_{\gamma\gamma}\mathcal{L}_{(k)}-\partial_{\gamma\gamma}\bar{\mathcal{L}}\rVert$.
The bound $\lVert\partial_{\alpha\gamma}\mathcal{L}_{(k)}-\partial_{\alpha\gamma}\bar{\mathcal{L}}\rVert=O_{p}(N^{-\frac{1}{2}+\frac{1}{q}})$
follows similarly, which gives $\lVert\mathcal{H}_{(k)}-\bar{\mathcal{H}}\rVert=O_{p}(N^{-\frac{1}{2}+\frac{1}{q}})=O_{p}(N^{-\frac{1}{2}+2\epsilon})$.
The second result follows in the same way. We similarly show that
\begin{align*}
\lVert\partial_{\beta\phi'}\tilde{\mathcal{L}}_{(k)}\rVert & \leq\lVert\frac{1}{N-1}\sum_{i}\sum_{j\ne i}\partial_{\beta\phi'}\tilde{\ell}_{ij}1_{ij}^{k}\rVert+\lVert\frac{1}{N-1}\sum_{i}\partial_{\beta\phi'}\bar{\ell}_{(k),ii_{k}^{*}}\rVert\\
 & =O_{p}(\lVert\partial_{\beta\phi}\tilde{\mathcal{L}}\rVert)
\end{align*}
since $\big[\frac{1}{N-1}\sum_{i}\partial_{\beta\phi'}\bar{\ell}_{(k),ii_{k}^{*}}\big]_{s}=\frac{1}{N-1}\partial_{\beta\pi}\bar{\ell}_{(k),ss_{k}^{*}}$
for $s=1,\dots,N$ and $\frac{1}{N-1}\partial_{\beta\pi}\bar{\ell}_{(k),s_{k}^{\dagger}s}$
(where $s_{k}^{\dagger}$ is the sender for which $(s_{k}^{\dagger},s)$
is dropped in leave-out sample $k$) for $s=N+1,\dots,2N$. Then,
using these results and applying the same steps as in Lemma \ref{lem:H1_approx},
we can show that the same approximations hold in the leave-out samples.
\end{proof}
Using this result, the first-order expansion for $\widehat{\beta}_{(k)}$
can be shown to be given by 
\begin{align*}
NW_{N,(k)}(\widehat{\beta}_{(k)}-\beta) & =(\partial_{\beta}\mathcal{L}_{(k)})+(\partial_{\beta\phi'}\mathcal{L}_{(k)})\mathcal{H}_{(k)}^{-1}\mathcal{S}_{(k)}+\frac{1}{2}\mathcal{S}_{(k)}'\mathcal{H}_{(k)}^{-1}(\partial_{\beta\phi\phi'}\mathcal{L}_{(k)})\mathcal{H}_{(k)}^{-1}\mathcal{S}_{(k)}\\
 & \quad+\frac{1}{2}\mathcal{S}_{(k)}'\mathcal{H}_{(k)}^{-1}\sum_{f}(\partial_{\phi\phi'\phi_{f}}\mathcal{L}_{(k)})\big[\mathcal{H}_{(k)}^{-1}(\partial_{\beta\phi}\mathcal{L}_{(k)})\big]_{f}\mathcal{H}_{(k)}^{-1}\mathcal{S}_{(k)}+o_{p}(1)\\
 & =(\partial_{\beta}\mathcal{L}_{(k)})+(\partial_{\beta\phi'}\mathcal{\bar{L}})\mathcal{\bar{H}}^{-1}\mathcal{S}_{(k)}+(\partial_{\beta\phi'}\mathcal{\bar{L}})\mathcal{\bar{H}}^{-1}\tilde{\mathcal{H}}_{(k)}\mathcal{\bar{H}}^{-1}\mathcal{S}_{(k)}\\
 & \quad+\frac{1}{2}\mathcal{S}_{(k)}'\bar{\mathcal{H}}^{-1}(\partial_{\beta\phi\phi'}\bar{\mathcal{L}})\bar{\mathcal{H}}^{-1}\mathcal{S}_{(k)}\\
 & \quad+\frac{1}{2}\mathcal{S}_{(k)}'\bar{\mathcal{H}}^{-1}\sum_{f}(\partial_{\phi\phi'\phi_{f}}\bar{\mathcal{L}})\big[\bar{\mathcal{H}}^{-1}(\partial_{\beta\phi}\bar{\mathcal{L}})\big]_{f}\bar{\mathcal{H}}^{-1}\mathcal{S}_{(k)}+o_{p}(1)
\end{align*}
Since $\lVert W_{N,(k)}-\bar{W}_{N}\rVert=o_{p}(1)$ by Lemma \ref{lem:HW_approx_leaveout},
the same expansion up to first order applies to $N\bar{W}_{N}(\widehat{\beta}_{(k)}-\beta)$.

\subsection{Jackknife results for higher-order terms}

In the main appendix, the first-order terms of the jackknife estimator
$\widehat{\beta}_{J}$ are derived. Here we show that the remaining
terms up to $O_{p}(N^{-1})$ in the expansion for $N(\widehat{\beta}_{J}-\beta)$
are in fact $o_{p}(1)$ and so do not affect the asymptotic distribution
of the estimator. As can be seen in Section \ref{sec:full_expansion},
there is an extremely large number of terms in the asymptotic expansion
that must be considered. However, inspection of the terms shows that
they share a common structure, that is, they can be expressed as V-statistics
of a certain order that depend on sums of the derivatives of $\ell_{ij}$
up to sixth order. Here we prove the result for an example term, and
discuss how this same proof can be used to show that the remaining
terms will also be $o_{p}(1)$.

Consider the expansion term $\mathcal{S}'(\partial_{bbss'}\mathcal{L}_{(1)}^{*})\mathcal{S}\mathcal{S}_{\beta}$,
which contains the term
\[
\mathcal{S}'\mathcal{H}^{-1}\mathcal{H}_{bb}\mathcal{H}^{-1}\mathcal{S}\mathcal{S}_{\beta}
\]
Expanding out this terms gives
\begin{align*}
\mathcal{S}'\mathcal{H}^{-1}\mathcal{H}_{bb}\mathcal{H}^{-1}\mathcal{S}\mathcal{S}_{\beta} & =-W^{-2}W_{b}\mathcal{S}'\mathcal{H}^{-1}\sum_{f}(\partial_{\phi\phi'\phi_{f}}\mathcal{L})\big[\mathcal{H}^{-1}(\partial_{\beta\phi}\mathcal{L})\big]_{f}\mathcal{H}^{-1}\mathcal{S}\mathcal{S}_{\beta}\\
 & -W^{-2}W_{b}\mathcal{S}'\mathcal{H}^{-1}\mathcal{E}\mathcal{H}^{-1}\mathcal{S}\mathcal{S}_{\beta}\\
 & +W^{-2}\mathcal{S}'\mathcal{H}^{-1}\mathcal{E}^{2}\mathcal{H}^{-1}\mathcal{S}\mathcal{S}_{\beta}\\
 & +W^{-2}\mathcal{S}'\mathcal{H}^{-1}\sum_{f}(\partial_{\beta\phi\phi'\phi_{f}}\mathcal{L})\big[\mathcal{H}^{-1}(\partial_{\beta\phi}\mathcal{L})\big]_{f}\mathcal{H}^{-1}\mathcal{S}\mathcal{S}_{\beta}\\
 & -W^{-2}\mathcal{S}'\mathcal{H}^{-1}\sum_{f}(\partial_{\beta\phi\phi'\phi_{f}}\mathcal{L})\big[\mathcal{H}^{-1}(\partial_{\beta\phi}\mathcal{L})\big]_{f}\mathcal{H}^{-1}\mathcal{S}\mathcal{S}_{\beta}\\
 & -W^{-2}\mathcal{S}'\mathcal{H}^{-1}\sum_{e,f}(\partial_{\phi\phi'\phi_{e}\phi_{f}}\mathcal{L})\big[\mathcal{H}^{-1}(\partial_{\beta\phi}\mathcal{L})\big]_{e}\big[\mathcal{H}^{-1}(\partial_{\beta\phi}\mathcal{L})\big]_{f}\mathcal{H}^{-1}\mathcal{S}\mathcal{S}_{\beta}\\
 & -W^{-2}\mathcal{S}'\mathcal{H}^{-1}\sum_{f}(\partial_{\phi\phi'\phi_{f}}\mathcal{L})\big[\mathcal{H}^{-1}\mathcal{E}\mathcal{H}^{-1}(\partial_{\beta\phi}\mathcal{L})\big]_{f}\mathcal{H}^{-1}\mathcal{S}\mathcal{S}_{\beta}\\
 & -W^{-2}\mathcal{S}'\mathcal{H}^{-1}\sum_{e,f}(\partial_{\phi\phi'\phi_{f}}\mathcal{L})\big[\mathcal{H}^{-1}(\partial_{\phi\phi'\phi_{e}}\mathcal{L})\mathcal{H}^{-1}(\partial_{\beta\phi}\mathcal{L})\big]_{f}\big[\mathcal{H}^{-1}(\partial_{\beta\phi}\mathcal{L})\big]_{e}\mathcal{H}^{-1}\mathcal{S}\mathcal{S}_{\beta}\\
 & -W^{-2}\mathcal{S}'\mathcal{H}^{-1}\sum_{f}(\partial_{\phi\phi'\phi_{f}}\mathcal{L})\big[\mathcal{H}^{-1}(\partial_{\beta\beta\phi}\mathcal{L})\big]_{f}\mathcal{H}^{-1}\mathcal{S}\mathcal{S}_{\beta}\\
 & -W^{-2}\mathcal{S}'\mathcal{H}^{-1}\sum_{f}(\partial_{\phi\phi'\phi_{f}}\mathcal{L})\big[\mathcal{H}^{-1}\mathcal{E}\mathcal{H}^{-1}(\partial_{\beta\phi}\mathcal{L})\big]_{f}\mathcal{H}^{-1}\mathcal{S}\mathcal{S}_{\beta}
\end{align*}
Take the first term in this expression. We can replace $W$ with $\bar{W}+(W-\bar{W})$,
and similarly for $W_{b}$, $\mathcal{H}^{-1}$ and $\partial_{\phi\phi'\phi_{f}}\mathcal{L}$
to give
\[
\bar{W}^{-2}\bar{W}_{b}\mathcal{S}'\bar{\mathcal{H}}^{-1}\sum_{f}(\partial_{\phi\phi'\phi_{f}}\bar{\mathcal{L}})\big[\bar{\mathcal{H}}^{-1}(\partial_{\beta\phi}\mathcal{L})\big]_{f}\mathcal{\bar{H}}^{-1}\mathcal{S}\mathcal{S}_{\beta}+o_{p}(N^{-2})
\]
This result holds identically for the leave-out samples, replacing
$\mathcal{S}$ with $\mathcal{S}_{(k)}$ and similarly for $\partial_{\beta\phi}\mathcal{L}$
and $\mathcal{S}_{\beta}$. Ignoring the $\bar{W}^{-2}\bar{W}_{b}$
term for the moment, which will not be affected by the jackknifing,
we can decompose the above sum further into the components of $\phi=(\alpha,\gamma)$,
for example the first of these terms would be
\begin{align*}
\mathcal{S}_{\alpha}' & \bar{\mathcal{H}}_{\alpha\alpha}^{-1}\sum_{f=1}^{N}(\partial_{\alpha\alpha'\alpha_{f}}\bar{\mathcal{L}})\big[\bar{\mathcal{H}}_{\alpha\alpha}^{-1}(\partial_{\beta\alpha}\mathcal{L})+\bar{\mathcal{H}}_{\alpha\gamma}^{-1}(\partial_{\beta\gamma}\mathcal{L})\big]_{f}\mathcal{\bar{H}}_{\alpha\alpha}^{-1}\mathcal{S}_{\alpha}\mathcal{S}_{\beta}\\
= & \mathcal{S}_{\beta}\sum_{i}\sum_{j}\sum_{s}\sum_{t}\mathcal{S}_{\alpha_{i}}[\bar{\mathcal{H}}_{\alpha\alpha}^{-1}]_{ij}[\partial_{\alpha\alpha'\alpha_{j}}\bar{\mathcal{L}}]_{jj}\big([\bar{\mathcal{H}}_{\alpha\alpha}^{-1}]_{jt}(\partial_{\beta\alpha_{t}}\mathcal{L})+[\bar{\mathcal{H}}_{\alpha\gamma}^{-1}]_{jt}(\partial_{\beta\gamma_{t}}\mathcal{L})\big)[\mathcal{\bar{H}}_{\alpha\alpha}^{-1}]_{js}\mathcal{S}_{\alpha_{s}}\\
= & \mathcal{S}_{\beta}\frac{1}{(N-1)^{3}}\sum_{j}\sum_{i,i'\ne i}\sum_{s,s'\ne s}\sum_{t,t'\ne t}[\bar{\mathcal{H}}_{\alpha\alpha}^{-1}]_{ij}[\mathcal{\bar{H}}_{\alpha\alpha}^{-1}]_{js}[\partial_{\alpha\alpha'\alpha_{j}}\bar{\mathcal{L}}]_{jj} \\
&\qquad\times (\partial_{\pi}\ell_{i,i'})(\partial_{\pi}\ell_{s,s'})\big([\bar{\mathcal{H}}_{\alpha\alpha}^{-1}]_{jt}(\partial_{\beta\pi}\ell_{t,t'})+[\bar{\mathcal{H}}_{\alpha\gamma}^{-1}]_{jt}(\partial_{\beta\pi}\ell_{t',t})\big)\\
= & \frac{1}{(N-1)^{4}}\sum_{j}\sum_{i,i'\ne i}\sum_{s,s'\ne s}\sum_{t,t'\ne t}\sum_{r,r'\ne r}[\bar{\mathcal{H}}_{\alpha\alpha}^{-1}]_{ij}[\mathcal{\bar{H}}_{\alpha\alpha}^{-1}]_{js}[\partial_{\alpha\alpha'\alpha_{j}}\bar{\mathcal{L}}]_{jj} \\
&\qquad\times (\partial_{\pi}\ell_{i,i'})(\partial_{\pi}\ell_{s,s'})\big([\bar{\mathcal{H}}_{\alpha\alpha}^{-1}]_{jt}(\partial_{\beta\pi}\ell_{t,t'})+[\bar{\mathcal{H}}_{\alpha\gamma}^{-1}]_{jt}(\partial_{\beta\pi}\ell_{t',t})\big)(\partial_{\beta}\ell_{r,r'})
\end{align*}
Define
\begin{align*}
Q_{ist}^{\alpha} & =\sum_{j}[\bar{\mathcal{H}}_{\alpha\alpha}^{-1}]_{ij}[\mathcal{\bar{H}}_{\alpha\alpha}^{-1}]_{js}[\partial_{\alpha\alpha'\alpha_{j}}\bar{\mathcal{L}}]_{jj}[\bar{\mathcal{H}}_{\alpha\alpha}^{-1}]_{jt}\\
Q_{ist}^{\gamma} & =\sum_{j}[\bar{\mathcal{H}}_{\alpha\alpha}^{-1}]_{ij}[\mathcal{\bar{H}}_{\alpha\alpha}^{-1}]_{js}[\partial_{\alpha\alpha'\alpha_{j}}\bar{\mathcal{L}}]_{jj}[\bar{\mathcal{H}}_{\alpha\gamma}^{-1}]_{jt}
\end{align*}
and note that, by Lemma \ref{lem:H_approx}, we have $\max_{i\ne s\ne t}Q_{ist}=O_{p}(N^{-2})$,
while the max such that two of $(i,s,t)$ are equal is $O_{p}(N^{-1})$
and $Q_{iii}=O_{p}(1)$. We can write the term now as
\begin{align*}
 & \frac{1}{(N-1)^{4}}\sum_{i,i'\ne i}\sum_{s,s'\ne s}\sum_{t,t'\ne t}\sum_{r,r'\ne r}Q_{ist}^{\alpha}(\partial_{\pi}\ell_{i,i'})(\partial_{\pi}\ell_{s,s'})(\partial_{\beta\pi}\ell_{t,t'})(\partial_{\beta}\ell_{r,r'})\\
 & +\frac{1}{(N-1)^{4}}\sum_{i,i'\ne i}\sum_{s,s'\ne s}\sum_{t,t'\ne t}\sum_{r,r'\ne r}Q_{ist}^{\gamma}(\partial_{\pi}\ell_{i,i'})(\partial_{\pi}\ell_{s,s'})(\partial_{\beta\pi}\ell_{t',t})(\partial_{\beta}\ell_{r,r'})
\end{align*}
which is the sum of two V-statistic type terms, of fourth order. We
now consider the effect of the jackknife operation on these terms.
Consider the first of the terms (the result will clearly be identical
for both terms) and note that $Q_{ist}^{\alpha}$ is fixed across
leave-out samples, so that the $k$-th leave-out version of this term
is
\[
\frac{1}{(N-2)^{4}}\sum_{i,i'\ne i}\sum_{s,s'\ne s}\sum_{t,t'\ne t}\sum_{r,r'\ne r}Q_{ist}^{\alpha}(\partial_{\pi}\ell_{i,i'})(\partial_{\pi}\ell_{s,s'})(\partial_{\beta\pi}\ell_{t,t'})(\partial_{\beta}\ell_{r,r'})1_{i,i'}^{k}1_{s,s'}^{k}1_{t,t'}^{k}1_{r,r'}^{k}
\]

The compute the average over the $N-1$ leave-out samples first define
\[
n(i,i',s,s',t,t',r,r')=N-1-\sum_{k}1_{i,i'}^{k}1_{s,s'}^{k}1_{t,t'}^{k}1_{r,r'}^{k}
\]
which counts the number of leave-out samples in which all four observations
appear -- this number depends on whether the four observations appear
in the same set $\mathcal{I}_{k}$, or in two, three of four different
sets. We can write the average as
\begin{align*}
\frac{1}{N-1} & \sum_{k}\frac{1}{(N-2)^{4}}\sum_{i,i'\ne i}\sum_{s,s'\ne s}\sum_{t,t'\ne t}\sum_{r,r'\ne r}Q_{ist}^{\alpha}(\partial_{\pi}\ell_{i,i'})(\partial_{\pi}\ell_{s,s'})(\partial_{\beta\pi}\ell_{t,t'})(\partial_{\beta}\ell_{r,r'})1_{i,i'}^{k}1_{s,s'}^{k}1_{t,t'}^{k}1_{r,r'}^{k}\\
 & =\frac{1}{(N-1)(N-2)^{4}}\sum_{i,i'\ne i}\sum_{s,s'\ne s}\sum_{t,t'\ne t}\sum_{r,r'\ne r}Q_{ist}^{\alpha}(\partial_{\pi}\ell_{i,i'})(\partial_{\pi}\ell_{s,s'})(\partial_{\beta\pi}\ell_{t,t'})(\partial_{\beta}\ell_{r,r'}) \\
 &\qquad\times \big(N-1-n(i,i',s,s',t,t',r,r')\big)
\end{align*}
The jackknifed term is then equal to
\begin{align*}
\vert\mathcal{J}\vert= & \lvert\frac{N-1}{(N-1)^{4}}\sum_{i,i'\ne i}\sum_{s,s'\ne s}\sum_{t,t'\ne t}\sum_{r,r'\ne r}Q_{ist}^{\alpha}(\partial_{\pi}\ell_{i,i'})(\partial_{\pi}\ell_{s,s'})(\partial_{\beta\pi}\ell_{t,t'})(\partial_{\beta}\ell_{r,r'})\\
- & \frac{N-2}{(N-1)(N-2)^{4}}\sum_{i,i'\ne i}\sum_{s,s'\ne s}\sum_{t,t'\ne t}\sum_{r,r'\ne r}Q_{ist}^{\alpha}(\partial_{\pi}\ell_{i,i'})(\partial_{\pi}\ell_{s,s'})(\partial_{\beta\pi}\ell_{t,t'})(\partial_{\beta}\ell_{r,r'}) \\
&\qquad\times \big(N-1-n(i,i',s,s',t,t',r,r')\big)\rvert\\
\leq & \frac{C}{N^{3}}\sum_{j=1}^{4}\lvert\sum_{i,i'\ne i}\sum_{s,s'\ne s}\sum_{t,t'\ne t}\sum_{r,r'\ne r}Q_{ist}^{\alpha}(\partial_{\pi}\ell_{i,i'})(\partial_{\pi}\ell_{s,s'})(\partial_{\beta\pi}\ell_{t,t'})(\partial_{\beta}\ell_{r,r'})1\{n(i,i',s,s',t,t',r,r')=j\}\rvert
\end{align*}
since $(\frac{1}{(N-1)^{3}}-\frac{N-1-n(i,i',s,s',t,t',r,r')}{(N-1)(N-2)^{3}})=O(N^{-3})$.

To bound the summation, we decompose the sum based on the order of
$Q_{ist}^{\alpha}$, using the shorthand notation $1_{n=j}=1\{n(i,i',s,s',t,t',r,r')=j\}$.
For some $j$ we have
\begin{align*}
\vert & \sum_{i,i'\ne i}\sum_{s,s'\ne s}\sum_{t,t'\ne t}\sum_{r,r'\ne r}Q_{ist}^{\alpha}(\partial_{\pi}\ell_{i,i'})(\partial_{\pi}\ell_{s,s'})(\partial_{\beta\pi}\ell_{t,t'})(\partial_{\beta}\ell_{r,r'})1_{n=j}\vert\\
 & \leq\vert\sum_{i,i'\ne i}\sum_{s\ne i,s'\ne i}\sum_{t\ne(i,s),t'\ne t}\sum_{r,r'\ne r}Q_{ist}^{\alpha}(\partial_{\pi}\ell_{i,i'})(\partial_{\pi}\ell_{s,s'})(\partial_{\beta\pi}\ell_{t,t'})(\partial_{\beta}\ell_{r,r'})1_{n=j}\vert\\
 & +\vert\sum_{i,i'\ne i}\sum_{s'\ne i}\sum_{t\ne i,t'\ne t}\sum_{r,r'\ne r}Q_{iit}^{\alpha}(\partial_{\pi}\ell_{i,i'})(\partial_{\pi}\ell_{i,s'})(\partial_{\beta\pi}\ell_{t,t'})(\partial_{\beta}\ell_{r,r'})1_{n=j}\vert\\
 & +\vert\sum_{i,i'\ne i}\sum_{s\ne i,s'\ne s}\sum_{t'\ne i}\sum_{r,r'\ne r}Q_{isi}^{\alpha}(\partial_{\pi}\ell_{i,i'})(\partial_{\pi}\ell_{s,s'})(\partial_{\beta\pi}\ell_{i,t'})(\partial_{\beta}\ell_{r,r'})1_{n=j}\vert\\
 & +\vert\sum_{i,i'\ne i}\sum_{s\ne i,s'\ne s}\sum_{t'\ne s}\sum_{r,r'\ne r}Q_{iss}^{\alpha}(\partial_{\pi}\ell_{i,i'})(\partial_{\pi}\ell_{s,s'})(\partial_{\beta\pi}\ell_{s,t'})(\partial_{\beta}\ell_{r,r'})1_{n=j}\vert\\
 & +\vert\sum_{i,i'\ne i}\sum_{s'\ne i}\sum_{t'\ne i}\sum_{r,r'\ne r}Q_{iii}^{\alpha}(\partial_{\pi}\ell_{i,i'})(\partial_{\pi}\ell_{i,s'})(\partial_{\beta\pi}\ell_{i,t'})(\partial_{\beta}\ell_{r,r'})1_{n=j}\vert
\end{align*}
The first term is
\begin{align*}
\bar{E}\Big[ & \Big(\sum_{i,i'\ne i}\sum_{s\ne i,s'\ne i}\sum_{t\ne(i,s),t'\ne t}\sum_{r,r'\ne r}Q_{ist}^{\alpha}(\partial_{\pi}\ell_{i,i'})(\partial_{\pi}\ell_{s,s'})(\partial_{\beta\pi}\ell_{t,t'})(\partial_{\beta}\ell_{r,r'})1_{n=j}\Big)^{2}\Big]\\
= & \sum_{i,i'\ne i}\sum_{s\ne i,s'\ne i}\sum_{t\ne(i,s),t'\ne t}\sum_{r,r'\ne r}\sum_{j,j'\ne j}\sum_{a\ne j,a'\ne j}\sum_{b\ne(j,a),b'\ne b}\sum_{c,c'\ne c}Q_{ist}^{\alpha}Q_{jab}^{\alpha}\\
 & \times\bar{E}\big[(\partial_{\pi}\ell_{i,i'})(\partial_{\pi}\ell_{s,s'})(\partial_{\beta\pi}\ell_{t,t'})(\partial_{\beta}\ell_{r,r'})(\partial_{\pi}\ell_{j,j'})(\partial_{\pi}\ell_{a,a'})(\partial_{\beta\pi}\ell_{b,b'})(\partial_{\beta}\ell_{c,c'})\big]1_{n=j}1_{n'=j}
\end{align*}
Since we have $\max_{i\ne s\ne t}Q_{ist}=O_{p}(N^{-2})$, and $\bar{E}[\partial_{\pi}\ell_{ii'}]=\bar{E}[\partial_{\beta}\ell_{ii'}]=0$,
this term is $O_{p}(N^{6})$. Similarly, for the final term we have
\begin{align*}
\bar{E}\Big[\Big( & \sum_{i,i'\ne i}\sum_{s'\ne i}\sum_{t'\ne i}\sum_{r,r'\ne r}Q_{iii}^{\alpha}(\partial_{\pi}\ell_{i,i'})(\partial_{\pi}\ell_{i,s'})(\partial_{\beta\pi}\ell_{i,t'})(\partial_{\beta}\ell_{r,r'})1_{n=j}\Big)^{2}\Big]\\
= & \sum_{i,i'\ne i}\sum_{s'\ne i}\sum_{t'\ne i}\sum_{r,r'\ne r}\sum_{j,j'\ne j}\sum_{a'\ne j}\sum_{b'\ne j}\sum_{c,c'\ne c}Q_{iii}^{\alpha}Q_{jjj}^{\alpha}\\
 & \times\bar{E}\big[(\partial_{\pi}\ell_{i,i'})(\partial_{\pi}\ell_{i,s'})(\partial_{\beta\pi}\ell_{i,t'})(\partial_{\beta}\ell_{r,r'})(\partial_{\pi}\ell_{j,j'})(\partial_{\pi}\ell_{j,a'})(\partial_{\beta\pi}\ell_{j,b'})(\partial_{\beta}\ell_{c,c'})\big]\\
= & O_{p}(N^{8})
\end{align*}
and hence the final term is $O_{p}(N^{4})$. The remaining terms can
be shown to also be $O_{p}(N^{4})$. This gives $\mathcal{J}=O_{p}(N)$.
Nearly identical steps apply to the remaining components of the term
$\mathcal{S}'\bar{\mathcal{H}}^{-1}\sum_{f}(\partial_{\phi\phi'\phi_{f}}\bar{\mathcal{L}})\big[\bar{\mathcal{H}}^{-1}(\partial_{\beta\phi}\mathcal{L})\big]_{f}\mathcal{\bar{H}}^{-1}\mathcal{S}\mathcal{S}_{\beta}$
and adding back the terms $\bar{W}^{-2}\bar{W}_{b}$, we find that
the jackknife operator applied to the full term satisfies
\[
\mathcal{J}\Big[\bar{W}^{-2}\bar{W}_{b}\mathcal{S}'\bar{\mathcal{H}}^{-1}\sum_{f}(\partial_{\phi\phi'\phi_{f}}\bar{\mathcal{L}})\big[\bar{\mathcal{H}}^{-1}(\partial_{\beta\phi}\mathcal{L})\big]_{f}\mathcal{\bar{H}}^{-1}\mathcal{S}\mathcal{S}_{\beta}\Big]=o_{p}(1).
\]
Each of the remaining terms in $\mathcal{S}'\mathcal{H}^{-1}\mathcal{H}_{bb}\mathcal{H}^{-1}\mathcal{S}\mathcal{S}_{\beta}$
can similarly be shown to have the same V-statistic like structure,
up to fifth-order. Inspection of the expressions in Section \ref{sec:full_expansion}
shows that an essentially identical proof can be used for any of these
terms. That is, we first replace matrices with expansions around their
conditional expectations, and the expand the terms to give sums of
products of terms like the above. The jackknifed versions of these
sums can then straightforwardly be shown to have the same asymptotic
order as the original terms, and hence each of the jackknifed terms
is $o_{p}(1)$ as required.

\subsection{Results for weighted jackknife} \label{subsec:weighted_proof}

In order to complete the proof of Theorem 1 in Section \ref{subsec:thm1proof},
we must show:

(1) $(N-1)\lVert W_{J}-W_{N}\rVert=o_{p}(1)$

(2) $(N-1)(\widehat{W}_{J}-W_{J})(\widehat{\beta}-\beta)-(N-2)\frac{1}{N-1}\sum_{k}(\widehat{W}_{(k)}-W_{(k)})(\widehat{\beta}_{(k)}-\beta)=o_{p}(N^{-1})$

(3) $\lVert \widehat{W}_J - \overline{W}_N\rVert=o_p(1)$

\subsubsection*{(1) $(N-1)(W_{J}-W_{N})=o_{p}(1)$}

Expanding out this term we get
\begin{align*}
&(N-1)(W_{J}-W_{N}) \\
&=(N-1)\frac{1}{N}\Big\{\frac{1}{N-1}\sum_{k}\partial_{\beta\beta}\mathcal{L}_{(k)}\\
&+\frac{1}{N-1}\sum_{k}(\partial_{\beta\phi'}\mathcal{L}_{(k)})\mathcal{H}_{(k)}^{-1}(\partial_{\beta\phi}\mathcal{L}_{(k)})-\partial_{\beta\beta}\mathcal{L}-(\partial_{\beta\phi'}\mathcal{L})\mathcal{H}^{-1}(\partial_{\beta\phi}\mathcal{L})\Big\}\\
 & =\frac{N-1}{N}\Big\{\frac{1}{N-1}\sum_{k}(\partial_{\beta\phi'}\mathcal{L}_{(k)})\mathcal{H}_{(k)}^{-1}(\partial_{\beta\phi}\mathcal{L}_{(k)})-(\partial_{\beta\phi'}\mathcal{L})\mathcal{H}^{-1}(\partial_{\beta\phi}\mathcal{L})\Big\}\\
 & =\frac{N-1}{N}\Big\{\frac{1}{N-1}\sum_{k}(\partial_{\beta\phi'}\mathcal{L}_{(k)})\Big(\mathcal{\bar{H}}^{-1}-\mathcal{\bar{H}}^{-1}\tilde{\mathcal{H}}_{(k)}\mathcal{\bar{H}}^{-1}+\mathcal{\bar{H}}^{-1}\tilde{\mathcal{H}}_{(k)}\mathcal{\bar{H}}^{-1}\tilde{\mathcal{H}}_{(k)}\mathcal{\bar{H}}^{-1}\Big)(\partial_{\beta\phi}\mathcal{L}_{(k)})\\
 &-(\partial_{\beta\phi'}\mathcal{L})\Big(\mathcal{\bar{H}}^{-1}-\mathcal{\bar{H}}^{-1}\tilde{\mathcal{H}}\mathcal{\bar{H}}^{-1}+
 \mathcal{\bar{H}}^{-1}\tilde{\mathcal{H}}\mathcal{\bar{H}}^{-1}\tilde{\mathcal{H}}\mathcal{\bar{H}}^{-1}\Big)(\partial_{\beta\phi}\mathcal{L})\Big\} +o_p(1)
\end{align*}
where the remainder term is $o_p(1)$ since $\lVert\mathcal{H}^{-1}-\mathcal{\bar{H}}^{-1}+\mathcal{\bar{H}}^{-1}\tilde{\mathcal{H}}\mathcal{\bar{H}}^{-1}-\mathcal{\bar{H}}^{-1}\tilde{\mathcal{H}}\mathcal{\bar{H}}^{-1}\tilde{\mathcal{H}}\mathcal{\bar{H}}^{-1}\rVert=O_{p}(N^{-\frac{3}{2}+6\epsilon})$
by Lemma \ref{lem:H1_approx}, $\lVert\partial_{\beta\phi'}\mathcal{L}\rVert=O_{p}(N^{1/2})$,
and similarly for the leaveout versions.

Further expanding out the first of these approximation terms gives
\begin{align*}
\frac{1}{N-1} & \sum_{k}(\partial_{\beta\phi'}\mathcal{L}_{(k)})\mathcal{\bar{H}}^{-1}(\partial_{\beta\phi}\mathcal{L}_{(k)})-(\partial_{\beta\phi'}\mathcal{L})\mathcal{\bar{H}}^{-1}(\partial_{\beta\phi}\mathcal{L})\\
= & \frac{1}{N-1}\sum_{k}\frac{1}{(N-2)^{2}}\sum_{i}\sum_{j\ne i}\sum_{s}\sum_{t\ne s}\Gamma_{isjt}(\partial_{\beta\pi}\ell_{ij})(\partial_{\beta\pi}\ell_{st})1_{ij}^{k}1_{st}^{k}\\
 & -\frac{1}{(N-1)^{2}}\sum_{i}\sum_{j\ne i}\sum_{s}\sum_{t\ne s}\Gamma_{isjt}(\partial_{\beta\pi}\ell_{ij})(\partial_{\beta\pi}\ell_{st})\\
= & \frac{1}{(N-1)^{2}}\sum_{i}\sum_{j\ne i}\sum_{s}\sum_{t\ne s}\Gamma_{isjt}(\partial_{\beta\pi}\ell_{ij})(\partial_{\beta\pi}\ell_{st})\Big(\frac{(N-1)}{(N-2)^{2}}\sum_{k}1_{ij}^{k}1_{st}^{k}-1\Big)\\
= & \frac{1}{(N-1)^{2}(N-2)}\sum_{i}\sum_{j\ne i}\sum_{s}\sum_{t\ne s}\Gamma_{isjt}(\partial_{\beta\pi}\ell_{ij})(\partial_{\beta\pi}\ell_{st})I_{(ij)(st)}^{1}\\
 & -\frac{1}{(N-1)^{2}(N-2)^{2}}\sum_{i}\sum_{j\ne i}\sum_{s}\sum_{t\ne s}\Gamma_{isjt}(\partial_{\beta\pi}\ell_{ij})(\partial_{\beta\pi}\ell_{st})(1-I_{(ij)(st)}^{1}) \\
 &=o_p(1)
\end{align*}
since $\Gamma_{isjt}$ is $O_p(1)$ only when either $i=s$ or $j=t$, and is $O_p(N^{-1})$  otherwise.
Expanding out the next term in the approximation gives
\begin{align*}
\frac{1}{N-1} & \sum_{k}(\partial_{\beta\phi'}\mathcal{L}_{(k)})\mathcal{\bar{H}}^{-1}\tilde{\mathcal{H}}_{(k)}\mathcal{\bar{H}}^{-1}(\partial_{\beta\phi}\mathcal{L}_{(k)})-(\partial_{\beta\phi'}\mathcal{L})\mathcal{\bar{H}}^{-1}\tilde{\mathcal{H}}\mathcal{\bar{H}}^{-1}(\partial_{\beta\phi}\mathcal{L})\\
= & \frac{1}{N-1}\sum_{k}\frac{1}{(N-2)^{3}}\sum_{i,j,s}\sum_{i'\ne i}\sum_{j'\ne j}\sum_{s'\ne s}\Gamma_{ii'jj'}\Gamma_{jj'ss'}(\partial_{\beta\pi}\ell_{ii'})(\partial_{\beta\pi}\ell_{ss'})(\partial_{\pi^{2}}\ell_{jj'})1_{ii'}^{k}1_{jj'}^{k}1_{ss'}^{k}\\
 & -\frac{1}{(N-1)^{3}}\sum_{i,j,s}\sum_{i'\ne i}\sum_{j'\ne j}\sum_{s'\ne s}\Gamma_{ii'jj'}\Gamma_{jj'ss'}(\partial_{\beta\pi}\ell_{ii'})(\partial_{\beta\pi}\ell_{ss'})(\partial_{\pi^{2}}\ell_{jj'})\\
= & \frac{1}{(N-1)^{3}}\sum_{i,j,s}\sum_{i'\ne i}\sum_{j'\ne j}\sum_{s'\ne s}\Gamma_{ii'jj'}\Gamma_{jj'ss'}(\partial_{\beta\pi}\ell_{ii'})(\partial_{\beta\pi}\ell_{ss'})(\partial_{\pi^{2}}\ell_{jj'})\Big(\frac{(N-1)^{2}}{(N-2)^{3}}\sum_{k}1_{ii'}^{k}1_{jj'}^{k}1_{ss'}^{k}-1\Big)\\
= & O(N^{-1})\frac{1}{(N-1)^{3}}\sum_{i,j,s}\sum_{i'\ne i}\sum_{j'\ne j}\sum_{s'\ne s}\Gamma_{ii'jj'}\Gamma_{jj'ss'}(\partial_{\beta\pi}\ell_{ii'})(\partial_{\beta\pi}\ell_{ss'})(\partial_{\pi^{2}}\ell_{jj'})I_{(ii')(ss')(jj')}^{1}\\
 & +O(N^{-1})\frac{1}{(N-1)^{3}}\sum_{i,j,s}\sum_{i'\ne i}\sum_{j'\ne j}\sum_{s'\ne s}\Gamma_{ii'jj'}\Gamma_{jj'ss'}(\partial_{\beta\pi}\ell_{ii'})(\partial_{\beta\pi}\ell_{ss'})(\partial_{\pi^{2}}\ell_{jj'})I_{(ii')(ss')(jj')}^{2}\\
 & +O(N^{-2})\frac{1}{(N-1)^{3}}\sum_{i,j,s}\sum_{i'\ne i}\sum_{j'\ne j}\sum_{s'\ne s}\Gamma_{ii'jj'}\Gamma_{jj'ss'}(\partial_{\beta\pi}\ell_{ii'})(\partial_{\beta\pi}\ell_{ss'})(\partial_{\pi^{2}}\ell_{jj'})I_{(ii')(ss')(jj')}^{3}\\
= & o_{p}(1)
\end{align*}
where the third equality follows from the fact that $\big(\frac{(N-1)^{2}}{(N-2)^{3}}\sum_{k}1_{ii'}^{k}1_{jj'}^{k}1_{ss'}^{k}-1\big)$ is $O(N^{-1})$ whenever $\sum_{k}1_{ii'}^{k}1_{jj'}^{k}1_{ss'}^{k}$ is equal to $(N-2)$ or $(N-3)$, but is $O(N^{-2})$ when the three observations span three different $\mathcal{I}_k$, while the
final equality is due to the fact that: (1)
$\Gamma_{ii'jj'}\Gamma_{jj'ss'}$ is $O_{p}(1)$ only if one of $i=j,i'=j'$ holds and one of $j=s,j'=s'$ holds, and is $O_{p}(N^{-1})$ if only
one of the two holds, and is $O_{p}(N^{-2})$ if neither hold, and (2) when $I^1_{(ii')(ss')(jj')}=1$, the indices $s'$ and $j'$ are fixed given $i'$, while for $I^2_{(ii')(ss')(jj')}=1$, choosing two of these fixes the third. Similar steps show that
\[
\frac{1}{N-1}\sum_{k}(\partial_{\beta\phi'}\mathcal{L}_{(k)})\mathcal{\bar{H}}^{-1}\tilde{\mathcal{H}}_{(k)}\mathcal{\bar{H}}^{-1}\tilde{\mathcal{H}}_{(k)}\mathcal{\bar{H}}^{-1}(\partial_{\beta\phi}\mathcal{L}_{(k)})-(\partial_{\beta\phi'}\mathcal{L})\mathcal{\bar{H}}^{-1}\tilde{\mathcal{H}}\mathcal{\bar{H}}^{-1}\tilde{\mathcal{H}}\mathcal{\bar{H}}^{-1}(\partial_{\beta\phi}\mathcal{L})=o_{p}(1)
\]
This gives the required result.

\subsubsection*{(2) $(N-1)(\widehat{W}_{J}-W_{J})(\widehat{\beta}-\beta)-(N-2)\frac{1}{N-1}\sum_{k}(\widehat{W}_{(k)}-W_{(k)})(\widehat{\beta}_{(k)}-\beta)=o_{p}(N^{-1})$}

We first provide a first-order expansion for $\widehat{W}$. Recall
that we can write $\widehat{W}=\big(\partial_{bb}\mathcal{L}^{*}(0,0)\big)^{-1}$,
and so letting $\mathcal{W}^{*}(b,s)=\big(\partial_{bb}\mathcal{L}^{*}(b,s)\big)^{-1}$
and $\mathcal{W}^{*}=\big(\partial_{bb}\mathcal{L}^{*}(\mathcal{S}_{\beta},\mathcal{S})\big)^{-1}$
we can expand to give
\begin{align*}
\widehat{W} & =\mathcal{W}^{*}+(\partial_{b}\mathcal{W}^{*})\mathcal{S_{\beta}}+(\partial_{s'}\mathcal{W}^{*})\mathcal{S}\\
 & +\frac{1}{2}\mathcal{S}'(\partial_{ss'}\mathcal{W}^{*})\mathcal{S}+o_{p}(N^{-1})\\
\widehat{W}_{N}-W_{N} & =\frac{1}{N}(\partial_{b}\mathcal{W}^{*})\mathcal{S_{\beta}}+\frac{1}{N}(\partial_{s'}\mathcal{W}^{*})\mathcal{S}+\frac{1}{2}\frac{1}{N}\mathcal{S}'(\partial_{ss'}\mathcal{W}^{*})\mathcal{S}+o_{p}(N^{-1})
\end{align*}
Up to first order we find
\begin{align*}
N(\widehat{W}_{N}-W_{N}) & =(\partial_{b}\mathcal{W}^{*})\mathcal{S_{\beta}}+(\partial_{s'}\mathcal{W}^{*})\mathcal{S}+\frac{1}{2}\mathcal{S}'(\partial_{ss'}\mathcal{W}^{*})\mathcal{S}+o_{p}(1)\\
 & =\bar{W}_{b}\mathcal{S_{\beta}}+\bar{W}_{s'}\mathcal{S}+\bar{W}^{2}\mathcal{S}'\mathcal{\bar{H}}^{-1}\mathcal{\bar{H}}_{b}\mathcal{\bar{H}}^{-1}\mathcal{\bar{H}}_{b}\mathcal{\bar{H}}^{-1}\mathcal{S}\\
 & \quad-\frac{1}{2}\bar{W}^{2}\mathcal{S}'\mathcal{\bar{H}}^{-1}\mathcal{\bar{H}}_{bb}\mathcal{\bar{H}}^{-1}\mathcal{S}+o_{p}(1)
\end{align*}
and similarly for the leaveout versions $\widehat{W}_{(k)}$. We can
then write
\begin{align*}
N(\widehat{W}_{J}-W_{J}) &= N\frac{1}{N-1}\sum_k (\widehat{W}_{(k)}-W_{(k)}) \\ &=\frac{1}{N-1}\sum_{k}\Big(\bar{W}_{b}\mathcal{S}_{\beta,(k)}+\bar{W}_{s'}\mathcal{S}_{(k)}+\bar{W}^{2}\mathcal{S}_{(k)}'\mathcal{\bar{H}}^{-1}\mathcal{\bar{H}}_{b}\mathcal{\bar{H}}^{-1}\mathcal{\bar{H}}_{b}\mathcal{\bar{H}}^{-1}\mathcal{S}_{(k)}\\
 & \quad-\frac{1}{2}\bar{W}^{2}\mathcal{S}_{(k)}'\mathcal{\bar{H}}^{-1}\mathcal{\bar{H}}_{bb}\mathcal{\bar{H}}^{-1}\mathcal{S}_{(k)}\Big)+o_{p}(1)\\
 &= \frac{1}{N-1}\sum_k \big( V_{1,(k)} + V_{2,(k)} + V_{3,(k)} + V_{4,(k)} \big)+o_p(1) \\
N(\widehat{\beta}-\beta) & =\bar{W}_{N}^{-1}\big(U^{(0)}+U^{(1)}\big)+o_{p}(1)
\end{align*}
Using these expansions (and the equivalent leaveout versions) we can write
\begin{align*}
	&(N-1)(\widehat{W}_{J}-W_{J})(\widehat{\beta}-\beta)-(N-2)\frac{1}{N-1}\sum_{k}(\widehat{W}_{(k)}-W_{(k)})(\widehat{\beta}_{(k)}-\beta) \\
	&= \frac{N-1}{N^2} \frac{1}{N-1}\sum_k \big( V_{1,(k)} + V_{2,(k)} + V_{3,(k)} + V_{4,(k)} \big)\bar{W}_{N}^{-1}\big(U^{(0)}+U^{(1)}\big) \\
	&\quad - \frac{N-2}{N^2} \frac{1}{N-1}\sum_k \big( V_{1,(k)} + V_{2,(k)} + V_{3,(k)} + V_{4,(k)} \big)\bar{W}_{N}^{-1}\big(U^{(0)}_{(k)}+U^{(1)}_{(k)}\big) + o_p(N^{-1})
\end{align*}

Following similar steps to that of the higher-order jackknife expansion,
we can show that each of the terms in the product $N(\widehat{W}_{J}-W_{J})\times(\widehat{\beta}-\beta)$
is $o_{p}(1)$ when jackknifed. For example, take the first term in
the expansion of both $N(\widehat{W}_{J}-W_{J})$ and $(\widehat{\beta}-\beta)$
and apply the jackknife to give
\begin{align*}
	(N-1) & \frac{1}{N}\frac{1}{N-1}\sum_{k}V_{1,(k)}\frac{1}{N}\bar{W}_{N}^{-1}U^{(0)}-(N-2)\frac{1}{N-1}\sum_{k}\frac{1}{N}V_{1,(k)}\frac{1}{N}\bar{W}_{N}^{-1}U_{(k)}^{(0)}\\
	= & \frac{1}{N^{2}}(N-1)\bar{W}_{b}\mathcal{S}_{\beta}\bar{W}_{N}^{-1}\big(\partial_{\beta}\mathcal{L}+(\partial_{\beta\phi'}\bar{\mathcal{L}})\bar{\mathcal{H}}^{-1}\mathcal{S}\big)\\
	& -\frac{1}{N^{2}}(N-2)\frac{1}{N-1}\sum_{k}\bar{W}_{b}\mathcal{S}_{\beta,(k)}\bar{W}_{N}^{-1}\big(\partial_{\beta}\mathcal{L}_{(k)}+(\partial_{\beta\phi'}\bar{\mathcal{L}})\bar{\mathcal{H}}^{-1}\mathcal{S}_{(k)}\big)\\
	&= \frac{1}{N^2} \mathbf{J}\big[ \bar{W}_{b}\bar{W}_{N}^{-1}\mathcal{S}_{\beta}\big(\partial_{\beta}\mathcal{L}+(\partial_{\beta\phi'}\bar{\mathcal{L}})\bar{\mathcal{H}}^{-1}\mathcal{S}\big) \big]\\
	&= o_p(N^{-2})
\end{align*}
where the final line follows from application of Lemma \ref{lem:jack_r0} in the main appendix. Similarly, each of the remaining terms can be shown to be jackknife versions of products of up to four random sums, all of which are $o_p(N^{-1})$, from which the result follows.

\paragraph*{(3) $\lVert \widehat{W}_J - \overline{W}_N\rVert=o_p(1)$}

For the final result, we can apply: (a) the first result $\lVert W_{J}-W_{N}\rVert=o_p (1)$, (b) $\lVert \widehat{W}_{J}-W_{J}\rVert=o_p (1)$ (which following from parameter consistency and the fact that $W$ is a smooth function of parameters), and (c) $\lVert W_{N}-\overline{W}_{N}\rVert=o_p (1)$ from Lemma \ref{lem:H1_approx}, along with the triangle inequality to give the result.

\section{Asymptotic expansion for average effect} \label{app:avg_effect}

\subsection{Expansion for $\widehat{\phi}$}

In order to derive an expansion for the average effect parameter,
we first give an expansion for the fixed effect parameters $\phi=(\alpha',\gamma')'$.
Since we have $\partial_{s}\mathcal{L}^{*}(0,0)=-\widehat{\phi}$
and $\partial_{s}\mathcal{L}^{*}(\mathcal{S}_{\beta},\mathcal{S})=-\phi_{0}$,
we take a Taylor expansion to give

\begin{align*}\widehat{\phi}-\phi_{0} & =(\partial_{bs}\mathcal{L}^{*})\mathcal{S}_{\beta}+(\partial_{ss'}\mathcal{L}^{*})\mathcal{S}\\
	& -\frac{1}{2}(\partial_{bbs}\mathcal{L}^{*})\mathcal{S}_{\beta}^{2}\\
	& -(\partial_{bss'}\mathcal{L}^{*})\mathcal{S}\mathcal{S}_{\beta}\\
	& -\frac{1}{2}\sum_{g}(\partial_{ss's_{g}}\mathcal{L}^{*})\mathcal{S}\mathcal{S}_{g}\\
	& +\frac{1}{2}\sum_{g}(\partial_{bss's_{g}}\mathcal{L}_{(1)}^{*})\mathcal{S}_{g}\mathcal{S}\mathcal{S}_{\beta}\\
	& +\frac{1}{6}\sum_{f,g}(\partial_{ss's_{f}s_{g}}\mathcal{L}_{(1)}^{*})\mathcal{S}_{f}\mathcal{S}_{g}\mathcal{S}\\
	& -\frac{1}{24}\sum_{e,f,g}(\partial_{ss's_{e}s_{f}s_{g}}\mathcal{L}_{(1)}^{*})\mathcal{S}\mathcal{S}_{e}\mathcal{S}_{f}\mathcal{S}_{g}\\
	& +\tilde{R}_{\phi}
\end{align*}

where the remainder is

\begin{align*}\tilde{R}_{\phi} & =\frac{1}{6}\big(\partial_{bbbs}\mathcal{L}^{*}(\bar{b},\bar{s})\big)\mathcal{S}_{\beta}^{3}\\
	& +\frac{1}{2}\big(\partial_{bbss'}\mathcal{L}^{*}(\bar{b},\bar{s})\big)\mathcal{S}\mathcal{S}_{\beta}^{2}\\
	& +\frac{1}{2}\sum_{g}\big(\partial_{bss's_{g}}\mathcal{L}_{(2)}^{*}(\bar{b},\bar{s})\big)\mathcal{S}\mathcal{S}_{g}\mathcal{S}_{\beta}\\
	& +\frac{1}{6}\sum_{f,g}\big(\partial_{ss's_{f}s_{g}}\mathcal{L}_{(2)}^{*}(\bar{b},\bar{s})\big)\mathcal{S}\mathcal{S}_{f}\mathcal{S}_{g}\\
	& -\frac{1}{4}\sum_{g}\big(\partial_{bbss's_{g}}\mathcal{L}_{(1)}^{*}(\bar{b},\bar{s})\big)\mathcal{S}\mathcal{S}_{g}\mathcal{S}_{\beta}^{2}\\
	& -\frac{1}{6}\sum_{g,h}\big(\partial_{bss's_{g}s_{h}}\mathcal{L}_{(1)}^{*}(\bar{b},\bar{s})\big)\mathcal{S}\mathcal{S}_{g}\mathcal{S}_{h}\mathcal{S}_{\beta}\\
	& +\frac{1}{120}\sum_{e,f,g,h}(\partial_{ss's_{e}s_{f}s_{g}s_{h}}\mathcal{L}_{(1)}^{*})\mathcal{S}\mathcal{S}_{e}\mathcal{S}_{f}\mathcal{S}_{g}\mathcal{S}_{h}
\end{align*}

\begin{lem}
	Let Assumptions 1 and 2 hold. Then the remainder term in the expansion
	for $\widehat{\phi}$ satisfies
	\begin{align*}
		\lVert\tilde{R}_{\phi}\rVert_{q} & = O_{p}(N^{-5/2+18\epsilon})
	\end{align*}
\end{lem}
\begin{proof}
	The result follows from application of the bounds in \textbackslash ref\{sec:bounds\_ind\}
	to the expressions derived in \textbackslash ref\{sec:expressions\_ind\}.
	For the final term, we have that
	
	\begin{align*}\lVert\partial_{ss's_{e}s_{f}s_{g}s_{h}}\mathcal{L}_{(1)}^{*}\rVert_{q} & \leq36\lVert\mathcal{H}^{-1}\mathcal{H}_{s_{e}}\mathcal{H}^{-1}\mathcal{H}_{s_{f}}\mathcal{H}^{-1}\mathcal{H}_{s_{g}}\mathcal{H}^{-1}\mathcal{H}_{s_{h}}\mathcal{H}^{-1}\rVert_{q}\\
		& +9\lVert\mathcal{H}^{-1}\mathcal{H}_{s_{e}s_{f}}\mathcal{H}^{-1}\mathcal{H}_{s_{g}}\mathcal{H}^{-1}\mathcal{H}_{s_{h}}\mathcal{H}^{-1}\rVert_{q}\\
		& +9\lVert\mathcal{H}^{-1}\mathcal{H}_{s_{f}}\mathcal{H}^{-1}\mathcal{H}_{s_{e}s_{g}}\mathcal{H}^{-1}\mathcal{H}_{s_{h}}\mathcal{H}^{-1}\rVert_{q}\\
		& +9\lVert\mathcal{H}^{-1}\mathcal{H}_{s_{f}}\mathcal{H}^{-1}\mathcal{H}_{s_{g}}\mathcal{H}^{-1}\mathcal{H}_{s_{e}s_{h}}\mathcal{H}^{-1}\rVert_{q}\\
		& +6\lVert\mathcal{H}^{-1}\mathcal{H}_{s_{e}}\mathcal{H}^{-1}\mathcal{H}_{s_{f}s_{g}}\mathcal{H}^{-1}\mathcal{H}_{s_{h}}\mathcal{H}^{-1}\rVert_{q}\\
		& +2\lVert\mathcal{H}^{-1}\mathcal{H}_{s_{e}s_{f}s_{g}}\mathcal{H}^{-1}\mathcal{H}_{s_{h}}\mathcal{H}^{-1}\rVert_{q}\\
		& +2\lVert\mathcal{H}^{-1}\mathcal{H}_{s_{f}s_{g}}\mathcal{H}^{-1}\mathcal{H}_{s_{e}s_{h}}\mathcal{H}^{-1}\rVert_{q}\\
		& +6\lVert\mathcal{H}^{-1}\mathcal{H}_{s_{e}}\mathcal{H}^{-1}\mathcal{H}_{s_{g}}\mathcal{H}^{-1}\mathcal{H}_{s_{f}s_{h}}\mathcal{H}^{-1}\rVert_{q}\\
		& +2\lVert\mathcal{H}^{-1}\mathcal{H}_{s_{e}s_{g}}\mathcal{H}^{-1}\mathcal{H}_{s_{f}s_{h}}\mathcal{H}^{-1}\rVert_{q}\\
		& +2\lVert\mathcal{H}^{-1}\mathcal{H}_{s_{g}}\mathcal{H}^{-1}\mathcal{H}_{s_{e}s_{f}s_{h}}\mathcal{H}^{-1}\rVert_{q}\\
		& +6\lVert\mathcal{H}^{-1}\mathcal{H}_{s_{e}}\mathcal{H}^{-1}\mathcal{H}_{s_{f}}\mathcal{H}^{-1}\mathcal{H}_{s_{g}s_{h}}\mathcal{H}^{-1}\rVert_{q}\\
		& +2\lVert\mathcal{H}^{-1}\mathcal{H}_{s_{e}s_{f}}\mathcal{H}^{-1}\mathcal{H}_{s_{g}s_{h}}\mathcal{H}^{-1}\rVert_{q}\\
		& +2\lVert\mathcal{H}^{-1}\mathcal{H}_{s_{f}}\mathcal{H}^{-1}\mathcal{H}_{s_{e}s_{g}s_{h}}\mathcal{H}^{-1}\rVert_{q}\\
		& +2\lVert\mathcal{H}^{-1}\mathcal{H}_{s_{e}}\mathcal{H}^{-1}\mathcal{H}_{s_{f}s_{g}s_{h}}\mathcal{H}^{-1}\rVert_{q}\\
		& +\lVert\mathcal{H}^{-1}\mathcal{H}_{s_{e}s_{f}s_{g}s_{h}}\mathcal{H}^{-1}\rVert_{q}
	\end{align*}
	
	so that $\lVert\partial_{ssssss}\mathcal{L}_{(1)}^{*}\rVert_{q}=O_{p}(N^{8\epsilon})$.
	We then have
	\begin{align*}
		\lVert\sum_{e,f,g,h}\big(\partial_{ss's_{e}s_{f}s_{g}s_{h}}\mathcal{L}_{(1)}^{*}(\bar{b},\bar{s})\big)\mathcal{S}\mathcal{S}_{e}\mathcal{S}_{f}\mathcal{S}_{g}\mathcal{S}_{h}\rVert_{q} & \leq\lVert\partial_{sssss}\mathcal{L}^{*}\rVert_{q}\lVert\mathcal{S}\rVert_{q}^{5}\\
		& =O_{p}(N^{-5/2+18\epsilon})
	\end{align*}
\end{proof}

\subsection{\label{subsec:avg_effect_bounds}Bounds for derivation of expansion}

In order to provide bounds on the terms in the asymptotic expansion
for the average effect, we first state some additional bounds on terms
that are used in the expansion for $\widehat{\phi}$. The first set
of bounds shows that $N\partial_{\phi}\bar{\Delta}$ has $O_{p}(1)$
elements, while $N\partial_{\phi\phi'}\bar{\Delta}$ is a matrix
that satisfies the conditions in Lemma \ref{lem:HAH_bound}. In addition,
we show that $N\partial_{\phi}\tilde{\Delta}$ satisfies the conditions
in Lemma \ref{lem:S6}, while $N\partial_{\phi\phi'}\tilde{\Delta}$
satisfies a bound like that in Lemma \ref{lem:H_tilde}.

The next lemmas are used to bound the jackknifed versions of the new
forms of product terms appearing in the expansion for the average
effects. 
\begin{lem} \label{lem:Delta_12}
	Let $\lambda$ be a set of $r$ observations
	$(i,j)$ involving $p$ unique agents, and $\Lambda_{N}$ be the collection
	of all such $\lambda$ formed by permuting the agents in $\lambda$.
	Then, under Assumption 3,
	
	(i) $\partial_{\phi}\bar{\Delta}$ has $O_{p}(N^{-1})$ elements
	
	(ii) $\partial_{\alpha\alpha'}\bar{\Delta}$, $\partial_{\alpha\gamma'}\bar{\Delta}$,
	$\partial_{\gamma\alpha'}\bar{\Delta}$, and $\partial_{\gamma\gamma'}\bar{\Delta}$
	each have $O_{p}(N^{-1})$ diagonal elements and $O_{p}(N^{-2})$
	off-diagonal elements
\end{lem}
\begin{proof}
	Let $\lambda_{\alpha}$ denote the set of $p_{\alpha}$ sender agents
	in the observations within $\lambda$, and $\lambda_{\gamma}$ the
	set of $p_{\gamma}$ receiving agents. There are $\vert\Lambda_{N}\vert=\frac{N!}{(N-p)!}$
	ways of selecting the $p$ agents in $\lambda$. Among these permutations,
	agent $i$ is a sender $p_{\alpha}\frac{(N-1)!}{(N-p)!}$ times, while
	node $j$ is the receiver $p_{\gamma}\frac{(N-1)!}{(N-p)!}$ times.
	Using this, the first derivatives of $\bar{\Delta}$ with respect
	to the fixed effects are
	\begin{align*}
		\partial_{\alpha_{i}}\bar{\Delta} & =\frac{1}{\vert\Lambda_{N}\vert}\sum_{\lambda:i\in\lambda_{\alpha}}\partial_{\alpha_{i}}\bar{m}_{\lambda}=O_{p}(N^{-1})\\
		\partial_{\gamma_{i}}\bar{\Delta} & =\frac{1}{\vert\Lambda_{N}\vert}\sum_{\lambda:i\in\lambda_{\gamma}}\partial_{\gamma_{i}}\bar{m}_{\lambda}=O_{p}(N^{-1})
	\end{align*}
	where the $O_{p}(N^{-1})$ statements come from the fact that $p_{\alpha}\frac{(N-1)!}{(N-p)!}/\frac{N!}{(N-p)!}=p_{\alpha}/N$.
	An identical result applies to the diagonal elements of $\partial_{\phi\phi'}\bar{\Delta}$,
	i.e. $\partial_{\alpha_{i}\alpha_{i}}\bar{\Delta}=O_{p}(N^{-1})$
	and $\partial_{\gamma_{j}\gamma_{j}}\bar{\Delta}=O_{p}(N^{-1})$
	since they are sums over the same sets of $\lambda$. Also, if the
	presence of $i$ as a sender agent implies that $i$ is also a receiver
	in $\lambda$ (e.g. the cyclic triangle $\{(i,j),(j,k),(k,i)\}$)
	then it will be the case that $\partial_{\alpha_{i}\gamma_{i}}\bar{\Delta}=O_{p}(N^{-1})$
	also (if this is not true it will be lower order).
	
	Next, consider the off-diagonal components of $\partial_{\phi\phi'}\bar{\Delta}$.
	If $p_{\alpha}=1$ then $\partial_{\alpha_{i}\alpha_{j}}\bar{\Delta}=0$,
	otherwise, there are ${p_{\alpha} \choose 2}\frac{(N-2)!}{(N-p)!}$
	permutations that contain both $i$ and $j$ as senders. Similarly,
	for $p_{\gamma}\geq2$, there are ${p_{\gamma} \choose 2}\frac{(N-2)!}{(N-p)!}$
	permutations that contain both $i$ and $j$ as receivers. Finally,
	there are \emph{at most} $p_{\alpha}p_{\gamma}\frac{(N-2)!}{(N-p)!}$
	permutations in which $i$ is a sender and $j$ a receiver (this is
	an upper bound since with $i$ in a particular sender position, not
	all receiver positions may be valid for $j$). This, along with Assumption
	3, gives the results
	\begin{align*}
		\partial_{\alpha_{i}\alpha_{j}}\bar{\Delta} & =O_{p}(N^{-2})\\
		\partial_{\alpha_{i}\gamma_{j}}\bar{\Delta} & =O_{p}(N^{-2})\\
		\partial_{\gamma_{i}\gamma_{j}}\bar{\Delta} & =O_{p}(N^{-2})
	\end{align*}
	which demonstrates the lemma.
\end{proof}
\begin{lem}
	\label{lem:delta_bounds}Let Assumptions 1, 2 and 3 hold, for $s=\{0,1,2,3\}$
	\begin{align*}
		N\partial_{\beta^{s}\alpha_{i}}\tilde{\Delta} & =O_{p}(N^{-1/2})\\
		\max_{i}\vert N\partial_{\beta^{s}\alpha_{i}}\tilde{\Delta}\vert & =O_{p}(N^{-1/2+2\epsilon})
	\end{align*}
	and hence, 
	\begin{align*}
		\lVert\partial_{\beta^{s}\phi}\Delta\rVert_{q} & =O_{p}(N^{-1+2\epsilon})\\
		\lVert\partial_{\beta^{s}\phi}\tilde{\Delta}\rVert_{q} & =O_{p}(N^{-\frac{3}{2}+2\epsilon})\\
		\lVert\partial_{\beta^{s}\phi\phi'}\tilde{\Delta}\rVert_{q} & =O_{p}(N^{-\frac{3}{2}+4\epsilon})\\
		\lVert\partial_{\beta^{s}\phi\phi'}\Delta\rVert_{q} & =O_{p}(N^{-1+2\epsilon})\\
		\lVert\partial_{\beta\phi\phi\phi}\tilde{\Delta}\rVert_{q} & =O_{p}(N^{-\frac{3}{2}+4\epsilon})\\
		\lVert\partial_{\beta\phi\phi\phi}\Delta\rVert_{q} & =O_{p}(N^{-1+2\epsilon})\\
		\lVert\partial_{\phi\phi\phi\phi}\tilde{\Delta}\rVert_{q} & =O_{p}(N^{-\frac{3}{2}+4\epsilon})\\
		\lVert\partial_{\phi\phi\phi\phi}\Delta\rVert_{q} & =O_{p}(N^{-1+2\epsilon}) \\
			\lVert\partial_{\phi\phi\phi\phi\phi}\Delta\rVert_{q} & =O_{p}(N^{-1+2\epsilon})
	\end{align*}
\end{lem}
\begin{proof}
	For the first claim, note that
	\[
	\bar{E}\big[(N\partial_{\beta^{s}\alpha_{i}}\tilde{\Delta})^{q}\big]=\frac{N^{q}}{\vert\Lambda_{N}\vert^{q}}\sum_{\lambda_{1}:i\in\lambda_{\alpha}}\cdots\sum_{\lambda_{q}:i\in\lambda_{\alpha}}\bar{E}[(\partial_{\beta^{s}\alpha_{i}}\tilde{m}_{\lambda_{1}})\cdots(\partial_{\beta^{s}\alpha_{i}}\tilde{m}_{\lambda_{q}})]
	\]
	Since terms are independent unless they share a common dyad, each
	of the $\lambda_{j}$ must share a dyad with at least one other $\lambda_{k}$.
	The number of $\lambda$ that contain a particular dyad is of order
	$N^{-2}$ smaller than $\vert\Lambda_{N}\vert$ and hence the above
	term must be at most $O_{p}(N^{-q/2})$, and hence $N\partial_{\beta^{s}\alpha_{i}}\tilde{\Delta}=O_{p}(N^{-1/2})$.
	This also implies
	\begin{align*}
		\bar{E}\big[\max_{i}\vert N\partial_{\beta^{s}\alpha_{i}}\tilde{\Delta}\vert^{q}\big] & \leq\sum_{i}\bar{E}\big[\vert N\partial_{\beta^{s}\alpha_{i}}\tilde{\Delta}\vert^{q}\big]\\
		& =O_{p}(N^{1-q/2})\\
		\max_{i}\vert N\partial_{\beta^{s}\alpha_{i}}\tilde{\Delta}\vert & =O_{p}(N^{-\frac{1}{2}+\frac{1}{q}})=O_{p}(N^{-\frac{1}{2}+2\epsilon})
	\end{align*}
	
	The next results follow from the first two by
	\begin{align*}
		\bar{E}\big[\lVert\partial_{\beta^{s}\phi}\Delta\rVert_{q}^{q}\big] & \leq\bar{E}\big[\sum_{i}\vert\frac{1}{\vert\Lambda_{N}\vert}\sum_{\lambda\in\lambda_{\alpha}}m_{\beta^{s}\alpha_{i}}(\lambda)\vert^{q}\big]+\bar{E}\big[\sum_{i}\vert\frac{1}{\vert\Lambda_{N}\vert}\sum_{\lambda\in\lambda_{\gamma}}m_{\beta^{s}\gamma_{i}}(\lambda)\vert^{q}\big]\\
		& \leq\sum_{i}\frac{\vert\Lambda_{\alpha}\vert^{q-1}}{\vert\Lambda_{N}\vert^{q}}\sum_{\lambda\in\lambda_{\alpha}}\bar{E}\big[\vert m_{\beta^{s}\alpha_{i}}(\lambda)\vert^{q}\big]+\sum_{i}\frac{\vert\Lambda_{\alpha}\vert^{q-1}}{\vert\Lambda_{N}\vert}\sum_{\lambda\in\lambda_{\gamma}}\bar{E}\big[\vert m_{\beta^{s}\gamma_{i}}(\lambda)\vert^{q}\big]\\
		& =O_{p}(N^{1-q})\\
		\lVert\partial_{\beta^{s}\phi}\Delta\rVert_{q} & =O_{p}(N^{-1+\frac{1}{q}})
	\end{align*}
	and
	\begin{align*}
		\bar{E}\big[\lVert\partial_{\beta^{s}\phi}\tilde{\Delta}\rVert_{q}^{q}\big] & \leq\sum_{i}\bar{E}\big[\vert\partial_{\beta^{s}\alpha_{i}}\tilde{\Delta}\vert^{q}\big]+\sum_{i}\bar{E}\big[\vert\partial_{\beta^{s}\gamma_{i}}\tilde{\Delta}\vert^{q}\big]\\
		& =O_{p}(N^{1-\frac{3q}{2}})\\
		\lVert\partial_{\beta^{s}\phi}\tilde{\Delta}\rVert_{q} & =O_{p}(N^{-\frac{3}{2}+\frac{1}{q}})
	\end{align*}
	
	For the second result, we focus on the block $\partial_{\beta^{s}\alpha\alpha'}\tilde{\Delta}$,
	with other blocks shown similarly. First, we have
	\[
	\bar{E}\big[(N\partial_{\beta^{s}\alpha_{i}\alpha_{j}}\tilde{\Delta})^{q}\big]=\frac{N^{q}}{\vert\Lambda_{N}\vert^{q}}\sum_{\lambda_{1}:\{i,j\}\in\lambda_{\alpha}}\cdots\sum_{\lambda_{q}:\{i,j\}\in\lambda_{\alpha}}\bar{E}[(\partial_{\beta^{s}\alpha_{i}\alpha_{j}}\tilde{m}_{\lambda_{1}})\cdots(\partial_{\beta^{s}\alpha_{i}\alpha_{j}}\tilde{m}_{\lambda_{q}})]
	\]
	Again, terms are independent unless they share a common dyad. Given
	this, and the fact that there are $O(N^{-2}\vert\Lambda_{N}\vert)$
	$\lambda$ that contain both $\alpha_{i}$ and $\alpha_{j}$ (and
	of order $N$ fewer that share a common dyad), the above summation
	is $O_{p}(N^{-\frac{3}{2}q})$. Using this result, we then have
	\begin{align*}
		\bar{E}\big[(\max_{i,j\ne i}N\partial_{\beta^{s}\alpha_{i}\alpha_{j}}\tilde{\Delta})^{q}\big] & \leq\sum_{i}\sum_{j\ne i}\bar{E}\big[(N\partial_{\beta^{s}\alpha_{i}\alpha_{j}}\tilde{\Delta})^{q}\big]\\
		& =O_{p}(N^{2-\frac{3}{2}q})
	\end{align*}
	and so $\max_{i,j\ne i}(N\partial_{\beta^{s}\alpha_{i}\alpha_{j}}\tilde{\Delta})=O_{p}(N^{-\frac{3}{2}+2/q})$.
	Similarly to the second statement of this lemma, we can show $\max_{i}(N\partial_{\beta^{s}\alpha_{i}\alpha_{i}}\tilde{\Delta})=O_{p}(N^{-\frac{1}{2}+1/q})$.
	These bounds then imply (using Lemma S.4 in FW16)
	\begin{align*}
		\lVert\partial_{\beta^{s}\alpha\alpha'}\tilde{\Delta}\rVert_{q} & \leq\lVert\partial_{\beta^{s}\alpha\alpha'}\tilde{\Delta}\rVert_{\infty}\\
		& =\max_{i}\vert\sum_{j}\partial_{\beta^{s}\alpha_{i}\alpha_{j}}\tilde{\Delta}\vert\\
		& \leq\max_{i}\vert\partial_{\beta^{s}\alpha_{i}\alpha_{i}}\tilde{\Delta}\vert+N\max_{i\ne j}\vert\partial_{\beta^{s}\alpha_{i}\alpha_{j}}\tilde{\Delta}\vert\\
		& =O_{p}(N^{-\frac{3}{2}+\frac{2}{q}})
	\end{align*}
	and similarly for the remaining blocks of $\partial_{\beta^{s}\phi\phi'}\tilde{\Delta}$.
	Similarly, is is straightforward to show that $\bar{E}\big[(\partial_{\alpha_{i}\alpha_{i}}\Delta)^{q}\big]=O_{p}(N^{-q})$
	and $\bar{E}\big[(\sum_{j\ne i}\vert\partial_{\alpha_{i}\alpha_{j}}\Delta\vert)^{q}\big]=O_{p}(N^{-q})$
	for $i\ne j$, and hence
	\begin{align*}
		\max_{i}\vert\partial_{\alpha_{i}\alpha_{i}}\Delta\vert & =O_{p}(N^{-1+\frac{1}{q}})\\
		\max_{i}\sum_{j\ne i}\vert\partial_{\alpha_{i}\alpha_{j}}\Delta\vert & =O_{p}(N^{-1+\frac{1}{q}})
	\end{align*}
	from which the final result follows, since $\lVert\partial_{\alpha\alpha}\Delta\rVert_{q}\leq\lVert\partial_{\alpha\alpha}\Delta\rVert_{\infty}$,
	and the same can be shown for the remaining components of $\partial_{\phi\phi'}\Delta$.
	Finally, the result for $\partial_{\beta\phi\phi\phi}\mathcal{L}$
	and $\partial_{\phi\phi\phi\phi}\mathcal{L}$, we have
	\begin{align*}
		\lVert\partial_{\beta\phi\phi\phi}\Delta\rVert_{q} & =\lVert\sum_{g}(\partial_{\phi\phi\phi_{g}}\Delta)\rVert_{q}\\
		& \leq\lVert\sum_{g}(\partial_{\phi\phi\phi_{g}}\Delta)\rVert_{\infty}\\
		& =\max_{e}\vert\sum_{f}\sum_{g}(\partial_{\phi_{e}\phi_{f}\phi_{g}}\Delta)\vert\\
		& \leq\max_{e}\vert(\partial_{\phi_{e}\phi_{e}\phi_{e}}\Delta)\vert+\max_{e}\vert\sum_{f\ne e}(\partial_{\phi_{e}\phi_{f}\phi_{f}}\Delta)\vert\\
		& +\max_{e}\vert\sum_{f\ne e}(\partial_{\phi_{e}\phi_{e}\phi_{f}}\Delta)\vert+\max_{e}\vert\sum_{f\ne g\ne e}(\partial_{\phi_{e}\phi_{f}\phi_{g}}\Delta)\vert\\
		& =O_{p}(N^{-1+1/q})
	\end{align*}
	since we have that derivatives with respect to $j$ distinct indices
	in $\phi$ result in a sum that is $O_{p}(N^{-j})$. The bound for
	$\partial_{\phi\phi\phi\phi}\mathcal{L}$ can be shown similarly. The results for $\partial_{\beta\phi\phi\phi}\Delta$ and $\partial_{\phi\phi\phi\phi}\Delta$ follow identically to $\partial_{\beta\phi\phi}\Delta$.
\end{proof}

\subsection{Asymptotic expansion}
An expansion of the average effects estimator is

\begin{align*}
	\widehat{\Delta}-\Delta & =(\partial_{\beta}\Delta)(\widehat{\beta}-\beta_{0})\\
	& +(\partial_{\phi'}\Delta)(\widehat{\phi}-\phi_{0})\\
	& \quad+\frac{1}{2}(\partial_{\beta\beta}\Delta)(\widehat{\beta}-\beta_{0})^{2}\\
	& +(\partial_{\beta\phi'}\Delta)(\widehat{\phi}-\phi_{0})(\widehat{\beta}-\beta_{0})\\
	& \quad+\frac{1}{2}(\widehat{\phi}-\phi_{0})'(\partial_{\phi\phi'}\Delta)(\widehat{\phi}-\phi_{0})\\
	& \quad+\frac{1}{2}(\widehat{\phi}-\phi_{0})'(\partial_{\beta\phi\phi'}\Delta)(\widehat{\phi}-\phi_{0})(\widehat{\beta}-\beta_{0})\\
	& \quad+\frac{1}{6}\sum_{g}(\widehat{\phi}-\phi_{0})'(\partial_{\phi\phi'\phi_{g}}\Delta)(\widehat{\phi}-\phi_{0})(\widehat{\phi}_{g}-\phi_{0,g})\\
	& \quad+\frac{1}{24}\sum_{g,h}(\widehat{\phi}-\phi_{0})'\big(\partial_{\phi\phi'\phi_{g}\phi_{h}}\Delta\big)(\widehat{\phi}-\phi_{0})(\widehat{\phi}_{g}-\phi_{0,g})(\widehat{\phi}_{h}-\phi_{0,h})\\
	& \quad+\tilde{R}_{\Delta}
\end{align*}

where the remainder term satisfies
\begin{align*}\tilde{R}_{\Delta} & =\frac{1}{6}\big(\partial_{\beta\beta\beta}\Delta(\bar{\beta},\bar{\phi})\big)(\widehat{\beta}-\beta_{0})^{3}\\
	& +\frac{1}{2}(\partial_{\beta\beta\phi'}\Delta)(\widehat{\phi}-\phi_{0})(\widehat{\beta}-\beta_{0})^{2}\\
	& \quad+\frac{1}{6}\big(\partial_{\beta\beta\beta\phi'}\Delta(\bar{\beta},\bar{\phi})\big)(\widehat{\phi}-\phi_{0})(\widehat{\beta}-\beta_{0})^{3}\\
	& \quad+\frac{1}{4}(\widehat{\phi}-\phi_{0})'\big(\partial_{\beta\beta\phi\phi'}\Delta(\bar{\beta},\bar{\phi})\big)(\widehat{\phi}-\phi_{0})(\widehat{\beta}-\beta_{0})^{2}\\
	& \quad+\frac{1}{6}\sum_{g}(\widehat{\phi}-\phi_{0})'\big(\partial_{\beta\phi\phi'\phi_{g}}\Delta(\bar{\beta},\bar{\phi})\big)(\widehat{\phi}-\phi_{0})(\widehat{\phi}_{g}-\phi_{0,g})(\widehat{\beta}-\beta_{0})\\
	& \quad+\frac{1}{120}\sum_{f,g,h}(\widehat{\phi}-\phi_{0})'\big(\partial_{\phi\phi'\phi_{f}\phi_{g}\phi_{h}}\Delta(\bar{\beta},\bar{\phi})\big)(\widehat{\phi}-\phi_{0})(\widehat{\phi}_{f}-\phi_{0,f})(\widehat{\phi}_{g}-\phi_{0,g})(\widehat{\phi}_{h}-\phi_{0,h})\\
	& =o_{p}(N^{-2})
\end{align*}

Substitution of the expansion for $\widehat{\beta}$ up to order $N^{-2}$
and the expansion for $\widehat{\phi}$ shown above then gives an expansion with remainder $o_p(N^{-2})$, since for example
\begin{align*}
		\lVert(\partial_{\phi'}\Delta) \tilde{R}_\phi \rVert_{q} & \leq N^{1-2/q}\lVert\partial_{\phi'}\Delta\rVert_{q}\lVert\tilde{R}_\phi\rVert_{q}\\
		& =O_p(N^{-1/q})\rVert_{q}\lVert\tilde{R}_\phi\rVert_{q}=O_p(N^{-5/2+16\epsilon})=o_{p}(N^{-2})	
\end{align*}

It follows that the expansion up to first-order is 

\begin{align*}
N(\widehat{\Delta}-\Delta) & =(\partial_{\beta}\bar{\Delta})N(\widehat{\beta}-\beta_{0})+N(\partial_{\phi'}\Delta)(\partial_{bs}\mathcal{L}^{*})\mathcal{S}_{\beta}\\
 & \quad+N(\partial_{\phi'}\Delta)(\partial_{ss'}\mathcal{L}^{*})\mathcal{S}-\frac{1}{2}N(\partial_{\phi'}\Delta)\sum_{g}(\partial_{ss's_{g}}\mathcal{L}^{*})\mathcal{S}\mathcal{S}_{g}\\
 & \quad+\frac{1}{2}N\mathcal{S}'(\partial_{ss'}\mathcal{L}^{*})(\partial_{\phi\phi'}\Delta)(\partial_{ss'}\mathcal{L}^{*})\mathcal{S}+o_{p}(1)\\
\\
 & =(\partial_{\beta}\bar{\Delta})\bar{W}_{N}^{-1}\big(U^{(0)}+U^{(1)}\big)\\
 & \quad+\bar{W}_{N}^{-1}(\partial_{\beta\phi'}\bar{\mathcal{L}})\bar{\mathcal{H}}^{-1}(\partial_{\phi}\bar{\Delta})\mathcal{S}_{\beta}\\
 & \quad+N(\partial_{\phi}\bar{\Delta})\mathcal{\bar{H}}^{-1}\mathcal{S}+N(\partial_{\phi}\tilde{\Delta})\mathcal{\bar{H}}^{-1}\mathcal{S}+N(\partial_{\phi}\bar{\Delta})\mathcal{\bar{H}}^{-1}\tilde{\mathcal{H}}\mathcal{\bar{H}}^{-1}\mathcal{S}\\
 & \quad+W_{N}^{-1}(\partial_{\phi}\bar{\Delta})\mathcal{\bar{H}}^{-1}(\partial_{\beta\phi}\bar{\mathcal{L}})\Big((\partial_{\beta\phi'}\bar{\mathcal{L}})\mathcal{\bar{H}}^{-1}\mathcal{S}+(\partial_{\beta\phi'}\tilde{\mathcal{L}})\mathcal{\bar{H}}^{-1}\mathcal{S}+(\partial_{\beta\phi'}\bar{\mathcal{L}})\mathcal{\bar{H}}^{-1}\tilde{\mathcal{H}}\mathcal{\bar{H}}^{-1}\mathcal{S}\Big)\\
 & \quad+\frac{1}{2}N\mathcal{S}'\bar{\mathcal{H}}^{-1}\sum_{f}(\partial_{\phi\phi'\phi_{f}}\bar{\mathcal{L}})[\mathcal{\bar{H}}^{-1}(\partial_{\phi'}\bar{\Delta})]_{f}\mathcal{\bar{H}}^{-1}\mathcal{S}\\
 & \quad+\frac{1}{2}\bar{W}^{-1}(\partial_{\phi'}\bar{\Delta})\mathcal{\bar{H}}^{-1}(\partial_{\beta\phi}\mathcal{\bar{L}})\mathcal{S}'\bar{\mathcal{H}}^{-1}\sum_{f}(\partial_{\phi\phi'\phi_{f}}\bar{\mathcal{L}})[\mathcal{\bar{H}}^{-1}(\partial_{\phi'}\bar{\Delta})]_{f}\mathcal{\bar{H}}^{-1}\mathcal{S}\\
 & \quad+\frac{1}{2}N\mathcal{S}'\bar{\mathcal{H}}^{-1}(\partial_{\phi\phi'}\bar{\Delta})\bar{\mathcal{H}}^{-1}\mathcal{S}+o_{p}(1)\\
\\
 & =\Big((\partial_{\beta}\bar{\Delta})+(\partial_{\beta\phi'}\bar{\mathcal{L}})\bar{\mathcal{H}}^{-1}(\partial_{\phi}\bar{\Delta})\Big)\bar{W}_{N}^{-1}\big(U^{(0)}+U^{(1)}\big)\\
 & \quad+N(\partial_{\phi}\bar{\Delta})\mathcal{\bar{H}}^{-1}\mathcal{S}+N(\partial_{\phi}\tilde{\Delta})\mathcal{\bar{H}}^{-1}\mathcal{S}-N(\partial_{\phi}\bar{\Delta})\mathcal{\bar{H}}^{-1}\tilde{\mathcal{H}}\mathcal{\bar{H}}^{-1}\mathcal{S}\\
 & \quad+\frac{1}{2}N\mathcal{S}'\bar{\mathcal{H}}^{-1}\Big((\partial_{\phi\phi'}\bar{\Delta})+\sum_{f}(\partial_{\phi\phi'\phi_{f}}\bar{\mathcal{L}})[\mathcal{\bar{H}}^{-1}(\partial_{\phi'}\bar{\Delta})]_{f}\Big)\bar{\mathcal{H}}^{-1}\mathcal{S}+o_{p}(1)
\end{align*}

The expansion for the leave-out estimator follows similarly, replacing $m_\lambda$ with $\frac{N-1}{N-r-1} m_\lambda 1^k_\lambda$. Note that the same bounds on $(N\partial_{\beta^{s}\alpha_{i}}\tilde{\Delta})$ derivied in Lemma \ref{lem:delta_bounds} also apply to $(N\partial_{\beta^{s}\alpha_{i}}\tilde{\Delta}_{(k)})$, so that we may replace the $\partial_{\beta}\Delta_{(k)}$ and $\partial_{\alpha}\Delta_{(k)}$ terms with $\partial_{\beta}\bar{\Delta}$ and $\partial_{\alpha}\bar{\Delta}$ in the first-order expansion. This gives the expansion for $N(\widehat{\Delta}_{(k)} - \Delta_{(k)})$ as shown in Lemma 7.

\subsection{Jackknifing higher-order terms for average effect}

Analogously to the jackknife of the common parameter $\widehat{\beta},$
we demonstrate the effect of the jackknife for just one higher-order
term in the expansion for $\widehat{\Delta}$. Inspection of this
expansion shows that the structure of each of the terms is similar,
i.e. a V-statistic like term of some order, multiplied by $(\partial_{\phi'}\Delta)$
or $(\partial_{\phi\phi'}\Delta)$. The bound for the jackknife
version of each of these many terms is therefore highly similar to
the specific term we demonstrate here. Consider the expansion term
$(\partial_{\phi'}\Delta)(\partial_{bss'}\mathcal{L}_{(1)}^{*})\mathcal{S}\mathcal{S}_{\beta}$,
which contains the term
\[
(\partial_{\phi'}\Delta)\mathcal{H}^{-1}\mathcal{H}_{b}\mathcal{H}^{-1}\mathcal{S}\mathcal{S}_{\beta}
\]
Expanding out this terms gives
\begin{align*}
(\partial_{\phi'}\Delta)\mathcal{H}^{-1}\mathcal{H}_{b}\mathcal{H}^{-1}\mathcal{S}\mathcal{S}_{\beta} & =W^{-1}(\partial_{\phi'}\Delta)\mathcal{H}^{-1}(\partial_{\beta\phi\phi'}\mathcal{L})\mathcal{H}^{-1}\mathcal{S}\mathcal{S}_{\beta}\\
 & +W^{-1}(\partial_{\phi'}\Delta)\mathcal{H}^{-1}\sum_{f}(\partial_{\phi\phi'\phi_{f}}\mathcal{L})\big[\mathcal{H}^{-1}(\partial_{\beta\phi}\mathcal{L})\big]_{f}\mathcal{H}^{-1}\mathcal{S}\mathcal{S}_{\beta}
\end{align*}
Take the first term in this expression. Replacing $W$ with $\bar{W}+(W-\bar{W})$,
and similarly for $\mathcal{H}^{-1}$ and $\partial_{\beta\phi\phi'}\mathcal{L}$
we get
\[
\frac{1}{N}\bar{W}_{N}^{-1}(\partial_{\phi'}\Delta)\mathcal{\bar{H}}^{-1}(\partial_{\beta\phi\phi'}\bar{\mathcal{L}})\bar{\mathcal{H}}^{-1}\mathcal{S}\mathcal{S}_{\beta}+o_{p}(N^{-2})
\]
This result holds identically for the leave-out samples, replacing
$\mathcal{S}$ with $\mathcal{S}_{(k)}$ and similarly for $\partial_{\phi'}\Delta$
and $\mathcal{S}_{\beta}$. Ignoring the $\frac{1}{N}\bar{W}_{N}^{-1}$
term for the moment, which will not be affected by the jackknifing,
and letting $\bar{M}=\mathcal{\bar{H}}^{-1}(\partial_{\beta\phi\phi'}\bar{\mathcal{L}})\bar{\mathcal{H}}^{-1}$,
we can decompose the above sum further into the components of $\phi=(\alpha,\gamma)$,
for example the first of these terms would be
\begin{align*}
(\partial_{\alpha'}\Delta)\bar{M}_{\alpha\alpha}\mathcal{S}_{\alpha}\mathcal{S}_{\beta} & =\frac{1}{\lvert\Lambda_{N}\rvert}\frac{1}{(N-1)^{2}}\sum_{i}\sum_{\lambda:i\in\lambda_{\alpha}}\sum_{s}\sum_{t\ne s}\sum_{j}\sum_{l\ne j}M_{is}(\partial_{\alpha_{i}}m_{\lambda})(\partial_{\pi}\ell_{st})(\partial_{\beta}\ell_{jl})
\end{align*}
and note that, by Lemma \ref{lem:HAH_bound}, we have $\max_{i\ne s}\bar{M}_{is}=O_{p}(N^{-1})$,
and $\max_{i}\bar{M}_{ii}=O_{p}(1)$. This is again a V-statistic
like term, and we will show that the jackknifed version of this term
is $o_{p}(1)$. The $k$-th leave-out version of this term is
\begin{align*}
\frac{N-1}{N-r-1}&(\partial_{\alpha'}\Delta_{N,(k)})\bar{M}_{\alpha\alpha}\mathcal{S}_{\alpha,(k)}\mathcal{S}_{\beta,(k)}  \\
& =\frac{1}{\lvert\Lambda_{N}\rvert}\frac{N-1}{(N-2)^{2}(N-r-1)}\sum_{i}\sum_{\lambda:i\in\lambda_{\alpha}}\sum_{s}\sum_{t\ne s}\sum_{j}\sum_{l\ne j}M_{is}(\partial_{\alpha_{i}}m_{\lambda})(\partial_{\pi}\ell_{st})(\partial_{\beta}\ell_{jl})1_{\lambda}^{k}1_{st}^{k}1_{jl}^{k}
\end{align*}

The jackknifed term is then
\begin{align*}
\mathcal{J}= & (N-1)(\partial_{\alpha'}\Delta)\bar{M}_{\alpha\alpha}\mathcal{S}_{\alpha}\mathcal{S}_{\beta}-\frac{N-2}{N-1}\sum_{k}(\partial_{\alpha'}\Delta_{N,(k)})\bar{M}_{\alpha\alpha}\mathcal{S}_{\alpha,(k)}\mathcal{S}_{\beta,(k)}\\
= & \frac{1}{\lvert\Lambda_{N}\rvert}\frac{1}{N-1}\sum_{i}\sum_{\lambda:i\in\lambda_{\alpha}}\sum_{s}\sum_{t\ne s}\sum_{j}\sum_{l\ne j}M_{is}(\partial_{\alpha_{i}}m_{\lambda})(\partial_{\pi}\ell_{st})(\partial_{\beta}\ell_{jl})\\
 & -\frac{1}{\lvert\Lambda_{N}\rvert}\frac{1}{(N-2)(N-r-1)}\sum_{i}\sum_{\lambda:i\in\lambda_{\alpha}}\sum_{s}\sum_{t\ne s}\sum_{j}\sum_{l\ne j}M_{is}(\partial_{\alpha_{i}}m_{\lambda})(\partial_{\pi}\ell_{st})(\partial_{\beta}\ell_{jl})\big(\sum_{k}1_{\lambda}^{k}1_{st}^{k}1_{jl}^{k}\big)\\
= & \frac{1}{\lvert\Lambda_{N}\rvert}\sum_{i}\sum_{\lambda:i\in\lambda_{\alpha}}\sum_{s}\sum_{t\ne s}\sum_{j}\sum_{l\ne j}M_{is}(\partial_{\alpha_{i}}m_{\lambda})(\partial_{\pi}\ell_{st})(\partial_{\beta}\ell_{jl})\Big(\frac{1}{N-1}-\frac{\sum_{k}1_{\lambda}^{k}1_{st}^{k}1_{jl}^{k}}{(N-2)(N-r-1)}\Big)
\end{align*}
Note that, since $\lambda$ contains $r$ observations, the set $\{\lambda,(s,t),(j,l)\}$
spans at most $r+2$ different sets $\mathcal{I}_{k}$ and so $\sum_{k}1_{\lambda}^{k}1_{st}^{k}1_{jl}^{k}\in\{N-r-3,\dots,N-2\}$,
depending on how many of the $\mathcal{I}_{k}$ the observations are
contained in. Then, letting $1_{n}$ be an indicator variable that
is equal to one when the set $\{\lambda,(s,t),(j,l)\}$ covers $n$
different $\mathcal{I}_{k},$ we can write
\begin{align*}
\vert\mathcal{J}\vert & \leq\sum_{n=1}^{r+2}\vert\frac{1}{\lvert\Lambda_{N}\rvert}\sum_{i}\sum_{\lambda:i\in\lambda_{\alpha}}\sum_{s}\sum_{t\ne s}\sum_{j}\sum_{l\ne j}M_{is}(\partial_{\alpha_{i}}m_{\lambda})(\partial_{\pi}\ell_{st})(\partial_{\beta}\ell_{jl})\Big(\frac{1}{N-1}-\frac{N-n-1}{(N-2)(N-r-1)}\Big)1_{n}\vert\\
 & \leq C\frac{1}{N^{2}}\sum_{n=1}^{r+2}\vert\frac{1}{\lvert\Lambda_{N}\rvert}\sum_{i}\sum_{\lambda:i\in\lambda_{\alpha}}\sum_{s}\sum_{t\ne s}\sum_{j}\sum_{l\ne j}M_{is}(\partial_{\alpha_{i}}m_{\lambda})(\partial_{\pi}\ell_{st})(\partial_{\beta}\ell_{jl})1_{n}\vert\\
 & \leq C\frac{1}{N^{2}}\sum_{n=1}^{r+2}\vert\frac{1}{\lvert\Lambda_{N}\rvert}\sum_{i}\sum_{\lambda:i\in\lambda_{\alpha}}\sum_{s}\sum_{t\ne s}\sum_{j}\sum_{l\ne j}M_{is}(\partial_{\alpha_{i}}m_{\lambda})(\partial_{\pi}\ell_{st})(\partial_{\beta}\ell_{jl})1_{n}\vert
\end{align*}
The second moment of the RHS is
\begin{align*}
&\bar{E}\Big[ \Big(\frac{1}{\lvert\Lambda_{N}\rvert}\sum_{i}\sum_{\lambda:i\in\lambda_{\alpha}}\sum_{s}\sum_{t\ne s}\sum_{j}\sum_{l\ne j}M_{is}(\partial_{\alpha_{i}}m_{\lambda})(\partial_{\pi}\ell_{st})(\partial_{\beta}\ell_{jl})1_{n}\Big)^{2}\Big]\\
&=  \frac{1}{\lvert\Lambda_{N}\rvert^{2}}\sum_{i,i',s,s',j,j'}\sum_{\substack{\lambda:i\in\lambda_{\alpha}\\
\lambda':i'\in\lambda_{\alpha}}} \sum_{t\ne s}\sum_{t'\ne s'}\sum_{l\ne j}\sum_{l'\ne j'}M_{is}M_{i's'}\bar{E}\Big[(\partial_{\alpha_{i}}m_{\lambda})(\partial_{\alpha_{i'}}m_{\lambda'})(\partial_{\pi}\ell_{st})(\partial_{\pi}\ell_{s't'})(\partial_{\beta}\ell_{jl})(\partial_{\beta}\ell_{j'l'})1_{n}1_{n'}\Big]\\
&=  \frac{1}{\lvert\Lambda_{N}\rvert^{2}}\sum_{i,i',j,j'}\sum_{\substack{\lambda:i\in\lambda_{\alpha}\\
\lambda':i'\in\lambda_{\alpha}}} \sum_{t\ne i}\sum_{t'\ne i'}\sum_{l\ne j}\sum_{l'\ne j'}M_{ii}M_{i'i'}\bar{E}\Big[(\partial_{\alpha_{i}}m_{\lambda})(\partial_{\alpha_{i'}}m_{\lambda'})(\partial_{\pi}\ell_{it})(\partial_{\pi}\ell_{i't'})(\partial_{\beta}\ell_{jl})(\partial_{\beta}\ell_{j'l'})1_{n}1_{n'}\Big]\\
 & +2\frac{1}{\lvert\Lambda_{N}\rvert^{2}}\sum_{i,i',j,j'}\sum_{\substack{\lambda:i\in\lambda_{\alpha}\\
\lambda':i'\in\lambda_{\alpha}}} \sum_{s\ne i}\sum_{\substack{t\ne s,t'\ne i'\\
l\ne j,l'\ne j'}}M_{is}M_{i'i'}\bar{E}\Big[(\partial_{\alpha_{i}}m_{\lambda})(\partial_{\alpha_{i'}}m_{\lambda'})(\partial_{\pi}\ell_{st})(\partial_{\pi}\ell_{i't'})(\partial_{\beta}\ell_{jl})(\partial_{\beta}\ell_{j'l'})1_{n}1_{n'}\Big]\\
 & +\frac{1}{\lvert\Lambda_{N}\rvert^{2}}\sum_{i,i',j,j'}\sum_{\substack{\lambda:i\in\lambda_{\alpha}\\
\lambda':i'\in\lambda_{\alpha}}} \sum_{s\ne i,s'\ne i'}\sum_{\substack{t\ne s,t'\ne s'\\
l\ne j,l'\ne j'}} M_{is}M_{i's'}\bar{E}\Big[(\partial_{\alpha_{i}}m_{\lambda})(\partial_{\alpha_{i'}}m_{\lambda'})(\partial_{\pi}\ell_{st})(\partial_{\pi}\ell_{s't'})(\partial_{\beta}\ell_{jl})(\partial_{\beta}\ell_{j'l'})1_{n}1_{n'}\Big]
\end{align*}
Since $\bar{E}[\partial_{\pi}\ell_{ij}]=\bar{E}[\partial_{\beta}\ell_{ij}]=0$,
while $\max_{i\ne s}\bar{M}_{is}=O_{p}(N^{-1})$, and $\max_{i}\bar{M}_{ii}=O_{p}(1)$,
we can conclude that the above sum is $O_{p}(N^{2})$, and hence $\vert\mathcal{J}\vert=O_{p}(N^{-1})$.
The same steps can be shown for the remaining elements of this term,
from which we then conclude that the jackknifed version of $(\partial_{\phi'}\Delta)\mathcal{H}^{-1}\mathcal{H}_{b}\mathcal{H}^{-1}\mathcal{S}\mathcal{S}_{\beta}$
is $O_{p}(N^{-2})$, and hence does not show up in the asymptotic
distribution of $\widehat{\Delta}_{J}$. Inspection of the expansion
for $\widehat{\Delta}$ shows that the remaining terms of the
expansion share this similar V-statistic like structure, and so similar
proofs could be applied to each of the terms.

This then gives the result
\[
(N-1)(\widehat{\Delta}-\Delta)-(N-2)\frac{1}{N-1}\sum_{k}(\widehat{\Delta}_{(k)}-\Delta_{(k)})=\frac{1}{N-1}\sum_{i}\sum_{j\ne i}h_{ij}+o_{p}(1).
\]

\end{document}